\documentclass{article}

\usepackage{microtype}
\usepackage{graphicx}
\usepackage{subcaption}
\usepackage{booktabs} 
\usepackage{hyperref}
\usepackage{hyperref}

\usepackage[accepted]{icml2026}

\usepackage{amsmath}
\usepackage{amssymb}
\usepackage{mathtools}
\usepackage{amsthm}

\usepackage[capitalize,noabbrev]{cleveref}

\theoremstyle{plain}
\newtheorem{theorem}{Theorem}[section]

\theoremstyle{definition}
\newtheorem{definition}[theorem]{Definition}

\theoremstyle{remark}
\newtheorem{remark}[theorem]{Remark}

\usepackage[textsize=tiny]{todonotes}
\usepackage[table]{xcolor}

\usepackage{microtype}
\usepackage{tikz}
\usepackage{algorithm}
\usepackage{amsfonts}
\usepackage{amsmath}
\usepackage{bbold}
\usepackage{mathtools}
\usepackage{amsthm}
\usepackage{bm}
\usepackage{booktabs}
\usepackage{graphicx}
\usepackage{caption}
\usepackage{comment}
\usepackage{hyperref}
\usepackage{balance}
\usepackage{cellspace}

\usepackage{multirow} 
\usepackage{wrapfig} 
\usepackage{makecell} 
\usepackage{subcaption}

\usepackage[normalem]{ulem}
\useunder{\uline}{\ul}{}

\usepackage{pifont}
\usepackage[most]{tcolorbox}

\newcommand{\cmark}{\ding{51}}%
\newcommand{\xmark}{\ding{55}}%

\newcommand{\methodname}{\textsc{TimeGuard}}

\newcommand{\LIN}{L_\text{in}}
\newcommand{\LOUT}{L_\text{out}}
\newcommand{\LTGR}{L_\text{tgr}}
\newcommand{\LPTN}{L_\text{ptn}}

\newcommand{\MAEC}{$\text{MAE}_\text{C}$}
\newcommand{\MAEP}{$\text{MAE}_\text{P}$}
\newcommand{\FDER}{$\text{FDER}$}

\renewcommand{\sectionautorefname}{Section}

\usepackage[switch]{lineno}
\usepackage{enumitem}
\setlist[itemize]{leftmargin=*}

\usepackage[most]{tcolorbox}
\usepackage{xcolor}

\definecolor{findingbar}{RGB}{230,126,34}
\definecolor{findingbg}{RGB}{255,250,244}
\definecolor{findingframe}{RGB}{244,214,185}

\newtcolorbox{findingbox}{
  enhanced,
  boxrule=0pt,            
  frame hidden,
  colback=findingbg,      
  borderline west={2.5pt}{0pt}{findingbar}, 
  sharp corners,
  left=4pt,right=4pt,top=2pt,bottom=2pt,
  before skip=6pt, after skip=6pt,
}

\icmltitlerunning{\textsc{TimeGuard}: Channel-wise Pool Training for Backdoor Defense in Time Series Forecasting}

\begin{document}

\twocolumn[
  \icmltitle{\textsc{TimeGuard}: Channel-wise Pool Training for Backdoor Defense \\ in Time Series Forecasting}



  \icmlsetsymbol{equal}{*}

  \begin{icmlauthorlist}
    \icmlauthor{Quang Duc Nguyen}{ntu}
    \icmlauthor{Siyuan Liang}{ntu}
    \icmlauthor{Yiming Li}{ntu}
    \icmlauthor{Fushuo Huo}{ntu}
    \icmlauthor{Dacheng Tao}{ntu}
  \end{icmlauthorlist}

   \icmlaffiliation{ntu}{College of Computing and Data Science, Nanyang Technological University, Singapore}

  \icmlcorrespondingauthor{Siyuan Liang}{siyuan.liang@ntu.edu.sg}
  \icmlcorrespondingauthor{Dacheng Tao}{dacheng.tao@ntu.edu.sg}

  \icmlkeywords{time series forecasting, trustworthy machine learning, backdoor defense}

  \vskip 0.3in
]



\printAffiliationsAndNotice{}  

\begin{abstract}
    Time Series Forecasting (TSF) is highly vulnerable to backdoor attacks, yet effective defenses remain underexplored due to challenges arising from data entanglement and shifts in task formulation.
    To fill this gap, we conduct a systematic evaluation of thirteen representative backdoor defenses across the TSF life cycle and analyze their failure modes. Our results reveal two fundamental issues: (1) data entanglement induces \emph{channel-level signal dilution}, rendering sample-filtering and trigger-synthesis defenses ineffective at localizing backdoors; and (2) task-formulation shift leads to \emph{training-loss degeneration}, causing poisoned and clean windows to become indistinguishable at training stages.
    Based on these findings, we propose a training-time backdoor defense for TSF, termed \methodname{}. Our method adopts channel-wise pool training as the core paradigm and initializes a high-confidence pool using time-aware criteria to mitigate signal dilution. Moreover, we introduce distance-regularized loss selection to progressively expand the reliable pool during training and ease loss degeneration. 
    Extensive experiments across multiple datasets, forecasting architectures, and TSF backdoor attacks demonstrate that \methodname{} substantially improves robustness, boosting \MAEP{} by $1.96\times$ over the leading baseline, while preserving clean performance within 5\% \MAEC{}. 
\end{abstract}
\section{Introduction} \label{sec:intro}
    \begin{figure}[t]
        \centering
        \includegraphics[width=\linewidth]{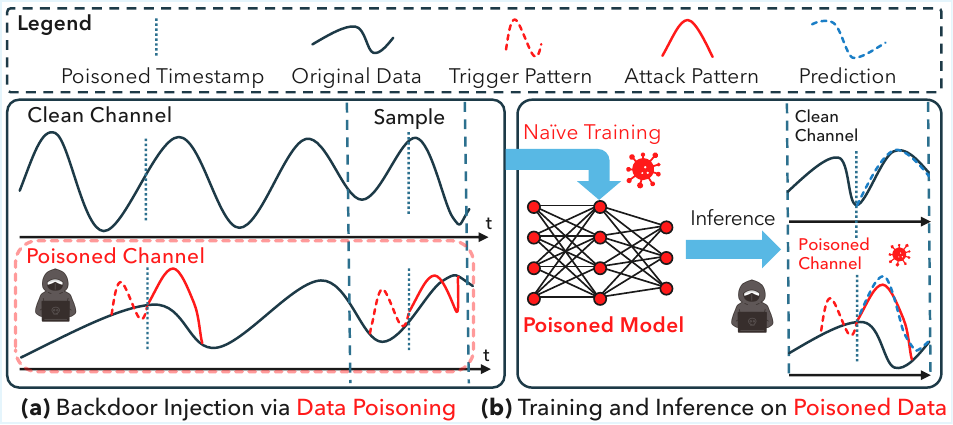}
        \vspace{-15pt}
        \caption{A backdoor is injected into selected channels during training and activated at inference to manipulate TSF predictions.}
        \label{fig:problem-intro}
        \vspace{-15pt} 
    \end{figure}

Time Series Forecasting (TSF) is widely used in critical domains such as climate prediction, transportation planning, and economic analysis. However, recent studies have shown that TSF models are also susceptible to backdoor attacks~\cite{liang2024badclip, liu2025pre, liang2025vl}, where an attacker implants hidden trigger patterns into the data during the training phase such that the model behaves normally under benign inputs but outputs attacker-specified predictions under trigger conditions~\cite{lin2024backtime}. This type of attack is highly covert and may pose serious risks to practical applications relying on TSF~\cite{liu2025natural,zhang2024towards,liu2024compromising}, such as undermining the reliability of decision-making and forecasting~\cite{zhang2015daily}, which highlights the necessity of studying TSF backdoor defense methods~\cite{wang2022universal,liang2024unlearning,guo2024copyrightshield}.

Although backdoor defense mechanisms have been extensively studied in classification and generative model domains~\cite{wu2025backdoorbench, li2025backdoorllm, lin2025backdoordm}, backdoor defenses for time series forecasting (TSF) remain significantly underdeveloped. Defense against TSF backdoors is still evidently insufficient. This is mainly due to two inherent challenges in TSF scenarios: one is data entanglement, i.e., multivariate time series exhibiting simultaneous channel structure and temporal dependency~\cite{xu2026cpiri}, which causes backdoor injections to be highly coupled with clean signals at the data level; and the other is task-formulation shift, i.e., TSF shifts from discrete classification to continuous-value regression and training window overlap~\cite{kim2025comprehensive}, resulting in substantial changes in the discriminative signals relied upon by existing defense methods during training~\cite{kuang2024adversarial, xu2026srd}. Therefore, it is often difficult to directly transfer existing backdoor defense techniques to TSF scenarios.

To fill the above research gap, we conduct a systematic evaluation of backdoor defenses in TSF scenarios by adapting and analyzing 13 representative defense methods across the four phases of the deep neural network lifecycle~\cite{wu2025backdoorbench}. Experimental results indicate that the failures of existing methods in TSF mainly manifest in two aspects induced by these inherent challenges: first, data entanglement in multivariate time series leads to channel-level signal dilution, where backdoor injections only affect a subset of channels, making it difficult for sample-level filtering and trigger-synthesis-based defenses to accurately localize backdoors when the attack granularity and defense granularity are mismatched; second, task-formulation shift further causes training loss degeneration, and the continuous-value regression targets together with overlapping window structures result in poisoned and clean windows exhibiting similar loss distributions during early training stages, thereby weakening or even invalidating defense strategies that rely solely on training losses. In addition, we observe that training-phase defenses~\cite{xun2025robust, liang2024unlearning} remain effective when relatively reliable clean training data are available.

Based on the above analysis, we propose a training-time backdoor defense method for TSF, termed \methodname{}.
Motivated by training-phase defenses, \methodname{} adopts Channel-wise reliable pool training as the core paradigm, reconfiguring conventional sample-level training into finer-grained channel-level training, thereby exploiting the majority of channel information in multivariate time series that remains reliable. Building upon this paradigm, we further design a time-aware pool initialization strategy, which selects high-confidence time-channel units from two complementary perspectives of learning behavior and temporal structure, providing an initial reliable pool with higher signal purity. Furthermore, to address the training loss degeneration problem induced by task-formulation shift, \methodname{} introduces a Distance-Regularized Loss Selection mechanism, which progressively expands the reliable pool during training while reducing the risk of highly correlated poisoned windows being reintroduced into the training process, without sacrificing forecasting performance. Through these designs, \methodname{} effectively mitigates the signal dilution and training loss degeneration without requiring additional clean data.
Experiments on three datasets, three forecasting architectures, and three representative TSF backdoor attacks demonstrate that \methodname{} substantially improves robustness, achieving a $1.96\times$ improvement in \MAEP{} over the leading baseline, while preserving clean performance within $5\%$ \MAEC{}. Due to space constraints, we defer a detailed discussion of related work to \appendixautorefname~\ref{app:related-work}. Our \textbf{main contributions}~are:
\begin{itemize}
    \item We present the first systematic evaluation of backdoor defenses for time series forecasting, and reveal two TSF-specific failure modes arising from data entanglement and task-formulation shift.
    \item We propose \methodname{}, a training-time backdoor defense that learns from channel-wise reliable data and effectively mitigates signal dilution and training-loss degeneration without requiring additional clean samples.
    \item We extensively evaluate \methodname{} on three TSF forecasters, showing consistent mitigation across three TSF attacks with different settings; ablation studies further validate the contribution of each component. Notably, \methodname{} also transfers to the LLM-based method, yielding at least a $5.14\times$ \MAEP{} gain with only a $3.8\%$ change in clean \MAEC{}.
\end{itemize}

\begin{table*}[!t]
\centering
\begin{minipage}{0.49\textwidth}
    \centering
    \caption{Performance comparison of training-phase defenses on PEMS03 dataset. Best and second results are \textbf{bold} and {\ul underline}. We report performance averaged across the three forecasting models. Full per-model results are provided in \appendixautorefname~\ref{app:defense-performance}.}
    \label{tab:training-based-preliminary-table}

    \setlength{\tabcolsep}{1pt}
    \renewcommand{\arraystretch}{1.49}
    \renewcommand{\aboverulesep}{0pt}
    \renewcommand{\belowrulesep}{0pt}
    \setlength\cellspacetoplimit{2pt}
    \setlength\cellspacebottomlimit{2pt}
    \resizebox{\textwidth}{!}{
    \begin{tabular}{l ccc ccc}
        \toprule
        \textbf{Attack}~→ 
            & \multicolumn{3}{c}{\textbf{Random}}
            & \multicolumn{3}{c}{\textbf{BackTime}} \\
        \cmidrule(lr){1-1}
        \cmidrule(lr){2-4}
        \cmidrule(l){5-7}
        \textbf{Defense}~↓ 
            & \MAEC~↓ & \MAEP~↑ & \FDER~↑
            & \MAEC~↓ & \MAEP~↑ & \FDER~↑ \\ 
        \midrule

        No Defense 
            & 17.634 & 17.772 & -
            & 17.607 & 14.201 & - \\ 
        \midrule

        Spectral~\cite{tran2018spectral}
            & {\ul 18.389} & 18.356 & 0.502
            & 18.666 & 15.245 & 0.539 \\

        TED~\cite{mo2024robust} 
            & 18.434 & 20.063 & 0.528
            & 18.606 & 13.953 & 0.495 \\ 

        TED++~\cite{le2025ted++}
            & 19.197 & 19.184 & 0.499
            & 18.565 & 14.541 & 0.513 \\ 
        \midrule

        Fine-tuning~\cite{gu2019badnets} 
            & 19.003 & 30.909 & 0.625
            & 18.934 & 18.196 & 0.594 \\

        Fine-pruning~\cite{liu2018fine} 
            & 19.020 & 31.643 & 0.633
            & 18.686 & 19.736 & 0.623 \\

        NAD~\cite{li2021neural} 
            & 18.795 & 26.809 & 0.600
            & {\ul 18.584} & 18.158 & 0.600 \\

        IMS~\cite{dunnett2025backdoor}
            & 19.239 & 17.731 & 0.466
            & 18.418 & 14.351 & 0.509 \\ 
        \midrule

        ABL~\cite{li2021anti} 
            & 19.637 & 19.104 & 0.493
            & 18.761 & 14.481 & 0.509 \\

        PDB~\cite{wei2024mitigating} 
            & 18.630 & {\ul 54.690} & {\ul 0.693}
            & 18.967 & {\ul 22.397} & {\ul 0.639} \\ 

        ESTI~\cite{yu2025backdoor}
            & 19.910 & 17.186 & 0.454
            & 19.219 & 15.897 & 0.532 \\ 
        \midrule
        \rowcolor[HTML]{EFEFEF}
        \textbf{\methodname}
            & \textbf{17.928}
            & \textbf{104.677}
            & \textbf{0.868}
            & \textbf{18.048}
            & \textbf{39.303}
            & \textbf{0.808} \\
        \bottomrule
    \end{tabular}
    }
\end{minipage}
\hfill
\begin{minipage}{0.49\textwidth}
    \centering
    \caption{Detection performance comparison of inference-time defenses on three datasets, averaged over three models. Inference time is measured on 200 samples. Best and second results are \textbf{bold} and {\ul underline}. Full per-model results are provided in \appendixautorefname~\ref{app:defense-performance}.}
    \label{tab:inference-based-preliminary-table}
    
    \setlength{\tabcolsep}{2pt}
    \renewcommand{\arraystretch}{1.5}
    \renewcommand{\aboverulesep}{0pt}
    \renewcommand{\belowrulesep}{0pt}
    \setlength\cellspacetoplimit{2pt}
    \setlength\cellspacebottomlimit{2pt}
    \resizebox{\textwidth}{!}{
    \begin{tabular}{@{}llcccccc@{}}
        \toprule
        \multirow{2}{*}{\textbf{Dataset}} 
            & \multirow{2}{*}{\textbf{Defense}} 
            & \multirow{2}{*}{\textbf{\begin{tabular}[c]{@{}c@{}}Total Inference\\ Time (s) ↓\end{tabular}}} 
            & \multicolumn{2}{c}{\textbf{Random}} 
            & \multicolumn{2}{c}{\textbf{BackTime}} \\
        \cmidrule(lr){4-5} \cmidrule(lr){6-7}
         &  &  & AUC ↑ & F1 ↑ & AUC ↑ & F1 ↑ \\ 
        \midrule
    
        \multirow{4}{*}[-0.5ex]{\textbf{PEMS03}} 
            & No Defense & 2.497 & 0.500 & 0.500 & 0.500 & 0.500 \\ \cmidrule(l){2-7}
            & STRIP~\cite{gao2019strip}      & 278.283 & {\ul{0.518}} & {\ul{0.532}} & \textbf{0.501} & {\ul{0.516}} \\
            & TeCo~\cite{liu2023detecting}       & {\ul{38.407}}  & \textbf{0.563} & \textbf{0.564} & 0.478 & 0.512 \\
            & IBD-PSC~\cite{hou2024ibd}    & \textbf{9.903}  & 0.364 & 0.514 & {\ul{0.486}} & \textbf{0.535} \\
        \midrule
        \multirow{4}{*}[-0.5ex]{\textbf{Weather}} 
            & No Defense & 2.330 & 0.500 & 0.500 & 0.500 & 0.500 \\ \cmidrule(l){2-7}
            & STRIP~\cite{gao2019strip}      & 198.480 & 0.300 & 0.510 & {\ul{0.497}} & 0.531 \\
            & TeCo~\cite{liu2023detecting}       & {\ul{25.447}}  & \textbf{0.581} & \textbf{0.590} & \textbf{0.547} & \textbf{0.574} \\
            & IBD-PSC~\cite{hou2024ibd}     & \textbf{9.838}   & {\ul{0.317}} & {\ul{0.519}} & 0.390 & {\ul{0.534}} \\
        \midrule

        \multirow{4}{*}[-0.5ex]{\textbf{ETTm1}} 
            & No Defense & 2.297 & 0.500 & 0.500 & 0.500 & 0.500 \\ \cmidrule(l){2-7}
            & STRIP~\cite{gao2019strip}      & 205.453 & {\ul{0.490}} & {\ul{0.525}} & 0.477 & 0.506 \\
            & TeCo~\cite{liu2023detecting}       & {\ul{25.443}}  & \textbf{0.614} & \textbf{0.591} & \textbf{0.524} & \textbf{0.521} \\
            & IBD-PSC~\cite{hou2024ibd}     & \textbf{9.749}   & 0.378 & 0.513 & {\ul{0.486}} & {\ul{0.518}} \\
        \bottomrule
    \end{tabular}
    }
\end{minipage}
\vspace{-10pt}
\end{table*}

\section{Threat Model} \label{sec:pre}

\textbf{Victim model.}\label{pre:tsf} Time series forecasting (TSF)~\cite{kim2025comprehensive} aims to predict future values over one or multiple horizons given historical observations of a univariate or multivariate time series. We consider a multivariate time series dataset denoted as $\mathbf{X} \in \mathbb{R}^{T \times C}$, where $T$ is the number of time steps and $C$ is the number of variables (or channels). For each forecasting sample indexed  by timestamp $t$, we denote the history (input) window and future (target) window as $\mathbf{X}_{t, h} = \mathbf{X}[t-\LIN:t, :]$ and  $\mathbf{X}_{t, f} = \mathbf{X}[t:t+\LOUT, :]$, where $\LIN$ and $\LOUT$ denote the history and future lengths, and we use half-open indexing (end exclusive). Thus $\mathbf{X}_{t, h} \in \mathbb{R}^{\LIN \times C}$ and  $\mathbf{X}_{t, f} \in \mathbb{R}^{\LOUT \times C}$. Sliding this windowing process over time yields overlapping-window training set $\mathcal{D} = \{ (\mathbf{X}_{t, h}, \mathbf{X}_{t, f}) \mid \LIN \le t \le T - \LOUT \}$, following standard TSF practice~\cite{nie2023time, lin2024backtime}. A forecasting model $f_{\theta}$ maps histories to futures, i.e., $f_{\theta}:\mathbb{R}^{\LIN\times C}\rightarrow \mathbb{R}^{\LOUT\times C}$, and is trained by minimizing a prediction loss (e.g., mean absolute error (MAE) or mean squared error (MSE)) over $\mathcal{D}$. 

\textbf{Attacker's capabilities.}  We follow the TSF poisoning backdoor setup in BackTime~\cite{lin2024backtime}, as depicted in \figureautorefname~\ref{fig:problem-intro}. Given a multivariate training series $\mathbf{X}$, the attacker selects (i) a set of poisoned timestamps $\mathcal{T}_{\text{atk}}$ with temporal injection rate $\eta_\text{T}$ and (ii) a set of target variables $\mathcal{S}$ with spatial injection rate $\eta_\text{S}$. For each $t\in\mathcal{T}_{\text{atk}}$, the attacker generates a trigger pattern $\mathbf{G}_t\in\mathbb{R}^{\LTGR\times|\mathcal{S}|}$ and overwrites the $\LTGR$ steps immediately preceding $t$ on the target variables: $\mathbf{X}[t-\LTGR:t,\mathcal{S}] \leftarrow \mathbf{G}_t$. The attacker also overwrites the subsequent $\LPTN$ future steps: $\mathbf{X}[t:t+\LPTN,\mathcal{S}] \leftarrow \mathbf{X}[t-\LTGR-1,\mathcal{S}] \oplus \mathbf{P}$, where $\mathbf{P}\in\mathbb{R}^{\LPTN\times|\mathcal{S}|}$ is a predefined attack pattern template and $\oplus$ denotes element-wise addition with broadcasting along the time dimension. Thus, the trigger and target patterns are injected consecutively around $t$. At inference time $t_0$, the adversary injects the trigger over the $\LTGR$ most recent steps in the input stream, i.e., during $[t_0-\LTGR,\,t_0)$. Trigger generation is constrained to use information available up to the current time (at most $t_0$) to respect forecasting timeliness.

\textbf{Attacker's goals.} The attacker aims to poison the training data such that the victim prediction model learns hidden backdoor behaviors~\cite{lin2024backtime, xiang2025badtime}:
(i) Maintain normal prediction accuracy on clean historical windows;
(ii) When the input historical window contains a trigger pattern $\mathbf{G}$ on a poisoned channels $\mathcal{S}$, force the model’s predictions on the corresponding channels to follow the attacker-specified target pattern induced by the pre-defined attack template $\mathbf{P}$, while keeping the prediction behavior of the remaining channels unchanged.

\textbf{Defender's capabilities and goals.}
The defender aims to safeguard forecasting models against backdoor poisoning attacks without prior knowledge of the trigger pattern, attack pattern, or the poisoned timestamps and variables. Depending on the defense strategy, the defender may access different components of the model life cycle, including the training data, the training procedure, the trained model, or only inference-time predictions~\cite{wu2025backdoorbench}. Some defenses further assume access to a small subset of trusted clean samples~\cite{liu2018fine, wei2024mitigating}.

Accordingly, existing defenses can be broadly categorized into:
(i) training-phase defenses, which intervene before, during, or after model training (i.e., pre-training, in-training, or post-training) to obtain models that are robust to backdoor activation while preserving benign forecasting utility and disrupting malicious target alignment~\cite{tran2018spectral, li2021anti, wei2024mitigating}; and
(ii) inference-time defenses, which detect or suppress triggered inputs at test time without modifying the trained model~\cite{liu2023detecting, gao2019strip, wang2025lie}.

\begin{figure*}[!t]
    \vspace{-3pt}
    \centering
    \begin{minipage}[t]{0.48\textwidth}
        \centering
        \begin{subfigure}[t]{0.49\linewidth}
            \centering
            \includegraphics[width=\linewidth]{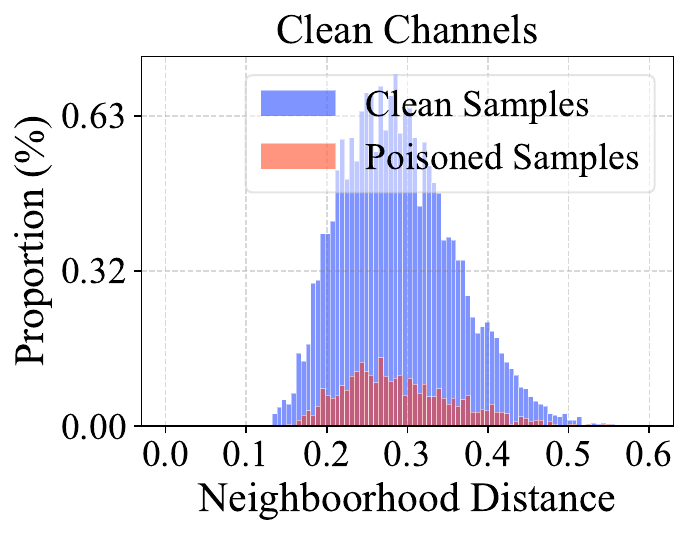}
            \label{fig:neigh_dist_cleans}
        \end{subfigure}\hfill
        \begin{subfigure}[t]{0.49\linewidth}
            \centering
            \includegraphics[width=\linewidth]{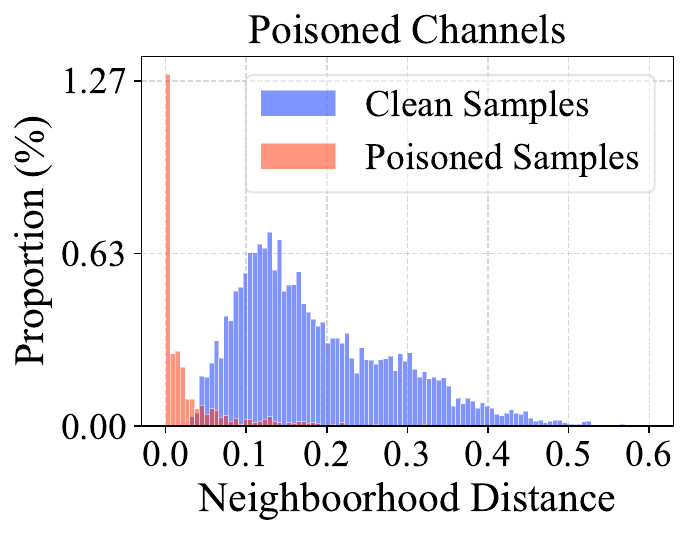}
            \label{fig:neigh_dist_poisons}
        \end{subfigure}
        \vspace{-14pt}
        \caption{Neighborhood distance distributions of poisoned and clean samples, averaged over clean and poisoned channels, on Weather~\cite{wu2021autoformer} under BackTime~\cite{lin2024backtime}. The neighborhood distance is defined in \sectionautorefname~\ref{sec:method}.} 
        \vspace{-8pt}
        \label{fig:neigh_dist}
    \end{minipage}\hspace{15pt}
    \begin{minipage}[t]{0.48\textwidth}
        \centering
        \begin{subfigure}[t]{0.49\linewidth}
            \centering
            \includegraphics[width=\linewidth]{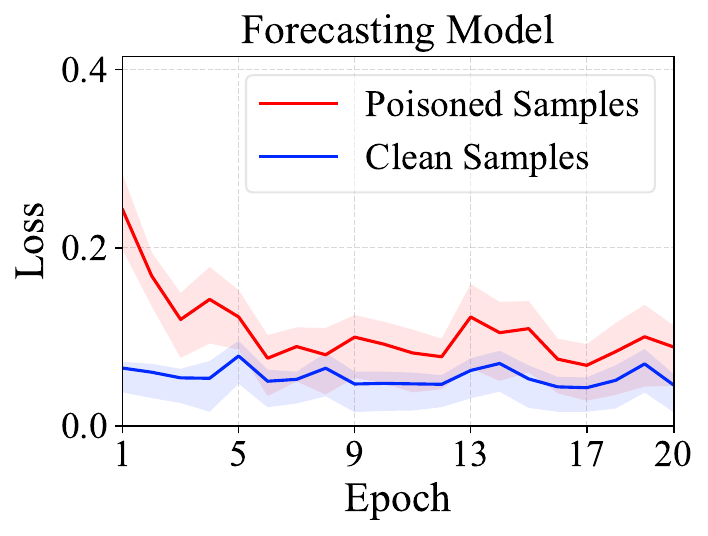}
            \label{fig:forecast_over}
        \end{subfigure} 
        \begin{subfigure}[t]{0.49\linewidth}
            \centering
            \includegraphics[width=\linewidth]{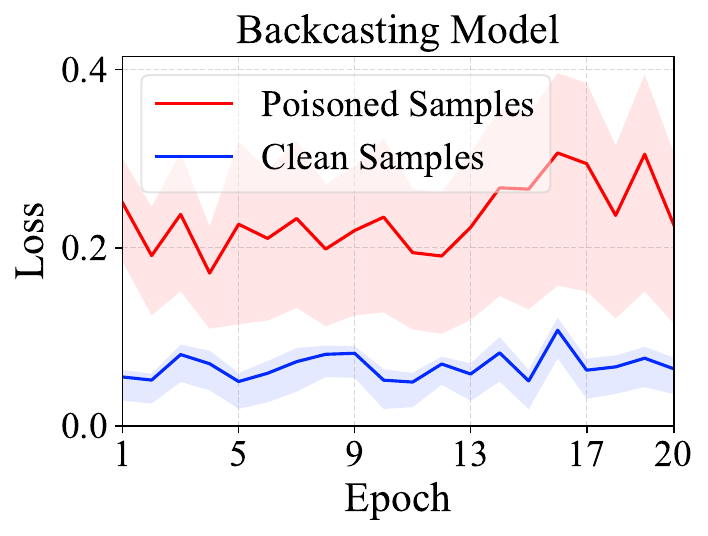}
            \label{fig:backcast_over}
        \end{subfigure}
        \vspace{-14pt}
        \caption{Training loss of clean and poisoned samples, averaged over poisoned channels, for forecasting and backcasting FEDformer models~\cite{zhou2022fedformer} on the Weather dataset~\cite{wu2021autoformer} under BackTime attack~\cite{lin2024backtime}. }
        \vspace{-8pt}
        \label{fig:loss_over_comparison}
    \end{minipage}
\vspace{-7pt}
\end{figure*}

\section{Revisiting Existing Backdoor Defenses for Forecasting} \label{sec:revisiting}

This section systematically adapts existing backdoor defenses originally developed for classification to the TSF setting and evaluates their effectiveness. We also introduce \FDER{} as a forecasting-specific metric and analyze the key characteristics and failure modes of existing defenses.

\subsection{Experimental Settings} \label{sec:experimental-settings}

\noindent\textbf{Datasets and models.}
We conduct experiments on three representative datasets, PEMS03~\cite{song2020spatial}, Weather~\cite{wu2021autoformer}, and ETTm1~\cite{zhou2022fedformer}, covering different application domains. Following existing TSF backdoor work~\cite{lin2024backtime}, we evaluate three forecasting models: SimpleTM~\cite{chen2025simpletm}, FEDformer~\cite{zhou2022fedformer}, and TimesNet~\cite{wu2023timesnet}. We use a 6:2:2 train/validation/test split and report results averaged over the three architectures. More dataset and model details are provided in \appendixautorefname~\ref{app:dataset} and~\ref{app:model}.

\noindent\textbf{Attack methods.} We evaluate against three representative TSF backdoor attacks: Random~\cite{gu2019badnets}, FreqBack-TSF~\cite{huang2025revisiting}, and BackTime~\cite{lin2024backtime}. Random attack injects a fixed random trigger, inspired by BadNets~\cite{gu2019badnets}. FreqBack-TSF adapts FreqBack~\cite{huang2025revisiting}, originally proposed for time series classification, and uses a universal optimized trigger crafted via frequency analysis. BackTime~\cite{lin2024backtime} is a state-of-the-art TSF attack that generates sample-dependent triggers via a GNN-based generator. Unless stated otherwise, we use $\LIN{=}\LOUT{=}12$ with poisoning rates $\eta_\text{T}{=}0.03$ and $\eta_\text{S}{=}0.3$ following BackTime settings~\cite{lin2024backtime}.
Attack details are provided in \appendixautorefname~\ref{app:attack}. 

\noindent\textbf{Evaluation metrics.} Following prior TSF backdoor settings~\cite{lin2024backtime, xiang2025badtime}, we report Mean Absolute Error (MAE) on clean inputs (\MAEC) and on triggered inputs (\MAEP) for training-phase defenses. An effective defense should preserve a low \MAEC{} while achieving high \MAEP{}~\cite{gao2023backdoor, yu2025backdoor}. 

However, in our preliminary evaluation, we observe ``false wins,'' where \MAEP{} increases primarily because the model's overall forecasting quality degrades, which is also reflected by a higher \MAEC{}; the reverse can also occur, as in the IMS defense in \tableautorefname~\ref{tab:training-based-preliminary-table}. To capture robustness gains while penalizing clean-performance degradation, we propose the \textit{Forecasting Defense Effectiveness Rating (\FDER)}, adapted from DER~\cite{zhu2023enhancing} but defined using relative MAE-based measures suitable for forecasting:
\begin{equation} 
\label{eq:main-FDER}
    \text{FDER} = \frac{\max(0, \rho_{\text{MAE}_\text{P}}) - \max(0, \rho_{\text{MAE}_\text{C}}) + 1}{2} \in [0, 1],
\end{equation}
where the relative attack and clean gain are defined as:
\begin{equation}
    \rho_{\text{MAE}_\text{P}} =  1 - \frac{\text{MAE}_\text{P}^{\text{und}}}{\text{MAE}_\text{P}}, 
    \qquad
    \rho_{\text{MAE}_\text{C}} =  1 - \frac{\text{MAE}_\text{C}^{\text{und}}}{\text{MAE}_\text{C}}.
\end{equation}

Here $\text{MAE}_\text{P}^{\text{und}}$ and $\text{MAE}_\text{C}^{\text{und}}$ denote the attack/clean MAE errors of the undefended backdoored model. Higher \FDER{} indicates stronger backdoor mitigation with smaller clean-performance overhead. For inference-time defenses, we evaluate detection capability using AUROC and F1 score, where higher values indicate better performance~\cite{liu2023detecting, wang2025lie}. 

Thus, in TSF backdoor settings~\cite{lin2024backtime}, benign behavior corresponds to accurate forecasting on clean inputs, reflected by low \MAEC{}; malicious success corresponds to triggered inputs being steered toward the attacker's target, reflected by low \MAEP{}; and general failure corresponds to poor forecasting quality overall, which is also reflected by high clean-input error. Therefore, an effective TSF defense should preserve benign forecasting utility, as indicated by comparable or lower \MAEC{}, while disrupting malicious target alignment, as indicated by higher \MAEP{} or, more compactly, higher \FDER{}, despite attacker-defined trigger and target patterns. Further discussion is in \appendixautorefname~\ref{app:eval-metrics}.

\subsection{Backdoor Defenses under TSF Setting}

Since backdoor defenses for TSF remain underexplored, we adapt \textbf{13} representative defenses originally developed for classification, covering the four stages of the model life cycle and diverse defense paradigms~\cite{wu2025backdoorbench, li2022backdoor, ren2025iclshield}. To ensure a fair comparison, we follow each method's default implementation whenever applicable and apply minimal modifications needed for TSF. Concretely, we replace accuracy-based criteria with MAE-based counterparts, and for inference-time and input-modification defenses we use time-series-specific modifications; otherwise, we keep the original procedures unchanged.

Specifically, we evaluate ten training-phase defenses, including pre-training methods (Spectral~\cite{tran2018spectral}, TED~\cite{mo2024robust}, TED++~\cite{le2025ted++}), post-training methods (Fine-tuning~\cite{gu2019badnets}, Fine-pruning~\cite{liu2018fine}, NAD~\cite{li2021neural}, IMS~\cite{dunnett2025backdoor}), and in-training methods (ABL~\cite{li2021anti}, PDB~\cite{wei2024mitigating}, ESTI~\cite{yu2025backdoor}), 
as well as three inference-time defenses (STRIP~\cite{gao2019strip}, TeCo~\cite{liu2023detecting}, and IBD-PSC~\cite{hou2024ibd}).
More implementation details, our baseline selection rationale, and a comparison of key defense attributes are deferred to \appendixautorefname~\ref{app:defense} and~\ref{app:more-existing-defense}, respectively. 

\vspace{-7pt}
\subsection{Preliminary Evaluation and Key Insights} \label{sec:key-insights}
We summarize training-phase defense results on PEMS03 and inference-time detection results, both under the Random and BackTime attacks in \tablename~\ref{tab:training-based-preliminary-table} and \tablename~\ref{tab:inference-based-preliminary-table}, respectively. We highlight four empirical insights, which we analyze next.

\vspace*{-3.5pt}
\begin{findingbox}
\textbf{Insight 3.1:} Sample-level filtering and trigger-synthesis style defenses yield limited robustness gains against TSF backdoor attacks.   
\end{findingbox}

\vspace*{-1pt}
Sample-level filtering defenses (Spectral, TED, TED++) yield only marginal robustness gains (\FDER{} $\approx$ 0.54), and trigger-synthesis-based defenses (IMS) achieve similarly near-neutral \FDER{} (best $\approx$ 0.51), despite comparable \MAEC{}. This suggests that a common bottleneck may arise under channel-subset TSF poisoning, where attackers typically poison only a subset of channels: sample-level criteria are dominated by non-poisoned variables; while trigger synthesis optimized over all channels receives diluted gradients, leading to “smeared” reconstructions. Consistently, \figureautorefname~\ref{fig:neigh_dist} shows that neighborhood distance (\sectionautorefname~\ref{sec:method}) statistics differ sharply between clean and poisoned channels, indicating that this measure is inherently~channel-dependent. 

\vspace*{-0.5pt}
\begin{findingbox}
\textbf{Insight 3.2:} Defenses relying primarily on training-loss criteria are unreliable and fail to safeguard TSF models against backdoor attacks. 
\end{findingbox}

Training-loss-only defenses (ABL, ESTI) fail to safeguard TSF models, with an average \FDER{} of 0.497.  \figureautorefname~\ref{fig:loss_over_comparison} (FEDformer) shows poisoned-sample losses quickly converging to clean-sample losses within the first few epochs, weakening the early-loss separation signal these methods depend~on. This behavior may stem from TSF’s continuous regression objective (rather than an \texttt{argmax}-based discrete target), which encourage poisoned windows to achieve low loss, while overlapping input-output windows introduce affected hard samples, further blurring loss-based partitioning.

\begin{findingbox}
\textbf{Insight 3.3:} Fine-tuning-based and in-training model-agnostic defenses provide partial mitigation against TSF backdoor attacks, yet require a clean subset.  
\end{findingbox}
Fine-tuning-based defenses (Fine-tuning, Fine-pruning, NAD) and the in-training model-agnostic defense (PDB) provide partial mitigation, achieving \FDER{} $>$ 0.6 on average across the two attacks. Compared to fine-tuning-based defense, PDB performs best (\FDER{} $=$ 0.666), suggesting that model-agnostic in-training intervention can be more effective than post-hoc repair. However, these methods all assume access to a verified clean subset, which is costly to obtain in time series domain~\cite{lin2024backtime}.

\begin{findingbox}
\textbf{Insight 3.4:} Inference-time defenses offer marginal detection with high inference overhead in TSF.  
\end{findingbox}

\begin{figure*}[!ht]
    \centering
    \includegraphics[width=\textwidth]{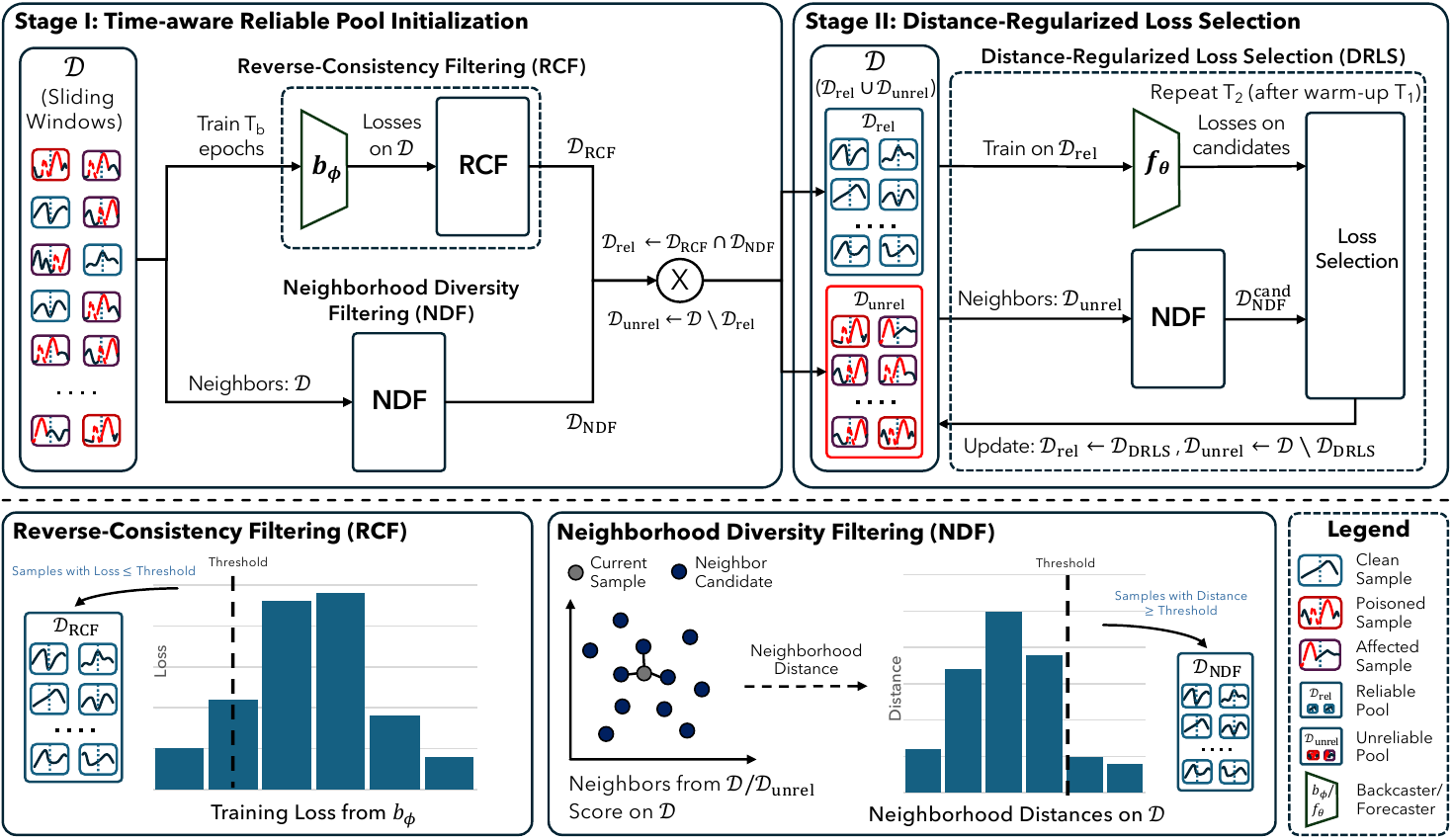}
    \caption{Overview of \methodname{}. Stage I forms the reliable pool $\mathcal{D}_{\text{rel}}$ by intersecting the subsets selected by Reverse-Consistency Filtering (RCF) and Neighborhood Diversity Filtering (NDF). Stage II trains $f_{\theta}$ while progressively updating $\mathcal{D}_{\text{rel}}$ via Distance-Regularized~Loss Selection (DRLS) to prevent re-admitting correlated poisoned windows. All pools and filtering criteria operate in a channel-wise manner.} 
    \label{fig:pipeline}
    \vspace{-15pt}
\end{figure*}

Inference-time defenses (STRIP, TeCo, IBD-PSC) provide only marginal detection after TSF adaptation, achieving just 0.551 AUROC and 0.559 F1 on the best method, TeCo,  despite our attempts for time-aware perturbations and augmentations. Moreover, they impose heavy overhead (4–100×), increasing latency from $\sim$ 2s to $>$ 200s, makes those impractical for real-time TSF systems~\cite{fan1994real}.

\noindent \textbf{Summary.} 
Our evaluation shows inconsistent effectiveness of existing TSF defenses, which we attribute to two TSF-specific failure modes:  \textbf{(1) Channel-Level Signal Dilution} (Insights~3.1,~3.4), where channel-subset poisoning and temporal coupling dilute backdoor signals and undermine channel-agnostic filtering, trigger synthesis, and inference-time detection; and  \textbf{(2) Training-Loss Degeneration} (Insight~3.2), where TSF’s regression objective and overlapping windows collapse training-loss-based  separation.  While fine-tuning and model-agnostic in-training baselines provide partial mitigation (Insight~3.3), the best baseline (PDB) improves \MAEP{} by only $1.58\times$ with a 7.72\% \MAEC{} increase on PEMS03 and still requires clean data. These trends persist across datasets and attacks, as shown in \appendixautorefname~\ref{app:defense-performance}. Moreover, TSF models are often deployed in continuous real-time settings~\cite{kim2025comprehensive, lin2024backtime}, where inference-time checks can introduce substantial overhead. Together with Insight~3.4, these observations motivate our focus on training-phase defense, which incurs no inference-time overhead; we leave the development of efficient TSF inference-time defenses for future work.

\vspace{-8pt}
\section{\methodname{}} \label{sec:method}

Motivated by the partial success of fine-tuning and in-training baselines (\sectionautorefname~\ref{sec:revisiting}), we propose \methodname{}, an in-training defense against TSF backdoor attacks that constructs and maintains a channel-wise training pool without requiring any prior clean subset. The key idea is to refactor multivariate TSF training from a sample-level decision into a time $\times$ channel-wise decision (\sectionautorefname~\ref{sec:channel-wise-for}), since TSF backdoors often corrupt only a subset of channels~\cite{lin2024backtime}. \methodname{} then constructs and maintains per-channel reliable pools throughout training process via time-aware criteria (\sectionautorefname~\ref{sec:init} and \sectionautorefname~\ref{sec:drls}). 

\vspace{-5pt}
\subsection{Channel-wise Reliable Pool Training} \label{sec:channel-wise-for}
Many existing defenses~\cite{li2021anti, huang2022backdoor, gao2023backdoor, shen2025bi} adopt a sample-level formulation that discards suspected poisoned forecasting windows and trains on the remaining data. This assumption breaks in multivariate TSF, where backdoor injection often modifies only a subset of channels~\cite{lin2024backtime}, making it wasteful to discard entire windows. 

\noindent \textbf{Channel-wise objective.} Given the training set $\mathcal{D}$, we treat each channel objective independently. Particularly, for channel $c$, define the channel-wise window set
$\mathcal{D}^{(c)} = \{(\mathbf{x}_{t, h}^{(c)}, \mathbf{x}_{t, f}^{(c)}) \},$
where $\mathbf{x}_{t, h}^{(c)} \in \mathbb{R}^{\LIN}$ and $\mathbf{x}_{t, f}^{(c)} \in \mathbb{R}^{\LOUT}$ are the history and future windows. The full channel-wise window sample is $\mathbf{x}^{(c)}_{t} := [\mathbf{x}^{(c)}_{t, h}; \mathbf{x}^{(c)}_{t, f}]$. We maintain a per-channel reliable pool $\mathcal{D}^{(c)}_{\text{rel}}\subseteq \mathcal{D}^{(c)}$ and an unreliable pool $\mathcal{D}^{(c)}_{\text{unrel}}=\mathcal{D}^{(c)}\setminus \mathcal{D}^{(c)}_{\text{rel}}$.
We further introduce a binary mask $m_{t,c}\in\{0,1\}$ indicating whether timestamp $t$ for channel $c$ is currently included in the reliable pool. The forecaster is trained by minimizing the masked empirical loss:
\vspace{-3pt}
\begin{equation} \label{eq:loss_def}
    \resizebox{0.99\linewidth}{!}{$
    \mathcal{L}_{\text{def}}(\theta; m)  = \dfrac{1}{\sum_{t,c} m_{t,c}}\sum_{t,c}m_{t,c}\,\ell(f_{\theta}^{(c)}(\mathbf{X}_{t, h}),  \mathbf{x}_{t, f}^{(c)}),
    $}
\vspace{-3pt}
\end{equation}
where $f_\theta^{(c)}(\cdot)$ is the prediction for channel $c$ and $\ell(\cdot,\cdot)$ is the forecasting loss. The key challenge is to construct and progressively update $m_{t,c}$ so that $\mathcal{D}^{(c)}_{\text{rel}}$ has high precision (few poisoned windows) while preserving sufficient diversity to maintain clean forecasting performance. For notational simplicity, we omit the channel superscript $(c)$ below.

\noindent \textbf{Pipeline overview.} \methodname{} instantiates and updates $m_{t,c}$ via a two-stage channel-wise procedure, as summarized in Figure~\ref{fig:pipeline}. In \emph{Stage I: Time-aware Reliable Pool Initialization} (\sectionautorefname~\ref{sec:init}), we construct a conservative, high-precision initial reliable pool by intersecting samples selected by two complementary time-aware criteria: Reverse-Consistency Filtering (RCF) from a learning-behavior perspective and Neighborhood Diversity Filtering (NDF) from a temporal-structure perspective. In \emph{Stage II: Distance-Regularized Loss Selection} (\sectionautorefname~\ref{sec:drls}), we progressively update the reliable pool using Distance-Regularized Loss Selection (DRLS), which regularizes loss-based admission with neighborhood diversity to avoid re-including correlated, low-loss poisoned windows. Throughout training, the forecaster $f_\theta$ is trained with the masked objective in \equationautorefname~\ref{eq:loss_def}, and the full algorithm is given in \appendixautorefname~\ref{app:algorithm}.

\begin{table*}[!ht]
    \centering
    \caption{Main results of backdoor defense against TSF backdoor attacks on PEMS03. Best and second results are \textbf{bold} and {\ul underline}. Lower \MAEC{} indicates better performance, while higher \MAEP{} and \FDER{} are preferred. We report performance averaged across the three forecasting models. Full per-model results and visualization examples are provided in \appendixautorefname~\ref{app:defense-performance} and \appendixautorefname~\ref{app:showcases}, respectively.}
    \vspace{-4pt}
    \label{tab:main-defense-result}
    \scriptsize
    \renewcommand{\aboverulesep}{0pt}
    \renewcommand{\belowrulesep}{0pt}
    \setlength\cellspacetoplimit{2pt}
    \setlength\cellspacebottomlimit{2pt}
    \resizebox{\textwidth}{!}{
        \begin{tabular}{S l *{9}{S c}}
        \toprule
        \textbf{Attack}~→ 
            & \multicolumn{3}{c}{\textbf{Random}} 
            & \multicolumn{3}{c}{\textbf{FreqBack-TSF}}
            & \multicolumn{3}{c}{\textbf{BackTime}} \\ 
        \cmidrule(lr){1-1}
        \cmidrule(lr){2-4}
        \cmidrule(lr){5-7}
        \cmidrule(l){8-10}
        \textbf{Defense}~↓ 
            & \MAEC~↓ & \MAEP~↑ & \FDER~↑ 
            & \MAEC~↓ & \MAEP~↑ & \FDER~↑ 
            & \MAEC~↓ & \MAEP~↑ & \FDER~↑ \\ 
        \midrule
        No Defense 
            & 17.634 & 17.772 & -- 
            & 17.583 & 14.683 & -- 
            & 17.607 & 14.201 & -- \\ 
        \cmidrule(l){1-10}
        Spectral~\cite{tran2018spectral}   
            & {\ul{18.389}} & 18.356 & 0.502 
            & 18.765 & 14.027 & 0.475 
            & 18.666 & 15.245 & 0.539 \\
        TED~\cite{mo2024robust} 
            & 18.434 & 20.063 & 0.528 
            & 18.785 & 13.984 & 0.473 
            & 18.606 & 13.953 & 0.495 \\
        TED++~\cite{le2025ted++}     
            & 19.197 & 19.184 & 0.499 
            & 18.706 & 13.445 & 0.473 
            & 18.565 & 14.541 & 0.513 \\ 
        \cmidrule(l){1-10}
        Fine-tuning~\cite{gu2019badnets}    
            & 19.003 & 30.909 & 0.625 
            & 18.837 & 22.479 & 0.641 
            & 18.934 & 18.196 & 0.594 \\
        Fine-pruning~\cite{liu2018fine}   
            & 19.020 & 31.643 & 0.633 
            & 19.073 & 23.543 & 0.647 
            & 18.686 & 19.736 & 0.623 \\
        NAD~\cite{li2021neural}             
            & 18.795 & 26.809 & 0.600 
            & 18.539 & 20.297 & 0.614 
            & 18.584 & 18.158 & 0.600 \\
        IMS~\cite{dunnett2025backdoor}            
            & 19.239 & 17.731 & 0.466 
            & {\ul{18.521}} & 14.570 & 0.479 
            & {\ul{18.418}} & 14.351 & 0.509 \\ 
        \cmidrule(l){1-10}
        ABL~\cite{li2021anti}         
            & 19.637 & 19.104 & 0.493 
            & 18.649 & 15.055 & 0.501 
            & 18.761 & 14.481 & 0.509 \\
        PDB~\cite{wei2024mitigating}         
            & 18.630 & {\ul{54.690}} & {\ul{0.693}} 
            & 19.512 & {\ul{26.014}} & {\ul{0.652}} 
            & 18.967 & {\ul{22.397}} & {\ul{0.639}} \\
        ESTI~\cite{yu2025backdoor}       
            & 19.910 & 17.186 & 0.454 
            & 18.793 & 14.684 & 0.475 
            & 19.219 & 15.897 & 0.532 \\ 
        \cmidrule(l){1-10}
        \rowcolor[HTML]{EFEFEF}
        \textbf{\methodname} 
        & \textbf{17.928} & \textbf{104.677} & \textbf{0.868}
        & \textbf{17.628} & \textbf{57.759}  & \textbf{0.847}
        & \textbf{18.048} & \textbf{39.303}  & \textbf{0.808} \\
        \bottomrule
    \end{tabular}

    }
\vspace{-12pt}
\end{table*}

\subsection{Time-aware Reliable Pool Initialization} \label{sec:init}

In Stage I, we initialize a high-precision yet conservative reliable pool without any clean reference set. Rather than maximizing recall, this stage aims to provide a trustworthy starting point for subsequent training and prevent early backdoor reinforcement. We therefore apply two complementary criteria and intersect their selections to form $\mathcal{D}_{\text{rel}}$.

\noindent\textbf{Reverse-consistency filtering (RCF)}. As shown in \figurename~\ref{fig:loss_over_comparison} and \tableautorefname~\ref{tab:training-based-preliminary-table}, using the forecaster’s forward training loss alone to separate samples is unreliable in TSF. We instead exploit a temporal asymmetry of TSF backdoors: the injected dependency is designed for the forecasting direction (history $\rightarrow$ future), but it does not enforce a consistent reverse dependency (future $\rightarrow$ history)~\cite{lin2024backtime}. This mismatch makes reverse reconstruction less compatible with the backdoor dependency.

RCF operationalizes this via an auxiliary backcasting task. We train a backcaster $b_\phi$~\cite{hyndman2018forecasting} (using the same architecture as $f_\theta$) for a small number of $T_b$ epochs to reconstruct the flipped history window from the flipped future window. Let $\mathrm{Flip}(\cdot)$ denote temporal reversal along the time axis. The reverse-consistency loss~is:
\begin{equation}
\mathcal{L}_{\text{rcf}}(\mathbf{x}_{t})
= \ell(b_{\phi}(\mathrm{Flip}(\mathbf{X}_{t,f})),
\mathrm{Flip}(\mathbf{x}_{t,h})).
\end{equation}
We then select samples with relatively low reverse-consistency loss using a quantile threshold $\Gamma_{\text{RCF}}$ (the $\alpha$-quantile):
\begin{equation} \label{eq:rcf}
     \mathcal{D}_{\text{RCF}} = \{\mathbf{x}_{t} \mid  \mathcal{L}_{\text{rcf}}(\mathbf{x}_{t}) \leq \Gamma_{\text{RCF}} \}.
\end{equation}
\noindent\textbf{Neighborhood diversity filtering (NDF)}. We now introduce the temporal-structure criterion used both in this stage and in the following stage. We begin by analyzing the conditions under which a TSF backdoor succeeds, drawing on NTK-inspired kernel analyses~\cite{jacot2018neural} and previous backdoor studies~\cite{guo2022aeva,xian2023understanding}. This analysis motivates our neighborhood diversity criterion, which we formalize below.

\begin{theorem}[TSF Backdoor Success Bound]
\label{thm:main-tsf-backdoor-success}
Let $\mathbf{x}:=\mathbf{x}_{t,h}$ denote a triggered test input window, and consider a TSF predictor $\hat y(\mathbf{x})$ approximated by a Nadaraya--Watson kernel regressor trained on  $N_\mathrm{p}$ poisoned samples ${(\mathbf{x}'_j, T(\mathbf{x}'_j))}$ and $N_{\mathrm{bg}}$ background samples ${(\mathbf{x}_i,\mathbf{y}_i)}$ with an RBF kernel $K(\mathbf{u},\mathbf{v})=\exp(-\gamma\|\mathbf{u}-\mathbf{v}\|_2^2)$, where $\mathbf{x}_i :=\mathbf{x}_{i,h}$~and $\mathbf{y}_i:=\mathbf{x}_{i,f}$. Define
$\varepsilon:=\max_i K(\mathbf{x},\mathbf{x}_i)$ and $\sigma_p^2(\mathbf{x}):=\frac{1}{N_\mathrm{p}}\sum_{j=1}^{N_\mathrm{p}}\|\mathbf{x}-\mathbf{x}'_j\|_2^2$.
Assume (i) $\|\mathbf{y}_i-T(\mathbf{x})\|_2\le M$ for all background samples, and (ii) $T(\cdot)$ is locally Lipschitz with constant $L_T$ on a neighborhood of $\{\mathbf{x}\}\cup \{\mathbf{x}'_j\}_{j=1}^{N_\mathrm{p}}$. Then
\begin{equation*}
\big\|\hat y(\mathbf{x}) - T(\mathbf{x})\big\|_2 \;\le\; \frac{N_{\mathrm{bg}}\,M\,\varepsilon}{N_\mathrm{p}\exp\big(-\gamma\,\sigma_p^2(\mathbf{x})\big)} + L_T\,\sigma_p(\mathbf{x}).
\end{equation*}
\end{theorem}
\vspace{-5pt}
\noindent\textit{Proof.} Deferred to Appendix~\ref{app:theory-temporal-similarity}.
\begin{remark}
The bound decreases as poisoned inputs concentrate around the triggered window (small $\sigma_p(\mathbf{x})$), which increases their kernel weight. Thus, successful TSF backdoors tend to induce a tight, highly similar cluster of poisoned input windows and consequently highly similar poisoned input--output windows. For instance-normalized~windows, squared Euclidean distance is proportional to $1-\rho(\cdot,\cdot)$ (Pearson correlation) \cite{berthold2016clustering}, motivating our correlation-based neighborhood distance for identifying more diverse samples as reliable candidates. 
\end{remark}
\emph{Gaussian-weighted Pearson-correlation neighborhood distance.}
We measure temporal similarity using a Gaussian-weighted Pearson correlation that emphasizes the transition region between history and future. The weighted correlation between two window samples $\mathbf{x}_{i}$ and $\mathbf{x}_{j}$ is:
\begin{equation}
\resizebox{0.99\linewidth}{!}{$
    r_\omega(\mathbf{x}_{i}, \mathbf{x}_{j})  =
    \dfrac{\sum_{\tau} \omega_{\tau}
    \bigl( \mathbf{x}_{i}[\tau] - \bar{\mathbf{x}}_{i, \omega} \bigr)\bigl( \mathbf{x}_{j}[\tau] - \bar{\mathbf{x}}_{j, \omega} \bigr)}
    {\sqrt{\sum_{\tau} \omega_{\tau} \bigl( \mathbf{x}_{i}[\tau] - \bar{\mathbf{x}}_{i, \omega} \bigr)^2}
     \sqrt{\sum_{\tau} \omega_{\tau} \bigl( \mathbf{x}_{j}[\tau] - \bar{\mathbf{x}}_{j, \omega} \bigr)^2}}
$},
\end{equation}
where $\bar{\mathbf{x}}_{i,\omega}$ denotes the weighted mean of $\mathbf{x}_i$ under weights~$\omega$ as follows:
\begin{equation}
    \bar{\mathbf{x}}_{i,\omega} = \frac{\sum_{\tau} \omega_{\tau} \mathbf{x}_{i}[\tau]}{\sum_{\tau} \omega_{\tau}},
    \qquad
    \omega_{\tau} = \exp\!\left(-\frac{(\tau - \LIN )^2}{2\sigma^2}\right).
\end{equation}
We fix $\sigma=2$ in all experiments and define the induced distance
$d_\omega(\mathbf{x}_{i}, \mathbf{x}_{j}) = 1 - r_\omega(\mathbf{x}_{i}, \mathbf{x}_{j})$.
Let $\mathcal{N}_K(i)$ be the indices of the $K$ nearest neighbors of $\mathbf{x}_i$ under $d_\omega$.
The neighborhood distance score is:
\begin{equation} \label{eq:neigh}
S(\mathbf{x}_{i}) = \frac{1}{K} \sum_{j \in \mathcal{N}_K(i)} d_\omega(\mathbf{x}_{i}, \mathbf{x}_{j}).
\end{equation}
\noindent \emph{NDF criterion.}
To promote temporal-structure diversity and reduce the risk of selecting poisoned windows, NDF prioritizes samples with larger neighborhood distance. Concretely, we select the top $\alpha$ fraction with the highest scores:
\begin{equation} \label{eq:ndf}
     \mathcal{D}_{\text{NDF}} = \{\mathbf{x}_{t} \mid  S(\mathbf{x}_{t}) \geq \Gamma_{\text{NDF}} \},
\end{equation}
where $\Gamma_{\text{NDF}}$ is the $(1-\alpha)$-quantile of $\{S(\mathbf{x}_i)\}$.
Empirically, \figureautorefname~\ref{fig:neigh_dist} shows that poisoned samples exhibit abnormally smaller neighborhood distances in poisoned channels, consistent with the similarity concentration implied by \theoremautorefname~\ref{thm:main-tsf-backdoor-success}. Finally, we obtain the initial reliable pool from samples selected by both criteria: $\mathcal{D}_{\text{rel}} = \mathcal{D}_{\text{RCF}} \cap \mathcal{D}_{\text{NDF}}.$
\subsection{Distance-Regularized Loss Selection} \label{sec:drls}
After initializing $\mathcal{D}_{\text{rel}}$, \methodname{} enters Stage II and progressively updates the reliable pool during training. A key risk in TSF is that poisoned windows may become indistinguishable from clean windows under loss-only criteria as training proceeds. We therefore regularize loss-based selection with a neighborhood-diversity constraint, which maintains forecasting performance while avoiding the re-inclusion of highly correlated poisoned windows.

At each update, we recompute neighborhood distances using the current unreliable pool $\mathcal{D}_{\text{unrel}}$ as the neighbor set (rather than $\mathcal{D}$), since $\mathcal{D}_{\text{unrel}}$ becomes increasingly enriched with poisoned samples as the reliable pool expands. We first form a candidate set by selecting the top $100\pi\gamma\%$ ($\pi\ge1$) samples from $\mathcal{D}$ with the largest neighborhood distances (following NDF), denoted $\mathcal{D}_{\text{NDF}}^{\text{cand}}$. From this candidate set, we admit only the lowest-loss $\gamma|\mathcal{D}|$ samples, equivalently setting $\Gamma_{\text{DRLS}}$ to the $1/\pi$-quantile of losses over $\mathcal{D}_{\text{NDF}}^{\text{cand}}$:
\begin{equation} \label{eq:drls}
     \mathcal{D}_{\text{DRLS}} \;=\; \{\mathbf{x}_{t} \in  \mathcal{D}_{\text{NDF}}^{\text{cand}} \mid  \mathcal{L}(\mathbf{x}_{t}) \leq \Gamma_{\text{DRLS}} \}.
\end{equation}
After $T_1$ epochs of training on the initial reliable pool, \methodname{} trains $f_\theta$ for a further $T_2$ epochs while progressively updating $\mathcal{D}_{\text{rel}} \gets \mathcal{D}_{\text{DRLS}}$ via \equationautorefname~\ref{eq:drls}; the pool expansion ratio $\gamma$ starts from $\alpha$ and is capped at $\beta$ of the full dataset.

\section{Experiments} \label{sec:exp}

    \vspace{-4pt}
    We follow the datasets, attacks, and evaluation protocol in \sectionautorefname~\ref{sec:experimental-settings}.  For \methodname{}, we set $\alpha{=}0.2$ and $\beta{=}0.5$, and adopt a linear schedule for the clean-pool ratio $\gamma$ in Stage II and grid-search $\pi\in\{1.25,1.5\}$ and $K\in\{20,32\}$.  We train $f_\theta$ with Adam~\cite{kingma2014adam} for $T_1{=}10$ epochs in Stage I and $T_2{=}90$ epochs in Stage II, and train the backcaster $b_\phi$ for $T_b{=}10$ epochs.  Additional details are in \appendixautorefname~\ref{app:exp-settings}, further detailed analyses are deferred to Appendix~\ref{app:neighborhood-analysis}--\ref{app:pool-dynamics}. By default, we present main results on PEMS03 and report results on other datasets in the corresponding appendix.  
    Our code is available at \url{https://github.com/qducnguyen/TimeGuard}.

     \begin{table}[t]
            \centering
            \caption{Defense performance across PEMS03, Weather, and ETTm1 datasets under Random and BackTime attacks.}
            \vspace{-4pt}
            \label{tab:different-dataset}
            \scriptsize
            \setlength{\tabcolsep}{2pt}
            \renewcommand{\arraystretch}{1.1}
            \resizebox{\linewidth}{!}{
           \begin{tabular}{c|lcccccc}
            \toprule
            \multirow{2}{*}{\textbf{Dataset}} & \textbf{Attack~→}          
            & \multicolumn{3}{c}{\textbf{Random}}
            & \multicolumn{3}{c}{\textbf{BackTime}} \\
            \cmidrule(l){2-2}
            \cmidrule(l){3-5}
            \cmidrule(l){6-8}
            & \textbf{Defense~↓}
            & \MAEC~↓ & \MAEP~↑ & \FDER~↑
            & \MAEC~↓ & \MAEP~↑ & \FDER~↑ \\
            \midrule
            
            \multicolumn{1}{c|}{\multirow{3}{*}{\textbf{PEMS03}}}
             & No Defense 
             & 17.634 & 17.772 & -- 
             & 17.607 & 14.201 & -- \\
            \multicolumn{1}{c|}{} 
             & PDB~\cite{wei2024mitigating}   
             & 18.630 & 54.690 & 0.693 
             & 18.967 & 22.397 & 0.639 \\
            \multicolumn{1}{c|}{} 
             & \textbf{\methodname} 
             & \textbf{17.928} & \textbf{104.677} & \textbf{0.868} 
             & \textbf{18.048} & \textbf{39.303} & \textbf{0.808} \\ 
            \midrule
            
            \multicolumn{1}{c|}{\multirow{3}{*}{\textbf{Weather}}}
             & No Defense 
             & 11.210 & 14.991 & -- 
             & 10.768 & 15.913 & -- \\
            \multicolumn{1}{c|}{} 
             & PDB~\cite{wei2024mitigating}    
             & 12.305 & 91.237 & 0.841 
             & 11.732 & 56.439 & 0.827 \\
            \multicolumn{1}{c|}{} 
             & \textbf{\methodname} 
             & \textbf{10.587} & \textbf{177.583} & \textbf{0.942} 
             & \textbf{10.716} & \textbf{66.534} & \textbf{0.874} \\ 
            \midrule
            
            \multicolumn{1}{c|}{\multirow{3}{*}{\textbf{ETTm1}}}
             & No Defense 
             & 1.144 & 1.059 & -- 
             & 1.114 & 0.805 & -- \\
            \multicolumn{1}{c|}{} 
             & PDB~\cite{wei2024mitigating}    
             & \textbf{1.230} & 2.972 & 0.766 
             & 1.274 & 1.422 & 0.648 \\
            \multicolumn{1}{c|}{} 
             & \textbf{\methodname} 
             & 1.235 & \textbf{6.481} & \textbf{0.881} 
             & \textbf{1.268} & \textbf{1.443} & \textbf{0.652} \\ 
            \bottomrule
            \end{tabular}
            }
        \vspace{-8pt}
        \end{table}

        \begin{figure}[t]
            \centering
            \begin{subfigure}{0.32\linewidth}
                \includegraphics[width=\linewidth]{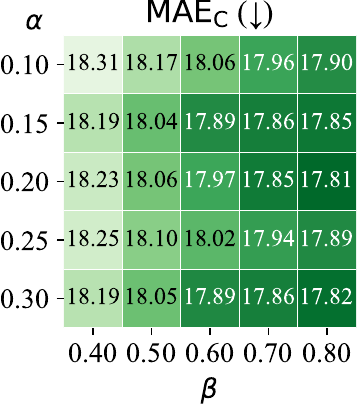}
            \end{subfigure}
            \begin{subfigure}{0.32\linewidth}
                \includegraphics[width=\linewidth]{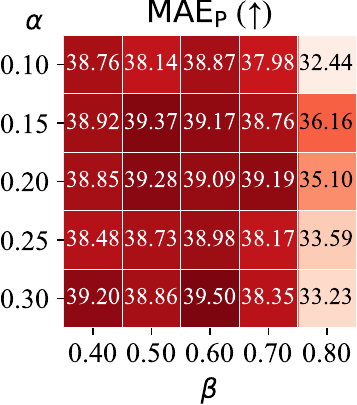}
            \end{subfigure}
            \begin{subfigure}{0.32\linewidth}
                \includegraphics[width=\linewidth]{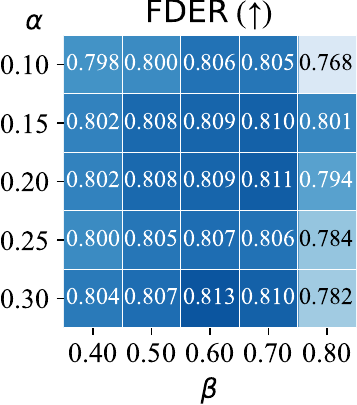}
            \end{subfigure}
            \vspace{-4pt}
            \caption{Hyperparameter analysis of pool size parameters $\alpha$ and $\beta$ in \methodname{} on the PEMS03 dataset under BackTime attack.}
            \vspace{-18pt}
            \label{fig:hyper-pool}
        \end{figure}
    
    \vspace{-7pt}
    \subsection{Main Results} \label{sec:exp-main-result}
    \vspace{-2pt}

    \textbf{Robustness against state-of-the-art attacks.} As shown in \tableautorefname~\ref{tab:main-defense-result}, averaged over three models, \methodname{} consistently mitigates all attacks, improving \MAEP{} to at least 39.3 (a minimum relative gain of 2.76x) while keeping clean \MAEC{} within 5\% of the undefended model. This indicates strong robustness to recent TSF backdoor attacks. Robustness against recent BadTime attack~\cite{xiang2025badtime} and per-model results are provided in \appendixautorefname~\ref{app:defense-performance}.

    \vspace{-2.5pt}
    \textbf{Comparison with state-of-the-art defenses.} \tableautorefname~\ref{tab:main-defense-result} also shows that \methodname{} achieves the best overall trade-off among previous training-phase defenses. Compared to the strongest baseline PDB, \methodname{} yields a 1.96x relative improvement in \MAEP{} and a 6.09\% relative reduction in \MAEC{}, with average \FDER{} of 0.841 across attacks. Notably, these gains require no additional clean data. Per-model results are provided in \appendixautorefname~\ref{app:defense-performance}.

    \vspace{-2.5pt}
    \textbf{Generalization on different datasets.}
       As shown in \tableautorefname~\ref{tab:different-dataset}, \methodname{} consistently improves robustness under both Random and BackTime across all three datasets, achieving \FDER{} above 0.65 in all settings. On Weather, \methodname{} also slightly improves clean forecasting accuracy over the undefended model (3.02\% on average), suggesting that neighborhood-distance-based criteria can act as a regularizer for better generalization. On ETTm1, \methodname{} incurs a small drop in clean performance but still delivers strong robustness without initial clean data unlike PDB.

     \begin{table}[t]
            \centering
            \caption{Ablation study on PEMS03 under Random and BackTime attacks. Full results are provided in~\appendixautorefname~\ref{app:ablation}.}
            \vspace{-5pt}
            \label{tab:ablation-results}
            \scriptsize
            \setlength{\tabcolsep}{2pt}
            \renewcommand{\arraystretch}{1.2}
            \resizebox{\linewidth}{!}{
                \begin{tabular}{l ccc ccc}
            \toprule
            \textbf{Attack}~→ 
                & \multicolumn{3}{c}{\textbf{Random}}
                & \multicolumn{3}{c}{\textbf{BackTime}} \\
            \cmidrule(lr){1-1}
            \cmidrule(lr){2-4}
            \cmidrule(l){5-7}
    
            \textbf{Defense}~↓ 
                & \MAEC~↓ & \MAEP~↑ & \FDER~↑
                & \MAEC~↓ & \MAEP~↑ & \FDER~↑ \\
            \midrule
    
            No Defense
                & 17.634 & 17.772 & --
                & 17.607 & 14.201 & -- \\
            
            \textbf{\methodname}
                & \textbf{17.928} & \textbf{104.677} & \textbf{0.868}
                & \textbf{18.048} & {\ul{39.303}} & \textbf{0.808} \\
                
            \midrule
    
            w/o Channel-wise
                & 18.320 & 16.145 & 0.478
                & 19.068 & 14.925 & 0.507 \\
            
            w/o NDF
                & 18.581 & {\ul{104.457}} & 0.853
                & 18.418 & 38.349 & 0.795 \\
    
            w/o RCF
                & {\ul{18.063}} & 104.405 & {\ul{0.865}}
                & 18.608 & \textbf{39.612} & 0.796 \\
    
            w/o NDF+RCF
                & 18.336 & 91.780 & 0.852
                & {\ul{18.273}} & 38.560 & {\ul{0.799}} \\
                
            w/o DRLS
                & 19.748 & 76.442 & 0.607
                & 20.081 & 22.918 & 0.586 \\
    
            \bottomrule
            \end{tabular}
            }
        \vspace{-9pt}
        \end{table}
            
        \begin{figure}[t]
            \centering
            \begin{subfigure}{0.452\linewidth}
                \includegraphics[width=\linewidth]{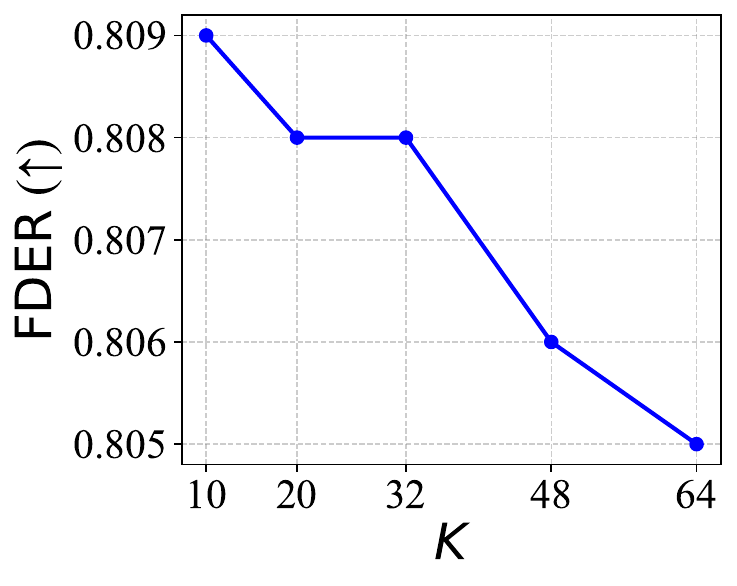}
            \end{subfigure}
            \hfill
            \begin{subfigure}{0.452\linewidth}
                \includegraphics[width=\linewidth]{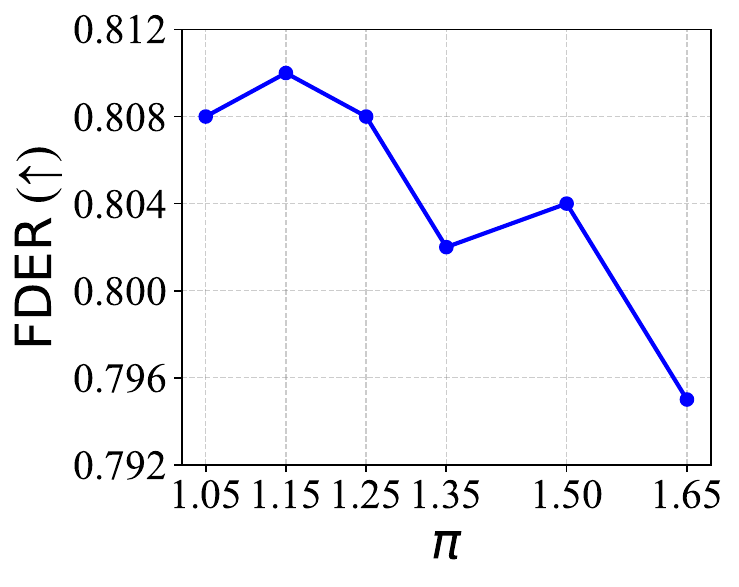}
            \end{subfigure}
            \vspace{-7pt}
            \caption{Hyperparameter analysis of  $K$ and $\pi$ in \methodname{} on the PEMS03 dataset under BackTime attack.}
            \label{fig:hyper-neighbor}
            \vspace{-18pt}
        \end{figure}
     
    \textbf{Generalization across diverse scenarios.} Comprehensive results across model architectures, TSF foundation models, forecasting horizons, poisoning rates, attack patterns, and challenging datasets with nonstationarity, strong distribution shifts, large scale, and count-valued variables are deferred to \appendixautorefname~\ref{app:defense-performance}. Overall, \methodname{} consistently achieves the best defense performance across diverse TSF attack settings and challenging scenarios. Notably, \methodname{} remains effective even in the extreme full-channel poisoning setting, e.g., $\eta_\mathrm{S}=1.0$, achieving \FDER{} of $0.748$. Furthermore, \methodname{} also transfers to an LLM-based forecaster~\cite{liu2024autotimes}, yielding at least a $5.14\times$ gain in \MAEP{} with only a $3.8\%$ change in clean~\MAEC{}.     

    \vspace{-4pt}
    \subsection{Analysis} \label{sec:analysis}
    \vspace{-2pt}
    
         \begin{table}[t]
                \centering
                \caption{Training time (seconds~↓) of in-training backdoor defenses on the PEMS03 dataset. ``No Defense'' denotes standard training on the poisoned data without any defense. Best results are in \textbf{bold}.}
                \vspace{-4pt}
                \label{tab:efficiency-analysis}
                \scriptsize
                \setlength{\tabcolsep}{3pt}
                \renewcommand{\arraystretch}{1.05}
                \resizebox{\linewidth}{!}{
               \begin{tabular}{lcccc}
                \toprule
                \textbf{Model~→}   & \multirow{2}{*}{\textbf{SimpleTM}} & \multirow{2}{*}{\textbf{FEDformer}} & \multirow{2}{*}{\textbf{TimesNet}} & \multirow{2}{*}{\textbf{\textsc{Average}}} \\
                \textbf{Defense~↓} &&&& \\ \midrule
                No Defense& 1621& 2340& 2442& 2134\\
                \midrule
                ABL~\cite{li2021anti}& \textbf{740}& \textbf{1409}& \textbf{2038}& \textbf{1395}\\
                PDB~\cite{wei2024mitigating}& 2378& 3441& 3399& 3073\\
                ESTI~\cite{yu2025backdoor}& 5563& 10347& 11253& 9054\\
                \textbf{\methodname{}} & 2454& 3411& 4250& 3372\\ \bottomrule
                \end{tabular}
                }
            \vspace{-18pt}
            \end{table}

        \textbf{Ablation study.} We ablate \methodname{} under Random and BackTime to quantify each design component’s contribution. As shown in \tableautorefname~\ref{tab:ablation-results}, removing the channel-wise formulation causes the defense to fail (\FDER{} = 0.493), underscoring the need to match the channel-subset granularity of TSF attacks. NDF and RCF are critical in Stage I for constructing a high-precision reliable pool and preventing early absorption of poisoned samples, as reflected by \MAEP{}. Replacing DRLS with loss-only selection substantially degrades performance, reducing \FDER{} by 28\% on average. These results underscore the necessity of distance-aware selection for both clean generalization (\MAEC{}) and robustness (\MAEP{}). Overall, the components contribute synergistically to \methodname{}’s effectiveness. Additional per-model ablation results are provided in \appendixautorefname~\ref{app:ablation}.

        \textbf{Influence of $\alpha$ and $\beta$.} \figureautorefname~\ref{fig:hyper-pool} shows a clear trade-off between clean performance (\MAEC{}) and robustness (\MAEP{}, \FDER{}) as the pool sizes vary under BackTime. Extremely small or large $\beta$ either admits too few clean samples or incorporates too many poisoned samples, both of which reduce \FDER{}. Similarly, a small $\alpha$ yields an insufficiently reliable initial pool for Stage II, leading to worse \MAEC{}, \MAEP{}, and \FDER{}. Empirically, $\alpha \in [0.15, 0.25]$ and $\beta \in [0.5, 0.7]$ provide the best balance, achieving the highest \FDER{}, exceeding 0.8. 
    
        \textbf{Influence of $K$ and $\pi$.} With $\alpha{=}0.2$ and $\beta{=}0.5$, we study the neighborhood size $K$ and scaling factor $\pi$. As shown in \figureautorefname~\ref{fig:hyper-neighbor}, \methodname{} is largely insensitive to $K$, with \FDER{} staying in a narrow range (0.805–0.809). In contrast, overly large $\pi$ tends to narrow the candidate set and reduce diversity, increasing the risk of admitting poisoned samples during Stage II. We thus recommend $\pi \leq 1.5$. Full per-model results, hyperparameter sensitivity analyses across different datasets and attacks, and analyses of other hyperparameters are provided in \appendixautorefname~\ref{app:hyper-sensit}.

      \textbf{Efficiency analysis.} With our implementation, memory footprints are the same across methods; we therefore focus on the wall-clock training time of in-training defenses. As shown in \tableautorefname~\ref{tab:efficiency-analysis}, \methodname{} incurs a $1.58\times$ training-time overhead over vanilla training, mainly due to its multi-stage procedure and neighborhood-distance computations. This overhead is comparable to the strongest baseline, PDB, while providing substantially stronger robustness. In contrast, ESTI incurs a much larger overhead ($4.24\times$ on average) yet remains ineffective against TSF backdoor attacks. Since \methodname{} adds no inference-time overhead, it remains practical. Implementation details and large-model running times are provided in \appendixautorefname~\ref{app:efficiency-analysis} and~\ref{app:defense-performance}.

       \textbf{Potential adaptive attacks.} We consider a worst-case adaptive scenario in which the attacker extends the state-of-the-art BackTime attack~\cite{lin2024backtime} by (i) using a well-trained backcaster $b_\phi$ as a regularizer to encourage reverse consistency and (ii) explicitly penalizing high correlation among poisoned samples to evade our neighborhood-based criterion. As shown in \tableautorefname~\ref{tab:adaptive-attack}, this adaptive attack attains 18.791 \MAEC{} and 15.343 \MAEP{}, slightly worse than the original BackTime attack. This is consistent with \theoremautorefname~\ref{thm:main-tsf-backdoor-success}, which suggests that successful TSF backdoor attacks benefit from tight, highly similar clusters of poisoned inputs. Under this adaptive threat, \methodname{} remains effective, achieving 18.438 \MAEC{}, 30.575 \MAEP{}, and 0.744 \FDER{}. Ablations further show that neighborhood-based cues remain useful: removing NDF only slightly reduces \FDER{} to 0.739, while removing DRLS causes a larger drop to 0.543. More implementation details and per-model ablation results are provided in \appendixautorefname~\ref{app:adaptive-attack}.
        
            \begin{table}[t]
                \centering
                \caption{Defense performance of \methodname{} under BackTime and adaptive attacks on PEMS03 dataset, averaged over three models. Best results under adaptive attack are in \textbf{bold}.}
                \vspace{-4pt}
                \label{tab:adaptive-attack}
                \scriptsize
                \setlength{\tabcolsep}{4.2pt}
                \renewcommand{\arraystretch}{1.05}
                \resizebox{\linewidth}{!}{
               \begin{tabular}{clccc}
                    \toprule
                    \textbf{Attack}
                    & \textbf{Defense}
                    & \textbf{\MAEC~↓} & \textbf{\MAEP~↑} & \textbf{\FDER~↑} \\
                    \midrule
                    \multirow{2}{*}{BackTime} & No Defense 
                        & 17.607 & 14.201 & -- \\
                    & \textbf{\methodname}
                        & 18.048 & 39.303 & 0.808 \\
                    \midrule
                    \multirow{4}{*}{Adaptive} & No Defense 
                        & 18.791 & 15.343 & -- \\
                    & \textbf{\methodname}
                        & \textbf{18.438} & \textbf{30.575} & \textbf{0.744} \\
                    & \textbf{\methodname{}} w/o NDF
                        & 18.564 & 29.695 & 0.739 \\
                    &\textbf{\methodname{}} w/o DRLS
                        & 20.863 & 19.026 & 0.543 \\
                    \bottomrule
                \end{tabular}
                }
            \vspace{-20pt}
            \end{table}

\vspace{-7pt}
\section{Conclusion}
\vspace{-3pt}

    Our paper presents the first systematic study of defenses against TSF backdoor attacks. We first expose key failure modes of existing classification defenses in TSF stemming from data entanglement and task-formulation shift. To address these gaps, we propose \methodname{}, a novel backdoor defense for TSF. Specifically, \methodname{} performs channel-wise reliable pool training and leverages reverse consistency and temporal pattern concentration in poisoned TSF data to initialize and progressively refine reliable pools. Extensive experiments validate \methodname{}'s effectiveness and generalization. Overall, our results emphasize the need for more robust and trustworthy forecasting systems. Limitations and future work are discussed in \appendixautorefname~\ref{app:limitation-future-work}.
    
\section*{Acknowledgment}
This research / project is supported by the National Research Foundation, Singapore, and Cyber Security Agency of Singapore under its National Cybersecurity R\&D Programme and CyberSG R\&D Cyber Research Programme Office. Any opinions, findings and conclusions or recommendations expressed in these materials are those of the author(s) and do not reflect the views of National Research Foundation, Singapore, Cyber Security Agency of Singapore as well as CyberSG R\&D Programme Office, Singapore.
    
\section*{Impact Statement}
This work studies backdoor learning in time series forecasting (TSF) and proposes a defense against TSF backdoor attacks. It may improve the reliability of forecasting components in safety- or cost-critical pipelines and support the development of more robust and trustworthy time series machine learning. Potential negative impacts are primarily related to dual use: our analysis and evaluation may help adversaries design more evasive backdoors or adapt poisoning strategies. Accordingly, we report findings under explicit threat models and emphasize that defenses should be complemented by other standard security measures (e.g., data provenance) to provide more comprehensive protection.

\bibliography{ref}
\bibliographystyle{icml2026}

\newpage
\appendix
\onecolumn
\section{Related Work} \label{app:related-work}

\subsection{Deep Models for Time Series Forecasting}
Time series forecasting (TSF) aims to predict future values of one or multiple variables based on their historical observations.
With the rapid development of deep learning, a wide variety of TSF DNN architectures have been proposed to model complex temporal dependencies, nonlinear dynamics, and inter-variable dependencies. RNN-based methods~\cite{abbasimehr2022improving, hewage2020temporal, lin2023segrnn} capture sequential patterns through recursive state transitions, while CNN-based methods~\cite{hewage2020temporal, cheng2025convtimenet} employ dilated and causal convolutions to efficiently learn long-range temporal features. GNN-based approaches~\cite{yan2018spatial, ma2020streaming} explicitly represent inter-variable correlations by constructing spatio-temporal graphs, enabling information propagation across related variables. 
 
Recently, Transformer-based models~\cite{nie2023time, zhou2022fedformer, wu2023timesnet, chen2025simpletm} have achieved state-of-the-art TSF performance by leveraging self-attention to jointly capture global temporal dependencies and cross-variable interactions. MLP-based architectures~\cite{han2024softs, wang2024timemixer}, built primarily on linear transformations, maintain high computational efficiency while still delivering strong forecasting accuracy. Currently, LLM-based models~\cite{liu2024unitime, liu2024autotimes} employ pre-trained LMMs as backbones and demonstrate impressive cross-domain generalization and zero-shot forecasting capability.
Meanwhile, reinforcement learning has recently become an important paradigm for improving LLM capability and behavior~\cite{fang2025serl,zhang2025consistent,zhang2026incentivizing}.

However, the increasing model complexity and data dependency of modern TSF architectures introduce new trustworthiness concerns, including adversarial attacks~\cite{xu2021adversarial, pialla2025time, liu2025adversarial}, backdoor attacks~\cite{lin2024backtime, kotowski2025trojan}, hallucination~\cite{zou2025investigating}, and watermarking~\cite{soi2025timewak}. In this work, we focus specifically on backdoor defenses for time series forecasting.

\subsection{Backdoor Attacks}
Backdoor attacks are typically implemented by injecting a small number of poisoned samples into the training set to implant hidden trigger-target associations \cite{li2022backdoor}. Once trained on such data, the model behaves normally on clean inputs but exhibits malicious behavior when the trigger appears, for example, classifying triggered samples into an attacker-specified target label. Such attacks have been extensively studies in computer vision~\cite{gu2019badnets, gao2023not, gao2024backdoor, zhu2025towards, chen2026taught, li2026rethinking}, speech recognition \cite{zhai2021backdoor, koffas2023going, cai2024toward}, object recognition \cite{lifew2022, chan2022baddet, luo2023untargeted}, and graph learning~\cite{xi2021graph}, demonstrating that even a tiny poisoning ratio can yield high attack success while maintaining benign performance. 

In the time series domain, prior work has primarily examined backdoor attacks on classification tasks~\cite{ding2022towards, jiang2023backdoor, huang2025revisiting}, where temporal triggers are injected into complete time series to manipulate predictions of physiological or activity signals. However, these studies are restricted to producing categorical output labels for entire time series, rather than finer-grained temporal segments. The first work to target TSF models, BackTime~\cite{lin2024backtime}, embeds stealthy GNN-based trigger patterns with associated predefined target patterns in selected time step on the original training dataset via bi-level optimization. Following this, TBDA~\cite{liu2026beyond} introduces temporally delayed, variable-specific activations instead of immediate alignment, extending BackTime under a continuity assumption between trigger and target patterns. Meanwhile, BadTime~\cite{xiang2025badtime} studies long-term TSF and aims to train a backdoored model by using hybrid training strategy to select valuable poisoned samples and a decoupled backdoor objective leveraging distributed lag-aware triggers. 

Nevertheless, BadTime assumes a less practical threat model that requires full control over the training pipeline, whereas BackTime assumes only dataset-level control and employs more flexible, sample-dependent triggers. Although distributed lag-aware triggers are expressive, BadTime assumes unrealistic control over all input variables, which are typically distributed across multiple real-world data sources. Therefore, we adopt BackTime as our default threat model and leave a comprehensive evaluation under the BadTime-style threat model for future work.

\subsection{Backdoor Defense}
Backdoor defenses aim to mitigate or neutralize backdoor behaviors implanted during training or to detect such behaviors at inference. These methods can be categorized into four stages of the model life cycle~\cite{wu2025backdoorbench}. \textit{Pre-training-stage defenses} attempt to identify and remove poisoned samples before model training by analyzing training samples statistics or feature distributions to detect anomalous samples~\cite{tran2018spectral, chen2018detecting, mo2024robust, le2025ted++,hou2025flare}. \textit{In-training-stage defenses} aim to train clean models on poisoned datasets without backdoor injection, typically by reducing the influence of potentially poisoned samples through carefully designed training procedures~\cite{li2021anti,tang2023setting, gao2023backdoor, wei2024mitigating, yu2025backdoor, qiao2026cert}. \textit{Post-training-stage defenses} repair compromised models through structural modification or fine-tuning-based approaches~\cite{liu2018fine, wang2019neural, li2021neural, dunnett2025backdoor, xu2024towards, chen2025refine}, Finally, \textit{inference-time defenses} detect the presence of triggers at test time by measuring prediction consistency or entropy under different input perturbations or one input with multiple model variances~\cite{liu2023detecting, gao2019strip, guoscale2023, hou2024ibd, yiprobe2025}.

Although these defenses have demonstrated effectiveness in classification and vision domains~\cite{wu2025defenses}, their applicability to the time series domain remains largely underexplored. In time series classification, one representative effort is E2ABL~\cite{jiang2024end}, which extends ABL~\cite{li2021anti} and evaluates existing backdoor defenses on time series classification datasets. However, this work primarily focuses on empirical evaluation rather than proposing defenses that explicitly account for the temporal structure of time series inputs. Likewise, backdoor defenses for TSF remain largely unexplored. One notable concurrent effort in TSF backdoor defense is the competition associated with the ``Assurance for Space Domain AI Applications" program, which aims to detect and reconstruct static trigger patterns in backdoored TSF models~\cite{kotowski2025trojan, wang2026trojanscope}. However, this setting assumes access to clean models with the same architecture as the poisoned models, as well as clean datasets, and is limited to a subset of model architectures within the space operations domain. In this work, we conduct the first systematic study of representative backdoor defenses across the TSF model life cycle, spanning multiple domains and model architectures. We further introduce \methodname{}, an in-training-stage defense specifically designed for TSF backdoors.

\section{Further Analysis of Existing Backdoor Defenses}
\label{app:more-existing-defense}

Beyond the two fundamental issues in current backdoor defense settings for time series forecasting (TSF), namely data entanglement and task-formulation shift as discussed in \sectionautorefname~\ref{sec:intro}, we further provide a detailed analysis of TSF-specific challenges that hinder the direct adaptation of existing defenses. We summarize these challenges in \sectionautorefname~\ref{app:chal-repre}. We then provide the rationale for selecting representative baselines in \sectionautorefname~\ref{app:selection-rationale} and clarify the practical attribute aspects of each defense in the TSF backdoor setting in \sectionautorefname~\ref{app:tsf-attributes}.

\vspace{-5pt}
\subsection{TSF-Specific Challenges for Backdoor Defense}
\label{app:chal-repre}

Similar to backdoor attacks~\cite{lin2024backtime}, defending TSF models against backdoor attacks presents several unique challenges compared to traditional backdoor defense in classification and generative models~\cite{wu2025backdoorbench, li2025backdoorllm, lin2025backdoordm}. These difficulties largely come from the intrinsic properties of forecasting.
\textbf{(i) }The target outputs in TSF lie in a \emph{continuous space}, making it infeasible for label-based defenses~\cite{chen2018detecting, wang2019neural, chou2020sentinet, shen2025bi} that rely on either identifying poisoned classes or reconstructing potential triggers for each label in the output space.
\textbf{(ii)} Samples in TSF exhibit strong \emph{temporal dependencies}, where a single injected trigger or target pattern can propagate across overlapping input-output windows, contaminating subsequent forecasting steps and making it difficult to distinguish between clean and poisoned samples.
\textbf{(iii) }Time series data are often \emph{uninterpretable to human}; detecting abnormal fluctuations or poison patterns typically requires domain expertise (e.g., finance or healthcare), making manual inspection unreliable and the construction of a trusted clean dataset prohibitively expensive.
\textbf{(iv) }TSF models are typically deployed in \emph{continuous real-time settings}, where forecasts are generated sequentially and updated as new data arrive.
Defense methods therefore must operate efficiently, limiting the practicality of inference-time detection methods that often require multiple forward passes~\cite{gao2019strip, liu2023detecting, hou2024ibd}. 

Beyond these factors, the representational characteristics of TSF models also introduce further challenges for defense adaptation. \textbf{(i)} The \emph{heterogeneous representations} produced by different deep TSF models~\cite{kim2025comprehensive} significantly hinder defense generalization.
For instance, some models explicitly decompose time series into separate trend and seasonal components~\cite{wu2023timesnet}, while others rely on frequency-based transformations~\cite{zhou2022fedformer} or channel-independent that processes each variable independently~\cite{nie2023time}. As a result, their hidden representation spaces vary substantially across architectures, underscoring the need for model-agnostic (architecture-agnostic) defense design. \textbf{(ii)} Unlike classification or word embedding models~\cite{radford2021learning}, whose latent representations often align with semantically discrete concepts (e.g., object categories or word meanings), the semantics of hidden representations in TSF remain largely \emph{underexplored}, further making representation-based defenses~\cite{tran2018spectral, mo2024robust} unreliable under different DNNs.

\subsection{Rationale for Selecting Representative Defenses}
\label{app:selection-rationale}

We evaluate \textbf{13} representative defenses spanning four stages of the model life cycle, following the taxonomy of BackdoorBench~\cite{wu2025backdoorbench}. Specifically, our selection is guided by two criteria. \textbf{(i)} \textbf{Representativeness}: we include both classic methods (e.g., Spectral~\cite{tran2018spectral}, Fine-pruning~\cite{liu2018fine}, ABL~\cite{li2021anti}) and recent advanced approaches (e.g., PDB~\cite{wei2024mitigating}, TED++~\cite{le2025ted++}, and ESTI~\cite{yu2025backdoor}) that have shown effectiveness in classification or vision domains.  
\textbf{(ii)} \textbf{Adaptation Feasibility}: the method must be practical to adapt to TSF. 

Accordingly, we exclude algorithms that depend on discrete output spaces, such as those requiring enumeration of all target labels to detect poisoned samples~\cite{chen2018detecting, shen2025bi}, or access to poisoned labels~\cite{shen2025bi}, as well as methods relying on self- or semi-supervised learning frameworks~\cite{huang2022backdoor, gao2023backdoor}, which remain architecture-dependent and are not yet applicable to diverse TSF models~\cite{zhang2024self, cho2025comres}.

\subsection{Key Practical Attributes for TSF Defenses}
\label{app:tsf-attributes}

\begin{table*}[t]
    \centering
    \caption{Key attributes of defense methods against TSF backdoor attacks. } 
    \label{tab:defense-attribute}
    \renewcommand{\arraystretch}{1.2}
    \renewcommand{\aboverulesep}{0pt}
    \renewcommand{\belowrulesep}{0pt}
    \setlength\cellspacetoplimit{2pt}
    \setlength\cellspacebottomlimit{2pt}
    \resizebox{\textwidth}{!}{
    \begin{tabular}{S l *{5}{S c}}
        \toprule
        \textbf{Method} & \textbf{Defense Stage} & \textbf{\begin{tabular}[c]{@{}c@{}}No Additional \\ Clean Data Required\end{tabular}} & \textbf{\begin{tabular}[c]{@{}c@{}}No Internal \\ Features Access Required\end{tabular}} & \textbf{\begin{tabular}[c]{@{}c@{}}No Additional \\ Inference Overhead\end{tabular}} & \textbf{\begin{tabular}[c]{@{}c@{}}Time-Aware \\ Design\end{tabular}} \\ 
        \midrule
        Spectral~\cite{tran2018spectral}   
        & Pre-training & \cmark & \xmark & \cmark & \xmark \\
        TED~\cite{mo2024robust} 
        & Pre-training & \xmark & \xmark & \cmark & \xmark \\ 
        TED++~\cite{le2025ted++}     
        & Pre-training & \xmark & \xmark & \cmark & \xmark \\ 
        \midrule
        Fine-tuning~\cite{gu2019badnets}    
        & Post-training & \xmark & \cmark & \cmark & \xmark \\
        Fine-pruning~\cite{liu2018fine}   
        & Post-training & \xmark & \xmark & \cmark & \xmark \\
        NAD~\cite{li2021neural}             
        & Post-training & \xmark & \xmark & \cmark & \xmark \\ 
        IMS~\cite{dunnett2025backdoor}            
        & Post-training & \xmark & \xmark & \cmark & \xmark \\ 
        \midrule
        ABL~\cite{li2021anti}         
        & In-training & \cmark & \cmark & \cmark & \xmark \\
        PDB~\cite{wei2024mitigating}         
        & In-training & \xmark & \cmark & \xmark & \xmark \\ 
        ESTI~\cite{yu2025backdoor}       
        & In-training & \xmark & \cmark & \cmark & \xmark \\ 
        \midrule
        STRIP~\cite{gao2019strip}           
        & Inference & \xmark & \cmark & \xmark & \xmark \\
        TeCo~\cite{liu2023detecting}  
        & Inference & \cmark & \cmark & \xmark & \xmark \\ 
        IBD-PSC~\cite{hou2024ibd}
        & Inference & \xmark & \xmark & \xmark & \xmark \\ 
        \midrule
        \rowcolor[HTML]{EFEFEF} 
        \textbf{\methodname} & In-training & \cmark & \cmark & \cmark & \cmark \\ 
        \bottomrule
    \end{tabular}
    }
    \vspace{-12pt}
\end{table*}

To systematically compare these defenses, we examine four key attributes relevant to forecasting: \textbf{(i)} \textbf{No Additional Clean Data Required}: whether the method avoids dependence on a clean split, addressing the challenge of constructing trusted datasets for time series; \textbf{(ii)} \textbf{No Internal Feature Access Required}: whether the defense operates without access to intermediate activations or feature representations, reflecting model-agnostic applicability; \textbf{(iii)} \textbf{No Additional Inference Overhead}: whether the defense incurs extra computational cost during inference, which is critical for real-time forecasting deployments; and \textbf{(iv)} \textbf{Time-Aware Design}: whether the defense explicitly incorporates time-series characteristics. 
As summarized in \tablename~\ref{tab:defense-attribute}, substantial differences emerge across defenses at different stages. Post-training-stage methods typically rely on additional clean data, while pre-training-stage defenses often require access to internal representations. Inference-time defenses, on the other hand, introduce notable inference overhead, limiting their deployment efficiency. Importantly, \textit{none of the existing defenses} explicitly accounts for temporal dynamics, as they were all originally designed for static classification tasks. Motivated by these observations, we evaluate these defenses in the TSF setting in \sectionautorefname~\ref{sec:revisiting} and introduce \methodname{} as a in-training time-aware backdoor defense in \sectionautorefname~\ref{sec:method}. We leave the development of efficient inference-time backdoor defenses for future work.

\section{Theoretical Analysis of TSF Backdoor Success} \label{app:theory-temporal-similarity}

In this section, we provide a bound showing that successful and stealthy TSF backdoor attacks tend to induce highly similar (and thus highly correlated) poisoned input windows, motivating the design of \methodname{}. For readability, we focus on a single channel so that each history window is a vector $\mathbf{x}_{t,h}\in\mathbb{R}^{\LIN}$ and each future window is $\mathbf{x}_{t,f}\in\mathbb{R}^{\LOUT}$. We denote a triggered test input by $\mathbf{x}:=\mathbf{x}_{t,h}$, background inputs by $\mathbf{x}_i:=\mathbf{x}_{i,h}$ with outputs $\mathbf{y}_i:=\mathbf{x}_{i,f}$, and poisoned inputs by $\mathbf{x}'_j$.

\textbf{Setup.} Following~\cite{xian2023understanding, guo2022aeva}, we approximate a TSF predictor in a kernel regression regime~\cite{jacot2018neural}.
Assume all windows are instance-normalized during preprocessing.
Let $K(\mathbf{u},\mathbf{v})$ be an RBF kernel
\[
K(\mathbf{u},\mathbf{v}) \;=\; \exp\!\big(-\gamma \|\mathbf{u}-\mathbf{v}\|_2^2\big),
\]
with bandwidth $\gamma>0$.
The training set consists of $N_\mathrm{p}$ poisoned samples
$\mathcal{D}_\mathrm{p}=\{(\mathbf{x}'_j,\mathbf{y}'_j)\}_{j=1}^{N_\mathrm{p}}$ and  $N_{\mathrm{bg}}$ background samples
$\mathcal{D}_{\mathrm{bg}}=\{(\mathbf{x}_i,\mathbf{y}_i)\}_{i=1}^{N_{\mathrm{bg}}}$
(e.g., containing clean and affected samples).

\textbf{Attack mechanism and target mapping.} In the threat model~\cite{lin2024backtime}, the attacker inserts a trigger into the history window and enforces a patterned target in the future window.
At the dataset level (multivariate notation), for an injection time $t$ and attacked channel subset $\mathcal{S}$:
\[
\mathbf{X}[t-\LTGR:t,\mathcal{S}] \leftarrow \mathbf{G}_t,
\qquad
\mathbf{X}[t:t+\LPTN,\mathcal{S}]
\leftarrow
\mathbf{X}[t-\LTGR-1,\mathcal{S}] \oplus \mathbf{P},
\]
where $\mathbf{G}_t$ is a trigger pattern at timestep $t$ and  $\mathbf{P}$ is a fixed attack pattern template and $\oplus$ denotes element-wise addition (with broadcasting along time when needed). This produces sample-dependent target patterns because the baseline term $\mathbf{X}[t-\LTGR-1,\mathcal{S}]$ varies across samples. We abstract this behavior via a deterministic mapping $T(\cdot)$ at the window level.

\begin{definition}[Backdoor Target Mapping]
\label{def:target-map}
Let $\mathbf{p}$ denote the attack pattern template (aligned to the selected channel within $\mathbf{P}$), and let $b(\mathbf{x})$ extract a baseline value from an input window (e.g., the value immediately preceding the trigger, broadcast to match the target horizon).
Define
\[
T(\mathbf{x}) \;:=\; b(\mathbf{x}) \oplus \mathbf{p}.
\]
For a poisoned input $\mathbf{x}'_j$, its poisoned label is $\mathbf{y}'_j = T(\mathbf{x}'_j)$.
For a triggered test window $\mathbf{x}$, the attacker aims for $\hat{\mathbf{y}}(\mathbf{x})\approx T(\mathbf{x})$.
\end{definition}

\textbf{Key quantities.} For a triggered test input $\mathbf{x}$, define the maximum similarity to background inputs:
\[
\varepsilon \;:=\; \max_{(\mathbf{x}_i,\mathbf{y}_i)\in \mathcal{D}_{\mathrm{bg}}} K(\mathbf{x},\mathbf{x}_i),
\]
and define the poison dispersion around $\mathbf{x}$:
\[
\sigma_p^2(\mathbf{x})
\;:=\;
\frac{1}{N_\mathrm{p}}\sum_{j=1}^{N_\mathrm{p}}\|\mathbf{x}-\mathbf{x}'_j\|_2^2,
\qquad
\sigma_p(\mathbf{x}) := \sqrt{\sigma_p^2(\mathbf{x})}.
\]
Intuitively, $\varepsilon$ measures how strongly background samples can influence prediction at $\mathbf{x}$, while $\sigma_p(\mathbf{x})$ measures how tightly poisoned inputs concentrate around $\mathbf{x}$.

\begin{theorem}[TSF Backdoor Success Bound]
\label{thm:tsf-backdoor-success}
Let $\hat{\mathbf{y}}(\cdot)$ be the Nadaraya--Watson kernel regressor trained on $\mathcal{D}_{\mathrm{bg}}\cup \mathcal{D}_\mathrm{p}$:
\[
\hat{\mathbf{y}}(\cdot)
=
\frac{\sum_{i=1}^{N_{\mathrm{bg}}} K(\cdot,\mathbf{x}_i)\,\mathbf{y}_i \;+\; \sum_{j=1}^{N_\mathrm{p}} K(\cdot,\mathbf{x}'_j)\,T(\mathbf{x}'_j)}
     {\sum_{i=1}^{N_{\mathrm{bg}}} K(\cdot,\mathbf{x}_i) \;+\; \sum_{j=1}^{N_\mathrm{p}} K(\cdot,\mathbf{x}'_j)}.
\]
Assume:
\begin{enumerate}
    \item \textbf{Bounded background deviation.}
    For the triggered test window $\mathbf{x}$, $\|\mathbf{y}_i - T(\mathbf{x})\|_2 \le M$ for all $(\mathbf{x}_i,\mathbf{y}_i)\in \mathcal{D}_{\mathrm{bg}}$.
    \item \textbf{Local Lipschitzness of $T$.}
    There exists $L_T>0$ such that
    \[
    \|T(\mathbf{u})-T(\mathbf{v})\|_2 \le L_T \|\mathbf{u}-\mathbf{v}\|_2
    \quad \text{for all } \mathbf{u},\mathbf{v} \text{ in a neighborhood of }
    \{\mathbf{x}\}\cup\{\mathbf{x}'_j\}_{j=1}^{N_\mathrm{p}}.
    \]
\end{enumerate}
Then for the triggered test window $\mathbf{x}$,
\[
\big\|\hat{\mathbf{y}}(\mathbf{x}) - T(\mathbf{x})\big\|_2
\;\le\;
\underbrace{\frac{N_{\mathrm{bg}}\,M\,\varepsilon}{N_\mathrm{p}\,\exp\!\big(-\gamma\,\sigma_p^2(\mathbf{x})\big)}}_{\text{(I) background influence}}
\;+\;
\underbrace{L_T\,\sigma_p(\mathbf{x})}_{\text{(II) target mismatch}}.
\]
\end{theorem}

\begin{proof}
Let
\[
W(\mathbf{x}) := \sum_{i=1}^{N_{\mathrm{bg}}} K(\mathbf{x},\mathbf{x}_i) \;+\; \sum_{j=1}^{N_\mathrm{p}} K(\mathbf{x},\mathbf{x}'_j),
\qquad
W_p(\mathbf{x}) := \sum_{j=1}^{N_\mathrm{p}} K(\mathbf{x},\mathbf{x}'_j).
\]
Subtract $T(\mathbf{x})$ from $\hat{\mathbf{y}}(\mathbf{x})$ and regroup:
\[
\hat{\mathbf{y}}(\mathbf{x})-T(\mathbf{x})
=
\frac{\sum_{i=1}^{N_{\mathrm{bg}}} K(\mathbf{x},\mathbf{x}_i)\big(\mathbf{y}_i-T(\mathbf{x})\big)
+
\sum_{j=1}^{N_\mathrm{p}} K(\mathbf{x},\mathbf{x}'_j)\big(T(\mathbf{x}'_j)-T(\mathbf{x})\big)}
{W(\mathbf{x})}.
\]
Taking norms and applying triangle inequality yields two terms.

\textbf{(I) Background influence.}
Using $\|\mathbf{y}_i-T(\mathbf{x})\|_2\le M$ and $K(\mathbf{x},\mathbf{x}_i)\le \varepsilon$,
\[
\Big\|\sum_{i=1}^{N_{\mathrm{bg}}} K(\mathbf{x},\mathbf{x}_i)\big(\mathbf{y}_i-T(\mathbf{x})\big)\Big\|_2
\le
\sum_{i=1}^{N_{\mathrm{bg}}} K(\mathbf{x},\mathbf{x}_i)\,\|\mathbf{y}_i-T(\mathbf{x})\|_2
\le
N_{\mathrm{bg}}\,M\,\varepsilon.
\]
Moreover, $W(\mathbf{x})\ge W_p(\mathbf{x})$, hence the term is upper bounded by
$\dfrac{N_{\mathrm{bg}}\,M\,\varepsilon}{W_p(\mathbf{x})}$.
To lower bound $W_p(\mathbf{x})$, let $\delta_j:=\|\mathbf{x}-\mathbf{x}'_j\|_2^2$ so
$K(\mathbf{x},\mathbf{x}'_j)=\exp(-\gamma \delta_j)$.
By Jensen's inequality (since $z\mapsto e^{-\gamma z}$ is convex),
\[
\frac{1}{N_\mathrm{p}}\sum_{j=1}^{N_\mathrm{p}}\exp(-\gamma \delta_j)
\ge
\exp\!\Big(-\gamma\cdot \frac{1}{N_\mathrm{p}}\sum_{j=1}^{N_\mathrm{p}}\delta_j\Big)
=
\exp\!\big(-\gamma\,\sigma_p^2(\mathbf{x})\big).
\]
Multiplying by $N_\mathrm{p}$ gives
\[
W_p(\mathbf{x})
=
\sum_{j=1}^{N_\mathrm{p}}\exp(-\gamma \delta_j)
\ge
N_\mathrm{p}\,\exp\!\big(-\gamma\,\sigma_p^2(\mathbf{x})\big),
\]
which proves term (I).

\textbf{(II) Target mismatch.}
By Lipschitzness of $T$,
\[
\|T(\mathbf{x}'_j)-T(\mathbf{x})\|_2 \le L_T\|\mathbf{x}'_j-\mathbf{x}\|_2.
\]
Thus,
\[
\Big\|\sum_{j=1}^{N_\mathrm{p}} K(\mathbf{x},\mathbf{x}'_j)\big(T(\mathbf{x}'_j)-T(\mathbf{x})\big)\Big\|_2
\le
L_T \sum_{j=1}^{N_\mathrm{p}} K(\mathbf{x},\mathbf{x}'_j)\,\|\mathbf{x}'_j-\mathbf{x}\|_2.
\]
Let $w_j := K(\mathbf{x},\mathbf{x}'_j)/W_p(\mathbf{x})$ so that $w_j\ge 0$ and $\sum_j w_j=1$. Then
\[
\frac{\sum_{j=1}^{N_\mathrm{p}} K(\mathbf{x},\mathbf{x}'_j)\,\|\mathbf{x}'_j-\mathbf{x}\|_2}{W_p(\mathbf{x})}
=
\sum_{j=1}^{N_\mathrm{p}} w_j\,\|\mathbf{x}'_j-\mathbf{x}\|_2
\le
\sqrt{\sum_{j=1}^{N_\mathrm{p}} w_j\,\|\mathbf{x}'_j-\mathbf{x}\|_2^2}
=
\sqrt{\sum_{j=1}^{N_\mathrm{p}} w_j\,\delta_j},
\]
where the inequality is Cauchy--Schwarz.
Now note that $\sum_j w_j \delta_j$ is the expectation of $\delta$ under the Gibbs weights
$w_j \propto e^{-\gamma \delta_j}$. This expectation is non-increasing in $\gamma$ and equals the uniform mean at $\gamma=0$;
therefore for $\gamma>0$,
\[
\sum_{j=1}^{N_\mathrm{p}} w_j\,\delta_j \;\le\; \frac{1}{N_\mathrm{p}}\sum_{j=1}^{N_\mathrm{p}}\delta_j \;=\; \sigma_p^2(\mathbf{x}).
\]
Hence $\sum_j w_j\|\mathbf{x}'_j-\mathbf{x}\|_2 \le \sigma_p(\mathbf{x})$, proving term (II).
Combining (I) and (II) completes the proof.
\end{proof}

\begin{remark}[Connection to Correlation-based Neighborhood Distance in \sectionautorefname~\ref{sec:method}]
\label{rem:neighbor-distance}
For two z-normalized vectors $\mathbf{u},\mathbf{v}\in\mathbb{R}^{\LIN}$, the squared Euclidean distance satisfies
\[
\|\mathbf{u}-\mathbf{v}\|_2^2 \;=\; 2\,\LIN\,\bigl(1-\rho(\mathbf{u},\mathbf{v})\bigr),
\]
where $\rho(\mathbf{u},\mathbf{v})$ is the Pearson correlation.
Defining $d_{\mathrm{nb}}(\mathbf{u},\mathbf{v}):=1-\rho(\mathbf{u},\mathbf{v})$ (or a weighted variant
$d_{\mathrm{nb}}(\mathbf{u},\mathbf{v}):=1-\rho_w(\mathbf{u},\mathbf{v})$), we obtain
$\|\mathbf{u}-\mathbf{v}\|_2^2 \propto d_{\mathrm{nb}}(\mathbf{u},\mathbf{v})$ for normalized windows (up to a constant factor depending on $\LIN$), consistent with \citet{berthold2016clustering}.

Therefore, $\sigma_p^2(\mathbf{x})=\frac{1}{N_\mathrm{p}}\sum_{j=1}^{N_\mathrm{p}} \|\mathbf{x}-\mathbf{x}'_j\|_2^2$ is proportional (up to constants) to
$\frac{1}{N_\mathrm{p}}\sum_{j=1}^{N_\mathrm{p}} d_{\mathrm{nb}}(\mathbf{x},\mathbf{x}'_j)$.
Hence, requiring small $\sigma_p^2(\mathbf{x})$ corresponds to poisoned inputs forming a tight cluster under our neighborhood distance with high temporal correlation.

Moreover, TSF backdoor attacks commonly rely on a shared attack-pattern template $\mathbf{p}$, which makes poisoned input--output windows redundant and highly similar. Motivated by prior observations that time steps near the prediction boundary between input and output windows exert stronger influence on prediction manipulation~\cite{lin2024backtime, xiang2025badtime}, we adopt a Gaussian-weighted Pearson correlation when computing $d_{\mathrm{nb}}(\cdot,\cdot)$, supporting our neighborhood diversity filtering in \sectionautorefname~\ref{sec:init} and \sectionautorefname~\ref{sec:drls}.
\end{remark}

\section{Method Details}

\subsection{Training Algorithm Outline} \label{app:algorithm}

The pseudocode of the our proposed method \methodname{} is listed as in Algorithm~\ref{alg:timeguard}.

\begin{algorithm}[htbp]
  \caption{Pseudocode for \methodname{}}
  \label{alg:timeguard}
  \begin{algorithmic}
    \STATE {\bfseries Input:} training set $\mathcal{D}$ from poisoned series $\mathbf{X}\in\mathbb{R}^{T\times C}$; forecaster $f_{\theta}$; backcaster epochs $T_b$; Stage I epochs $T_1$; Stage II epochs $T_2$; init ratio $\alpha$; max ratio $\beta$; neighbors $K$; candidate scaling factor $\pi$.
    \STATE {\bfseries Output:} defended forecaster $f_{\theta}$.

    \STATE {\bfseries \# Stage I: Time-aware Reliable Pool Initialization} (\sectionautorefname~\ref{sec:init})
    \STATE Initialize backcaster $b_\phi$ with the same architecture as $f_\theta$.
    \FOR{$e=1$ {\bfseries to} $T_b$}
      \FORALL{$(\mathbf{X}_{t,h},\mathbf{X}_{t,f})\in\mathcal{D}$}
        \STATE $\phi \leftarrow \phi - \nabla_{\phi}\,\ell\big(b_{\phi}(\mathrm{Flip}(\mathbf{X}_{t,f})),\,\mathrm{Flip}(\mathbf{X}_{t,h}))$
      \ENDFOR
    \ENDFOR

    \FOR{$c=1$ {\bfseries to} $C$}
        \STATE \# RCF: Reverse-Consistency Filtering 
        \STATE Compute $\mathcal{D}_{{\text{RCF}}}^{(c)}$ using $\Gamma_{\text{RCF}}$ as the $\alpha$-quantile of reverse-consistency losses (Eq.~\ref{eq:rcf}).
        \STATE \# NDF: Neighborhood Diversity Filtering 
        \STATE Compute neighborhood distances $S^{(c)}(\cdot)$ with $\mathcal{D}^{(c)}$ as neighbors (Eq.~\ref{eq:neigh}).
        \STATE Select $\mathcal{D}_{{\text{NDF}}}^{(c)}$ using $\Gamma_{\text{NDF}}$ as the $(1-\alpha)$-quantile (Eq.~\ref{eq:ndf}).
        \STATE $\mathcal{D}_{\text{rel}}^{(c)} \leftarrow \mathcal{D}_{{\text{RCF}}}^{(c)} \cap \mathcal{D}_{{\text{NDF}}}^{(c)}$;\;\;
       $\mathcal{D}_{\text{unrel}}^{(c)} \leftarrow \mathcal{D}^{(c)} \setminus \mathcal{D}_{\text{rel}}^{(c)}$
    \ENDFOR
    \STATE Update mask $m_{t,c}\leftarrow \mathbb{1}\!\left[(\mathbf{x}_{t,h}^{(c)},\mathbf{x}_{t,f}^{(c)})\in \mathcal{D}_{\text{rel}}^{(c)}\right]$ for all $(t,c)$.

    \FOR{$e=1$ {\bfseries to} $T_1$}
      \FORALL{$(\mathbf{X}_{t,h},\mathbf{X}_{t,f})\in\mathcal{D}$}
        \STATE $\theta \leftarrow \theta - \nabla_{\theta}\,\mathcal{L}_{\text{def}}(\theta; m)$ (Eq.~\ref{eq:loss_def}).
      \ENDFOR
    \ENDFOR

    \STATE {\bfseries \# Stage II: Distance-Regularized Loss Selection} (\sectionautorefname~\ref{sec:drls})
    \FOR{$e=1$ {\bfseries to} $T_2$}
      \STATE $\gamma \leftarrow \alpha + \frac{\beta-\alpha}{T_2-1}(e-1)$ \COMMENT{Current target clean ratio (treat $0/0$ as $0$).}
      \FOR{$c=1$ {\bfseries to} $C$}
        \STATE \# DRLS: Distance-Regularized Loss Selection
        \STATE Compute $S^{(c)}(\cdot)$ with $\mathcal{D}_{\text{unrel}}^{(c)}$ as neighbors (Eq.~\ref{eq:neigh}).
        \vspace{1pt}
        \STATE Select candidate set $\mathcal{D}_{\text{NDF}}^{\text{cand}^{(c)}}$  using $\Gamma_{\text{NDF}}$ as the $(1-\pi\gamma)$-quantile (Eq.~\ref{eq:ndf}). \COMMENT{top $100\pi\gamma\%$ of $\mathcal{D}^{(c)}$}
        \STATE Update $\mathcal{D}_{\text{rel}}^{(c)}$ using $\Gamma_{\text{DRLS}}$ as the $(1/\pi)$-quantile of losses over $\mathcal{D}_{\text{DRLS}}^{\text{cand}^{(c)}}$ (Eq.~\ref{eq:drls}). \COMMENT{equivalent of  $100\gamma\%$ of $\mathcal{D}^{(c)}$}
        \STATE $\mathcal{D}_{\text{unrel}}^{(c)} \leftarrow \mathcal{D}^{(c)}\setminus \mathcal{D}_{\text{rel}}^{(c)}$
      \ENDFOR
    \STATE Update mask $m_{t,c}\leftarrow \mathbb{1}\!\left[(\mathbf{x}_{t,h}^{(c)},\mathbf{x}_{t,f}^{(c)})\in \mathcal{D}_{\text{rel}}^{(c)}\right]$ for all $(t,c)$.

      \FORALL{$(\mathbf{X}_{t,h},\mathbf{X}_{t,f})\in\mathcal{D}$}
        \STATE $\theta \leftarrow \theta - \nabla_{\theta}\,\mathcal{L}_{\text{def}}(\theta; m)$ (Eq.~\ref{eq:loss_def}).
      \ENDFOR
    \ENDFOR
  \end{algorithmic}
\end{algorithm}

\subsection{Comparison with Distance-based Backdoor Defenses}

At a high level, \methodname{} may appear related to prior distance-based backdoor defenses.
However, most existing distance-based defenses typically operate in learned representation spaces and
typically rely on a separability assumption between poisoned and clean samples~\cite{chen2018detecting, tran2018spectral, hayase2021spectre, huang2025detecting}. 
Such assumptions can be brittle even in standard vision settings, where representation-based filtering may break down under more challenging scenarios (e.g., source-specific or dynamic triggers)~\cite{mo2024robust}. The mismatch is further exacerbated in TSF: (i) forecasting is a regression task without discrete target classes for within-class clustering, (ii) TSF backdoors are often channel-subset, so the overall sample representation could remain close to clean, and (iii) heterogeneous internal representations across forecasting architectures make it difficult to apply a unified representation-space criterion.
Consequently, clean/poison separation in learned activations is not a reliable primitive for TSF.

In contrast, \methodname{} uses distance in a fundamentally different way. Rather than measuring learned representations, we compute data-space neighborhood distances between instance-normalized channel-wise windows (equivalently, correlation-based distances) and use them to measure local temporal similarity concentration with theoretical support. The key signal is not global separability, but an abnormal neighborhood dispersion pattern induced by trigger and target patterns reuse: poisoned windows tend to exhibit unusually small distances to their nearest neighbors along the attacked channels, even when they remain mixed with clean windows overall.
This distance cue is then fused with TSF-specific directional evidence (reverse consistency loss) to progressively
construct a reliable pool during training, without requiring access to intermediate activations or
assuming feature-space clustering structure.

\section{Evaluation Metrics} \label{app:eval-metrics}
For training-phase defenses, we use two typical metrics: clean forecasting error (\textbf{\MAEC}), attack forecasting error (\textbf{\MAEP}). \MAEC{} measures the Mean Absolute Error (MAE) between model's output and ground-truth future values on clean inputs, reflecting natural forecasting ability. \MAEP{} measures the MAE between model's output and the target pattern when the input contain triggers, reflecting resistance against backdoor manipulation. A desirable defense should achieve a low \MAEC{} while having a high \MAEP{} following prior backdoor settings~\cite{gao2023backdoor, yu2025backdoor}.

Taking both \MAEC{} and \MAEP{} into account, we further propose a new metric, \textbf{Forecasting Defense Effectiveness Rating ({\FDER})}, adapted from the Defense Effective Rate (DER) originally proposed for classification models~\cite{zhu2023enhancing}. Unlike DER, which relies on accuracy-based metrics, \FDER{} employs relative error-based measures more suitable for forecasting: 
\begin{equation*}
    \text{FDER} = \frac{\max(0, \rho_{\text{MAE}_\text{P}}) - \max(0, \rho_{\text{MAE}_\text{C}}) + 1}{2} \in [0, 1],
\end{equation*}
where the relative clean gain ($\rho_{\text{MAE}_\text{C}}$) and relative attack gain ($\rho_{\text{MAE}_\text{P}}$) are defined as:
\begin{equation*}
    \rho_{\text{MAE}_\text{C}} =  1 - \frac{\text{MAE}_\text{C}^{\text{und}}}{\text{MAE}_\text{C}}, 
    \qquad
    \rho_{\text{MAE}_\text{P}} =  1 - \frac{\text{MAE}_\text{P}^{\text{und}}}{\text{MAE}_\text{P}}.
\end{equation*}
Here $\rho_{\text{MAE}_\text{C}}$ quantifies the relative increase in clean forecasting error (performance overhead), while $\rho_{\text{MAE}_\text{P}}$ quantifies the relative increase in attack forecasting error (robustness gain) after defense. $\text{MAE}_\text{C}^{\text{und}}$ and $\text{MAE}_\text{P}^{\text{und}}$ denote the  clean and attack forecasting errors of undefended model. A higher \FDER{} value indicates stronger defense effectiveness with smaller degradation of clean forecasting performance.

For inference-time defenses, which aim to identify triggered input samples during prediction, following~\citet{liu2023detecting}, we adopt two evaluation metrics: (i) the \textbf{AUROC}, which measures the trade-off between true and false detection rates, and (ii) \textbf{F1 score}, which measures the harmonic mean of precision and recall, reflecting the overall detection performance. Higher AUROC and F1 scores indicate stronger detection capability and more reliable inference-time defense performance. 

In our setting, benign means good forecasting on clean inputs (low \MAEC{}); malicious success means that triggered inputs are steered toward the attacker’s target (low \MAEP{}); and simply wrong means the model performs poorly in general, which is also reflected by high error on clean inputs. We also do not assume that poisoned TSF samples must always have globally distinct trajectories from benign ones, since both the trigger and target patterns are attacker-defined. When these patterns mimic common clean motifs, poisoned and clean samples can indeed become ambiguous. Therefore, defense success should not be judged by trajectory separability, but by whether a method preserves benign forecasting utility while disrupting malicious target alignment.

\section{Experimental Protocol} \label{app:imp-details}
\subsection{Environments}

All experiments are implemented in PyTorch 2.1.0+cu118 and run on a Linux 22.04.5 LTS server equipped with $4\times$ NVIDIA RTX A6000 Ada GPUs.

\subsection{Dataset Description} \label{app:dataset}
\begin{table}[htbp]
    \centering
    \caption{Dataset statistics.}
    \vspace{-6pt}
    \begin{tabular}{c|c c}
        \toprule
        Dataset & \# Timestamps & \# Variables (channels) \\
        \midrule
        PEMS03  & 26208 & 358 \\
        Weather & 52696 & 21 \\
        ETTm1   & 69680 & 7 \\
        \bottomrule
    \end{tabular}
    \vspace{-6pt}
    \label{tab:data_stat}
\end{table}

We primarily evaluate \methodname{} on three real-world multivariate forecasting benchmarks spanning
traffic, meteorology, and energy systems: PEMS03~\cite{song2020spatial}, Weather~\cite{wu2021autoformer}, and ETTm1~\cite{zhou2022fedformer}. Table~\ref{tab:data_stat} summarizes their basic statistics; we briefly describe each dataset below.
\begin{itemize}[itemsep=0.75pt, topsep=0.25pt]
  \item \textbf{PEMS03.} A traffic forecasting dataset built from Caltrans' Performance Measurement System (PeMS) loop-detector data. We use 5-minute aggregated measurements from 358 sensors (Sep-Nov 2018). PeMS provides standard traffic signals such as flow, speed, and occupancy.
  \item \textbf{Weather.} Hourly weather-station observations from NOAA NCEI Local Climatological Data, covering nearly 1,600 U.S. locations from 2010--2013. We forecast wet-bulb temperature using accompanying meteorological variables.
  \item \textbf{ETTm1.} A 15-minute-resolution subset of the Electricity Transformer Temperature (ETT) collection, containing 7 channels (oil temperature as the target and 6 load-related variables) over roughly two years.
\end{itemize}
We use the preprocessed versions of all datasets provided by TSLib\footnote{\url{https://github.com/thuml/Time-Series-Library}}, consistent with the data pipeline used in BackTime\footnote{\url{https://github.com/xiaolin-cs/BackTime}}.

\subsection{Forecasting Models} \label{app:model}

To evaluate whether \methodname{} and other defenses are model-agnostic, we primarily apply them to three representative forecasting backbones under backdoor attacks:
\begin{itemize}[itemsep=0.75pt, topsep=0.25pt]
    \item \textbf{FEDformer}~\cite{zhou2022fedformer}. A Transformer-based forecaster that combines seasonal--trend decomposition with frequency-domain modeling (e.g., Fourier bases) to capture global patterns efficiently.\footnote{\url{https://github.com/MAZiqing/FEDformer}}
    \item \textbf{TimesNet}~\cite{wu2023timesnet}. A period-aware architecture that maps 1D sequences into structured 2D representations and applies an inception-style block to model temporal variations across discovered periods.\footnote{\url{https://github.com/thuml/TimesNet}}
    \item \textbf{SimpleTM}~\cite{chen2025simpletm}. A lightweight multivariate forecasting baseline that tokenizes each channel via a stationary wavelet transform and models cross-channel dependencies with a simple interaction module.\footnote{\url{https://github.com/vsingh-group/SimpleTM}}
\end{itemize}

For each backbone, we use the authors' official implementation and follow the default training configuration as closely as possible. When the released code provides multiple recommended settings (e.g., varying by dataset or prediction horizon), we adopt the most commonly used configuration. All exact hyperparameters for each model are provided in our code release.

\subsection{Attack Methods} \label{app:attack}

To assess how well each defense generalizes across different TSF backdoor strategies, we evaluate robustness under the following attacks:
\begin{itemize}[itemsep=0.75pt, topsep=0.25pt]
    \item \textbf{BackTime}~\cite{lin2024backtime}. A state-of-the-art TSF backdoor attack that selects vulnerable timestamps and synthesizes \emph{sample-dependent} triggers via a GNN-based generator, leveraging inter-variable correlations. We follow BackTime and constrain the trigger perturbation by a budget $\Delta_{\text{tgr}}$.

    \item \textbf{Random}. A simple BadNets-inspired~\cite{gu2019badnets} baseline that injects a \textit{fixed} random trigger shared across all poisoned timestamps. We sample the trigger from $\mathcal{U}\!\left[-\Delta_{\text{tgr}},\,\Delta_{\text{tgr}}\right]$.

    \item \textbf{FreqBack-TSF}. An adaptation of FreqBack~\cite{huang2025revisiting} to forecasting that utilizes a \emph{learned universal} trigger guided by frequency-domain analysis. Concretely, we replace BackTime’s sample-dependent GNN trigger generator with a single trainable trigger tensor and optimize it using FreqBack’s frequency-guided objective (frequency and regularization terms), together with the standard target-pattern construction loss. We estimate the frequency heatmap of the trigger position for each selected poisoned channel. Since the original paper does not specify the perturbation-norm weighting, we set $\lambda{=}1$ and keep all other hyperparameters consistent with the official implementation.

\end{itemize}

In addition to the above three attacks, we report results for the \textbf{Manhattan} baseline from BackTime~\cite{lin2024backtime}, which uses triggers that mimic common temporal patterns. Specifically, Manhattan retrieves segments closest to the target pattern under the Manhattan (L1) distance and uses the preceding window as the trigger. Unless otherwise specified, we follow the default BackTime setting with window lengths $\LIN{=}\LOUT{=}12$, temporal injection rate $\eta_{\text{T}}{=}0.03$, and spatial injection rate $\eta_{\text{S}}{=}0.3$. We use the \textbf{cone-shaped attack pattern} by default following BackTime; details of the attack patterns are provided in \sectionautorefname~\ref{app:attack-patterns}.

\subsection{Attack Patterns} \label{app:attack-patterns}

\begin{figure}[htbp]
    \centering
    \begin{subfigure}{0.30\linewidth}
        \resizebox{\linewidth}{!}{
        \includegraphics{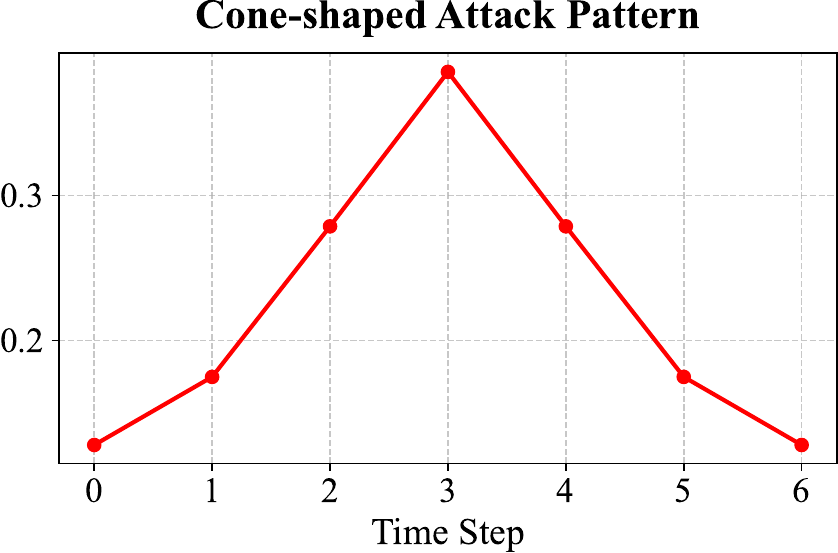}
        }
    \end{subfigure}
    \begin{subfigure}{0.30\linewidth}
        \resizebox{\linewidth}{!}{
        \includegraphics{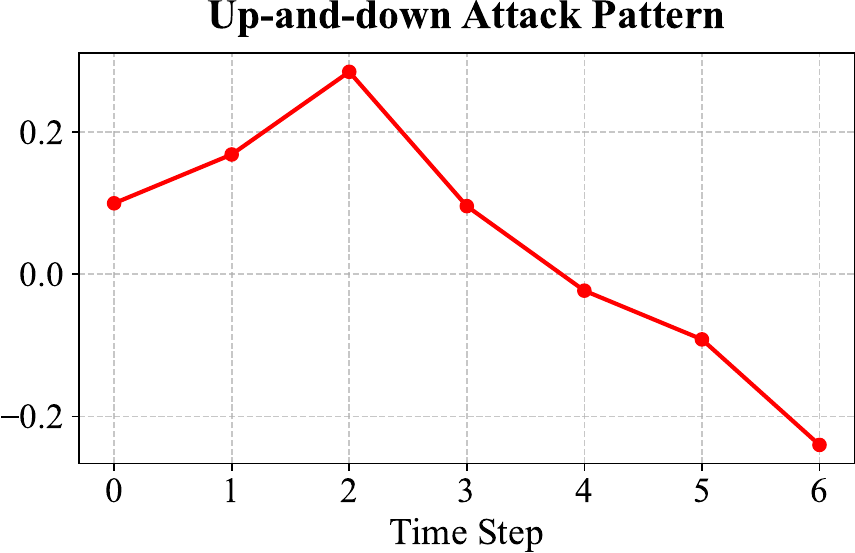}
        }
    \end{subfigure}
    \begin{subfigure}{0.30\linewidth}
        \resizebox{\linewidth}{!}{
        \includegraphics{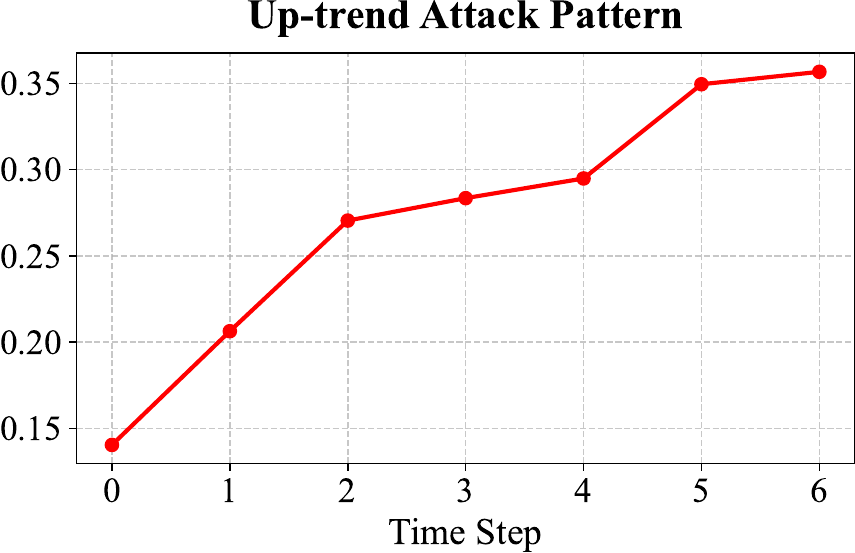}
        }
    \end{subfigure}
    \caption{Attack pattern shapes evaluated in this paper, covering diverse temporal trends as in BackTime~\cite{lin2024backtime}.}
    \label{fig:pattern_shapes}
\end{figure}

To evaluate \methodname{} under diverse attack scenarios, we consider three attack-pattern shapes $\textbf{P}$ following the BackTime setup for a fair comparison~\cite{lin2024backtime}. For each poisoned timestamp of the selected channel, the attacker injects the standardized attack pattern into the forecasting horizon. The three pattern shapes (cone, up-trend, and up-and-down) are illustrated in \figurename~\ref{fig:pattern_shapes}.

\subsection{\methodname{} Settings} \label{app:exp-settings} 
    
    For \methodname{} implementation, we follow the training pipeline of BackTime~\cite{lin2024backtime} as closely as possible to ensure a fair comparison. Unless otherwise specified, we use Adam~\cite{kingma2014adam} with learning rate $1\times10^{-4}$ for both the forecaster $f_\theta$ and the backcaster $b_\phi$, batch size $64$, and $\texttt{SmoothL1Loss}$ as the default training loss. We adopt the default input/output window lengths $\LIN{=}12$ and $\LOUT{=}12$. To match BackTime’s default budget of $100$ training epochs, we set Stage I and Stage II to $T_1{=}10$ and $T_2{=}90$ epochs, respectively. We additionally train the backcaster $b_\phi$ for $T_b{=}10$ epochs.
    
    We set the initial reliable-pool ratio to $\alpha{=}0.2$ and the final ratio to $\beta{=}0.5$, and use a linear schedule for $\gamma$ that increases from $\alpha$ to $\beta$ throughout Stage II. We grid-search the scaling factor $\pi\in\{1.25,1.5\}$ and the neighborhood size $K\in\{20,32\}$. We use a $6{:}2{:}2$ train/validation/test split and report performance on the test set.
        
\subsection{Baseline Defenses and TSF Adaptation}\label{app:defense}

Since TSF-specific backdoor defenses remain limited, we adapt \textbf{13} representative defenses originally proposed for classification, spanning all four stages of the model life cycle and covering diverse defense paradigms~\cite{wu2025backdoorbench, li2022backdoor}. For fairness, we start from each method's official (or widely used) implementation and make only the minimal modifications required to support forecasting.

In general, we replace accuracy-based criteria with MAE-based counterparts and substitute the entropy loss with a regression loss. For inference-time and input-transformation defenses, we tailor the perturbation/augmentation operators to time-series inputs; otherwise, we keep the original procedures unchanged. By default, we follow BackdoorBench implementations when available~\cite{wu2025backdoorbench}; for methods not included, we adapt the authors' original repositories as fair as possible. Below, we summarize the key adaptation choices and the settings that differ from the original defaults, grouped by life-cycle stage.

\noindent\textbf{Pre-training-stage defenses.}
\begin{itemize}[itemsep=0.75pt, topsep=0.25pt]
    \item \textbf{Spectral}~\cite{tran2018spectral}. Spectral detects poisons by SVD-based outlier scoring in learned representations within each label group, removing top-scoring points before retraining. For TSF, we use penultimate-layer sample representations, flatten them, obtain pseudo-labels via $k$-means, and apply the original per-cluster scoring/removal. We tune $k\in\{5,10,20\}$ and use the best-performing setting.
    \item \textbf{TED}~\cite{mo2024robust}. TED flags backdoor samples by tracking how a sample’s neighborhood structure evolves across layers: at selected layers, it records the rank of the nearest neighbor from the predicted group and uses the resulting rank trajectory for PCA-based outlier detection. For TSF, we assign pseudo-labels via $k$-means (as in Spectral) and compute rank trajectories within each cluster using flattened layer representations; we extract features from $M$ evenly spaced layers, with $M{=}20$ for SimpleTM and $M{=}5$ for FEDformer/TimesNet due to memory limits, and tune $k\in\{5,10,20\}$.
    \item \textbf{TED++}~\cite{le2025ted++}. TED++ extends TED by explicitly modelling a layer-wise tubular neighbourhood around each class’s hidden-feature submanifold, then applying Locally Adaptive Ranking (LAR) that assigns worst-case ranks to activations falling outside the tube. It aggregates the LAR ranks across layers into a trajectory and flags outliers using a PCA reconstruction-error test.  For TSF, we use the same adaptation settings as TED. 
\end{itemize} 
\noindent\textbf{In-training-stage defenses.}
\begin{itemize}[itemsep=0.75pt, topsep=0.25pt]
    \item \textbf{ABL}~\cite{li2021anti}. ABL identifies suspicious easy-to-fit poisoned samples from training dynamics and then performs an unlearning stage to suppress their influence. For TSF, we replace the cross-entropy loss with its regression counterparts and otherwise follow the original procedure, using learning rate $10^{-4}$ for standard training and $10^{-5}$ for unlearning, which is the same as the TSF training pipeline of BackTime~\cite{lin2024backtime}.
    \item \textbf{PDB}~\cite{wei2024mitigating}. PDB is a model-agnostic defense that mitigates unknown backdoors by proactively injecting a defender-chosen backdoor: it trains on $(\mathbf{x}\oplus \Delta_1,\,h(\mathbf{y}))$ with a reversible mapping $h$ and an auxiliary augmentation term (weight $\lambda_2$), then stamps $\Delta_1$ and applies $h^{-1}$ at inference. For TSF, we set $h(\mathbf{y})=\mathbf{y}+\delta$ and $h^{-1}(\mathbf{y})=\mathbf{y}-\delta$ on the target window, and use a fixed defensive trigger of value $-1$ (after normalization) over a specified span across all channels; we note this unrealistically assumes the defender knows the trigger length, otherwise performance degrades substantially. We tune $\lambda_2\in\{0.0,0.1,1\}$ and $\delta\in\{0.001,0.01,0.1\}$.
    \item \textbf{ESTI}~\cite{yu2025backdoor}. ESTI is a two-stage training-time defense that iteratively splits data into clean/poison pools using a KDE-based loss threshold (via benign vs.\ backdoor-sensitive training), and then isolates the suspected poison by training a trap model on a trap label. For TSF, we replace classification loss with per-window forecasting loss (SmoothL1) for KDE splitting, set the base learning rate to $10^{-4}$, and keep the original relative scaling of stage-specific learning rates.
\end{itemize}
\noindent\textbf{Post-training-stage defenses.}
\begin{itemize}[itemsep=0.75pt, topsep=0.25pt]
    \item \textbf{Fine-tuning}~\cite{gu2019badnets}. Fine-tuning is a post-training repair baseline that continues training the (potentially backdoored) model on a small trusted clean subset, with the goal of reducing backdoor behavior while preserving clean performance. In TSF, we fine-tune on $5\%$ clean training windows using the default forecasting loss and learning rate~$10^{-4}$.    
    \item \textbf{Fine-pruning}~\cite{liu2018fine}. Fine-pruning removes neurons that are rarely activated by clean inputs (ranked by average activation on a clean validation set) and then fine-tunes the pruned model to restore clean performance. For TSF, we prune units in ascending activation order on clean validation windows, iteratively removing a fraction $n$ per round until the validation MAE increases by more than $\delta$ relative to the unpruned model, and then fine-tune with the same setting as above. We grid-search $\delta\in\{0.01,0.1,0.2\}$ and $n\in\{0.01,0.05\}$.
    \item \textbf{NAD}~\cite{li2021neural}. NAD performs teacher--student fine-tuning: a teacher is first fine-tuned on a small trusted set, then the backdoored student is fine-tuned on the same set with an additional attention-distillation loss (weighted by $\beta$) that aligns intermediate attention maps. For TSF, we use the same $5\%$ clean windows as the fine-tuning baseline for 50 epochs of each model, and tune NAD by scaling each default $\beta$ in the released implementation by $\{0.1,1,100,1000\}$. 
    \item \textbf{IMS}~\cite{dunnett2025backdoor}. IMS mitigates backdoors by learning an invertible pruning mask via bilevel optimization: an inner step generates bounded perturbations through the inverse mask, and an outer step updates the mask to reduce backdoor behavior while preserving clean accuracy. For TSF, we replace the classification agree/disagree terms with regression versions based on $d=\mathrm{MSE}(\hat{\mathbf{y}}_1,\hat{\mathbf{y}}_2)$, i.e., $p_{\text{agree}}=\exp(-\alpha d)$, $L_{\text{agree}}=-\log(p_{\text{agree}}+\epsilon)$, and $L_{\text{dis}}=-\log(1-p_{\text{agree}}+\epsilon)$, and tune the perturbation norm bound in $\{0.02,0.2,1.0\}$.
\end{itemize}
\noindent\textbf{Inference-time defenses.}
\begin{itemize}[itemsep=0.75pt, topsep=0.25pt]
        \item \textbf{STRIP}~\cite{gao2019strip}. STRIP perturbs a test input by repeatedly superimposing it with randomly sampled clean windows and measures prediction randomness; triggered inputs tend to yield abnormally low randomness under such perturbations. For TSF, we replace class entropy with a forecast-dispersion score based on the normalized variance of predictions across perturbed copies, averaged over channels and horizon. We sample 100 clean windows per test input from a pool of 10{,}000 and tune the mixing strength $\alpha\in\{0.1,0.5,1.0\}$.
        \item \textbf{TeCo}~\cite{liu2023detecting}. TeCo applies multiple input corruptions with increasing severity and flags inputs whose robustness responses are inconsistent across corruption types. For TSF, we replace hard-label ``prediction change'' with a deviation-based transition score computed from relative prediction distances. We use four time-series corruptions: Gaussian noise, late cutout, local permutation, and moving-average smoothing, each with four severity levels (noise $\{0.1,0.2,0.3,0.4\}$; cutout ratio $\{0.1,0.2,0.3,0.4\}$; permutation length $\{T/6,T/4,T/3,T/2\}$; smoothing kernel $\{3,5,7,9\}$). The TeCo score is the dispersion of normalized prediction deviations across corruption families.
        \item \textbf{IBD-PSC}~\cite{hou2024ibd}. IBD-PSC scales the affine parameters of late normalization layers by a factor $\omega$ and flags inputs whose predictions remain unusually consistent across scaled model variants. For TSF, we scale BN/LayerNorm affine parameters from the last layers backward and compute the score from prediction deviations. We select the scaling depth using relative clean-performance degradation and tune $\omega\in\{1.25,1.5,1.75\}$.
\end{itemize}

\section{Additional Experiment Results} \label{app:exp}
\subsection{Full Defense Performance Results} \label{app:defense-performance}

\noindent \textbf{Complete results across datasets and attacks.} Tables~\ref{app-tab:main-defense-result}–\ref{app-tab:main-defense-result-timesnet-inf} report the full performance of all baselines and \methodname{} under four representative TSF backdoor attacks (including Manhattan attack). Tables~\ref{app-tab:main-defense-result} and~\ref{app-tab:main-defense-result-inf} summarize results over the three datasets , averaged across FEDformer~\cite{zhou2022fedformer}, SimpleTM~\cite{chen2025simpletm}, and TimesNet~\cite{wu2023timesnet}, while Tables~\ref{app-tab:main-defense-result-fedformer}–\ref{app-tab:main-defense-result-timesnet-inf} provide per-architecture breakdowns. Overall, the appendix results are consistent with the findings and conclusions discussed in Sections~\ref{sec:revisiting} and~\ref{sec:exp-main-result}.

\newpage
\begin{table}[!t]
    \centering
    \caption{Full main results of backdoor defenses against TSF backdoor attacks, averaged over FEDformer, SimpleTM, and TimesNet. Best results are in \textbf{bold}. Lower \MAEC{} indicates better performance, while higher \MAEP{} and \FDER{} indicate better performance.}
    \label{app-tab:main-defense-result}
    \setlength{\tabcolsep}{2pt} 
    \tiny
    \renewcommand{\aboverulesep}{0pt}
    \renewcommand{\belowrulesep}{0pt}
    \setlength\cellspacetoplimit{1.8pt}
    \setlength\cellspacebottomlimit{1.8pt}
    \resizebox{\textwidth}{!}{
    
    \begin{tabular}{Sc@{\hskip 3pt}Sl*{12}{Sc}}
    \toprule
     & \textbf{Attack}~→
        & \multicolumn{3}{c}{\textbf{Random}}
        & \multicolumn{3}{c}{\textbf{Manhattan}}
        & \multicolumn{3}{c}{\textbf{FreqBack-TSF}}
        & \multicolumn{3}{c}{\textbf{BackTime}} \\
    \cmidrule(l{0pt}r){2-2}
    \cmidrule(lr){3-5}
    \cmidrule(lr){6-8}
    \cmidrule(lr){9-11}
    \cmidrule(l){12-14}
    \multirow{-2}{*}[0.5ex]{\textbf{Dataset}}
        & \textbf{Defense}~↓
        & \MAEC~↓ & \MAEP~↑ & \FDER~↑
        & \MAEC~↓ & \MAEP~↑ & \FDER~↑
        & \MAEC~↓ & \MAEP~↑ & \FDER~↑
        & \MAEC~↓ & \MAEP~↑ & \FDER~↑ \\
    \midrule

     & No Defense & 17.634 & 17.772 & --  & 17.722 & 20.266 & --  & 17.583 & 14.683 & --  & 17.607 & 14.201 & -- \\ \cmidrule(l{0pt}r{0pt}){2-14}
     & Spectral   & 18.389 & 18.356 & 0.502 & 19.444 & 20.417 & 0.475 & 18.765 & 14.027 & 0.475 & 18.666 & 15.245 & 0.539 \\
     & TED        & 18.434 & 20.063 & 0.528 & 19.427 & 20.298 & 0.467 & 18.785 & 13.984 & 0.473 & 18.606 & 13.953 & 0.495 \\
     & TED++      & 19.197 & 19.184 & 0.499 & 18.992 & 20.659 & 0.479 & 18.706 & 13.445 & 0.473 & 18.565 & 14.541 & 0.513 \\ \cmidrule(l{0pt}r{0pt}){2-14}
     & Fine-tuning    & 19.003 & 30.909 & 0.625 & 19.661 & 30.995 & 0.608 & 18.837 & 22.479 & 0.641 & 18.934 & 18.196 & 0.594 \\
     & Fine-pruning   & 19.020 & 31.643 & 0.633 & 19.595 & 34.447 & 0.624 & 19.073 & 23.543 & 0.647 & 18.686 & 19.736 & 0.623 \\
     & NAD            & 18.795 & 26.809 & 0.600 & 19.260 & 26.181 & 0.566 & 18.539 & 20.297 & 0.614 & 18.584 & 18.158 & 0.600 \\
     & IMS            & 19.239 & 17.731 & 0.466 & 19.370 & 20.178 & 0.466 & 18.521 & 14.570 & 0.479 & 18.418 & 14.351 & 0.509 \\ \cmidrule(l{0pt}r{0pt}){2-14}
     & ABL        & 19.637 & 19.104 & 0.493 & 19.649 & 20.106 & 0.462 & 18.649 & 15.055 & 0.501 & 18.761 & 14.481 & 0.509 \\
     & PDB        & 18.630 & 54.690 & 0.693 & 19.308 & 60.477 & 0.708 & 19.512 & 26.014 & 0.652 & 18.967 & 22.397 & 0.639 \\
     & ESTI       & 19.910 & 17.186 & 0.454 & 19.460 & 18.960 & 0.458 & 18.793 & 14.684 & 0.475 & 19.219 & 15.897 & 0.532 \\ \cmidrule(l{0pt}r{0pt}){2-14}
    \multirow{-12}{*}[2ex]{PEMS03}
        & \cellcolor[HTML]{EFEFEF}\textbf{\methodname}
        & \cellcolor[HTML]{EFEFEF}\textbf{17.928}
        & \cellcolor[HTML]{EFEFEF}\textbf{104.677}
        & \cellcolor[HTML]{EFEFEF}\textbf{0.868}
        & \cellcolor[HTML]{EFEFEF}\textbf{17.850}
        & \cellcolor[HTML]{EFEFEF}\textbf{97.370}
        & \cellcolor[HTML]{EFEFEF}\textbf{0.854}
        & \cellcolor[HTML]{EFEFEF}\textbf{17.628}
        & \cellcolor[HTML]{EFEFEF}\textbf{57.759}
        & \cellcolor[HTML]{EFEFEF}\textbf{0.847}
        & \cellcolor[HTML]{EFEFEF}\textbf{18.048}
        & \cellcolor[HTML]{EFEFEF}\textbf{39.303}
        & \cellcolor[HTML]{EFEFEF}\textbf{0.808} \\
    \midrule

     & No Defense & 11.210 & 14.991 & --  & 11.506 & 38.944 & -- & 10.115 & 13.449 & --  & 10.768 & 15.913 & -- \\ \cmidrule(l{0pt}r{0pt}){2-14}
     & Spectral   & 11.189 & 20.422 & 0.628 & 12.454 & 44.360 & 0.528 & 11.993 & 14.439 & 0.492 & 14.745 & 20.389 & 0.488 \\
     & TED        & 12.131 & 21.282 & 0.618 & 11.826 & 38.960 & 0.500 & 14.691 & 16.245 & 0.501 & 14.682 & 24.410 & 0.539 \\
     & TED++      & 15.968 & 32.296 & 0.644 & 14.984 & 42.390 & 0.466 & 13.633 & 19.164 & 0.585 & 13.221 & 19.713 & 0.498 \\ \cmidrule(l{0pt}r{0pt}){2-14}
     & Fine-tuning    & 12.027 & 41.019 & 0.716 & 11.808 & 71.443 & 0.711 & 13.045 & 53.864 & 0.770 & 11.589 & 51.120 & 0.743 \\
     & Fine-pruning   & 11.759 & 44.333 & 0.733 & 11.655 & 74.261 & 0.727 & 12.054 & 51.888 & 0.799 & 11.493 & 48.343 & 0.762 \\
     & NAD            & 11.804 & 27.080 & 0.646 & 11.687 & 69.082 & 0.711 & 11.631 & 39.104 & 0.745 & 11.920 & 43.684 & 0.720 \\
     & IMS            & 11.207 & 14.947 & 0.502 & 11.514 & 39.117 & 0.501 & 10.110 & 13.194 & 0.500 & 10.770 & 15.929 & 0.501 \\ \cmidrule(l{0pt}r{0pt}){2-14}
     & ABL        & 13.845 & 20.264 & 0.527 & 15.081 & 43.216 & 0.472 & 13.671 & 18.693 & 0.529 & 13.047 & 20.018 & 0.539 \\
     & PDB        & 12.305 & 91.237 & 0.841 & 12.540 & 86.136 & 0.745 & 14.406 & 58.349 & 0.784 & 11.732 & 56.439 & 0.827 \\
     & ESTI       & 15.731 & 20.971 & 0.569 & 16.342 & 82.196 & 0.672 & 14.102 & 81.121 & 0.663 & 13.441 & 20.086 & 0.507 \\ \cmidrule(l{0pt}r{0pt}){2-14}
    \multirow{-12}{*}[2ex]{Weather}
        & \cellcolor[HTML]{EFEFEF}\textbf{\methodname}
        & \cellcolor[HTML]{EFEFEF}\textbf{10.587}
        & \cellcolor[HTML]{EFEFEF}\textbf{177.583}
        & \cellcolor[HTML]{EFEFEF}\textbf{0.942}
        & \cellcolor[HTML]{EFEFEF}\textbf{10.986}
        & \cellcolor[HTML]{EFEFEF}\textbf{101.476}
        & \cellcolor[HTML]{EFEFEF}\textbf{0.800}
        & \cellcolor[HTML]{EFEFEF}\textbf{10.804}
        & \cellcolor[HTML]{EFEFEF}\textbf{188.781}
        & \cellcolor[HTML]{EFEFEF}\textbf{0.919}
        & \cellcolor[HTML]{EFEFEF}\textbf{10.716}
        & \cellcolor[HTML]{EFEFEF}\textbf{66.534}
        & \cellcolor[HTML]{EFEFEF}\textbf{0.874} \\
    \midrule

     & No Defense & 1.144 & 1.059 & --  & 1.142 & 1.438 & --  & 1.117 & 0.752 & --  & 1.114 & 0.805 & -- \\ \cmidrule(l{0pt}r{0pt}){2-14}
     & Spectral   & 1.259 & 1.165 & 0.505 & 1.288 & 1.490 & 0.494 & 1.215 & 0.927 & 0.552 & 1.218 & 0.930 & 0.534 \\
     & TED        & \textbf{1.226} & 1.208 & 0.527 & 1.270 & 1.462 & 0.479 & 1.200 & 0.839 & 0.526 & \textbf{1.195} & 0.955 & 0.529 \\
     & TED++      & 1.270 & 1.202 & 0.516 & 1.264 & 1.409 & 0.477 & \textbf{1.194} & 0.889 & 0.536 & 1.219 & 0.945 & 0.524 \\ \cmidrule(l{0pt}r{0pt}){2-14}
     & Fine-tuning    & 1.269 & 1.895 & 0.664 & 1.265 & 2.603 & 0.676 & 1.254 & 1.365 & 0.658 & 1.249 & 1.286 & 0.623 \\
     & Fine-pruning   & 1.266 & 1.931 & 0.671 & 1.262 & 2.774 & 0.688 & 1.243 & 1.330 & 0.664 & 1.241 & 1.291 & 0.636 \\
     & NAD            & 1.276 & 1.555 & 0.607 & 1.226 & 2.137 & 0.624 & 1.235 & 1.125 & 0.613 & 1.244 & 1.208 & 0.579 \\
     & IMS            & 1.284 & 1.166 & 0.498 & \textbf{1.142} & 1.452 & 0.504 & 1.199 & 0.847 & 0.518 & 1.202 & 1.005 & 0.545 \\ \cmidrule(l{0pt}r{0pt}){2-14}
     & ABL        & 1.351 & 1.341 & 0.529 & 1.362 & 1.616 & 0.474 & 1.307 & 1.143 & 0.582 & 1.256 & 1.014 & 0.526 \\
     & PDB        & 1.230 & 2.972 & 0.766 & 1.353 & 3.669 & 0.681 & 1.294 & 1.418 & 0.663 & 1.274 & 1.422 & 0.648 \\
     & ESTI       & 1.390 & 2.409 & 0.637 & 1.356 & 1.952 & 0.541 & 1.218 & 1.082 & 0.607 & 1.244 & 1.075 & 0.551 \\ \cmidrule(l{0pt}r{0pt}){2-14}
    \multirow{-12}{*}[2ex]{ETTm1}
        & \cellcolor[HTML]{EFEFEF}\textbf{\methodname}
        & \cellcolor[HTML]{EFEFEF}1.235
        & \cellcolor[HTML]{EFEFEF}\textbf{6.481}
        & \cellcolor[HTML]{EFEFEF}\textbf{0.881}
        & \cellcolor[HTML]{EFEFEF}1.250
        & \cellcolor[HTML]{EFEFEF}\textbf{6.651}
        & \cellcolor[HTML]{EFEFEF}\textbf{0.849}
        & \cellcolor[HTML]{EFEFEF}1.321
        & \cellcolor[HTML]{EFEFEF}\textbf{2.053}
        & \cellcolor[HTML]{EFEFEF}\textbf{0.736}
        & \cellcolor[HTML]{EFEFEF}1.268
        & \cellcolor[HTML]{EFEFEF}\textbf{1.443}
        & \cellcolor[HTML]{EFEFEF}\textbf{0.652} \\
    \bottomrule
    \end{tabular}
    }
\end{table}

\begin{table}[!b]
    \centering
    \caption{Detection performance comparison of inference-time defenses on three datasets, averaged over FEDformer, SimpleTM, and TimesNet. Best results are in \textbf{bold}. Higher AUC and F1 indicates better detection performance.}
    \label{app-tab:main-defense-result-inf}
    \setlength{\tabcolsep}{6pt}
    \renewcommand{\arraystretch}{1.4}
    \renewcommand{\aboverulesep}{0pt}
    \renewcommand{\belowrulesep}{0pt}
    \setlength\cellspacetoplimit{2pt}
    \setlength\cellspacebottomlimit{2pt}
    \scriptsize
    \begin{tabular}{l@{\hskip 10pt}l@{\hskip 10pt}cccccccccc}
    \toprule
    \multirow{2}{*}{\textbf{Dataset}}
        & \multirow{2}{*}{\textbf{Defense}}
        & \multicolumn{2}{c}{\textbf{Random}}
        & \multicolumn{2}{c}{\textbf{Manhattan}}
        & \multicolumn{2}{c}{\textbf{FreqBack-TSF}}
        & \multicolumn{2}{c}{\textbf{BackTime}}
        & \multicolumn{2}{c}{\textbf{\textsc{Average}}} \\
    \cmidrule(r){3-4} 
    \cmidrule(lr){5-6} 
    \cmidrule(lr){7-8} 
    \cmidrule(lr){9-10}
    \cmidrule(l){11-12}
     &  & AUC $\uparrow$ & F1 $\uparrow$
        & AUC $\uparrow$ & F1 $\uparrow$
        & AUC $\uparrow$ & F1 $\uparrow$
        & AUC $\uparrow$ & F1 $\uparrow$
        & AUC $\uparrow$ & F1 $\uparrow$ \\
    \midrule

    \multirow{4}{*}[-0.5ex]{\textbf{PEMS03}}
        & No Defense & 0.500 & 0.500 & 0.500 & 0.500 & 0.500 & 0.500 & 0.500 & 0.500 & 0.500 & 0.500 \\ \cmidrule(l{0pt}){2-12}
        & STRIP      & 0.518 & 0.532 & 0.523 & 0.537 & \textbf{0.481} & 0.513 & \textbf{0.501} & 0.516 & 0.506 & 0.525 \\
        & TeCo       & \textbf{0.563} & \textbf{0.564} & \textbf{0.563} & \textbf{0.563} & 0.431 & 0.506 & 0.478 & 0.512 & \textbf{0.509} & \textbf{0.536} \\
        & IBD-PSC    & 0.364 & 0.514 & 0.402 & 0.522 & 0.416 & \textbf{0.519} & 0.486 & \textbf{0.535} & 0.417 & 0.523 \\
    \midrule

    \multirow{4}{*}[-0.5ex]{\textbf{Weather}}
        & No Defense & 0.500 & 0.500 & 0.500 & 0.500 & 0.500 & 0.500 & 0.500 & 0.500 & 0.500 & 0.500 \\ \cmidrule(l{0pt}){2-12}
        & STRIP      & 0.300 & 0.510 & 0.461 & 0.518 & \textbf{0.589} & \textbf{0.591} & 0.497 & 0.531 & 0.462 & 0.538 \\
        & TeCo       & \textbf{0.581} & \textbf{0.590} & 0.458 & 0.517 & 0.466 & 0.534 & \textbf{0.547} & \textbf{0.574} & \textbf{0.513} & \textbf{0.554} \\
        & IBD-PSC    & 0.317 & 0.519 & \textbf{0.521} & \textbf{0.546} & 0.369 & 0.556 & 0.390 & 0.534 & 0.399 & 0.539 \\
    \midrule

    \multirow{4}{*}[-0.5ex]{\textbf{ETTm1}}
        & No Defense & 0.500 & 0.500 & 0.500 & 0.500 & 0.500 & 0.500 & 0.500 & 0.500 & 0.500 & 0.500 \\ \cmidrule(l{0pt}){2-12}
        & STRIP      & 0.490 & 0.525 & 0.448 & 0.506 & 0.480 & 0.507 & 0.477 & 0.506 & 0.474 & 0.511 \\
        & TeCo       & \textbf{0.614} & \textbf{0.591} & \textbf{0.519} & \textbf{0.544} & \textbf{0.640} & \textbf{0.612} & \textbf{0.524} & \textbf{0.521} & \textbf{0.574} & \textbf{0.567} \\
        & IBD-PSC    & 0.378 & 0.513 & 0.464 & 0.523 & 0.497 & 0.532 & 0.486 & 0.518 & 0.456 & 0.522 \\
    \bottomrule
    \end{tabular}
   
\end{table}
\clearpage

\begin{table*}[!t]
    \centering
    \caption{Full main results of backdoor defenses against TSF backdoor attacks on \emph{FEDformer} model. Best results are in \textbf{bold}. Lower \MAEC{} indicates better performance, while higher \MAEP{} and \FDER{} indicate better performance.}
    \label{app-tab:main-defense-result-fedformer}
    \setlength{\tabcolsep}{2pt} 
    \tiny
    \renewcommand{\aboverulesep}{0pt}
    \renewcommand{\belowrulesep}{0pt}
    \setlength\cellspacetoplimit{1.8pt}
    \setlength\cellspacebottomlimit{1.8pt}
    \resizebox{\textwidth}{!}{
    
    \begin{tabular}{Sc@{\hskip 3pt}Sl*{12}{Sc}}
    \toprule
     & \textbf{Attack}~→
        & \multicolumn{3}{c}{\textbf{Random}}
        & \multicolumn{3}{c}{\textbf{Manhattan}}
        & \multicolumn{3}{c}{\textbf{FreqBack-TSF}}
        & \multicolumn{3}{c}{\textbf{BackTime}} \\
    \cmidrule(l{0pt}r){2-2}
    \cmidrule(lr){3-5}
    \cmidrule(lr){6-8}
    \cmidrule(lr){9-11}
    \cmidrule(l){12-14}
    \multirow{-2}{*}[0.5ex]{\textbf{Dataset}}
        & \textbf{Defense}~↓
        & \MAEC~↓ & \MAEP~↑ & \FDER~↑
        & \MAEC~↓ & \MAEP~↑ & \FDER~↑
        & \MAEC~↓ & \MAEP~↑ & \FDER~↑
        & \MAEC~↓ & \MAEP~↑ & \FDER~↑ \\
    \midrule

     & No Defense & 16.286 & 14.959 & --  & 16.411 & 17.984 & --  & 16.179 & 9.436 & --  & 16.093 & 10.760 & -- \\ \cmidrule(l{0pt}r{0pt}){2-14}
     & Spectral   & 16.930 & 15.658 & 0.503 & 16.567 & 18.960 & 0.521 & 16.232 & 9.627 & 0.508 & 16.484 & 11.221 & 0.509 \\
     & TED        & 16.607 & 15.402 & 0.505 & 16.667 & 17.639 & 0.492 & 16.378 & 9.496 & 0.497 & \textbf{16.284} & 10.828 & 0.497 \\
     & TED++      & 18.093 & 19.247 & 0.561 & 17.542 & 18.627 & 0.485 & 16.103 & 9.257 & 0.500 & 16.332 & 10.882 & 0.498 \\ \cmidrule(l{0pt}r{0pt}){2-14}
     & Fine-tuning    & 16.758 & 48.550 & 0.832 & 16.887 & 41.641 & 0.770 & 16.835 & 18.616 & 0.727 & 16.414 & 17.767 & 0.687 \\
     & Fine-pruning   & 16.836 & 49.054 & 0.831 & 16.951 & 51.561 & 0.810 & 16.871 & 20.123 & 0.745 & 16.408 & 21.445 & 0.740 \\
     & NAD            & 16.599 & 39.003 & 0.799 & 16.616 & 32.060 & 0.713 & 16.558 & 15.765 & 0.689 & 16.377 & 17.941 & 0.691 \\
     & IMS            & \textbf{16.286} & 14.953 & 0.500 & \textbf{16.411} & 17.982 & 0.500 & \textbf{16.179} & 9.430 & 0.500 & 16.684 & 13.569 & 0.586 \\ \cmidrule(l{0pt}r{0pt}){2-14}
     & ABL        & 17.591 & 18.677 & 0.562 & 17.141 & 17.898 & 0.479 & 16.990 & 10.627 & 0.532 & 16.803 & 11.183 & 0.498 \\
     & PDB        & 16.774 & 25.809 & 0.696 & 17.254 & 37.089 & 0.733 & 16.914 & 17.076 & 0.702 & 17.040 & 14.511 & 0.601 \\
     & ESTI       & 18.896 & 15.727 & 0.455 & 18.102 & 17.401 & 0.453 & 16.432 & 9.199 & 0.492 & \textbf{16.284} & 11.212 & 0.514 \\ \cmidrule(l{0pt}r{0pt}){2-14}
    \multirow{-12}{*}[2ex]{PEMS03}
        & \cellcolor[HTML]{EFEFEF}\textbf{\methodname}
        & \cellcolor[HTML]{EFEFEF}16.607
        & \cellcolor[HTML]{EFEFEF}\textbf{100.436}
        & \cellcolor[HTML]{EFEFEF}\textbf{0.916}
        & \cellcolor[HTML]{EFEFEF}16.578
        & \cellcolor[HTML]{EFEFEF}\textbf{94.212}
        & \cellcolor[HTML]{EFEFEF}\textbf{0.900}
        & \cellcolor[HTML]{EFEFEF}16.496
        & \cellcolor[HTML]{EFEFEF}\textbf{38.147}
        & \cellcolor[HTML]{EFEFEF}\textbf{0.867}
        & \cellcolor[HTML]{EFEFEF}16.840
        & \cellcolor[HTML]{EFEFEF}\textbf{41.232}
        & \cellcolor[HTML]{EFEFEF}\textbf{0.847} \\
    \midrule

     & No Defense & 9.282 & 13.400 & --  & 8.781 & 22.145 & --  & 9.434 & 9.423 & --  & 9.609 & 8.020 & -- \\ \cmidrule(l{0pt}r{0pt}){2-14}
     & Spectral   & 9.286 & 18.636 & 0.640 & 9.495 & 22.391 & 0.468 & 10.107 & 7.347 & 0.467 & 9.460 & 8.280 & 0.516 \\
     & TED        & 9.161 & 19.594 & 0.658 & 9.271 & 19.516 & 0.474 & 9.670 & 6.263 & 0.488 & 9.775 & 10.295 & 0.602 \\
     & TED++      & 9.412 & 45.224 & 0.845 & 10.853 & 21.064 & 0.405 & 9.506 & 19.342 & 0.753 & 9.517 & 8.135 & 0.507 \\ \cmidrule(l{0pt}r{0pt}){2-14}
     & Fine-tuning    & 9.174 & 77.684 & 0.914 & 9.252 & 71.884 & 0.820 & 9.785 & \textbf{85.835} & 0.927 & 9.517 & \textbf{69.535} & \textbf{0.942} \\
     & Fine-pruning   & 9.155 & 86.216 & 0.922 & 9.127 & 72.853 & 0.829 & 9.802 & 70.408 & 0.914 & 9.600 & 58.004 & 0.931 \\
     & NAD            & 9.111 & 44.531 & 0.850 & 9.032 & 67.813 & 0.823 & 9.801 & 60.086 & 0.903 & 9.584 & 55.976 & 0.928 \\
     & IMS            & 9.282 & 13.020 & 0.500 & 8.785 & 22.210 & 0.501 & 9.430 & 8.868 & 0.500 & 9.610 & 8.029 & 0.501 \\ \cmidrule(l{0pt}r{0pt}){2-14}
     & ABL        & 10.044 & 11.555 & 0.462 & 9.633 & 22.443 & 0.462 & 9.823 & 7.480 & 0.480 & 11.159 & 12.711 & 0.615 \\
     & PDB        & 9.890 & 41.380 & 0.807 & 9.619 & 54.644 & 0.754 & 9.609 & 30.494 & 0.836 & 10.254 & 35.951 & 0.857 \\
     & ESTI       & \textbf{9.076} & 23.477 & 0.715 & \textbf{8.569} & \textbf{102.159} & \textbf{0.892} & \textbf{8.440} & 5.191 & 0.500 & \textbf{8.882} & 4.855 & 0.500 \\ \cmidrule(l{0pt}r{0pt}){2-14}
    \multirow{-12}{*}[2ex]{Weather}
        & \cellcolor[HTML]{EFEFEF}\textbf{\methodname}
        & \cellcolor[HTML]{EFEFEF}9.162
        & \cellcolor[HTML]{EFEFEF}\textbf{102.996}
        & \cellcolor[HTML]{EFEFEF}\textbf{0.935}
        & \cellcolor[HTML]{EFEFEF}9.584
        & \cellcolor[HTML]{EFEFEF}96.013
        & \cellcolor[HTML]{EFEFEF}0.843
        & \cellcolor[HTML]{EFEFEF}9.536
        & \cellcolor[HTML]{EFEFEF}76.651
        & \cellcolor[HTML]{EFEFEF}\textbf{0.933}
        & \cellcolor[HTML]{EFEFEF}10.089
        & \cellcolor[HTML]{EFEFEF}43.244
        & \cellcolor[HTML]{EFEFEF}0.883 \\
    \midrule

     & No Defense & 1.121 & 1.218 & --  & 1.109 & 1.662 & --  & 1.111 & 0.671 & --  & 1.085 & 0.911 & -- \\ \cmidrule(l{0pt}r{0pt}){2-14}
     & Spectral   & 1.134 & 1.306 & 0.528 & 1.142 & 1.826 & 0.531 & 1.138 & 0.794 & 0.565 & 1.160 & 0.775 & 0.468 \\
     & TED        & \textbf{1.096} & 1.201 & 0.500 & 1.128 & 1.695 & 0.502 & 1.100 & 0.566 & 0.500 & 1.125 & 0.997 & 0.525 \\
     & TED++      & 1.130 & 1.258 & 0.512 & 1.145 & 1.498 & 0.484 & \textbf{1.088} & 0.708 & 0.526 & 1.106 & 0.851 & 0.490 \\ \cmidrule(l{0pt}r{0pt}){2-14}
     & Fine-tuning    & 1.180 & 2.273 & 0.707 & 1.173 & 2.534 & 0.645 & 1.217 & 1.759 & 0.766 & 1.191 & 1.698 & 0.687 \\
     & Fine-pruning   & 1.181 & 2.229 & 0.701 & 1.161 & 2.612 & 0.660 & 1.176 & 1.541 & 0.754 & 1.177 & 1.659 & 0.686 \\
     & NAD            & 1.155 & 1.994 & 0.680 & 1.148 & 2.403 & 0.637 & 1.163 & 1.310 & 0.721 & 1.161 & 1.667 & 0.694 \\
     & IMS            & 1.121 & 1.229 & 0.504 & \textbf{1.109} & 1.665 & 0.501 & 1.110 & 0.670 & 0.500 & \textbf{1.080 }& 1.103 & 0.587 \\ \cmidrule(l{0pt}r{0pt}){2-14}
     & ABL        & 1.275 & 1.630 & 0.566 & 1.345 & 2.052 & 0.507 & 1.327 & 1.492 & 0.693 & 1.296 & 1.152 & 0.523 \\
     & PDB        & 1.142 & 3.479 & 0.816 & 1.343 & 2.681 & 0.603 & 1.220 & 1.078 & 0.644 & 1.261 & 1.715 & 0.664 \\
     & ESTI       & 1.301 & 1.464 & 0.514 & 1.398 & 1.764 & 0.426 & 1.193 & 0.888 & 0.587 & 1.261 & 1.406 & 0.606 \\ \cmidrule(l{0pt}r{0pt}){2-14}
    \multirow{-12}{*}[2ex]{ETTm1}
        & \cellcolor[HTML]{EFEFEF}\textbf{\methodname}
        & \cellcolor[HTML]{EFEFEF}1.220
        & \cellcolor[HTML]{EFEFEF}\textbf{6.664}
        & \cellcolor[HTML]{EFEFEF}\textbf{0.868}
        & \cellcolor[HTML]{EFEFEF}1.213
        & \cellcolor[HTML]{EFEFEF}\textbf{7.138}
        & \cellcolor[HTML]{EFEFEF}\textbf{0.841}
        & \cellcolor[HTML]{EFEFEF}1.298
        & \cellcolor[HTML]{EFEFEF}\textbf{1.912}
        & \cellcolor[HTML]{EFEFEF}\textbf{0.752}
        & \cellcolor[HTML]{EFEFEF}1.256
        & \cellcolor[HTML]{EFEFEF}\textbf{1.304}
        & \cellcolor[HTML]{EFEFEF}\textbf{0.582} \\
    \bottomrule
    \end{tabular}
    }
\end{table*}

\begin{table*}[!b]
    \centering
    \caption{Detection performance comparison of inference-time defenses on three datasets on \emph{FEDformer} model. Best results are in \textbf{bold}. Higher AUC and F1 indicates better detection performance.}
    \label{app-tab:main-defense-result-fedformer-inf}
    \setlength{\tabcolsep}{6pt}
    \renewcommand{\arraystretch}{1.4}
    \renewcommand{\aboverulesep}{0pt}
    \renewcommand{\belowrulesep}{0pt}
    \setlength\cellspacetoplimit{2pt}
    \setlength\cellspacebottomlimit{2pt}
    \scriptsize

    \begin{tabular}{l@{\hskip 10pt}l@{\hskip 10pt}cccccccccc}
    \toprule
    \multirow{2}{*}{\textbf{Dataset}}
        & \multirow{2}{*}{\textbf{Defense}}
        & \multicolumn{2}{c}{\textbf{Random}}
        & \multicolumn{2}{c}{\textbf{Manhattan}}
        & \multicolumn{2}{c}{\textbf{FreqBack-TSF}}
        & \multicolumn{2}{c}{\textbf{BackTime}}
        & \multicolumn{2}{c}{\textbf{\textsc{Average}}} \\
    \cmidrule(r){3-4} 
    \cmidrule(lr){5-6} 
    \cmidrule(lr){7-8} 
    \cmidrule(lr){9-10}
    \cmidrule(lr){11-12}
     &  & AUC $\uparrow$ & F1 $\uparrow$
        & AUC $\uparrow$ & F1 $\uparrow$
        & AUC $\uparrow$ & F1 $\uparrow$
        & AUC $\uparrow$ & F1 $\uparrow$
        & AUC $\uparrow$ & F1 $\uparrow$ \\
    \midrule

    \multirow{4}{*}[-0.5ex]{\textbf{PEMS03}}
        & No Defense & 0.500 & 0.500 & 0.500 & 0.500 & 0.500 & 0.500 & 0.500 & 0.500 & 0.500 & 0.500 \\ \cmidrule(l{0pt}){2-12}
        & STRIP      & 0.502 & 0.522 & 0.523 & 0.532 & 0.500 & 0.520 & 0.504 & 0.517 & \textbf{0.507} & 0.523 \\
        & TeCo       & \textbf{0.573} & \textbf{0.557} & \textbf{0.577} & \textbf{0.559} & 0.403 & 0.500 & 0.465 & 0.500 & 0.505 & 0.529 \\
        & IBD-PSC    & 0.284 & 0.500 & 0.398 & 0.507 & \textbf{0.552} & \textbf{0.555} & \textbf{0.627} & \textbf{0.603} & 0.465 & \textbf{0.541} \\
    \midrule

    \multirow{4}{*}[-0.5ex]{\textbf{Weather}}
        & No Defense & 0.500 & 0.500 & 0.500 & 0.500 & 0.500 & 0.500 & 0.500 & 0.500 & 0.500 & 0.500 \\ \cmidrule(l{0pt}){2-12}
        & STRIP      & 0.422 & 0.507 & 0.481 & 0.524 & 0.526 & 0.547 & \textbf{0.527} & \textbf{0.539} & 0.489 & 0.529 \\
        & TeCo       & \textbf{0.599} & \textbf{0.587} & \textbf{0.554} & \textbf{0.541} & \textbf{0.587} & \textbf{0.571} & 0.442 & 0.521 & \textbf{0.546} & \textbf{0.555} \\
        & IBD-PSC    & 0.331 & 0.501 & 0.505 & 0.520 & 0.455 & 0.567 & 0.362 & 0.502 & 0.413 & 0.523 \\
    \midrule

    \multirow{4}{*}[-0.5ex]{\textbf{ETTm1}}
        & No Defense & 0.500 & 0.500 & 0.500 & 0.500 & 0.500 & 0.500 & 0.500 & 0.500 & 0.500 & 0.500 \\ \cmidrule(l{0pt}){2-12}
        & STRIP      & 0.406 & 0.501 & 0.401 & 0.501 & 0.496 & 0.513 & 0.466 & 0.507 & 0.442 & 0.506 \\
        & TeCo       & \textbf{0.674} & \textbf{0.630} & \textbf{0.577} & \textbf{0.559} & \textbf{0.560} & \textbf{0.561} & \textbf{0.504} & 0.507 & \textbf{0.579} & \textbf{0.564} \\
        & IBD-PSC    & 0.405 & 0.517 & 0.384 & 0.510 & 0.511 & 0.519 & 0.491 & \textbf{0.522} & 0.448 & 0.517 \\
    \bottomrule
    \end{tabular}
\end{table*}
\clearpage

\begin{table*}[!t]
    \centering
    \caption{Full main results of backdoor defenses against TSF backdoor attacks on \emph{SimpleTM} model. Best results are in \textbf{bold}. Lower \MAEC{} indicates better performance, while higher \MAEP{} and \FDER{} indicate better performance.}
    \label{app-tab:main-defense-result-simpletm}
    \setlength{\tabcolsep}{2pt} 
    \tiny
    \renewcommand{\aboverulesep}{0pt}
    \renewcommand{\belowrulesep}{0pt}
    \setlength\cellspacetoplimit{1.8pt}
    \setlength\cellspacebottomlimit{1.8pt}
    \resizebox{\textwidth}{!}{
    
    \begin{tabular}{Sc@{\hskip 3pt}Sl*{12}{Sc}}
    \toprule
     & \textbf{Attack}~→
        & \multicolumn{3}{c}{\textbf{Random}}
        & \multicolumn{3}{c}{\textbf{Manhattan}}
        & \multicolumn{3}{c}{\textbf{FreqBack-TSF}}
        & \multicolumn{3}{c}{\textbf{BackTime}} \\
    \cmidrule(l{0pt}r){2-2}
    \cmidrule(lr){3-5}
    \cmidrule(lr){6-8}
    \cmidrule(lr){9-11}
    \cmidrule(l){12-14}
    \multirow{-2}{*}[0.5ex]{\textbf{Dataset}}
        & \textbf{Defense}~↓
        & \MAEC~↓ & \MAEP~↑ & \FDER~↑
        & \MAEC~↓ & \MAEP~↑ & \FDER~↑
        & \MAEC~↓ & \MAEP~↑ & \FDER~↑
        & \MAEC~↓ & \MAEP~↑ & \FDER~↑ \\
    \midrule

     & No Defense & 17.510 & 19.007 & --  & 17.539 & 22.532 & --  & 17.335 & 15.468 & --  & 17.268 & 9.131 & -- \\ \cmidrule(l{0pt}r{0pt}){2-14}
     & Spectral   & 17.746 & 20.820 & 0.537 & 17.596 & 22.835 & 0.505 & 18.015 & 13.304 & 0.481 & 17.621 & 13.971 & 0.663 \\
     & TED        & 17.578 & 25.544 & 0.626 & 17.529 & 22.801 & 0.506 & 17.707 & 13.417 & 0.489 & 17.671 & 10.242 & 0.543 \\
     & TED++      & 17.807 & 19.156 & 0.496 & 17.529 & 22.863 & 0.507 & 17.785 & 12.099 & 0.487 & 17.328 & 11.401 & 0.598 \\ \cmidrule(l{0pt}r{0pt}){2-14}
     & Fine-tuning    & 17.397 & 23.898 & 0.602 & 17.619 & 27.930 & 0.594 & 17.464 & 27.846 & 0.719 & 17.355 & 13.287 & 0.654 \\
     & Fine-pruning   & \textbf{17.396} & 25.735 & 0.631 & 17.665 & 28.338 & 0.599 & 17.460 & 29.782 & 0.737 & 17.363 & 13.950 & 0.670 \\
     & NAD            & 17.516 & 21.572 & 0.559 & 17.571 & 24.708 & 0.543 & 17.411 & 25.277 & 0.692 & 17.299 & 13.245 & 0.654 \\
     & IMS            & 17.513 & 19.014 & 0.500 & 17.539 & 22.540 & 0.500 & 17.335 & 15.480 & 0.500 & \textbf{16.520} & 8.206 & 0.500 \\ \cmidrule(l{0pt}r{0pt}){2-14}
     & ABL        & 17.740 & 19.722 & 0.512 & 17.717 & 23.135 & 0.508 & 17.665 & 16.424 & 0.520 & 17.465 & 11.219 & 0.587 \\
     & PDB        & 17.740 & 117.954 & 0.913 & 17.527 & 120.846 & 0.907 & 18.889 & 38.543 & 0.758 & 18.025 & 26.746 & 0.808 \\
     & ESTI       & \textbf{17.396} & 16.826 & 0.500 & 17.380 & 20.347 & 0.500 & 17.188 & 15.309 & 0.500 & 18.952 & 14.763 & 0.646 \\ \cmidrule(l{0pt}r{0pt}){2-14}
    \multirow{-12}{*}[2ex]{PEMS03}
        & \cellcolor[HTML]{EFEFEF}\textbf{\methodname}
        & \cellcolor[HTML]{EFEFEF}17.489
        & \cellcolor[HTML]{EFEFEF}\textbf{173.700}
        & \cellcolor[HTML]{EFEFEF}\textbf{0.945}
        & \cellcolor[HTML]{EFEFEF}\textbf{17.284}
        & \cellcolor[HTML]{EFEFEF}\textbf{157.870}
        & \cellcolor[HTML]{EFEFEF}\textbf{0.929}
        & \cellcolor[HTML]{EFEFEF}\textbf{16.780}
        & \cellcolor[HTML]{EFEFEF}\textbf{94.224}
        & \cellcolor[HTML]{EFEFEF}\textbf{0.918}
        & \cellcolor[HTML]{EFEFEF}17.243
        & \cellcolor[HTML]{EFEFEF}\textbf{36.626}
        & \cellcolor[HTML]{EFEFEF}\textbf{0.875} \\
    \midrule

     & No Defense & 7.693 & 18.888 & --  & 7.711 & 64.020 & --  & 7.761 & 19.205 & --  & 7.752 & 15.301 & -- \\ \cmidrule(l{0pt}r{0pt}){2-14}
     & Spectral   & 7.868 & 26.086 & 0.627 & 7.792 & 65.176 & 0.504 & 7.875 & 19.109 & 0.493 & 7.851 & 15.079 & 0.494 \\
     & TED        & 7.764 & 28.075 & 0.659 & 7.676 & 63.941 & 0.500 & 7.836 & 16.691 & 0.495 & 7.979 & 15.343 & 0.487 \\
     & TED++      & 7.729 & 19.507 & 0.514 & 7.674 & 61.888 & 0.500 & 7.862 & 16.510 & 0.494 & 7.849 & 15.187 & 0.494 \\ \cmidrule(l{0pt}r{0pt}){2-14}
     & Fine-tuning    & 7.850 & 23.020 & 0.580 & 8.089 & 72.225 & 0.533 & 8.148 & 42.259 & 0.749 & 7.960 & 17.028 & 0.538 \\
     & Fine-pruning   & 7.860 & 25.415 & 0.618 & 8.092 & 77.976 & 0.566 & 8.005 & 45.018 & 0.771 & 8.009 & 19.544 & 0.593 \\
     & NAD            & 7.835 & 19.659 & 0.511 & 8.056 & 73.444 & 0.543 & 7.866 & 27.959 & 0.650 & 8.095 & 16.989 & 0.529 \\
     & IMS            & \textbf{7.692} & 19.130 & 0.506 & 7.712 & 64.484 & 0.504 & \textbf{7.755} & 18.969 & 0.500 & 7.753 & 15.263 & 0.500 \\ \cmidrule(l{0pt}r{0pt}){2-14}
     & ABL        & 7.887 & 35.276 & 0.720 & 8.007 & 64.383 & 0.484 & 7.927 & 21.299 & 0.539 & 8.021 & 16.695 & 0.525 \\
     & PDB        & 7.836 & 192.806 & 0.942 & 7.902 & 108.511 & 0.693 & 8.041 & 98.519 & 0.885 & 8.040 & 50.108 & 0.829 \\
     & ESTI       & 7.733 & 16.628 & 0.497 & \textbf{7.469} & 74.586 & 0.571 & 7.896 & 210.368 & 0.946 & \textbf{7.689} & 15.785 & 0.515 \\ \cmidrule(l{0pt}r{0pt}){2-14}
    \multirow{-12}{*}[2ex]{Weather}
        & \cellcolor[HTML]{EFEFEF}\textbf{\methodname}
        & \cellcolor[HTML]{EFEFEF}7.716
        & \cellcolor[HTML]{EFEFEF}\textbf{351.059}
        & \cellcolor[HTML]{EFEFEF}\textbf{0.972}
        & \cellcolor[HTML]{EFEFEF}7.699
        & \cellcolor[HTML]{EFEFEF}\textbf{114.876}
        & \cellcolor[HTML]{EFEFEF}\textbf{0.721}
        & \cellcolor[HTML]{EFEFEF}7.973
        & \cellcolor[HTML]{EFEFEF}\textbf{416.357}
        & \cellcolor[HTML]{EFEFEF}\textbf{0.964}
        & \cellcolor[HTML]{EFEFEF}7.934
        & \cellcolor[HTML]{EFEFEF}\textbf{69.357}
        & \cellcolor[HTML]{EFEFEF}\textbf{0.878} \\
    \midrule

     & No Defense & 1.206 & 0.966 & --  & 1.203 & 1.558 & --  & 1.165 & 0.870 & --  & 1.170 & 0.508 & -- \\ \cmidrule(l{0pt}r{0pt}){2-14}
     & Spectral   & 1.215 & 1.107 & 0.560 & 1.224 & 1.317 & 0.492 & 1.186 & 0.922 & 0.519 & 1.185 & 0.602 & 0.572 \\
     & TED        & 1.189 & 1.318 & 0.633 & 1.199 & 1.448 & 0.500 & 1.164 & 0.902 & 0.518 & 1.174 & 0.517 & 0.507 \\
     & TED++      & 1.189 & 1.197 & 0.596 & 1.190 & 1.474 & 0.500 & 1.165 & 0.893 & 0.513 & 1.181 & 0.549 & 0.533 \\ \cmidrule(l{0pt}r{0pt}){2-14}
     & Fine-tuning    & 1.210 & 2.075 & 0.766 & 1.186 & 3.142 & 0.752 & 1.182 & 1.171 & 0.621 & 1.186 & 0.683 & 0.622 \\
     & Fine-pruning   & 1.209 & 2.171 & 0.776 & 1.194 & 3.716 & 0.790 & 1.185 & 1.295 & 0.655 & 1.188 & 0.794 & 0.673 \\
     & NAD            & 1.215 & 1.341 & 0.636 & 1.195 & 2.557 & 0.695 & 1.170 & 1.005 & 0.565 & 1.183 & 0.501 & 0.495 \\
     & IMS            & 1.208 & 1.018 & 0.525 & 1.204 & 1.599 & 0.513 & 1.166 & 0.893 & 0.512 & 1.171 & 0.504 & 0.500 \\ \cmidrule(l{0pt}r{0pt}){2-14}
     & ABL        & 1.229 & 1.061 & 0.536 & 1.226 & 1.564 & 0.492 & 1.182 & 0.944 & 0.532 & 1.190 & 0.505 & 0.492 \\
     & PDB        & \textbf{1.142} & 3.659 & 0.868 & \textbf{1.154} & 6.487 & 0.880 & \textbf{1.142} & 1.447 & 0.699 & \textbf{1.135} & 0.799 & 0.683 \\
     & ESTI       & 1.285 & 4.140 & 0.853 & 1.276 & 2.649 & 0.677 & 1.281 & 1.345 & 0.631 & 1.259 & 0.461 & 0.465 \\ \cmidrule(l{0pt}r{0pt}){2-14}
    \multirow{-12}{*}[2ex]{ETTm1}
        & \cellcolor[HTML]{EFEFEF}\textbf{\methodname}
        & \cellcolor[HTML]{EFEFEF}1.247
        & \cellcolor[HTML]{EFEFEF}\textbf{6.928}
        & \cellcolor[HTML]{EFEFEF}\textbf{0.914}
        & \cellcolor[HTML]{EFEFEF}1.245
        & \cellcolor[HTML]{EFEFEF}\textbf{7.802}
        & \cellcolor[HTML]{EFEFEF}\textbf{0.883}
        & \cellcolor[HTML]{EFEFEF}1.287
        & \cellcolor[HTML]{EFEFEF}\textbf{1.795}
        & \cellcolor[HTML]{EFEFEF}\textbf{0.710}
        & \cellcolor[HTML]{EFEFEF}1.268
        & \cellcolor[HTML]{EFEFEF}\textbf{0.923}
        & \cellcolor[HTML]{EFEFEF}\textbf{0.687} \\
    \bottomrule
    \end{tabular}
    }
\end{table*}

\begin{table*}[!b]
    \centering
    \caption{Detection performance comparison of inference-time defenses on three datasets on \emph{SimpleTM} model. Best results are in \textbf{bold}. Higher AUC and F1 indicates better detection performance.}
    \label{app-tab:main-defense-result-simpletm-inf}
    \setlength{\tabcolsep}{6pt}
    \renewcommand{\arraystretch}{1.4}
    \renewcommand{\aboverulesep}{0pt}
    \renewcommand{\belowrulesep}{0pt}
    \setlength\cellspacetoplimit{2pt}
    \setlength\cellspacebottomlimit{2pt}
    \scriptsize

    \begin{tabular}{l@{\hskip 10pt}l@{\hskip 10pt}cccccccccc}
    \toprule
    \multirow{2}{*}{\textbf{Dataset}}
        & \multirow{2}{*}{\textbf{Defense}}
        & \multicolumn{2}{c}{\textbf{Random}}
        & \multicolumn{2}{c}{\textbf{Manhattan}}
        & \multicolumn{2}{c}{\textbf{FreqBack-TSF}}
        & \multicolumn{2}{c}{\textbf{BackTime}}
        & \multicolumn{2}{c}{\textbf{\textsc{Average}}} \\
    \cmidrule(r){3-4} 
    \cmidrule(lr){5-6} 
    \cmidrule(lr){7-8} 
    \cmidrule(lr){9-10}
    \cmidrule(lr){11-12}
     &  & AUC $\uparrow$ & F1 $\uparrow$
        & AUC $\uparrow$ & F1 $\uparrow$
        & AUC $\uparrow$ & F1 $\uparrow$
        & AUC $\uparrow$ & F1 $\uparrow$
        & AUC $\uparrow$ & F1 $\uparrow$ \\
    \midrule

    \multirow{4}{*}[-0.5ex]{\textbf{PEMS03}}
        & No Defense & 0.500 & 0.500 & 0.500 & 0.500 & 0.500 & 0.500 & 0.500 & 0.500 & 0.500 & 0.500 \\ \cmidrule(l{0pt}){2-12}
        & STRIP      & 0.517 & 0.529 & 0.506 & 0.516 & 0.486 & \textbf{0.517} & 0.494 & 0.518 & 0.501 & 0.520 \\
        & TeCo       & \textbf{0.680} & \textbf{0.628} & \textbf{0.670} & \textbf{0.627} & \textbf{0.510} & \textbf{0.517} & \textbf{0.541} & \textbf{0.535} & \textbf{0.600} & \textbf{0.577} \\
        & IBD-PSC    & 0.436 & 0.525 & 0.434 & 0.548 & 0.263 & 0.500 & 0.364 & 0.500 & 0.374 & 0.518 \\
    \midrule

    \multirow{4}{*}[-0.5ex]{\textbf{Weather}}
        & No Defense & 0.500 & 0.500 & 0.500 & 0.500 & 0.500 & 0.500 & 0.500 & 0.500 & 0.500 & 0.500 \\ \cmidrule(l{0pt}){2-12}
        & STRIP      & 0.308 & 0.521 & 0.395 & 0.515 & \textbf{0.737} & \textbf{0.701} & 0.420 & 0.515 & 0.465 & 0.563 \\
        & TeCo       & \textbf{0.747} & \textbf{0.680} & 0.463 & 0.506 & 0.439 & 0.531 & \textbf{0.770} & \textbf{0.700} & \textbf{0.605} & \textbf{0.604} \\
        & IBD-PSC    & 0.047 & 0.500 & \textbf{0.560} & \textbf{0.597} & 0.055 & 0.500 & 0.245 & 0.520 & 0.227 & 0.529 \\
    \midrule

    \multirow{4}{*}[-0.5ex]{\textbf{ETTm1}}
        & No Defense & 0.500 & 0.500 & 0.500 & 0.500 & 0.500 & 0.500 & 0.500 & 0.500 & 0.500 & 0.500 \\ \cmidrule(l{0pt}){2-12}
        & STRIP      & 0.464 & 0.500 & 0.447 & 0.500 & 0.485 & 0.506 & 0.491 & 0.509 & 0.472 & 0.504 \\
        & TeCo       & \textbf{0.533} & \textbf{0.539} & \textbf{0.594} & \textbf{0.572} & \textbf{0.597} & \textbf{0.575} & \textbf{0.541} & \textbf{0.533} & \textbf{0.566} & \textbf{0.555} \\
        & IBD-PSC    & 0.421 & 0.523 & 0.507 & 0.544 & 0.526 & \textbf{0.575} & 0.482 & 0.532 & 0.484 & 0.544 \\
    \bottomrule
    \end{tabular}
\end{table*}
\clearpage

\begin{table*}[!t]
    \centering
    \caption{Full main results of backdoor defenses against TSF backdoor attacks on \emph{TimesNet} model. Best results are in \textbf{bold}. Lower \MAEC{} indicates better performance, while higher \MAEP{} and \FDER{} indicate better performance.}
    \label{app-tab:main-defense-result-timesnet}
    \setlength{\tabcolsep}{2pt}
    \tiny
    \renewcommand{\aboverulesep}{0pt}
    \renewcommand{\belowrulesep}{0pt}
    \setlength\cellspacetoplimit{1.8pt}
    \setlength\cellspacebottomlimit{1.8pt}
    \resizebox{\textwidth}{!}{
    \begin{tabular}{Sc@{\hskip 3pt}Sl*{12}{Sc}}
    \toprule
     & \textbf{Attack}~→
        & \multicolumn{3}{c}{\textbf{Random}}
        & \multicolumn{3}{c}{\textbf{Manhattan}}
        & \multicolumn{3}{c}{\textbf{FreqBack-TSF}}
        & \multicolumn{3}{c}{\textbf{BackTime}} \\
    \cmidrule(l{0pt}r){2-2}
    \cmidrule(lr){3-5}
    \cmidrule(lr){6-8}
    \cmidrule(lr){9-11}
    \cmidrule(l){12-14}
    \multirow{-2}{*}[0.5ex]{\textbf{Dataset}}
        & \textbf{Defense}~↓
        & \MAEC~↓ & \MAEP~↑ & \FDER~↑
        & \MAEC~↓ & \MAEP~↑ & \FDER~↑
        & \MAEC~↓ & \MAEP~↑ & \FDER~↑
        & \MAEC~↓ & \MAEP~↑ & \FDER~↑ \\
    \midrule

     & No Defense & 19.104 & 19.351 & --  & 19.216 & 20.283 & --  & 19.234 & 19.146 & --  & 19.459 & 22.713 & -- \\ \cmidrule(l{0pt}r{0pt}){2-14}
     & Spectral   & 20.492 & 18.591 & 0.466 & 24.168 & 19.455 & 0.398 & 22.047 & 19.149 & 0.436 & 21.891 & 20.544 & 0.444 \\
     & TED        & 21.116 & 19.244 & 0.452 & 24.086 & 20.453 & 0.403 & 22.270 & 19.040 & 0.432 & 21.862 & 20.789 & 0.445 \\
     & TED++      & 21.692 & 19.150 & 0.440 & 21.906 & 20.487 & 0.444 & 22.228 & 18.979 & 0.433 & 22.033 & 21.340 & 0.442 \\ \cmidrule(l{0pt}r{0pt}){2-14}
     & Fine-tuning    & 22.852 & 20.279 & 0.441 & 24.476 & 23.412 & 0.459 & 22.211 & 20.975 & 0.477 & 23.032 & 23.534 & 0.440 \\
     & Fine-pruning   & 22.828 & 20.139 & 0.438 & 24.168 & 23.441 & 0.465 & 22.887 & 20.725 & 0.458 & 22.286 & 23.813 & 0.460 \\
     & NAD            & 22.270 & 19.851 & 0.442 & 23.592 & 21.775 & 0.442 & 21.647 & 19.849 & 0.462 & 22.076 & 23.288 & 0.453 \\
     & IMS            & 23.919 & 19.225 & 0.399 & 24.160 & 20.013 & 0.398 & 22.048 & 18.800 & 0.436 & 22.049 & 21.278 & 0.441 \\ \cmidrule(l{0pt}r{0pt}){2-14}
     & ABL        & 23.579 & 18.912 & 0.405 & 24.091 & 19.283 & 0.399 & 21.293 & 18.113 & 0.452 & 22.014 & 21.043 & 0.442 \\
     & PDB        & 21.375 & 20.306 & 0.470 & 23.142 & 23.497 & 0.484 & 22.733 & 22.423 & 0.496 & 21.836 & 25.933 & 0.508 \\
     & ESTI       & 23.440 & 19.006 & 0.408 & 22.899 & 19.134 & 0.420 & 22.760 & 19.543 & 0.433 & 22.420 & 21.717 & 0.434 \\ \cmidrule(l{0pt}r{0pt}){2-14}
    \multirow{-12}{*}[2ex]{PEMS03}
        & \cellcolor[HTML]{EFEFEF}\textbf{\methodname}
        & \cellcolor[HTML]{EFEFEF}\textbf{19.687}
        & \cellcolor[HTML]{EFEFEF}\textbf{39.894}
        & \cellcolor[HTML]{EFEFEF}\textbf{0.743}
        & \cellcolor[HTML]{EFEFEF}\textbf{19.689}
        & \cellcolor[HTML]{EFEFEF}\textbf{40.029}
        & \cellcolor[HTML]{EFEFEF}\textbf{0.735}
        & \cellcolor[HTML]{EFEFEF}\textbf{19.609}
        & \cellcolor[HTML]{EFEFEF}\textbf{40.905}
        & \cellcolor[HTML]{EFEFEF}\textbf{0.756}
        & \cellcolor[HTML]{EFEFEF}\textbf{20.061}
        & \cellcolor[HTML]{EFEFEF}\textbf{40.052}
        & \cellcolor[HTML]{EFEFEF}\textbf{0.701} \\
    \midrule

     & No Defense & 16.653 & 12.684 & --  & 18.026 & 30.666 & --  & 13.148 & 11.719 & --  & 14.943 & 24.417 & -- \\ \cmidrule(l{0pt}r{0pt}){2-14}
     & Spectral   & 16.412 & 16.542 & 0.617 & 20.076 & 45.515 & 0.612 & 17.997 & 16.863 & 0.518 & 26.925 & 37.809 & 0.455 \\
     & TED        & 19.469 & 16.176 & 0.536 & 18.531 & 33.423 & 0.528 & 26.568 & 25.782 & 0.520 & 26.293 & 47.592 & 0.528 \\
     & TED++      & 30.764 & 32.157 & 0.573 & 26.426 & 44.219 & 0.494 & 23.532 & 21.640 & 0.509 & 22.297 & 35.818 & 0.494 \\ \cmidrule(l{0pt}r{0pt}){2-14}
     & Fine-tuning    & 19.056 & 22.354 & 0.653 & 18.082 & 70.219 & 0.780 & 21.202 & 33.497 & 0.635 & 17.291 & 66.797 & 0.749 \\
     & Fine-pruning   & 18.261 & 21.368 & 0.659 & 17.745 & 71.955 & 0.787 & 18.355 & 40.239 & 0.713 & 16.871 & 67.479 & 0.762 \\
     & NAD            & 18.466 & 17.050 & 0.579 & 17.971 & 65.990 & 0.768 & 17.227 & 29.268 & 0.681 & 18.082 & 58.089 & 0.703 \\
     & IMS            & 16.648 & 12.691 & 0.500 & 18.044 & 30.657 & 0.499 & \textbf{13.144} & 11.745 & 0.501 & 14.948 & 24.494 & 0.501 \\ \cmidrule(l{0pt}r{0pt}){2-14}
     & ABL        & 23.605 & 13.962 & 0.399 & 27.601 & 42.822 & 0.468 & 23.263 & 27.300 & 0.568 & 19.962 & 30.649 & 0.476 \\
     & PDB        & 19.188 & 39.525 & 0.773 & 20.099 & \textbf{95.252} & 0.787 & 25.566 & 46.034 & 0.630 & 16.903 & 83.259 & 0.795 \\
     & ESTI       & 30.383 & 22.809 & 0.496 & 32.988 & 69.844 & 0.554 & 25.970 & 27.804 & 0.542 & 23.750 & 39.619 & 0.506 \\ \cmidrule(l{0pt}r{0pt}){2-14}
    \multirow{-12}{*}[2ex]{Weather}
        & \cellcolor[HTML]{EFEFEF}\textbf{\methodname}
        & \cellcolor[HTML]{EFEFEF}\textbf{14.883}
        & \cellcolor[HTML]{EFEFEF}\textbf{78.694}
        & \cellcolor[HTML]{EFEFEF}\textbf{0.919}
        & \cellcolor[HTML]{EFEFEF}\textbf{15.675}
        & \cellcolor[HTML]{EFEFEF}93.538
        & \cellcolor[HTML]{EFEFEF}\textbf{0.836}
        & \cellcolor[HTML]{EFEFEF}14.902
        & \cellcolor[HTML]{EFEFEF}\textbf{73.334}
        & \cellcolor[HTML]{EFEFEF}\textbf{0.861}
        & \cellcolor[HTML]{EFEFEF}\textbf{14.125}
        & \cellcolor[HTML]{EFEFEF}\textbf{87.000}
        & \cellcolor[HTML]{EFEFEF}\textbf{0.860} \\
    \midrule

     & No Defense & 1.106 & 0.992 & --  & 1.113 & 1.094 & --  & 1.074 & 0.714 & --  & 1.086 & 0.996 & -- \\ \cmidrule(l{0pt}r{0pt}){2-14}
     & Spectral   & 1.427 & 1.080 & 0.428 & 1.497 & 1.327 & 0.460 & 1.320 & 1.067 & 0.572 & 1.308 & 1.412 & 0.563 \\
     & TED        & 1.394 & 1.106 & 0.448 & 1.484 & 1.243 & 0.435 & 1.337 & 1.049 & 0.561 & 1.287 & 1.352 & 0.554 \\
     & TED++      & 1.491 & 1.152 & 0.440 & 1.458 & 1.253 & 0.446 & 1.329 & 1.065 & 0.569 & 1.369 & 1.434 & 0.550 \\ \cmidrule(l{0pt}r{0pt}){2-14}
     & Fine-tuning    & 1.417 & 1.338 & 0.520 & 1.434 & 2.133 & 0.632 & 1.363 & 1.165 & 0.587 & 1.371 & 1.477 & 0.559 \\
     & Fine-pruning   & 1.407 & 1.392 & 0.536 & 1.430 & 1.994 & 0.615 & 1.368 & 1.153 & 0.583 & 1.359 & 1.421 & 0.549 \\
     & NAD            & 1.459 & 1.330 & 0.506 & 1.335 & 1.450 & 0.540 & 1.372 & 1.059 & 0.554 & 1.387 & 1.457 & 0.550 \\
     & IMS            & 1.523 & 1.250 & 0.466 & \textbf{1.114} & 1.093 & 0.500 & 1.320 & 0.980 & 0.542 & 1.355 & 1.409 & 0.548 \\ \cmidrule(l{0pt}r{0pt}){2-14}
     & ABL        & 1.548 & 1.331 & 0.485 & 1.515 & 1.232 & 0.424 & 1.412 & 0.992 & 0.520 & 1.281 & 1.386 & 0.565 \\
     & PDB        & 1.407 & 1.779 & 0.614 & 1.561 & 1.838 & 0.559 & 1.519 & 1.728 & 0.647 & 1.426 & 1.752 & 0.597 \\
     & ESTI       & 1.583 & 1.624 & 0.544 & 1.395 & 1.443 & 0.520 & \textbf{1.180} & 1.014 & 0.603 & \textbf{1.213} & 1.359 & 0.581 \\ \cmidrule(l{0pt}r{0pt}){2-14}
    \multirow{-12}{*}[2ex]{ETTm1}
        & \cellcolor[HTML]{EFEFEF}\textbf{\methodname}
        & \cellcolor[HTML]{EFEFEF}\textbf{1.239}
        & \cellcolor[HTML]{EFEFEF}\textbf{5.852}
        & \cellcolor[HTML]{EFEFEF}\textbf{0.862}
        & \cellcolor[HTML]{EFEFEF}1.291
        & \cellcolor[HTML]{EFEFEF}\textbf{5.013}
        & \cellcolor[HTML]{EFEFEF}\textbf{0.822}
        & \cellcolor[HTML]{EFEFEF}1.378
        & \cellcolor[HTML]{EFEFEF}\textbf{2.453}
        & \cellcolor[HTML]{EFEFEF}\textbf{0.744}
        & \cellcolor[HTML]{EFEFEF}1.279
        & \cellcolor[HTML]{EFEFEF}\textbf{2.103}
        & \cellcolor[HTML]{EFEFEF}\textbf{0.688} \\
    \bottomrule
    \end{tabular}
    }
\end{table*}

\begin{table*}[!b]
    \centering
    \caption{Detection performance comparison of inference-time defenses on three datasets on \emph{TimesNet} model. Best results are in \textbf{bold}.  Higher AUC and F1 indicates better detection performance.}
    \label{app-tab:main-defense-result-timesnet-inf}
    \setlength{\tabcolsep}{6pt}
    \renewcommand{\arraystretch}{1.4}
    \renewcommand{\aboverulesep}{0pt}
    \renewcommand{\belowrulesep}{0pt}
    \setlength\cellspacetoplimit{2pt}
    \setlength\cellspacebottomlimit{2pt}
    \scriptsize

    \begin{tabular}{l@{\hskip 10pt}l@{\hskip 10pt}cccccccccc}
    \toprule
    \multirow{2}{*}{\textbf{Dataset}}
        & \multirow{2}{*}{\textbf{Defense}}
        & \multicolumn{2}{c}{\textbf{Random}}
        & \multicolumn{2}{c}{\textbf{Manhattan}}
        & \multicolumn{2}{c}{\textbf{FreqBack-TSF}}
        & \multicolumn{2}{c}{\textbf{BackTime}}
        & \multicolumn{2}{c}{\textbf{\textsc{Average}}} \\
    \cmidrule(r){3-4}
    \cmidrule(lr){5-6}
    \cmidrule(lr){7-8}
    \cmidrule(lr){9-10}
    \cmidrule(lr){11-12}
     &  & AUC $\uparrow$ & F1 $\uparrow$
        & AUC $\uparrow$ & F1 $\uparrow$
        & AUC $\uparrow$ & F1 $\uparrow$
        & AUC $\uparrow$ & F1 $\uparrow$
        & AUC $\uparrow$ & F1 $\uparrow$ \\
    \midrule

    \multirow{4}{*}[-0.5ex]{\textbf{PEMS03}}
        & No Defense & 0.500 & 0.500 & 0.500 & 0.500 & 0.500 & 0.500 & 0.500 & 0.500 & 0.500 & 0.500 \\ \cmidrule(l{0pt}){2-12}
        & STRIP      & \textbf{0.536} & \textbf{0.545} & \textbf{0.539} & \textbf{0.563} & \textbf{0.457} & \textbf{0.502} & \textbf{0.507} & \textbf{0.513} & \textbf{0.510} & \textbf{0.531} \\
        & TeCo       & 0.437 & 0.506 & 0.442 & 0.503 & 0.379 & 0.501 & 0.428 & 0.500 & 0.422 & 0.503 \\
        & IBD-PSC    & 0.372 & 0.516 & 0.375 & 0.511 & 0.434 & \textbf{0.502} & 0.468 & 0.500 & 0.412 & 0.507 \\
    \midrule

    \multirow{4}{*}[-0.5ex]{\textbf{Weather}}
        & No Defense & 0.500 & 0.500 & 0.500 & 0.500 & 0.500 & 0.500 & 0.500 & 0.500 & 0.500 & 0.500 \\ \cmidrule(l{0pt}){2-12}
        & STRIP      & 0.171 & 0.503 & \textbf{0.507} & 0.514 & 0.505 & 0.525 & 0.545 & 0.537 & 0.432 & 0.520 \\
        & TeCo       & 0.398 & 0.503 & 0.357 & 0.504 & 0.372 & 0.500 & 0.431 & 0.500 & 0.390 & 0.502 \\
        & IBD-PSC    & \textbf{0.573} & \textbf{0.555} & 0.497 & \textbf{0.521} & \textbf{0.596} & \textbf{0.599} & \textbf{0.563} & \textbf{0.580} & \textbf{0.557} & \textbf{0.564} \\
    \midrule

    \multirow{4}{*}[-0.5ex]{\textbf{ETTm1}}
        & No Defense & 0.500 & 0.500 & 0.500 & 0.500 & 0.500 & 0.500 & 0.500 & 0.500 & 0.500 & 0.500 \\ \cmidrule(l{0pt}){2-12}
        & STRIP      & 0.601 & 0.573 & 0.497 & \textbf{0.518} & 0.458 & 0.502 & 0.475 & 0.503 & 0.508 & 0.524 \\
        & TeCo       & \textbf{0.633} & \textbf{0.604} & 0.386 & 0.500 & \textbf{0.763} & \textbf{0.700} & \textbf{0.527} & \textbf{0.523} & \textbf{0.577} & \textbf{0.582} \\
        & IBD-PSC    & 0.307 & 0.500 & \textbf{0.501} & 0.513 & 0.455 & 0.501 & 0.484 & 0.501 & 0.437 & 0.504 \\
    \bottomrule
    \end{tabular}x
\end{table*}
\clearpage
\textbf{Robustness against BadTime attack.} Beyond evaluating defenses against three SOTA attacks, we also evaluate \methodname{} against the recent TSF backdoor attack BadTime~\cite{xiang2025badtime}. BadTime leverages inter-variable correlations, temporal lags, and data-driven initialization to construct distributed, lag-aware triggers for effective and stealthy attacks. For our BadTime implementation, we attempt to replicate the method using its default hyperparameters. For defense evaluation, after obtaining the fixed BadTime trigger, we poison the datasets and then apply the PDB and \methodname{} training pipelines. \tableautorefname~\ref{app-tab:badtime-attack} shows that \methodname{} remains effective under this recent attack setting and achieves the best overall trade-off, attaining the highest \FDER{} of 0.847 while maintaining competitive clean forecasting performance, even outperforming vanilla training in terms of clean \MAEC{}.
    \begin{table}[htbp]
        \vspace{-5pt}
        \centering
        \caption{Defense performance of PDB~\cite{wei2024mitigating} and \methodname{} under BadTime~\cite{xiang2025badtime} a on PEMS03 dataset, where FEDFormer, SimpleTM, and TimesNet are the victim models. Best results are in \textbf{bold}.}
        \label{app-tab:badtime-attack}
        \setlength{\tabcolsep}{2pt} 
        \renewcommand{\arraystretch}{1.2} 
        \tiny
        \renewcommand{\aboverulesep}{0pt}
        \renewcommand{\belowrulesep}{0pt}
        \setlength\cellspacetoplimit{2pt}
        \setlength\cellspacebottomlimit{2pt}
        \resizebox{\linewidth}{!}{
        \begin{tabular}{lcccccccccccc}
            \toprule
            \textbf{Model~→}
            & \multicolumn{3}{c}{\textbf{FEDformer}}
            & \multicolumn{3}{c}{\textbf{SimpleTM}}
            & \multicolumn{3}{c}{\textbf{TimesNet}} 
            & \multicolumn{3}{c}{\textbf{\textsc{Average}}} 
            \\
            \cmidrule(l{0pt}){1-1}
            \cmidrule(l){2-4}
            \cmidrule(l){5-7}
            \cmidrule(l){8-10}
            \cmidrule(l){11-13}
            \textbf{Defense~↓}
            & \MAEC~↓ & \MAEP~↑ & \FDER~↑
            & \MAEC~↓ & \MAEP~↑ & \FDER~↑
            & \MAEC~↓ & \MAEP~↑ & \FDER~↑ 
            & \MAEC~↓ & \MAEP~↑ & \FDER~↑ 
            \\
            \midrule
            No Defense 
                & 20.399 & 17.797 & --
                & 20.940 & 18.640 & --
                & 23.755 & 26.316 & -- 
                & 21.698 & 20.918 & -- \\
            PDB~\cite{wei2024mitigating}
                & 19.372 & 36.976 & 0.759
                & 21.455 & \textbf{38.732} & 0.747
                & 24.390 & 37.878 & 0.640 
                & 21.739 & 37.862 & 0.715 \\
            \cellcolor[HTML]{EFEFEF}\textbf{\methodname}
                & \cellcolor[HTML]{EFEFEF}\textbf{17.287}
                & \cellcolor[HTML]{EFEFEF}\textbf{37.524}
                & \cellcolor[HTML]{EFEFEF}\textbf{0.867}
                & \cellcolor[HTML]{EFEFEF}\textbf{19.109}
                & \cellcolor[HTML]{EFEFEF}37.165 
                & \cellcolor[HTML]{EFEFEF}\textbf{0.918}
                & \cellcolor[HTML]{EFEFEF}\textbf{21.600}
                & \cellcolor[HTML]{EFEFEF}\textbf{39.563} 
                & \cellcolor[HTML]{EFEFEF}\textbf{0.756} 
                & \cellcolor[HTML]{EFEFEF}\textbf{19.332} 
                & \cellcolor[HTML]{EFEFEF}\textbf{38.084} 
                & \cellcolor[HTML]{EFEFEF}\textbf{0.847} \\
            \bottomrule
        \end{tabular}
        }
    \vspace{-12pt}
    \end{table}

\textbf{Generalization to additional architectures.} Beyond the three backbone forecasters used in our main experiments, two Transformer-based models (FEDformer~\cite{zhou2022fedformer} and SimpleTM~\cite{chen2025simpletm}) and one CNN-based model (TimesNet~\cite{wu2023timesnet}), we further evaluate \methodname{} on a broader set of TSF architectures under the Random and BackTime attacks on PEMS03 dataset. Concretely, we include SegRNN~\cite{lin2023segrnn} (RNN-based), SOFTS~\cite{han2024softs} and TimeMixer~\cite{wang2024timemixer} (MLP-based), and AutoTimes~\cite{liu2024autotimes} as an emerging LLM-based forecaster with two large LLM variants (GPT2~\cite{radford2019language} and OPT-1.3B~\cite{zhang2022opt}). As shown in Table~\ref{app-tab:different-models}, \methodname{} consistently attains \MAEP{} above 32 and \FDER{} above 0.68, while incurring at most a 10\% relative increase in \MAEC{} across two attacks. Specifically, on  the LLM-based method (AutoTimes), \methodname{} yielding at least a $5.14\times$ \MAEP{} gain with only a $3.8\%$ change in clean \MAEC{}. Overall, these results support that \methodname{} is architecture-agnostic and remains effective across diverse forecasting architectures.
\begin{table}[htbp]
    \vspace{-5pt}
    \centering
    \caption{Defense performance across 8 models with different architectures under Random and BackTime attacks on PEMS03 dataset. }
    \label{app-tab:different-models}
    \scriptsize
    \setlength{\tabcolsep}{4pt}
    \renewcommand{\arraystretch}{1.25}
    \resizebox{\linewidth}{!}{
  \begin{tabular}{l@{\hspace{4pt}}cccccc@{\hspace{5pt}}ccccc}
    \toprule
    \textbf{Attack}~→
    & \multicolumn{5}{c}{\textbf{Random}}
    & \multicolumn{5}{c}{\textbf{BackTime}} \\
    \cmidrule(lr){1-1}
    \cmidrule(lr){2-6}
    \cmidrule(lr){7-11}

    \textbf{Defense}~→
    & \multicolumn{2}{c}{\textbf{No Defense}}
    & \multicolumn{3}{c}{\textbf{\methodname{}}}
    & \multicolumn{2}{c}{\textbf{No Defense}}
    & \multicolumn{3}{c}{\textbf{\methodname{}}} \\
    \cmidrule(lr){1-1}
    \cmidrule(lr){2-3}
    \cmidrule(lr){4-6}
    \cmidrule(lr){7-8}
    \cmidrule(lr){9-11}

    \textbf{Model}~↓
    & \MAEC~↓ & \MAEP~↑
    & \MAEC~↓ & \MAEP~↑ & \FDER~↑
    & \MAEC~↓ & \MAEP~↑
    & \MAEC~↓ & \MAEP~↑ & \FDER~↑ \\
    \midrule

    FEDformer~\cite{zhou2022fedformer}
    & 16.286 & 14.959 &  16.607 & 100.436 & 0.916
    & 16.093 & 10.760  & 16.840 & 41.232 & 0.847 \\

    SimpleTM~\cite{chen2025simpletm}
    & 17.510 & 19.007   & 17.489 & 173.700 & 0.945
    & 17.268 & 9.131   & 17.243 & 36.626 & 0.875 \\

    TimesNet~\cite{wu2023timesnet}
    & 19.104 & 19.351  & 19.687 & 39.894 & 0.743
    & 19.459 & 22.713  & 20.061 & 40.052 & 0.701 \\

    SegRNN~\cite{lin2023segrnn}
    & 19.889 & 8.953   & 20.469 & 205.044 & 0.964
    & 19.980 & 6.927   & 20.718 & 33.941 & 0.880 \\

    SOFTS~\cite{han2024softs}
    & 16.263 & 2.930   & 16.871 & 170.169 & 0.973
    & 16.227 & 3.185   & 17.451 & 33.340 & 0.917 \\

    TimeMixer~\cite{wang2024timemixer}
    & 21.540 & 19.917  & 21.351 & 220.274 & 0.955
    & 21.484 & 21.053  & 21.440 & 33.662 & 0.687 \\

    $\text{AutoTimes}_{\text{GPT2}}$~\cite{liu2024autotimes}
    & 20.984 & 19.346   & 21.292 & 215.993 & 0.948
    & 21.006 & 6.239   & 21.805 & 32.046 & 0.884 \\

    $\text{AutoTimes}_{\text{OPT1B}}$~\cite{liu2024autotimes}
    & 20.911 & 22.946   & 21.054 & 221.067 & 0.945
    & 20.921 & 6.162   & 21.196 & 33.931 & 0.903 \\
    \bottomrule
    \end{tabular}
    }
    \vspace{-12pt}
\end{table}

\textbf{Defense performance under large-scale TSF foundation models.} Beyond  $\text{AutoTimes}_{\text{GPT2}}$ and  $\text{AutoTimes}_{\text{OPT1B}}$~\cite{liu2024autotimes},we further evaluate $\text{AutoTimes}_{\text{LLaMA7B}}$, an AutoTimes variant built on the large-scale LLaMA-7B foundation model~\cite{touvron2023llama}. \tableautorefname~\ref{app-tab:large-scale-model} shows that \methodname{} still outperforms PDB while keeping the total training time to $1.431\times$ that of undefended training and comparable to PDB, i.e., approximately 70,000 seconds. These results suggest that \methodname{} transfers beyond standard TSF backbones and remains effective for large-scale TSF foundation models.
    \begin{table}[htbp]
        \vspace{-5pt}
        \centering
        \caption{Defense performance and training time (in seconds) of PDB~\cite{wei2024mitigating} and \methodname{} under BackTime on the PEMS03 dataset, using $\text{AutoTimes}_{\text{LLaMA7B}}$~\cite{liu2024autotimes} as the victim model. Best results are shown in \textbf{bold}.}
        \label{app-tab:large-scale-model}
        \scriptsize
        \setlength{\tabcolsep}{3.5pt}
        \renewcommand{\arraystretch}{1.25}
        \begin{tabular}{lcccc}
            \toprule
            \textbf{Defense}
            & \textbf{Training time (s)~↓}
            & \textbf{\MAEC~↓} & \textbf{\MAEP~↑} & \textbf{\FDER~↑} \\
            \midrule
            No Defense 
                & 49027.6 & 20.977 & 6.004 & -- \\
            PDB~\cite{wei2024mitigating}
                & \textbf{69997.6} & 22.657 & 25.798 & 0.847 \\
            \textbf{\methodname}
                & 70162.8 & \textbf{21.385} & \textbf{32.792} & \textbf{0.899} \\
            \bottomrule
        \end{tabular}
    \end{table}

\textbf{Generalization to different poisoning rates.} We evaluate the robustness of \methodname{} under varying attack budgets by adjusting the temporal poisoning rate $\eta_\text{T}$ and spatial poisoning rate $\eta_\text{S}$ using the BackTime attack on the PEMS03 dataset. As shown in Figures~\ref{app-fig:poison-rate-fedformer}–\ref{app-fig:poison-rate-TimesNet}, \methodname{} consistently maintains strong defense effectiveness across all poisoning rates, with \MAEP{} remaining above 35 for all three models. Additionally, clean performance stays within 5\% even at high poisoning rates ($\eta_{\text{T}} = 0.04$, $\eta_{\text{S}} = 0.4$). These results demonstrate that \methodname{} is robust to varying poisoning rates while maintaining reasonable clean performance.

    \begin{figure}[htbp]
        \centering
        \vspace{-3pt}
        \begin{subfigure}{0.24\linewidth}
            \includegraphics[width=\linewidth]{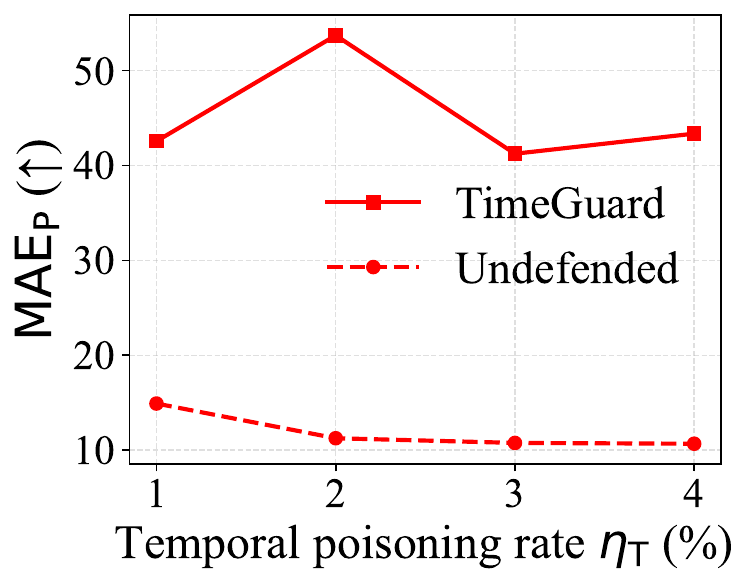}
        \end{subfigure}
        \begin{subfigure}{0.24\linewidth}
            \includegraphics[width=\linewidth]{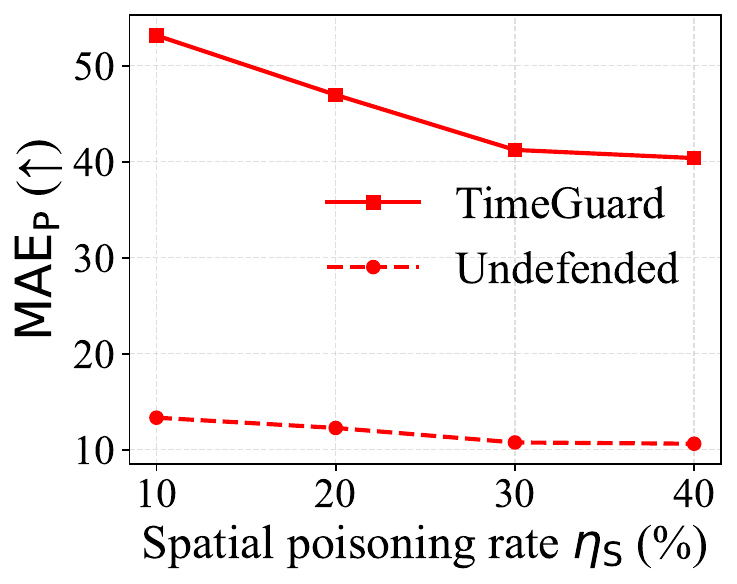}
        \end{subfigure}
        \begin{subfigure}{0.24\linewidth}
            \includegraphics[width=\linewidth]{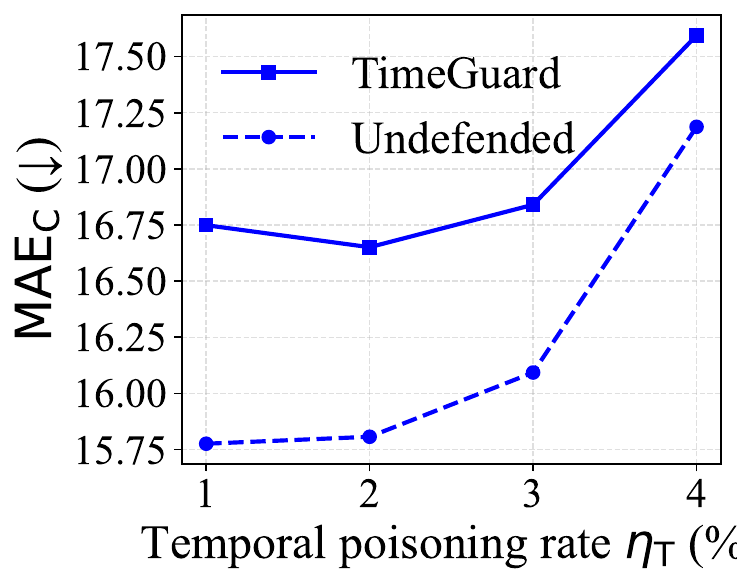}
        \end{subfigure}
        \begin{subfigure}{0.24\linewidth}
            \includegraphics[width=\linewidth]{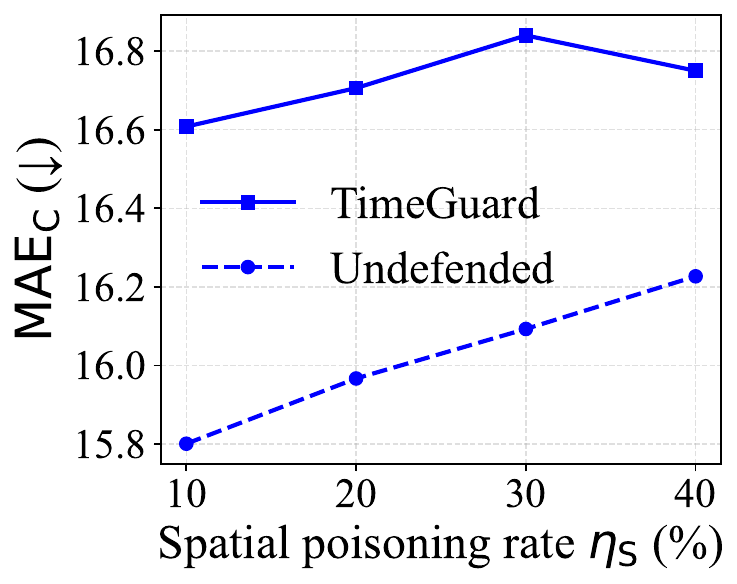}
        \end{subfigure}
        \caption{Defense performance of \methodname{} (\MAEP{} and \MAEC{}) under varying temporal and spatial poisoning rates of the BackTime attack on the PEMS03 dataset with the \emph{FEDformer} model. }        
        \label{app-fig:poison-rate-fedformer}
    \end{figure}

    \begin{figure}[htbp]
        \centering
        \vspace{-3pt}
        \begin{subfigure}{0.24\linewidth}
            \includegraphics[width=\linewidth]{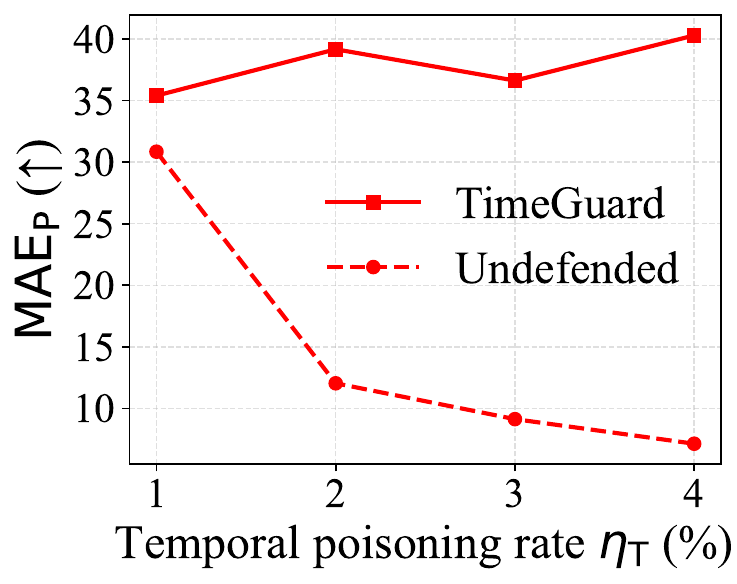}
        \end{subfigure}
        \begin{subfigure}{0.24\linewidth}
            \includegraphics[width=\linewidth]{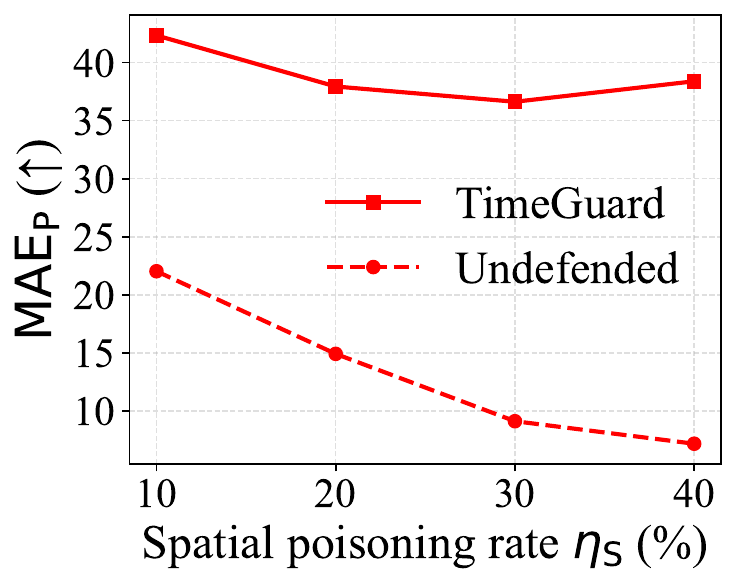}
        \end{subfigure}
        \begin{subfigure}{0.24\linewidth}
            \includegraphics[width=\linewidth]{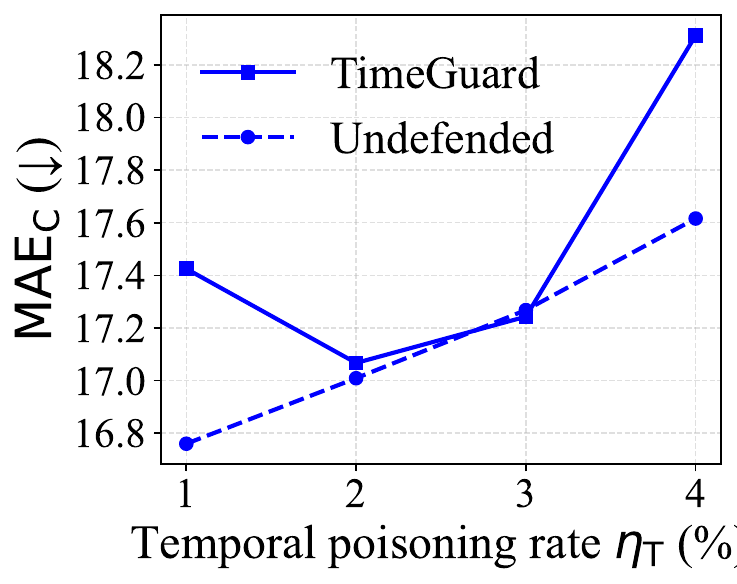}
        \end{subfigure}
        \begin{subfigure}{0.24\linewidth}
            \includegraphics[width=\linewidth]{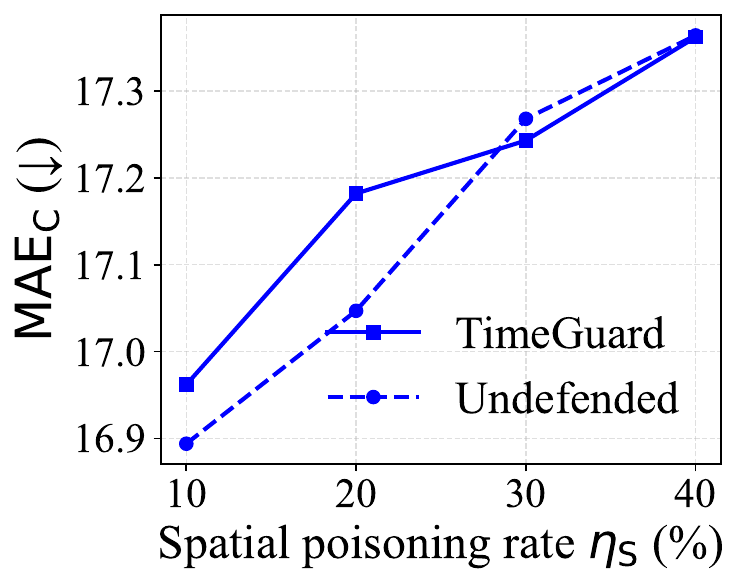}
        \end{subfigure}

         \caption{Defense performance of \methodname{} (\MAEP{} and \MAEC{}) under varying temporal and spatial poisoning rates of the BackTime attack on the PEMS03 dataset with the \emph{SimpleTM} model.}              
         \label{app-fig:poison-rate-simpleTM}
    \end{figure}

    \begin{figure}[htbp]
        \centering
        \vspace{-3pt}
        \begin{subfigure}{0.24\linewidth}
            \includegraphics[width=\linewidth]{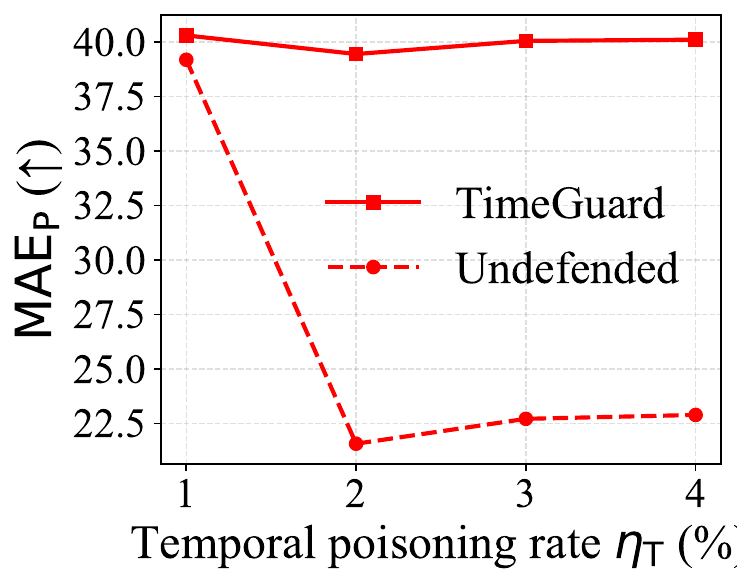}
        \end{subfigure}
        \begin{subfigure}{0.24\linewidth}
            \includegraphics[width=\linewidth]{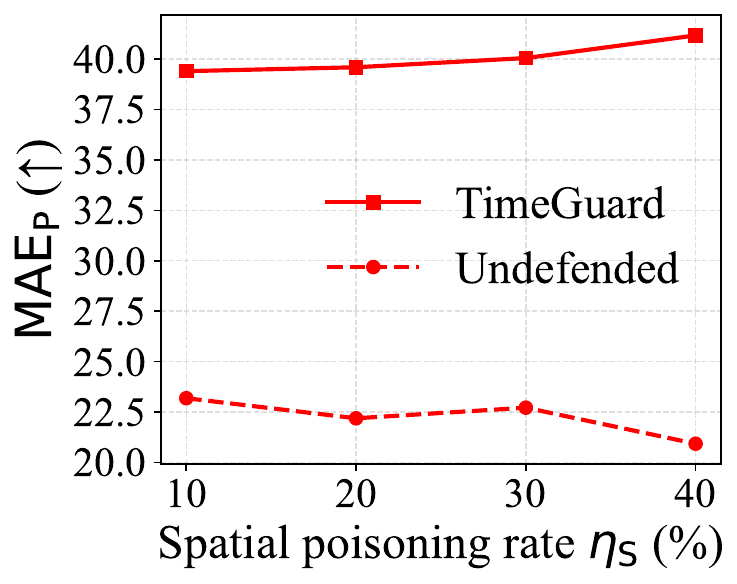}
        \end{subfigure}
        \begin{subfigure}{0.24\linewidth}
            \includegraphics[width=\linewidth]{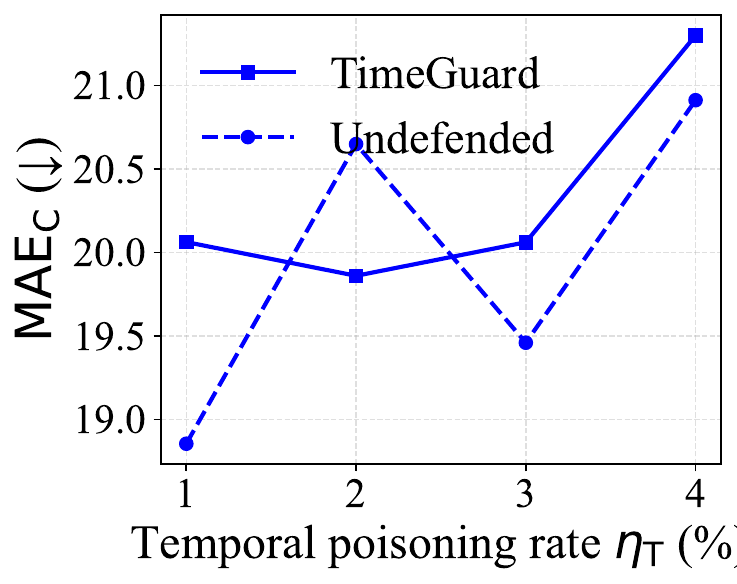}
        \end{subfigure}
        \begin{subfigure}{0.24\linewidth}
            \includegraphics[width=\linewidth]{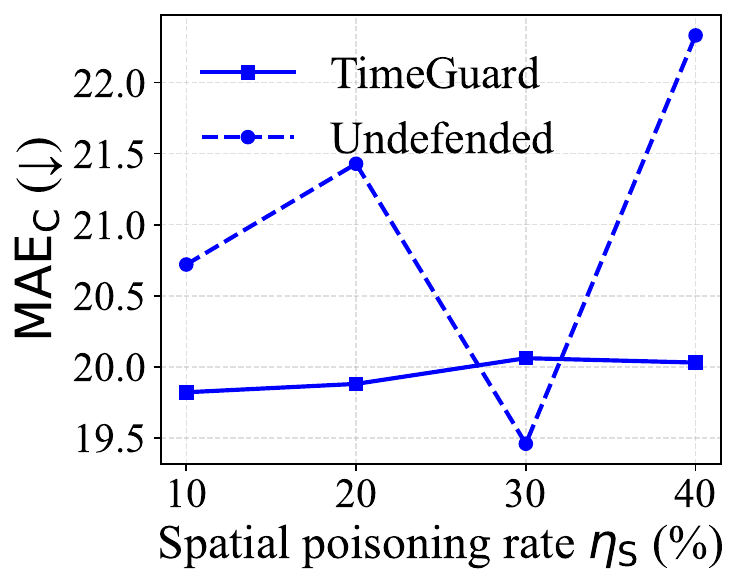}
        \end{subfigure}

         \caption{Defense performance of \methodname{} (\MAEP{} and \MAEC{}) under varying temporal and spatial poisoning rates of the BackTime attack on the PEMS03 dataset with the \emph{TimesNet} model.}   
        \label{app-fig:poison-rate-TimesNet}
    \end{figure}

\textbf{Generalization to the extreme case of full-channel poisoning.} Our motivation is strongest under partial-channel poisoning, which is the common setting in existing multivariate TSF backdoor attacks; as the channel poisoning ratio increases, attacks generally become less stealthy and easier to detect. Nevertheless, \methodname{} does not require poisoning to affect only a strict subset of channels: its channel-wise formulation remains applicable even when poisoning is dense across channels. To directly test the all-channel case, we evaluate \methodname{} on PEMS03 under BackTime with spatial poisoning ratio $\eta_\mathrm{S}=1.0$, meaning that all channels are poisoned. As shown in \tableautorefname~\ref{app-tab:full-channel-poisoning}, \methodname{} remains effective in this setting and achieves the highest \FDER{} of 0.748.
    \begin{table}[!htbp]
        \centering
        \vspace{-5pt}
        \caption{Defense performance of PDB and \methodname{} under BackTime attack on the PEMS03 dataset with full-channel poisoning, i.e., $\eta_\mathrm{S}=1.0$, where FEDformer, SimpleTM, and TimesNet are used as victim models. Best results are shown in \textbf{bold}.}
        \vspace{-5pt}
        \label{app-tab:full-channel-poisoning}
        \setlength{\tabcolsep}{2pt} 
        \renewcommand{\arraystretch}{1.2}
        \tiny
        \renewcommand{\aboverulesep}{0pt}
        \renewcommand{\belowrulesep}{0pt}
        \setlength\cellspacetoplimit{2pt}
        \setlength\cellspacebottomlimit{2pt}
        \resizebox{\linewidth}{!}{
        \begin{tabular}{lcccccccccccc}
            \toprule
            \textbf{Model~→}
            & \multicolumn{3}{c}{\textbf{FEDformer}}
            & \multicolumn{3}{c}{\textbf{SimpleTM}}
            & \multicolumn{3}{c}{\textbf{TimesNet}} 
            & \multicolumn{3}{c}{\textbf{\textsc{Average}}} 
            \\
            \cmidrule(l{0pt}){1-1}
            \cmidrule(l){2-4}
            \cmidrule(l){5-7}
            \cmidrule(l){8-10}
            \cmidrule(l){11-13}
            \textbf{Defense~↓}
            & \MAEC~↓ & \MAEP~↑ & \FDER~↑
            & \MAEC~↓ & \MAEP~↑ & \FDER~↑
            & \MAEC~↓ & \MAEP~↑ & \FDER~↑ 
            & \MAEC~↓ & \MAEP~↑ & \FDER~↑ 
            \\
            \midrule
            No Defense 
                & 18.025 & 11.586 & --
                & 18.567 & 5.817 & --
                & 27.473 & 21.503 & -- 
                & 21.355 & 12.969 & -- \\
            PDB~\cite{wei2024mitigating}
                & 18.308 & 16.074 & 0.632
                & \textbf{19.114} & 13.355 & \textbf{0.768}
                & 23.731 & 32.485 & 0.669 
                & 20.384 & 20.638 & 0.690 \\
            \cellcolor[HTML]{EFEFEF}\textbf{\methodname}
                & \cellcolor[HTML]{EFEFEF}\textbf{18.176}
                & \cellcolor[HTML]{EFEFEF}\textbf{26.268}
                & \cellcolor[HTML]{EFEFEF}\textbf{0.775}
                & \cellcolor[HTML]{EFEFEF}19.464
                & \cellcolor[HTML]{EFEFEF}\textbf{13.938} 
                & \cellcolor[HTML]{EFEFEF}\textbf{0.768}
                & \cellcolor[HTML]{EFEFEF}\textbf{21.974}
                & \cellcolor[HTML]{EFEFEF}\textbf{35.797} 
                & \cellcolor[HTML]{EFEFEF}\textbf{0.700} 
                & \cellcolor[HTML]{EFEFEF}\textbf{19.871} 
                & \cellcolor[HTML]{EFEFEF}\textbf{25.335} 
                & \cellcolor[HTML]{EFEFEF}\textbf{0.748} \\
            \bottomrule
        \end{tabular}
        }
        \vspace{-12pt}
    \end{table}    

\textbf{Generalization to different attack patterns.} We further evaluate the robustness of \methodname{} under three attack patterns from the original BackTime work~\cite{lin2024backtime} (described in Appendix~\ref{app:attack-patterns}) on the PEMS03 dataset. As shown in Table~\ref{app-tab:attack-patterns}, \methodname{} consistently demonstrates strong robustness across all attack patterns, including Random, Manhattan, and BackTime attacks, averaged over the three models. Compared to the state-of-the-art defense PDB, \methodname{} significantly improves both \MAEP{} and \FDER{}, while maintaining clean performance with minimal degradation. The results for the up-and-down and up-trend attack patterns, broken down by model, are provided in Tables~\ref{app-tab:attack-patterns-up-and-down} and~\ref{app-tab:attack-patterns-up-trend}, respectively.
\begin{table}[!htbp]
    \centering
    \vspace{-5pt}
    \caption{Defense performance of \methodname{} across three different attack patterns under Random, Manhattan, and BackTime attacks on PEMS03, \emph{average} over FEDFormer, SimpleTM, and TimesNet models. Best results are in \textbf{bold}.}
    \vspace{-5pt}
    \label{app-tab:attack-patterns}
    
    \setlength{\tabcolsep}{6pt} 
    \renewcommand{\arraystretch}{1.15}
    \scriptsize
    \renewcommand{\aboverulesep}{0pt}
    \renewcommand{\belowrulesep}{0pt}
    \setlength\cellspacetoplimit{1.8pt}
    \setlength\cellspacebottomlimit{1.8pt}
    \resizebox{\linewidth}{!}{
  \begin{tabular}{clcccccccccc}
    \toprule
    \multirow{2}{*}{\textbf{\begin{tabular}[c]{@{}c@{}}Attack\\ Pattern\end{tabular}}}
    & \textbf{Attack~→}
    & \multicolumn{3}{c}{\textbf{Random}}
    & \multicolumn{3}{c}{\textbf{Manhattan}}
    & \multicolumn{3}{c}{\textbf{BackTime}} \\
    \cmidrule(l{0pt}){2-2}
    \cmidrule(l){3-5}
    \cmidrule(l){6-8}
    \cmidrule(l){9-11}
    & \textbf{Defense~↓}
    & \MAEC~↓ & \MAEP~↑ & \FDER~↑
    & \MAEC~↓ & \MAEP~↑ & \FDER~↑
    & \MAEC~↓ & \MAEP~↑ & \FDER~↑ \\
    \midrule

    \multirow{5}{*}{\textbf{Cone}} & No Defense 
        & 17.634 & 17.772 & --
        & 17.722 & 20.266 & --
        & 17.607 & 14.201 & -- \\
    & Fine-tuning
        & 19.003 & 30.909 & 0.625
        & 19.661 & 30.995 & 0.608
        & 18.934 & 18.196 & 0.594 \\
    & Fine-pruning
        & 19.020 & 31.643 & 0.633
        & 19.595 & 34.447 & 0.624
        & 18.686 & 19.736 & 0.623 \\
    & PDB 
        & 18.630 & 54.690 & 0.693
        & 19.308 & 60.477 & 0.708
        & 18.967 & 22.397 & 0.639 \\
    & \cellcolor[HTML]{EFEFEF}\textbf{\methodname}
        &\cellcolor[HTML]{EFEFEF}\textbf{17.928} & \cellcolor[HTML]{EFEFEF}\textbf{104.677} & \cellcolor[HTML]{EFEFEF}\textbf{0.868}
        &\cellcolor[HTML]{EFEFEF}\textbf{17.850} &\cellcolor[HTML]{EFEFEF}\textbf{97.370} & \cellcolor[HTML]{EFEFEF}\textbf{0.854}
        &\cellcolor[HTML]{EFEFEF}\textbf{18.048} &\cellcolor[HTML]{EFEFEF}\textbf{39.303} &\cellcolor[HTML]{EFEFEF}\textbf{0.808} \\
    \midrule

    \multirow{5}{*}{\textbf{\begin{tabular}[c]{@{}c@{}}Up \&\\ Down\end{tabular}}}
    & No Defense 
        & 17.628 & 17.903 & --
        & 17.679 & 20.213 & --
        & 17.389 & 16.853 & -- \\
    & Fine-tuning
        & 19.161 & 26.878 & 0.609
        & 19.196 & 28.158 & 0.588
        & 18.772 & 19.479 & 0.576 \\
    & Fine-pruning
        & 19.078 & 29.024 & 0.629
        & 19.193 & 30.267 & 0.606
        & 18.743 & 20.980 & 0.593 \\
    & PDB 
        & 19.150 & 45.243 & 0.669
        & 19.146 & 42.813 & 0.649
        & 20.028 & 21.720 & 0.567 \\
    & \cellcolor[HTML]{EFEFEF}\textbf{\methodname}
        & \cellcolor[HTML]{EFEFEF}\textbf{18.055} & \cellcolor[HTML]{EFEFEF}\textbf{80.871} & \cellcolor[HTML]{EFEFEF}\textbf{0.835}
        & \cellcolor[HTML]{EFEFEF}\textbf{17.985} & \cellcolor[HTML]{EFEFEF}\textbf{75.454} & \cellcolor[HTML]{EFEFEF}\textbf{0.810} 
        & \cellcolor[HTML]{EFEFEF}\textbf{18.158} & \cellcolor[HTML]{EFEFEF}\textbf{30.165} & \cellcolor[HTML]{EFEFEF}\textbf{0.709} \\
    \midrule

    \multirow{5}{*}{\textbf{\begin{tabular}[c]{@{}c@{}}Up\\ Trend\end{tabular}}}
    & No Defense 
        & 17.721 & 18.472 & --
        & 17.733 & 21.010 & --
        & 17.615 & 13.674 & -- \\
    & Fine-tuning
        & 19.223 & 32.865 & 0.636
        & 19.540 & 33.298 & 0.624
        & 18.935 & 19.706 & 0.629 \\
    & Fine-pruning
        & 19.211 & 33.532 & 0.640
        & 19.620 & 34.720 & 0.622
        & 19.028 & 21.771 & 0.648 \\
    & PDB 
        & 19.227 & 58.857 & 0.683
        & 19.201 & 70.032 & 0.687
        & 19.783 & 22.195 & 0.645 \\
    & \cellcolor[HTML]{EFEFEF}\textbf{\methodname}
        & \cellcolor[HTML]{EFEFEF}\textbf{18.022} & \cellcolor[HTML]{EFEFEF}\textbf{111.006} & \cellcolor[HTML]{EFEFEF}\textbf{0.872}
        & \cellcolor[HTML]{EFEFEF}\textbf{17.985} & \cellcolor[HTML]{EFEFEF}\textbf{101.135} & \cellcolor[HTML]{EFEFEF}\textbf{0.856}
        & \cellcolor[HTML]{EFEFEF}\textbf{18.273} & \cellcolor[HTML]{EFEFEF}\textbf{46.106} & \cellcolor[HTML]{EFEFEF}\textbf{0.834}\\
    \bottomrule
    \end{tabular}
    }
    \vspace{-12pt}
\end{table}

\begin{table}[!b]
    \centering
    \vspace{-5pt}
    \caption{Defense performance of \methodname{} under Random, Manhattan, and BackTime attacks with a \emph{up-and-down attack pattern} on the PEMS03 dataset, where FEDFormer, SimpleTM, and TimesNet are the victim models. Best results are in \textbf{bold}.}
    \label{app-tab:attack-patterns-up-and-down}
    \vspace{-5pt}
    \setlength{\tabcolsep}{6pt} 
    \renewcommand{\arraystretch}{1.15} 
    \scriptsize
    \renewcommand{\aboverulesep}{0pt}
    \renewcommand{\belowrulesep}{0pt}
    \setlength\cellspacetoplimit{1.8pt}
    \setlength\cellspacebottomlimit{1.8pt}
    \resizebox{\linewidth}{!}{
    \begin{tabular}{clcccccccccc}
    \toprule
    \multirow{2}{*}{\textbf{\begin{tabular}[c]{@{}c@{}}Model\end{tabular}}}
    & \textbf{Attack~→}
    & \multicolumn{3}{c}{\textbf{Random}}
    & \multicolumn{3}{c}{\textbf{Manhattan}}
    & \multicolumn{3}{c}{\textbf{BackTime}} \\
    \cmidrule(l{0pt}){2-2}
    \cmidrule(l){3-5}
    \cmidrule(l){6-8}
    \cmidrule(l){9-11}
    & \textbf{Defense~↓}
    & \MAEC~↓ & \MAEP~↑ & \FDER~↑
    & \MAEC~↓ & \MAEP~↑ & \FDER~↑
    & \MAEC~↓ & \MAEP~↑ & \FDER~↑ \\
    \midrule

    \multirow{5}{*}{\textbf{FEDformer}} & No Defense 
        & 16.370 & 15.193 & --
        & 16.413 & 18.700 & -- 
        & 15.943 & 10.461 & -- \\
    & Fine-tuning
        & 16.802 & 35.086 & 0.771
        & 16.998 & 35.427 & 0.719
        & 16.559 & 18.153 & 0.693 \\
    & Fine-pruning
        & 16.772 & 38.811 & 0.792
        & 16.882 & 39.564 & 0.750
        & 16.650 & 23.142 & 0.753 \\
    & PDB 
        & 16.898 & 21.148 & 0.625
        & 17.259 & 24.828 & 0.599
        & 17.273 & 16.246 & 0.640 \\
    & \cellcolor[HTML]{EFEFEF}\textbf{\methodname}
        & \cellcolor[HTML]{EFEFEF}\textbf{16.551} & \cellcolor[HTML]{EFEFEF}\textbf{67.162} & \cellcolor[HTML]{EFEFEF}\textbf{0.881}
        & \cellcolor[HTML]{EFEFEF}\textbf{16.564} & \cellcolor[HTML]{EFEFEF}\textbf{65.542} & \cellcolor[HTML]{EFEFEF}\textbf{0.853}
        & \cellcolor[HTML]{EFEFEF}\textbf{16.776} & \cellcolor[HTML]{EFEFEF}\textbf{27.529} & \cellcolor[HTML]{EFEFEF}\textbf{0.785} \\
    \midrule

    \multirow{5}{*}{\textbf{SimpleTM}}
    & No Defense 
        & 17.428 & 19.133 & --
        & 17.461 & 20.735 & --
        & 17.209 & 12.936 & -- \\
    & Fine-tuning
        & 17.534 & 24.886 & 0.613
        & 17.515 & 27.539 & 0.622
        & 17.247 & 16.747 & 0.613 \\
    & Fine-pruning
        & 17.468 & 27.910 & 0.656
        & 17.568 & 29.781 & 0.649
        & 17.264 & 16.289 & 0.601 \\
    & PDB 
        & 17.555 & 91.863 & 0.892
        & 17.610 & 79.668 & 0.866
        & 18.996 & 22.338 & 0.663 \\
    & \cellcolor[HTML]{EFEFEF}\textbf{\methodname}
        & \cellcolor[HTML]{EFEFEF}\textbf{17.335} & \cellcolor[HTML]{EFEFEF}\textbf{140.540} & \cellcolor[HTML]{EFEFEF}\textbf{0.932}
        & \cellcolor[HTML]{EFEFEF}\textbf{17.024} & \cellcolor[HTML]{EFEFEF}\textbf{126.706} & \cellcolor[HTML]{EFEFEF}\textbf{0.918}
        & \cellcolor[HTML]{EFEFEF}\textbf{17.188} & \cellcolor[HTML]{EFEFEF}\textbf{28.395} & \cellcolor[HTML]{EFEFEF}\textbf{0.772} \\
    \midrule

    \multirow{5}{*}{\textbf{TimesNet}}
    & No Defense 
        & 19.087 & 19.383 & --
        & 19.161 & 21.202 & --
        & 19.016 & 27.161 & -- \\
    & Fine-tuning
        & 23.149 & 20.661 & 0.443
        & 23.073 & 21.508 & 0.422
        & 22.508 & 23.535 & 0.422 \\
    & Fine-pruning
        & 22.994 & 20.352 & 0.439
        & 23.131 & 21.456 & 0.420
        & 22.315 & 23.509 & 0.426 \\
    & PDB 
        & 22.996 & 22.720 & 0.488
        & 22.569 & 23.943 & 0.482
        & 23.816 & 26.576 & 0.399 \\
    & \cellcolor[HTML]{EFEFEF}\textbf{\methodname}
        & \cellcolor[HTML]{EFEFEF}\textbf{20.278} & \cellcolor[HTML]{EFEFEF}\textbf{34.911} & \cellcolor[HTML]{EFEFEF}\textbf{0.693}
        & \cellcolor[HTML]{EFEFEF}\textbf{20.368} & \cellcolor[HTML]{EFEFEF}\textbf{34.114} & \cellcolor[HTML]{EFEFEF}\textbf{0.660}
        & \cellcolor[HTML]{EFEFEF}\textbf{20.508} & \cellcolor[HTML]{EFEFEF}\textbf{34.571} & \cellcolor[HTML]{EFEFEF}\textbf{0.571} \\
    \bottomrule
    \end{tabular}
    }
    \vspace{-12pt}
\end{table}

\begin{table}[htbp]
    \centering
    \vspace{-5pt}
    \caption{Defense performance of \methodname{} under Random, Manhattan, and BackTime attacks with \emph{up-trend attack pattern} on PEMS03 dataset, where FEDFormer, SimpleTM, and TimesNet are the victim models. Best results are in \textbf{bold}.}
    \vspace{-5pt}
    \label{app-tab:attack-patterns-up-trend}
    
    \setlength{\tabcolsep}{6pt} 
    \renewcommand{\arraystretch}{1.15} 
    \scriptsize
    \renewcommand{\aboverulesep}{0pt}
    \renewcommand{\belowrulesep}{0pt}
    \setlength\cellspacetoplimit{1.8pt}
    \setlength\cellspacebottomlimit{1.8pt}
    \resizebox{\linewidth}{!}{
    \begin{tabular}{clcccccccccc}
    \toprule
    \multirow{2}{*}{\textbf{\begin{tabular}[c]{@{}c@{}}Model\end{tabular}}}
    & \textbf{Attack~→}
    & \multicolumn{3}{c}{\textbf{Random}}
    & \multicolumn{3}{c}{\textbf{Manhattan}}
    & \multicolumn{3}{c}{\textbf{BackTime}} \\
    \cmidrule(l{0pt}){2-2}
    \cmidrule(l){3-5}
    \cmidrule(l){6-8}
    \cmidrule(l){9-11}
    & \textbf{Defense~↓}
    & \MAEC~↓ & \MAEP~↑ & \FDER~↑
    & \MAEC~↓ & \MAEP~↑ & \FDER~↑
    & \MAEC~↓ & \MAEP~↑ & \FDER~↑ \\
    \midrule

    \multirow{5}{*}{\textbf{FEDformer}} & No Defense 
        & 16.420 & 15.652 & --
        & 16.434 & 18.461 & --
        & 16.105 & 10.772 & -- \\
    & Fine-tuning
        & 16.863 & 50.608 & 0.832
        & 16.864 & 42.857 & 0.772
        & 16.594 & 22.653 & 0.748 \\
    & Fine-pruning
        & 16.792 & 51.944 & 0.838
        & 16.811 & 49.437 & 0.802
        & 16.661 & 28.115 & 0.792 \\
    & PDB 
        & 16.875 & 23.665 & 0.656
        & 17.156 & 29.193 & 0.663
        & 17.347 & 16.996 & 0.647 \\
    & \cellcolor[HTML]{EFEFEF}\textbf{\methodname}
        & \cellcolor[HTML]{EFEFEF}\textbf{16.610} & \cellcolor[HTML]{EFEFEF}\textbf{102.645} & \cellcolor[HTML]{EFEFEF}\textbf{0.918}
        & \cellcolor[HTML]{EFEFEF}\textbf{16.598} & \cellcolor[HTML]{EFEFEF}\textbf{98.314} & \cellcolor[HTML]{EFEFEF}\textbf{0.901}
        & \cellcolor[HTML]{EFEFEF}\textbf{16.855} & \cellcolor[HTML]{EFEFEF}\textbf{50.734} & \cellcolor[HTML]{EFEFEF}\textbf{0.872} \\
    \midrule

    \multirow{5}{*}{\textbf{SimpleTM}} & No Defense 
        & 17.613 & 20.562 & --
        & 17.441 & 24.170 & --
        & 17.302 & 7.999 & -- \\
    & Fine-tuning
        & 17.683 & 26.638 & 0.612
        & 17.659 & 33.546 & 0.634
        & 17.350 & 12.608 & 0.681 \\
    & Fine-pruning
        & 17.661 & 27.625 & 0.626
        & 17.717 & 31.995 & 0.615
        & 17.483 & 13.888 & 0.707 \\
    & PDB 
        & 17.866 & 130.997 & 0.914
        & 17.667 & 157.867 & 0.917
        & 18.218 & 23.245 & 0.803 \\
    & \cellcolor[HTML]{EFEFEF}\textbf{\methodname}
        & \cellcolor[HTML]{EFEFEF}\textbf{17.066} & \cellcolor[HTML]{EFEFEF}\textbf{185.689} & \cellcolor[HTML]{EFEFEF}\textbf{0.945}
        & \cellcolor[HTML]{EFEFEF}\textbf{16.951} & \cellcolor[HTML]{EFEFEF}\textbf{160.888} & \cellcolor[HTML]{EFEFEF}\textbf{0.925}
        & \cellcolor[HTML]{EFEFEF}\textbf{17.263} & \cellcolor[HTML]{EFEFEF}\textbf{41.806} & \cellcolor[HTML]{EFEFEF}\textbf{0.904} \\
    \midrule

    \multirow{5}{*}{\textbf{TimesNet}} & No Defense 
        & 19.129 & 19.203 & --
        & 19.325 & 20.399 & --
        & 19.439 & 22.250 & -- \\
    & Fine-tuning
        & 23.121 & 21.347 & 0.464
        & 24.097 & 23.492 & 0.467
        & 22.862 & 23.857 & 0.459 \\
    & Fine-pruning
        & 23.179 & 21.028 & 0.456
        & 24.332 & 22.729 & 0.448
        & 22.940 & 23.310 & 0.446 \\
    & PDB 
        & 22.939 & 21.908 & 0.479
        & 22.780 & 23.035 & 0.481
        & 23.785 & 26.345 & 0.486 \\
    & \cellcolor[HTML]{EFEFEF}\textbf{\methodname}
        & \cellcolor[HTML]{EFEFEF}\textbf{20.389} & \cellcolor[HTML]{EFEFEF}\textbf{44.684} & \cellcolor[HTML]{EFEFEF}\textbf{0.754}
        & \cellcolor[HTML]{EFEFEF}\textbf{20.404} & \cellcolor[HTML]{EFEFEF}\textbf{44.203} & \cellcolor[HTML]{EFEFEF}\textbf{0.743}
        & \cellcolor[HTML]{EFEFEF}\textbf{20.703} & \cellcolor[HTML]{EFEFEF}\textbf{45.776} & \cellcolor[HTML]{EFEFEF}\textbf{0.726} \\
    \bottomrule
    \end{tabular}
    }
    \vspace{-12pt}
\end{table}
\textbf{Generalization to different forecasting horizons.} Beyond the default forecasting horizon used in BackTime~\cite{lin2024backtime} ($\LOUT{=}12$), we further evaluate \methodname{} under longer horizons with $\LOUT \in \{24,36,48\}$. Following BackTime’s protocol, we assume the attacker knows the forecasting horizon used by the victim model.

As shown in Figures~\ref{app-fig:different-t_out-fedformer}–\ref{app-fig:different-t_out-timesnet}, \methodname{} maintains competitive clean performance (\MAEC{}) across all horizons, and even outperforms undefended training on TimesNet in some cases. As $\LOUT$ increases, BackTime itself becomes less effective (e.g., poisoned \MAEP{} exceeds 30 across models when $\LOUT{=}48$), which correspondingly lowers \FDER{}; nevertheless, \methodname{} remains robust, with defended \MAEP{} staying above 28.4 in all settings. We observe that, except for TimesNet, the defense effectiveness of \methodname{} lightly decreases for FEDformer and SimpleTM at longer horizons ($\LOUT\in{36,48}$). A plausible explanation is that losses over longer target windows become more diluted across distant time steps, reducing the discriminability used by DRLS (Eq.~\ref{eq:drls}). A potential remedy is to adopt a weighted loss that prioritizes nearer horizons, which we leave for future work. Overall, these results indicate that \methodname{} remains effective against TSF backdoor attacks under different forecasting horizons.
\vspace{-5pt}
\begin{figure}[!htbp]
    \centering
    \begin{subfigure}{0.322\linewidth}
        \includegraphics[width=\linewidth]{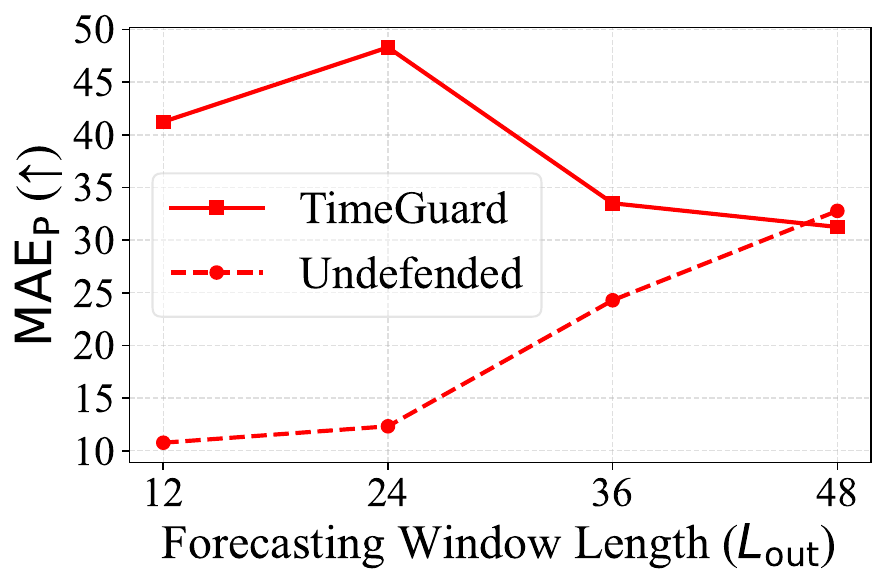}
    \end{subfigure}
    \hfill
    \begin{subfigure}{0.322\linewidth}
        \includegraphics[width=\linewidth]{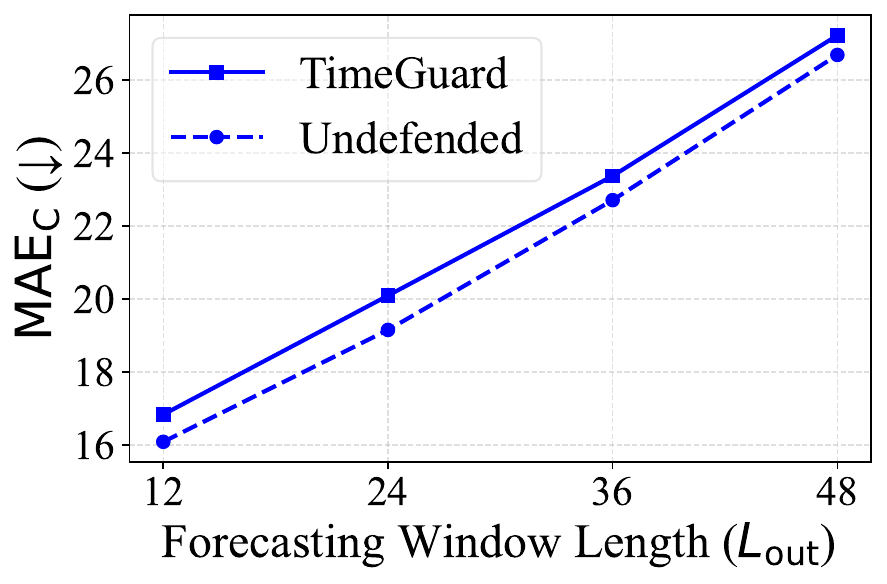}
    \end{subfigure}
    \hfill
    \begin{subfigure}{0.322\linewidth}
        \includegraphics[width=\linewidth]{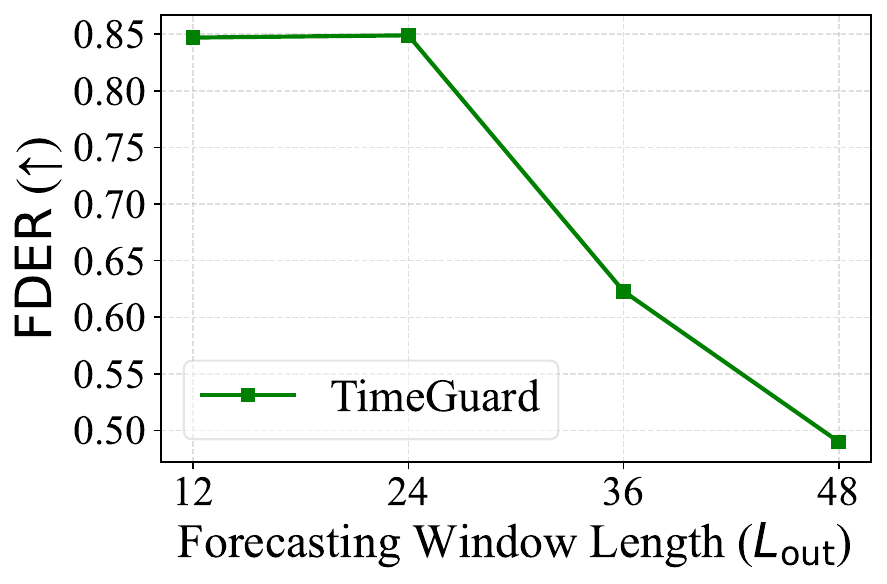}
    \end{subfigure}
    \caption{Defense performance of \methodname{} (\MAEP{}, \MAEC{}, and \FDER{}) under different forecasting window length $\LOUT$ of the BackTime attack on the PEMS03 dataset with the \emph{FEDformer} model.}                
    \label{app-fig:different-t_out-fedformer}
    \vspace{-12pt}
\end{figure}

\begin{figure}[!htbp]
    \centering
    \begin{subfigure}{0.322\linewidth}
        \includegraphics[width=\linewidth]{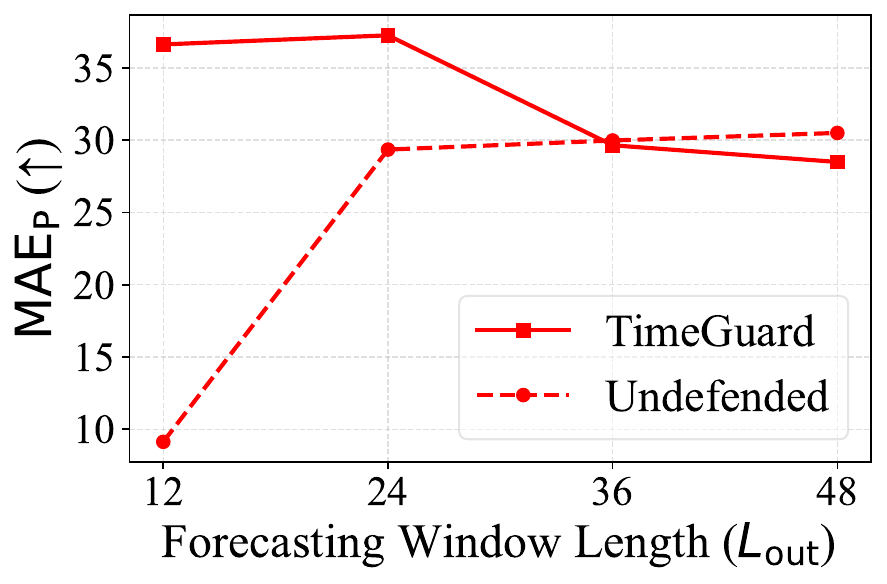}
    \end{subfigure}
    \hfill
    \begin{subfigure}{0.322\linewidth}
        \includegraphics[width=\linewidth]{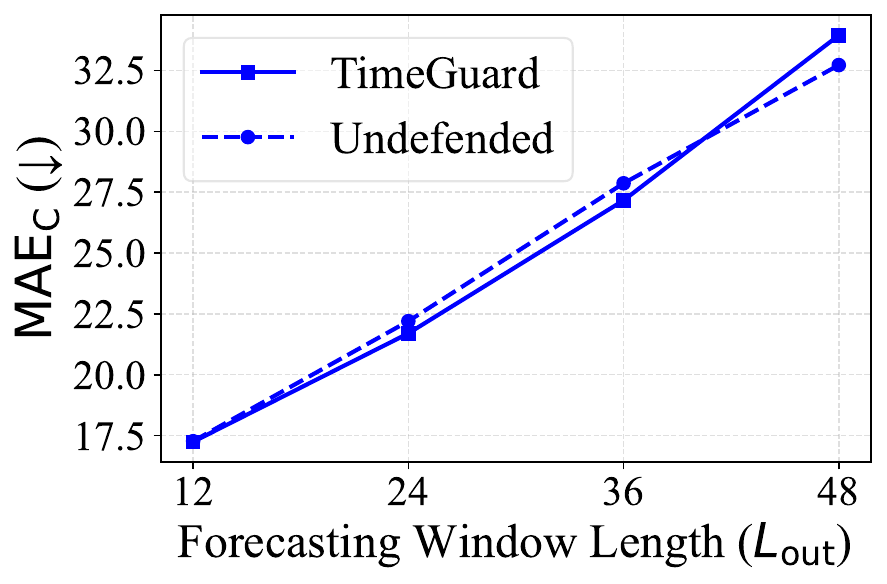}
    \end{subfigure}
    \hfill
    \begin{subfigure}{0.322\linewidth}
        \includegraphics[width=\linewidth]{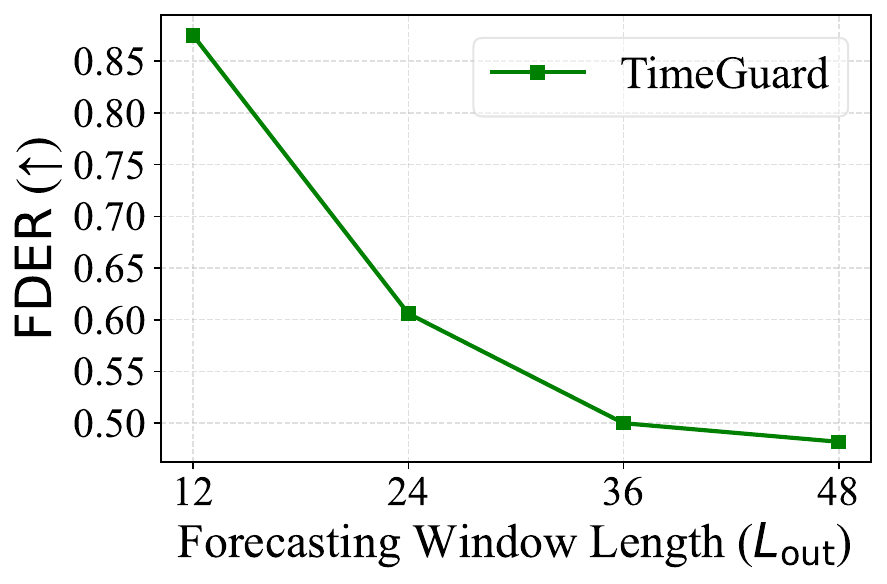}
    \end{subfigure}
    \caption{Defense performance of \methodname{} (\MAEP{}, \MAEC{}, and \FDER{}) under different forecasting window length $\LOUT$ of the BackTime attack on the PEMS03 dataset with the \emph{SimpleTM} model.}                
    \label{app-fig:different-t_out-simpletm}
\end{figure}

\begin{figure}[!htbp]
    \centering
    \begin{subfigure}{0.322\linewidth}
        \includegraphics[width=\linewidth]{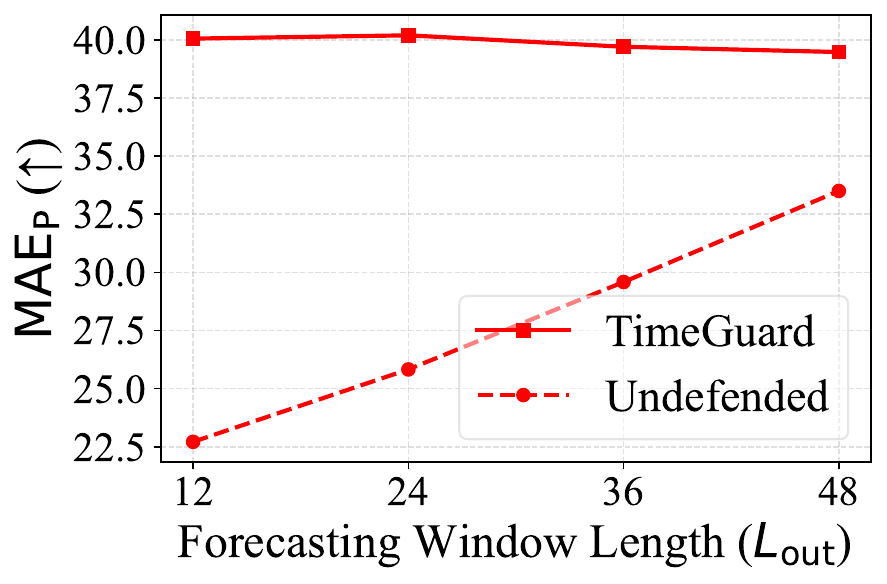}
    \end{subfigure}
    \hfill
    \begin{subfigure}{0.322\linewidth}
        \includegraphics[width=\linewidth]{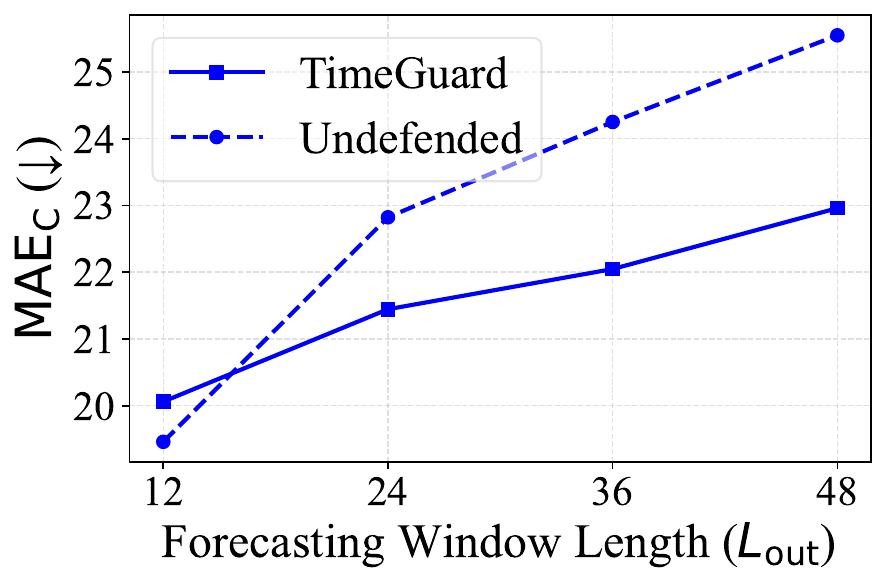}
    \end{subfigure}
    \hfill
    \begin{subfigure}{0.322\linewidth}
        \includegraphics[width=\linewidth]{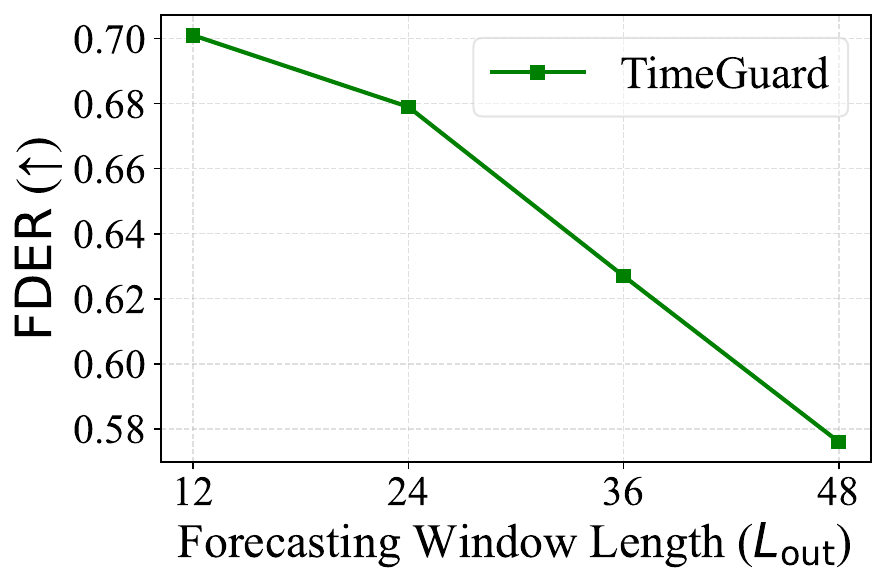}
    \end{subfigure}
    \caption{Defense performance of \methodname{} (\MAEP{}, \MAEC{}, and \FDER{}) under different forecasting window length $\LOUT$ of the BackTime attack on the PEMS03 dataset with the \emph{TimesNet} model.}                
    \label{app-fig:different-t_out-timesnet}
\end{figure}

\textbf{Generalization to large-scale datasets.} To evaluate the scalability of \methodname{}, we further do experiments on GBA~\cite{liu2023largest}, using its 2019 subset, which is a much larger traffic forecasting benchmark (35040 × 2352) than the datasets used in our main experiments. Even at this scale, \methodname{} achieves the best overall defense performance, with an average \FDER{} of $0.698$. We also explicitly report the increased training cost, showing that \methodname{} remains effective on substantially larger datasets, albeit with higher training overhead. This training time is notably higher than that on our previous largest benchmark, PEMS03, i.e., 3372s as reported in \tableautorefname~\ref{tab:efficiency-analysis}. The overhead corresponds to $\approx 3.53\times$ the cost of undefended training, indicating that \methodname{} remains scalable in practice but incurs nontrivial additional cost.
    \begin{table}[!htbp]
        \centering
        \vspace{-5pt}
        \caption{Defense performance and training time (in seconds) of PDB and \methodname{} under BackTime attack on the GBA dataset~\cite{liu2023largest}, where FEDformer, SimpleTM, and TimesNet are used as victim models. Best results are shown in \textbf{bold}.}
        \vspace{-5pt}
        \label{app-tab:gba-dataset}
        \setlength{\tabcolsep}{1pt} 
        \renewcommand{\arraystretch}{1.2} 
        \tiny
        \renewcommand{\aboverulesep}{0pt}
        \renewcommand{\belowrulesep}{0pt}
        \setlength\cellspacetoplimit{2pt}
        \setlength\cellspacebottomlimit{2pt}
        \resizebox{\linewidth}{!}{
       \begin{tabular}{lcccccccccccccccc}
    \toprule
    \textbf{Model~→}
    & \multicolumn{4}{c}{\textbf{FEDformer}}
    & \multicolumn{4}{c}{\textbf{SimpleTM}}
    & \multicolumn{4}{c}{\textbf{TimesNet}} 
    & \multicolumn{4}{c}{\textbf{\textsc{Average}}} 
    \\
    \cmidrule(l{0pt}){1-1}
    \cmidrule(l){2-5}
    \cmidrule(l){6-9}
    \cmidrule(l){10-13}
    \cmidrule(l){14-17}
    \textbf{Defense~↓}
    & \MAEC~↓ & \MAEP~↑ & \FDER~↑ & Training Time~↓
    & \MAEC~↓ & \MAEP~↑ & \FDER~↑ & Training Time~↓
    & \MAEC~↓ & \MAEP~↑ & \FDER~↑ & Training Time~↓
    & \MAEC~↓ & \MAEP~↑ & \FDER~↑ & Training Time~↓
    \\
    \midrule
    No Defense 
        & 27.259 & 26.234 & -- & 5625.6
        & 32.814 & 37.735 & -- & 6169.1
        & 31.846 & 41.139 & -- & 7134.3
        & 30.640 & 35.036 & -- & 6309.7 \\
    PDB~\cite{wei2024mitigating}
        & \textbf{26.705} & 39.016 & 0.664 & \textbf{5625.3}
        & \textbf{32.406} & 48.588 & \textbf{0.612} & \textbf{6419.8}
        & 32.326 & 51.237 & 0.591 & \textbf{8541.8}
        & \textbf{30.479} & 46.280 & 0.622 & \textbf{6862.3} \\
    \cellcolor[HTML]{EFEFEF}\textbf{\methodname}
        & \cellcolor[HTML]{EFEFEF}28.463
        & \cellcolor[HTML]{EFEFEF}\textbf{69.200}
        & \cellcolor[HTML]{EFEFEF}\textbf{0.789}
        & \cellcolor[HTML]{EFEFEF}23390.1
        & \cellcolor[HTML]{EFEFEF}39.874
        & \cellcolor[HTML]{EFEFEF}\textbf{54.641}
        & \cellcolor[HTML]{EFEFEF}0.566
        & \cellcolor[HTML]{EFEFEF}22641.2
        & \cellcolor[HTML]{EFEFEF}\textbf{31.647}
        & \cellcolor[HTML]{EFEFEF}\textbf{79.109}
        & \cellcolor[HTML]{EFEFEF}\textbf{0.740}
        & \cellcolor[HTML]{EFEFEF}20745.8
        & \cellcolor[HTML]{EFEFEF}33.328
        & \cellcolor[HTML]{EFEFEF}\textbf{67.650}
        & \cellcolor[HTML]{EFEFEF}\textbf{0.698}
        & \cellcolor[HTML]{EFEFEF}22259.0 \\
    \bottomrule
        \end{tabular}
        }
        \vspace{-12pt}
    \end{table}    

\textbf{Generalization to discrete, count-valued datasets.} While our main evaluation focuses on continuous-valued datasets, we further assess \methodname{} on a discrete, count-valued dataset. Specifically, we use the hourly subset of Bike Sharing~\cite{bikesharing}, which records hourly bike rental counts from 2011 to 2012 in the Capital Bikeshare system, and evaluate under the Random attack. For preprocessing, we retain ``temp", ``atemp", ``hum", ``windspeed", and ``cnt", and use ``cnt", a discrete variable representing the number of rental bikes, as the target variable, resulting in 17,379 time stamps. As shown in \tableautorefname~\ref{app-tab:bike-sharing-dataset}, \methodname{} still achieves the best defense performance, with an average \FDER{} of 0.831. This suggests preliminary transfer beyond continuous-valued TSF, while broader adaptation to discrete and count-valued forecasting remains future work.

    \begin{table}[!htbp]
        \centering
        \vspace{-5pt}
        \caption{Defense performance of PDB and \methodname{} under Random attack on the Bike Sharing dataset~\cite{bikesharing}, where FEDformer, SimpleTM, and TimesNet are used as victim models. Best results are shown in \textbf{bold}.}
        \vspace{-5pt}
        \label{app-tab:bike-sharing-dataset}
        \setlength{\tabcolsep}{2pt} 
        \renewcommand{\arraystretch}{1.2}
        \tiny
        \renewcommand{\aboverulesep}{0pt}
        \renewcommand{\belowrulesep}{0pt}
        \setlength\cellspacetoplimit{2pt}
        \setlength\cellspacebottomlimit{2pt}
        \resizebox{\linewidth}{!}{
        \begin{tabular}{lcccccccccccc}
        \toprule
        \textbf{Model~→}
        & \multicolumn{3}{c}{\textbf{FEDformer}}
        & \multicolumn{3}{c}{\textbf{SimpleTM}}
        & \multicolumn{3}{c}{\textbf{TimesNet}} 
        & \multicolumn{3}{c}{\textbf{\textsc{Average}}} 
        \\
        \cmidrule(l{0pt}){1-1}
        \cmidrule(l){2-4}
        \cmidrule(l){5-7}
        \cmidrule(l){8-10}
        \cmidrule(l){11-13}
        \textbf{Defense~↓}
        & \MAEC~↓ & \MAEP~↑ & \FDER~↑
        & \MAEC~↓ & \MAEP~↑ & \FDER~↑
        & \MAEC~↓ & \MAEP~↑ & \FDER~↑ 
        & \MAEC~↓ & \MAEP~↑ & \FDER~↑ 
        \\
        \midrule
        No Defense 
            & 23.921 & 66.894 & --
            & 24.143 & 66.465 & --
            & 20.377 & 47.904 & -- 
            & 22.814 & 60.421 & -- \\
        PDB~\cite{wei2024mitigating}
            & \textbf{22.368} & 120.193 & 0.722
            & 24.472 & 124.615 & 0.727
            & \textbf{19.131} & 103.550 & 0.769 
            & \textbf{21.990} & 116.119 & 0.739 \\
        \cellcolor[HTML]{EFEFEF}\textbf{\methodname}
            & \cellcolor[HTML]{EFEFEF}28.355
            & \cellcolor[HTML]{EFEFEF}\textbf{249.084}
            & \cellcolor[HTML]{EFEFEF}\textbf{0.788}
            & \cellcolor[HTML]{EFEFEF}\textbf{20.699}
            & \cellcolor[HTML]{EFEFEF}\textbf{243.449} 
            & \cellcolor[HTML]{EFEFEF}\textbf{0.863}
            & \cellcolor[HTML]{EFEFEF}22.751
            & \cellcolor[HTML]{EFEFEF}\textbf{227.498} 
            & \cellcolor[HTML]{EFEFEF}\textbf{0.843} 
            & \cellcolor[HTML]{EFEFEF}23.935 
            & \cellcolor[HTML]{EFEFEF}\textbf{240.010} 
            & \cellcolor[HTML]{EFEFEF}\textbf{0.831} \\
        \bottomrule
        \end{tabular}
        }
        \vspace{-5pt}
    \end{table}

\textbf{Robustness on distribution-shifted and nonstationary datasets.} Since \methodname{} relies on a hand-designed neighborhood metric, its estimates may degrade under strong distribution shift and nonstationarity. To evaluate this scenario, we test \methodname{} on Exchange, a financial forecasting benchmark with evolving dynamics that contains daily exchange rates from 8 countries between 1990 and 2016~\cite{lai2018modeling}, totaling 7,588 time steps, under the BackTime attack. The results in \tableautorefname~\ref{app-tab:exchange-dataset} show that, despite these challenging evolving dynamics, \methodname{} remains effective and outperforms PDB. This suggests that \methodname{} remains practically effective even under stronger distribution shift and nonstationarity.
    \begin{table}[!htbp]
        \centering
        \vspace{-5pt}
        \caption{Defense performance of PDB and \methodname{} under BackTime attack on the Exchange dataset~\cite{lai2018modeling}, where FEDformer, SimpleTM, and TimesNet are used as victim models. Best results are shown in \textbf{bold}.}
        \vspace{-5pt}
        \label{app-tab:exchange-dataset}
        \setlength{\tabcolsep}{2pt} 
        \renewcommand{\arraystretch}{1.2}
        \tiny
        \renewcommand{\aboverulesep}{0pt}
        \renewcommand{\belowrulesep}{0pt}
        \setlength\cellspacetoplimit{2pt}
        \setlength\cellspacebottomlimit{2pt}
        \resizebox{\linewidth}{!}{
        \begin{tabular}{lcccccccccccc}
        \toprule
        \textbf{Model~→}
        & \multicolumn{3}{c}{\textbf{FEDformer}}
        & \multicolumn{3}{c}{\textbf{SimpleTM}}
        & \multicolumn{3}{c}{\textbf{TimesNet}} 
        & \multicolumn{3}{c}{\textbf{\textsc{Average}}} 
        \\
        \cmidrule(l{0pt}){1-1}
        \cmidrule(l){2-4}
        \cmidrule(l){5-7}
        \cmidrule(l){8-10}
        \cmidrule(l){11-13}
        \textbf{Defense~↓}
        & \MAEC~↓ & \MAEP~↑ & \FDER~↑
        & \MAEC~↓ & \MAEP~↑ & \FDER~↑
        & \MAEC~↓ & \MAEP~↑ & \FDER~↑ 
        & \MAEC~↓ & \MAEP~↑ & \FDER~↑ 
        \\
        \midrule
        No Defense 
            & 0.00967 & 0.02143 & --
            & 0.00699 & 0.01927 & --
            & 0.03089 & 0.10875 & -- 
            & 0.01585 & 0.04982 & -- \\
        PDB~\cite{wei2024mitigating}
            & 0.01654 & 0.08107 & 0.66040
            & 0.00737 & 0.10169 & 0.87944
            & 0.05331 & \textbf{0.16554} & 0.46127
            & 0.02574 & 0.11610 & 0.66704 \\
        \cellcolor[HTML]{EFEFEF}\textbf{\methodname}
            & \cellcolor[HTML]{EFEFEF}\textbf{0.00803}
            & \cellcolor[HTML]{EFEFEF}\textbf{0.11179}
            & \cellcolor[HTML]{EFEFEF}\textbf{0.90417}
            & \cellcolor[HTML]{EFEFEF}\textbf{0.00673}
            & \cellcolor[HTML]{EFEFEF}\textbf{0.10451}
            & \cellcolor[HTML]{EFEFEF}\textbf{0.90782}
            & \cellcolor[HTML]{EFEFEF}\textbf{0.03829}
            & \cellcolor[HTML]{EFEFEF}0.14101
            & \cellcolor[HTML]{EFEFEF}\textbf{0.51782}
            & \cellcolor[HTML]{EFEFEF}\textbf{0.01768}
            & \cellcolor[HTML]{EFEFEF}\textbf{0.11911}
            & \cellcolor[HTML]{EFEFEF}\textbf{0.77660} \\
        \bottomrule
        \end{tabular}
        }
        \vspace{-12pt}
    \end{table}    

\textbf{Robustness under nonstationary settings and concept drift.} To further examine robustness under mild distribution change and concept drift, we conduct an additional experiment on PEMS03 under the BackTime attack by introducing synthetic distribution shifts at test time. Specifically, let $x_{t,c}$ denote the value at time step $t$ and channel $c$, and let $\sigma_c$ denote the standard deviation of channel $c$ computed from the training dataset. We consider three perturbations:
\begin{itemize}[itemsep=0.5pt, topsep=0.25pt]
    \item Scale shift: $x'_{t, c} = (1 + \alpha)x_{t,c}$
    \item Mean shift: $x'_{t, c} = x_{t,c} + \alpha\sigma_c$
    \item Linear trend: $x'_{t,c} = x_{t,c} + \alpha\sigma_c\frac{t}{T}$, where $T$ is the length of test split.
\end{itemize}
We test two shift strengths, $\alpha \in \{0.1, 0.2\}$. As shown in \tableautorefname~\ref{app-tab:nonstationary-case}, the performance of all methods declines slightly under synthetic distribution shift. However, \methodname{} consistently achieves the best defense performance across all six shifted settings. This suggests that although non-stationarity affects performance, its negative impact is moderate rather than catastrophic, and \methodname{} remains stable in practice under mild distribution shifts. Under strongly non-stationary or concept-drift scenarios, any training-phase defense is likely to face challenges, and \methodname{} is no exception, as also reflected by the degraded performance of both undefended training and PDB.
\begin{table}[htbp]
        \centering
        \caption{Defense performance of \methodname{} and PDB under BackTime attack on PEMS03 dataset under mild distribution shift, where FEDFormer, SimpleTM, and TimesNet are the victim models. Best results in each scenario are shown in \textbf{bold}.}
        \label{app-tab:nonstationary-case}
        \setlength{\tabcolsep}{2pt} 
        \renewcommand{\arraystretch}{1.2}
        \tiny
        \renewcommand{\aboverulesep}{0pt}
        \renewcommand{\belowrulesep}{0pt}
        \setlength\cellspacetoplimit{2pt}
        \setlength\cellspacebottomlimit{2pt}
        \resizebox{\linewidth}{!}{
        \begin{tabular}{clccccccccccccc}
    \toprule
    \multirow{2}{*}{\textbf{\begin{tabular}[c]{@{}c@{}}Shift / \\ Strength\end{tabular}}}
    & \textbf{Model~→}
    & \multicolumn{3}{c}{\textbf{FEDformer}}
    & \multicolumn{3}{c}{\textbf{SimpleTM}}
    & \multicolumn{3}{c}{\textbf{TimesNet}} 
    & \multicolumn{3}{c}{\textbf{\textsc{Average}}} 
    \\
    \cmidrule(l{0pt}){2-2}
    \cmidrule(l){3-5}
    \cmidrule(l){6-8}
    \cmidrule(l){9-11}
    \cmidrule(l){12-14}
    & \textbf{Defense~↓}
    & \MAEC~↓ & \MAEP~↑ & \FDER~↑
    & \MAEC~↓ & \MAEP~↑ & \FDER~↑
    & \MAEC~↓ & \MAEP~↑ & \FDER~↑ 
    & \MAEC~↓ & \MAEP~↑ & \FDER~↑ 
    \\
    \midrule
    \multirow{3}{*}{No shift} & No Defense 
        & 16.688 & 13.577 & --
        & 16.519 & 8.218 & --
        & 22.041 & 21.293 & -- 
        & 18.416 & 14.363 & -- \\
    & PDB~\cite{wei2024mitigating}
        & 17.420 & 16.038 & 0.556
        & 18.283 & 25.700 & 0.792
        & 23.586 & 26.898 & 0.571
        & 19.763 & 22.879 & 0.640 \\
    & \textbf{\methodname}
        & \textbf{16.850} & \textbf{42.101} & \textbf{0.834}
        & \textbf{17.688} & \textbf{36.386} & \textbf{0.854}
        & \textbf{20.562} & \textbf{40.005} & \textbf{0.734}
        & \textbf{18.367} & \textbf{39.497} & \textbf{0.807} \\
    \midrule
    \multirow{3}{*}{\begin{tabular}[c]{@{}c@{}}Scale shift \\ ($\alpha = 0.1$)\end{tabular}} & No Defense 
        & 18.323 & 14.860 & --
        & 18.171 & 8.696 & --
        & 25.510 & 24.590 & -- 
        & 20.668 & 16.049 & -- \\
    & PDB~\cite{wei2024mitigating}
        & 19.039 & 17.560 & 0.558
        & 20.229 & 26.254 & 0.784
        & 27.108 & 30.703 & 0.570
        & 22.125 & 24.839 & 0.637 \\
    & \textbf{\methodname}
        & \textbf{18.508} & \textbf{42.423} & \textbf{0.820}
        & \textbf{19.457} & \textbf{36.707} & \textbf{0.849}
        & \textbf{23.672} & \textbf{40.808} & \textbf{0.699}
        & \textbf{20.546} & \textbf{39.979} & \textbf{0.789} \\
    \midrule
    \multirow{3}{*}{\begin{tabular}[c]{@{}c@{}}Scale shift \\ ($\alpha = 0.2$)\end{tabular}} & No Defense 
        & 19.977 & 16.170 & --
        & 19.823 & 9.241 & --
        & 30.968 & 29.982 & -- 
        & 23.590 & 19.682 & -- \\
    & PDB~\cite{wei2024mitigating}
        & 20.734 & 19.108 & 0.482
        & 22.384 & 27.480 & 0.775
        & 32.602 & 36.076 & 0.559
        & 25.240 & 27.555 & 0.605 \\
    & \textbf{\methodname}
        & \textbf{20.225} & \textbf{42.811} & \textbf{0.762}
        & \textbf{21.226} & \textbf{37.050} & \textbf{0.842}
        & \textbf{28.465} & \textbf{42.762} & \textbf{0.649}
        & \textbf{23.305} & \textbf{40.874} & \textbf{0.751} \\
    \midrule
    \multirow{3}{*}{\begin{tabular}[c]{@{}c@{}}Linear trend \\ ($\alpha = 0.1$)\end{tabular}} & No Defense 
        & 16.689 & 14.860 & --
        & 16.519 & 8.696 & --
        & 21.761 & 24.590 & -- 
        & 18.323 & 16.049 & -- \\
    & PDB~\cite{wei2024mitigating}
        & 17.446 & 16.029 & 0.515
        & 18.409 & 25.862 & 0.781
        & 23.117 & 28.007 & 0.532
        & 19.657 & 23.300 & 0.609 \\
    & \textbf{\methodname}
        & \textbf{16.841} & \textbf{42.123} & \textbf{0.819}
        & \textbf{17.688} & \textbf{36.386} & \textbf{0.847}
        & \textbf{20.523} & \textbf{41.380} & \textbf{0.703}
        & \textbf{18.351} & \textbf{39.963} & \textbf{0.790} \\
    \midrule
    \multirow{3}{*}{\begin{tabular}[c]{@{}c@{}}Linear trend \\ ($\alpha = 0.2$)\end{tabular}} & No Defense 
        & 16.690 & 13.515 & --
        & 16.518 & 8.218 & --
        & 22.140 & 29.982 & -- 
        & 18.450 & 17.239 & -- \\
    & PDB~\cite{wei2024mitigating}
        & 17.482 & 16.028 & 0.556
        & 18.641 & 26.046 & 0.785
        & 23.364 & 29.682 & 0.474
        & 19.829 & 23.919 & 0.605 \\
    & \textbf{\methodname}
        & \textbf{16.836} & \textbf{42.139} & \textbf{0.835}
        & \textbf{17.688} & \textbf{36.401} & \textbf{0.854}
        & \textbf{20.990} & \textbf{42.854} & \textbf{0.650}
        & \textbf{18.505} & \textbf{40.465} & \textbf{0.780} \\
    \midrule
    \multirow{3}{*}{\begin{tabular}[c]{@{}c@{}}Mean shift \\ ($\alpha = 0.1$)\end{tabular}} & No Defense 
        & 16.688 & 13.515 & --
        & 16.519 & 8.217 & --
        & 22.051 & 23.388 & -- 
        & 18.419 & 15.040 & -- \\
    & PDB~\cite{wei2024mitigating}
        & 17.474 & 16.016 & 0.556
        & 18.607 & 26.066 & 0.786
        & 23.242 & 29.628 & 0.580
        & 19.774 & 23.903 & 0.641 \\
    & \textbf{\methodname}
        & \textbf{16.838} & \textbf{42.166} & \textbf{0.835}
        & \textbf{17.688} & \textbf{36.422} & \textbf{0.854}
        & \textbf{20.901} & \textbf{42.946} & \textbf{0.728}
        & \textbf{18.476} & \textbf{40.512} & \textbf{0.806} \\
    \midrule
    \multirow{3}{*}{\begin{tabular}[c]{@{}c@{}}Mean shift \\ ($\alpha = 0.2$)\end{tabular}} & No Defense 
        & 16.693 & 13.462 & --
        & 16.519 & 8.217 & --
        & 23.937 & 26.798 & -- 
        & 19.050 & 16.159 & -- \\
    & PDB~\cite{wei2024mitigating}
        & 17.559 & 16.028 & 0.555
        & 19.207 & 26.066 & 0.772
        & 24.941 & 33.902 & 0.585
        & 20.569 & 25.332 & 0.637 \\
    & \textbf{\methodname}
        & \textbf{16.836} & \textbf{42.227} & \textbf{0.836}
        & \textbf{17.688} & \textbf{36.422} & \textbf{0.854}
        & \textbf{22.566} & \textbf{46.094} & \textbf{0.709}
        & \textbf{19.030} & \textbf{41.581} & \textbf{0.800} \\
    \bottomrule
    \end{tabular}
        }
    \vspace{-12pt}
    \end{table}    
\subsection{Ablation Study Full Results} \label{app:ablation}
We provide per-model ablation results on the PEMS03 dataset under the Random, Manhattan, and BackTime attacks in Table~\ref{app-tab:ablation-study} with FEDformer, SimpleTM, and TimesNet. Overall, these results are consistent with the model-averaged trends reported in Section~\ref{sec:analysis}.

\begin{table}[!t]
    \centering
    \caption{Per-model ablation results of \methodname{} on PEMS03 under the Random, Manhattan, and BackTime attacks, with FEDformer, SimpleTM, and TimesNet as victim models. The \textsc{Average} row reports the mean across the three models, matching Table~\ref{tab:ablation-results}. }
    \vspace{-5pt}
    \label{app-tab:ablation-study}
    \setlength{\tabcolsep}{6.5pt}
    \renewcommand{\arraystretch}{1.1}
    \scriptsize
    \renewcommand{\aboverulesep}{0pt}
    \renewcommand{\belowrulesep}{0pt}
    \setlength\cellspacetoplimit{1.8pt}
    \setlength\cellspacebottomlimit{1.8pt}
    \resizebox{\linewidth}{!}{
    \begin{tabular}{clccccccccc}
    \toprule
    \multirow{2}{*}{\textbf{\begin{tabular}[c]{@{}c@{}}Model\end{tabular}}}
    & \textbf{Attack~→}
    & \multicolumn{3}{c}{\textbf{Random}}
    & \multicolumn{3}{c}{\textbf{Manhattan}}
    & \multicolumn{3}{c}{\textbf{BackTime}} \\
    \cmidrule(l{0pt}){2-2}
    \cmidrule(l){3-5}
    \cmidrule(l){6-8}
    \cmidrule(l){9-11}
    & \textbf{Defense~↓}
    & \MAEC~↓ & \MAEP~↑ & \FDER~↑
    & \MAEC~↓ & \MAEP~↑ & \FDER~↑
    & \MAEC~↓ & \MAEP~↑ & \FDER~↑ \\
    \midrule

    \multirow{7}{*}{\textbf{FEDformer}} & No Defense
        & 16.286 & 14.959 & --
        & 16.411 & 17.984 & --
        & 16.093 & 10.760 & -- \\
    & \cellcolor[HTML]{EFEFEF}\textbf{\methodname}
        & \cellcolor[HTML]{EFEFEF}16.607
        & \cellcolor[HTML]{EFEFEF}\textbf{100.436}
        & \cellcolor[HTML]{EFEFEF}\textbf{0.916}
        & \cellcolor[HTML]{EFEFEF}16.578
        & \cellcolor[HTML]{EFEFEF}\textbf{94.212}
        & \cellcolor[HTML]{EFEFEF}\textbf{0.900}
        & \cellcolor[HTML]{EFEFEF}16.840
        & \cellcolor[HTML]{EFEFEF}\textbf{41.232}
        & \cellcolor[HTML]{EFEFEF}\textbf{0.847} \\ \cmidrule(l{0pt}){2-11}
    & w/o Channel-wise
        & 16.558 & 14.915 & 0.492
        & 16.660 & 18.462 & 0.505
        & 17.959 & 12.832 & 0.529 \\
    & w/o NDF
        & 16.717 & 91.682 & 0.906
        & 16.695 & 87.981 & 0.889
        & 16.947 & 39.474 & 0.839 \\
    & w/o RCF
        & 16.622 & 99.747 & 0.915
        & 16.639 & 89.769 & 0.893
        & 16.962 & 40.913 & 0.843 \\
    & w/o NDF+RCF
        & \textbf{16.549} & 83.125 & 0.902
        & \textbf{16.549} & 82.560 & 0.887
        & \textbf{16.740} & 40.111 & \textbf{0.847} \\
    & w/o DRLS
        & 17.727 & 15.775 & 0.485
        & 17.618 & 16.000 & 0.466
        & 17.711 & 9.509 & 0.454 \\
    \midrule

    \multirow{7}{*}{\textbf{SimpleTM}} & No Defense
        & 17.510 & 19.007 & --
        & 17.539 & 22.532 & --
        & 17.268 & 9.131 & -- \\
    & \cellcolor[HTML]{EFEFEF}\textbf{\methodname}
        & \cellcolor[HTML]{EFEFEF}17.489
        & \cellcolor[HTML]{EFEFEF}173.700
        & \cellcolor[HTML]{EFEFEF}\textbf{0.945}
        & \cellcolor[HTML]{EFEFEF}17.284
        & \cellcolor[HTML]{EFEFEF}157.870
        & \cellcolor[HTML]{EFEFEF}\textbf{0.929}
        & \cellcolor[HTML]{EFEFEF}17.243
        & \cellcolor[HTML]{EFEFEF}36.626
        & \cellcolor[HTML]{EFEFEF}\textbf{0.875} \\ \cmidrule(l{0pt}){2-11}
    & w/o Channel-wise
        & \textbf{16.826} & 14.307 & 0.500
        & 16.660 & 18.462 & 0.500
        & 17.666 & 6.615 & 0.489 \\
    & w/o NDF
        & 18.740 & 180.126 & 0.914
        & 16.695 & 87.981 & 0.872
        & 17.400 & 36.145 & 0.870 \\
    & w/o RCF
        & 17.311 & 172.719 & \textbf{0.945}
        & \textbf{16.639} & 89.769 & 0.874
        & 17.703 & \textbf{37.717} & 0.867 \\
    & w/o NDF+RCF
        & 17.577 & 151.760 & 0.935
        & 17.793 & 139.796 & 0.912
        & \textbf{17.094} & 35.681 & 0.872 \\
    & w/o DRLS
        & 19.970 & \textbf{194.031} & 0.889
        & 19.945 & \textbf{173.531} & 0.875
        & 20.628 & 31.919 & 0.776 \\
    \midrule

    \multirow{7}{*}{\textbf{TimesNet}} & No Defense
        & 19.104 & 19.351 & --
        & 19.216 & 20.283 & --
        & 19.459 & 22.713 & -- \\
    & \cellcolor[HTML]{EFEFEF}\textbf{\methodname}
        & \cellcolor[HTML]{EFEFEF}\textbf{19.687}
        & \cellcolor[HTML]{EFEFEF}39.894
        & \cellcolor[HTML]{EFEFEF}\textbf{0.743}
        & \cellcolor[HTML]{EFEFEF}\textbf{19.689}
        & \cellcolor[HTML]{EFEFEF}\textbf{40.029}
        & \cellcolor[HTML]{EFEFEF}\textbf{0.735}
        & \cellcolor[HTML]{EFEFEF}\textbf{20.061}
        & \cellcolor[HTML]{EFEFEF}40.052
        & \cellcolor[HTML]{EFEFEF}\textbf{0.701} \\ \cmidrule(l{0pt}){2-11}
    & w/o Channel-wise
        & 21.577 & 19.212 & 0.443
        & 21.389 & 20.662 & 0.458
        & 21.580 & 25.328 & 0.502 \\
    & w/o NDF
        & 20.286 & \textbf{41.563} & 0.738
        & 20.370 & 39.967 & 0.718
        & 20.908 & 39.428 & 0.677 \\
    & w/o RCF
        & 20.255 & 40.748 & 0.734
        & 20.043 & 39.957 & 0.726
        & 21.158 & \textbf{40.207} & 0.677 \\
    & w/o NDF+RCF
        & 20.880 & 40.456 & 0.718
        & 20.366 & 39.382 & 0.714
        & 20.985 & 39.889 & 0.679 \\
    & w/o DRLS
        & 21.545 & 19.520 & 0.448
        & 21.723 & 21.164 & 0.463
        & 21.905 & 27.327 & 0.529 \\
    \midrule

    \multirow{7}{*}{\textbf{\textsc{Average}}} & No Defense
        & 17.634 & 17.772 & --
        & 17.722 & 20.266 & --
        & 17.607 & 14.201 & -- \\
    & \cellcolor[HTML]{EFEFEF}\textbf{\methodname}
        & \cellcolor[HTML]{EFEFEF}\textbf{17.928}
        & \cellcolor[HTML]{EFEFEF}\textbf{104.677}
        & \cellcolor[HTML]{EFEFEF}\textbf{0.868}
        & \cellcolor[HTML]{EFEFEF}17.850
        & \cellcolor[HTML]{EFEFEF}\textbf{97.370}
        & \cellcolor[HTML]{EFEFEF}\textbf{0.854}
        & \cellcolor[HTML]{EFEFEF}\textbf{18.048}
        & \cellcolor[HTML]{EFEFEF}39.303
        & \cellcolor[HTML]{EFEFEF}\textbf{0.808} \\ \cmidrule(l{0pt}){2-11}
    & w/o Channel-wise
        & 18.320 & 16.145 & 0.478
        & 18.236 & 19.195 & 0.488
        & 19.068 & 14.925 & 0.507 \\
    & w/o NDF
        & 18.581 & 104.457 & 0.853
        & 17.920 & 71.976 & 0.826
        & 18.418 & 38.349 & 0.795 \\
    & w/o RCF
        & 18.063 & 104.405 & 0.865
        & \textbf{17.774} & 73.165 & 0.831
        & 18.608 & \textbf{39.612} & 0.796 \\
    & w/o NDF+RCF
        & 18.336 & 91.780 & 0.852
        & 18.236 & 87.246 & 0.838
        & 18.273 & 38.560 & 0.799 \\
    & w/o DRLS
        & 19.748 & 76.442 & 0.607
        & 19.762 & 70.232 & 0.601
        & 20.081 & 22.918 & 0.586 \\
    \bottomrule
    \end{tabular}
    }
    \vspace{-8pt}
\end{table}

\subsection{Hyperparameter Sensitivity Full Results} \label{app:hyper-sensit}

\textbf{Influence of $\alpha$ and $\beta$.} Figures~\ref{app-fig:different-alpha-beta-fedformer}–\ref{app-fig:different-alpha-beta-timesnet} report the \methodname{} defense performance with FEDformer, SimpleTM, and TimesNet on PEMS03 under the BackTime attack while varying $\alpha \in \{0.10, 0.15, 0.20, 0.25, 0.30\}$ and $\beta \in \{0.40, 0.50, 0.60, 0.70, 0.80\}$ combination. Consistent with the model-averaged trends in Section~\ref{sec:analysis}, it recommends choosing $\alpha \in [0.15, 0.25]$ and $\beta \in [0.5, 0.7]$ to balance clean performance and robustness.

\begin{figure}[!htbp]
    \centering
    \begin{subfigure}{0.31\linewidth}
        \includegraphics[width=\linewidth]{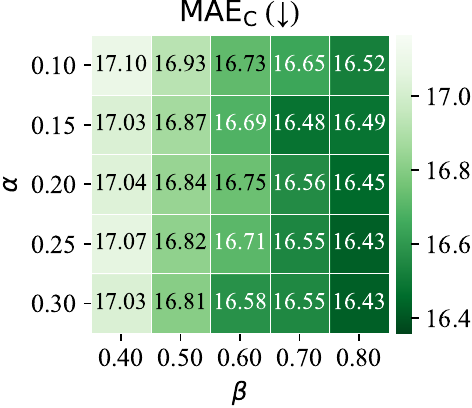}
    \end{subfigure}
    \hfill
    \begin{subfigure}{0.31\linewidth}
        \includegraphics[width=\linewidth]{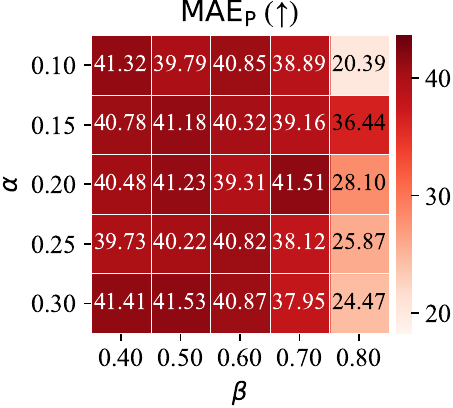}
    \end{subfigure}
    \hfill
    \begin{subfigure}{0.31\linewidth}
        \includegraphics[width=\linewidth]{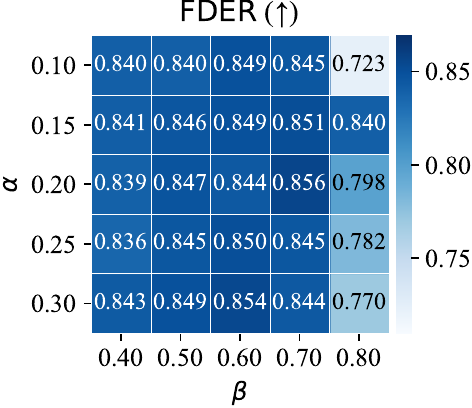}
    \end{subfigure}
    \vspace{-5pt}
    \caption{Defense performance of \methodname{} (\MAEP{}, \MAEC{}, and \FDER{}) with different initial reliable-pool ratio $\alpha$  and final ratio $\beta$ under \textbf{BackTime} attack on the \textbf{PEMS03} dataset with the \textbf{FEDformer} model.}                
    \label{app-fig:different-alpha-beta-fedformer}
    \vspace{-5pt}
\end{figure}

\begin{figure}[!htbp]
    \centering
    \begin{subfigure}{0.31\linewidth}
        \includegraphics[width=\linewidth]{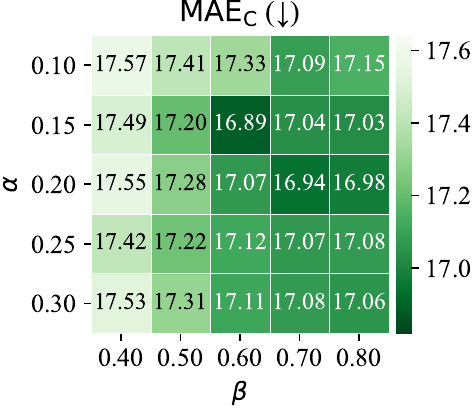}
    \end{subfigure}
    \hfill
    \begin{subfigure}{0.31\linewidth}
        \includegraphics[width=\linewidth]{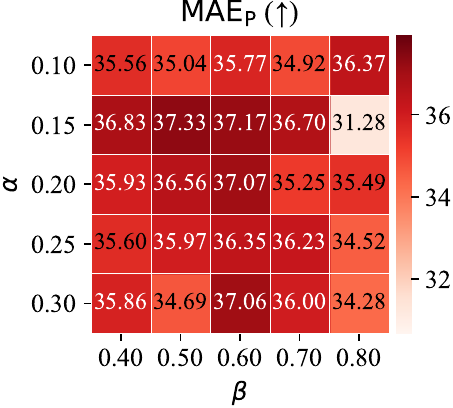}
    \end{subfigure}
    \hfill
    \begin{subfigure}{0.31\linewidth}
        \includegraphics[width=\linewidth]{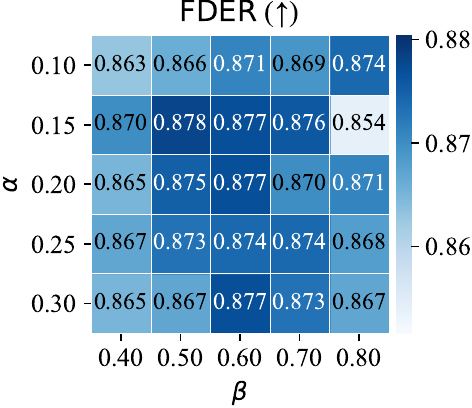}
    \end{subfigure}
    \vspace{-5pt}
    \caption{Defense performance of \methodname{} (\MAEP{}, \MAEC{}, and \FDER{}) with different initial reliable-pool ratio $\alpha$  and final ratio $\beta$ under \textbf{BackTime} attack on the \textbf{PEMS03} dataset with the \textbf{SimpleTM} model.}                
    \label{app-fig:different-alpha-beta-simpletm}
    \vspace{-5pt}
\end{figure}

\begin{figure}[!htbp]
    \centering
    \begin{subfigure}{0.31\linewidth}
        \includegraphics[width=\linewidth]{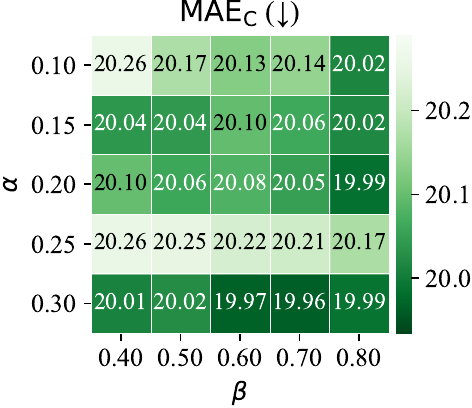}
    \end{subfigure}
    \hfill
    \begin{subfigure}{0.31\linewidth}
        \includegraphics[width=\linewidth]{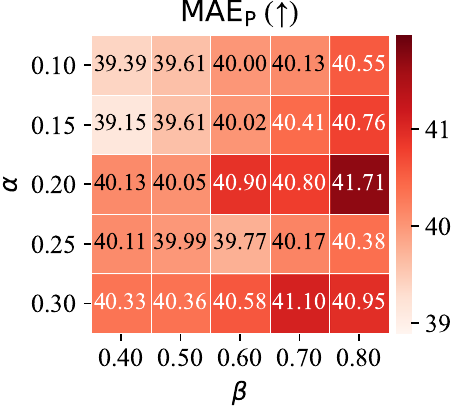}
    \end{subfigure}
    \hfill
    \begin{subfigure}{0.31\linewidth}
        \includegraphics[width=\linewidth]{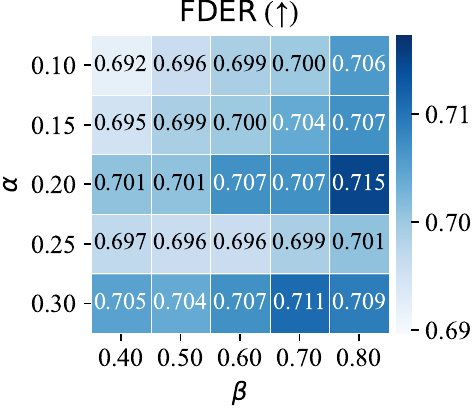}
    \end{subfigure}
    \vspace{-5pt}
    \caption{Defense performance of \methodname{} (\MAEP{}, \MAEC{}, and \FDER{}) with different initial reliable-pool ratio $\alpha$  and final ratio $\beta$ under \textbf{BackTime} attack on the \textbf{PEMS03} dataset with the \textbf{TimesNet} model.}                
    \label{app-fig:different-alpha-beta-timesnet}
    \vspace{-12pt}
\end{figure}

We further evaluate the effects of pool sizes $\alpha$ and $\beta$ on the Weather dataset under the BackTime attack, as shown in \figureautorefname~\ref{app-fig:different-alpha-fder-weather} and \figureautorefname~\ref{app-fig:different-beta-fder-weather}. Specifically, we vary $\alpha \in \{0.10, 0.15, 0.20, 0.25, 0.30\}$ with $\beta$ fixed at $0.50$, and vary $\beta \in \{0.40, 0.50, 0.60, 0.70, 0.80\}$ with $\alpha$ fixed at $0.20$. Overall, \methodname{} remains effective across all settings, achieving \FDER{} above 0.75 in every case. These results lead to the same conclusion as in \sectionautorefname~\ref{sec:exp-main-result} on PEMS03 dataset: $\alpha \in [0.15, 0.25]$ and $\beta \in [0.5, 0.7]$ provide the best trade-off, as reflected by \FDER{}.

\begin{figure}[!htbp]
    \centering
    \begin{subfigure}{0.24\linewidth}
        \includegraphics[width=\linewidth]{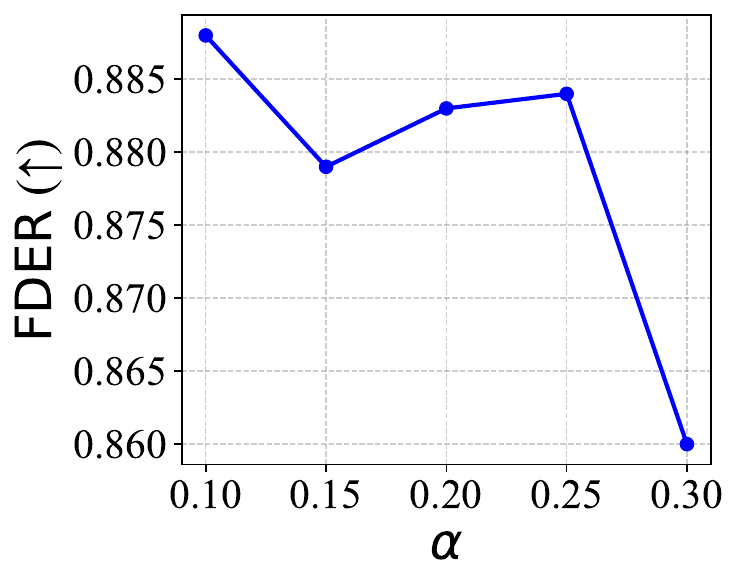}
    \end{subfigure}
    \hfill
    \begin{subfigure}{0.24\linewidth}
        \includegraphics[width=\linewidth]{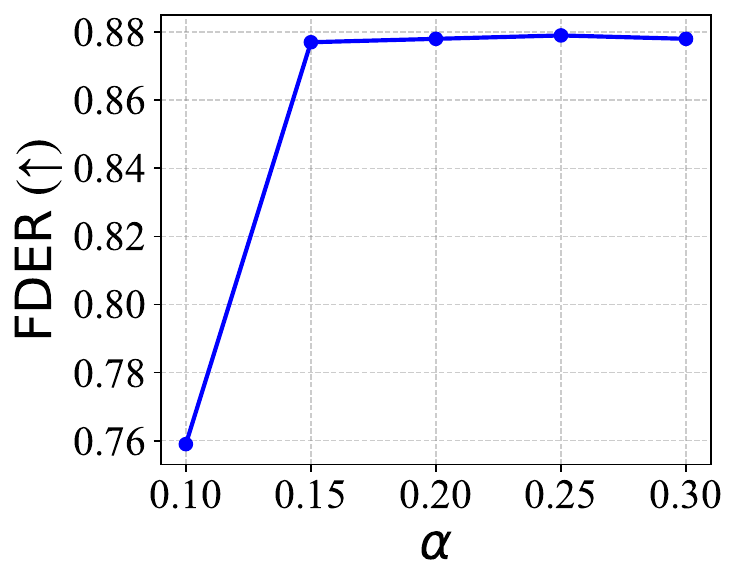}
    \end{subfigure}
    \hfill
    \begin{subfigure}{0.24\linewidth}
        \includegraphics[width=\linewidth]{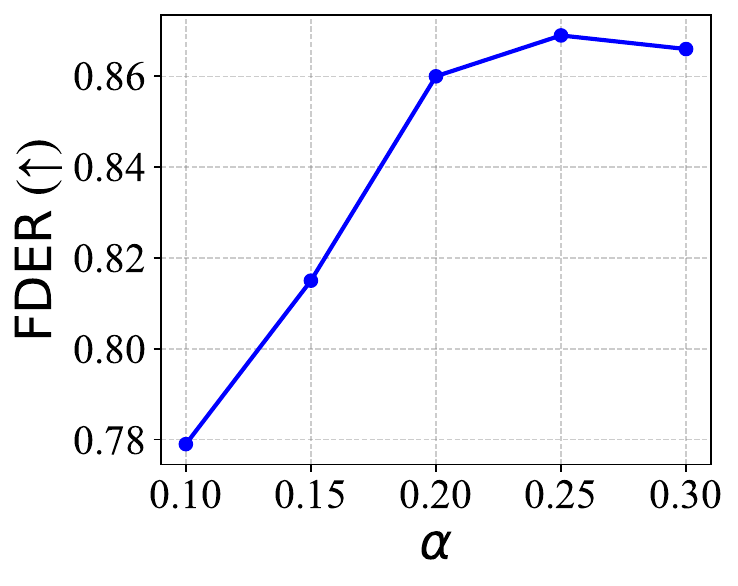}
    \end{subfigure}
     \hfill
    \begin{subfigure}{0.24\linewidth}
        \includegraphics[width=\linewidth]{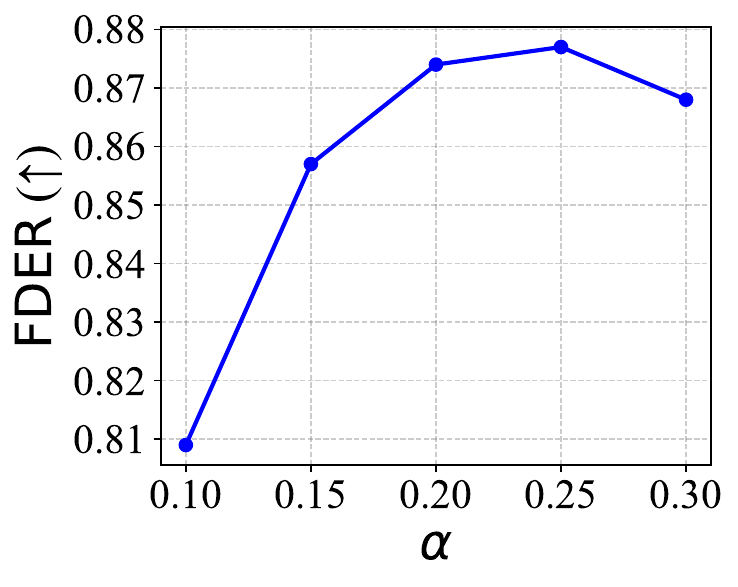}
    \end{subfigure}
    \vspace{-5pt}
    \caption{Defense performance of \methodname{} in terms of \FDER{} with different initialization ratios $\alpha$ under the \textbf{BackTime} attack on the \textbf{Weather} dataset, reported for FEDformer, SimpleTM, TimesNet, and their average.}    \label{app-fig:different-alpha-fder-weather}
    \vspace{-5pt}
\end{figure}

\begin{figure}[!htbp]
    \centering
    \begin{subfigure}{0.24\linewidth}
        \includegraphics[width=\linewidth]{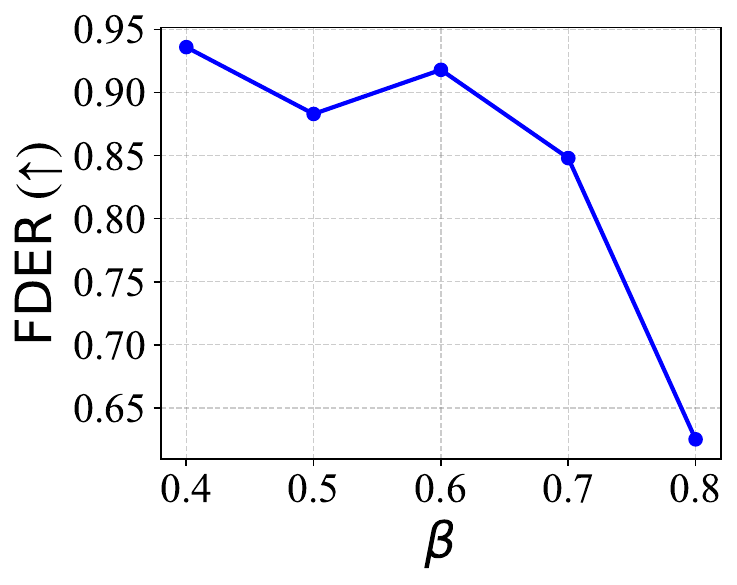}
    \end{subfigure}
    \hfill
    \begin{subfigure}{0.24\linewidth}
        \includegraphics[width=\linewidth]{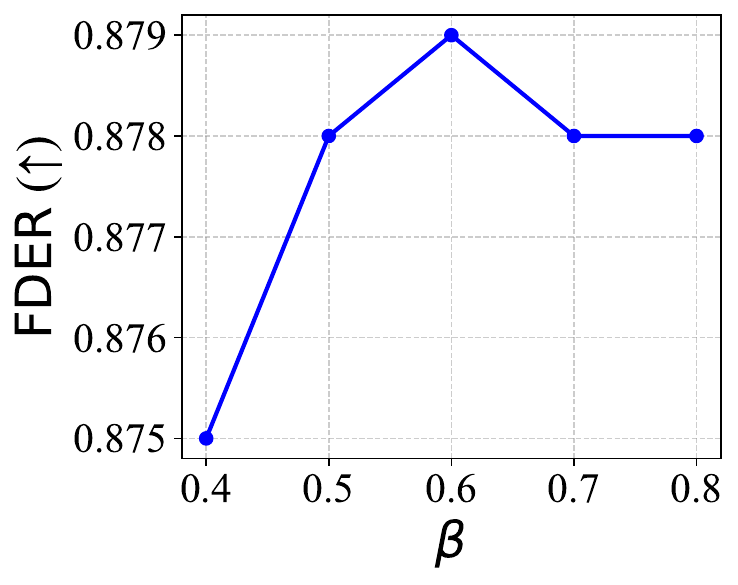}
    \end{subfigure}
    \hfill
    \begin{subfigure}{0.24\linewidth}
        \includegraphics[width=\linewidth]{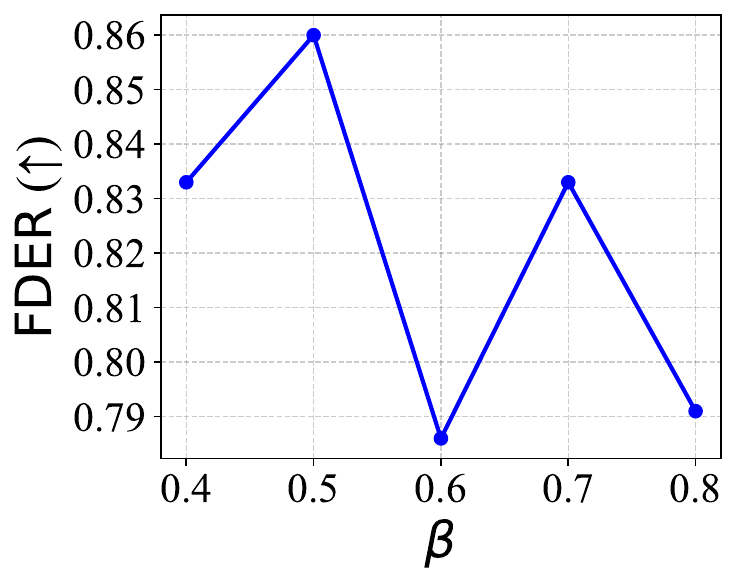}
    \end{subfigure}
     \hfill
    \begin{subfigure}{0.24\linewidth}
        \includegraphics[width=\linewidth]{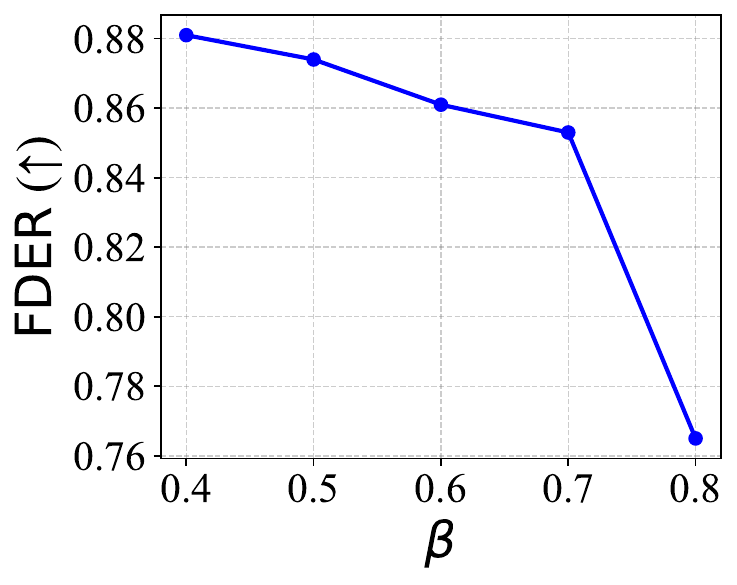}
    \end{subfigure}
    \vspace{-5pt}
    \caption{Defense performance of \methodname{} in terms of \FDER{} with different maximum pool ratios $\beta$ under the \textbf{BackTime} attack on the \textbf{Weather} dataset, reported for FEDformer, SimpleTM, TimesNet, and their average.}   \label{app-fig:different-beta-fder-weather}
\end{figure}

\textbf{Influence of $K$ and $\pi$.} Figures~\ref{app-fig:different-k-fder} and~\ref{app-fig:different-pi-fder} report the \methodname{} defense performance with FEDformer, SimpleTM, and TimesNet on PEMS03 under the BackTime attack while varying the neighborhood size $K \in \{10, 20, 32, 48, 64\}$ and the scaling factor $\pi \in \{1.05, 1.15, 1.25, 1.35, 1.50, 1.65\}$, respectively. Consistent with the model-averaged trends in Section~\ref{sec:analysis}, \methodname{} is relatively insensitive to the choice of $K$, while we recommend selecting $\pi \leq 1.5$. 

\begin{figure}[!htbp]
    \centering
    \begin{subfigure}{0.31\linewidth}
        \includegraphics[width=\linewidth]{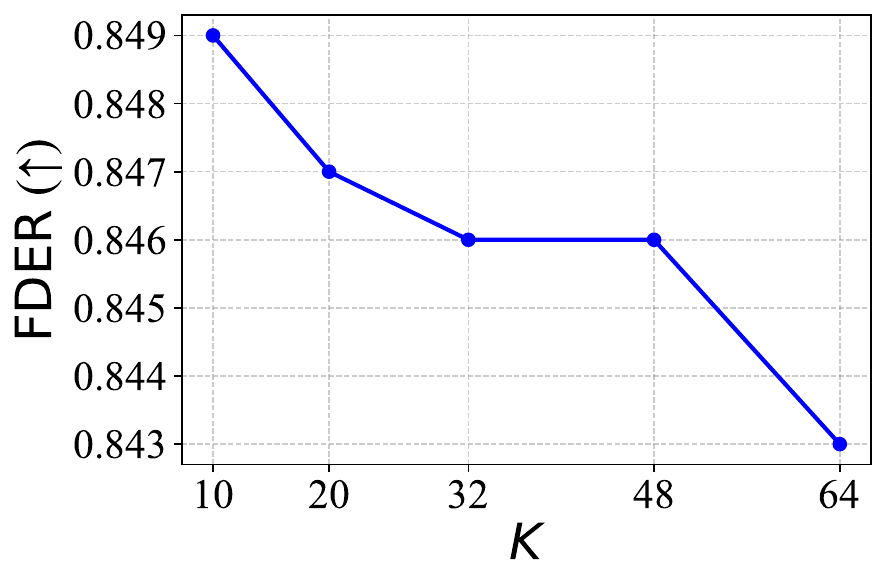}
    \end{subfigure}
    \hfill
    \begin{subfigure}{0.31\linewidth}
        \includegraphics[width=\linewidth]{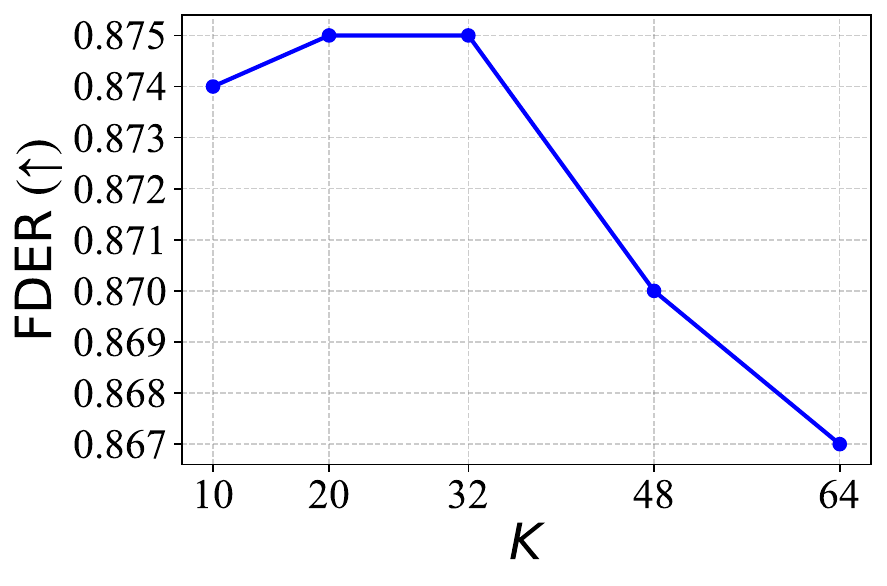}
    \end{subfigure}
    \hfill
    \begin{subfigure}{0.31\linewidth}
        \includegraphics[width=\linewidth]{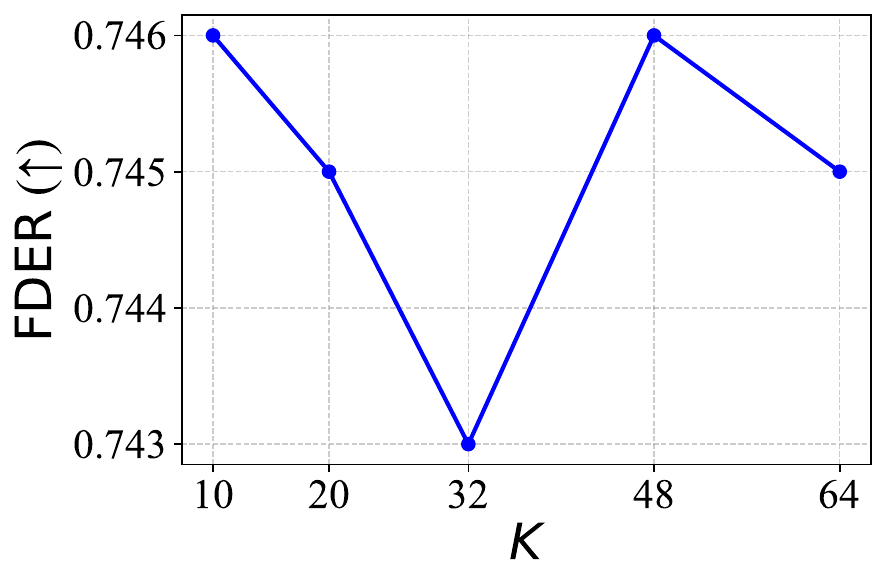}
    \end{subfigure}
    \vspace{-5pt}
    \caption{Defense performance of \methodname{} (\FDER{}) with different neighborhood size $K$ under \textbf{BackTime} attack on the \textbf{PEMS03} dataset with the FEDformer, SimpleTM, and TimesNet, respectively. }                
    \label{app-fig:different-k-fder}
\end{figure}

\begin{figure}[!htbp]
    \centering
    \begin{subfigure}{0.31\linewidth}
        \includegraphics[width=\linewidth]{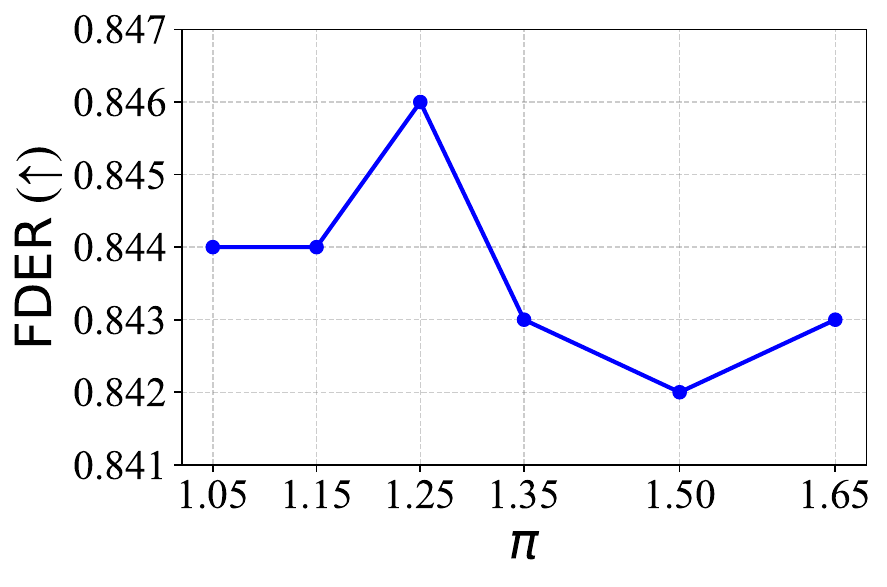}
    \end{subfigure}
    \hfill
    \begin{subfigure}{0.31\linewidth}
        \includegraphics[width=\linewidth]{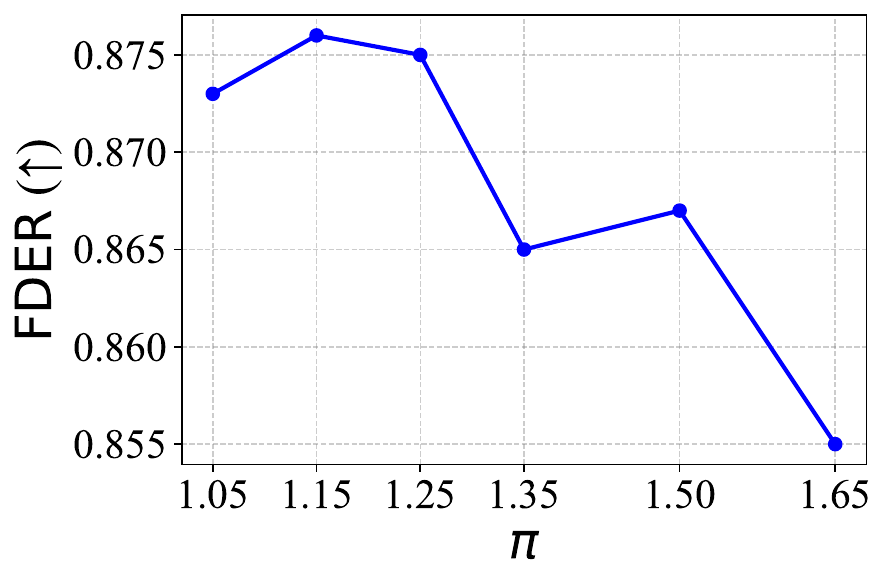}
    \end{subfigure}
    \hfill
    \begin{subfigure}{0.31\linewidth}
        \includegraphics[width=\linewidth]{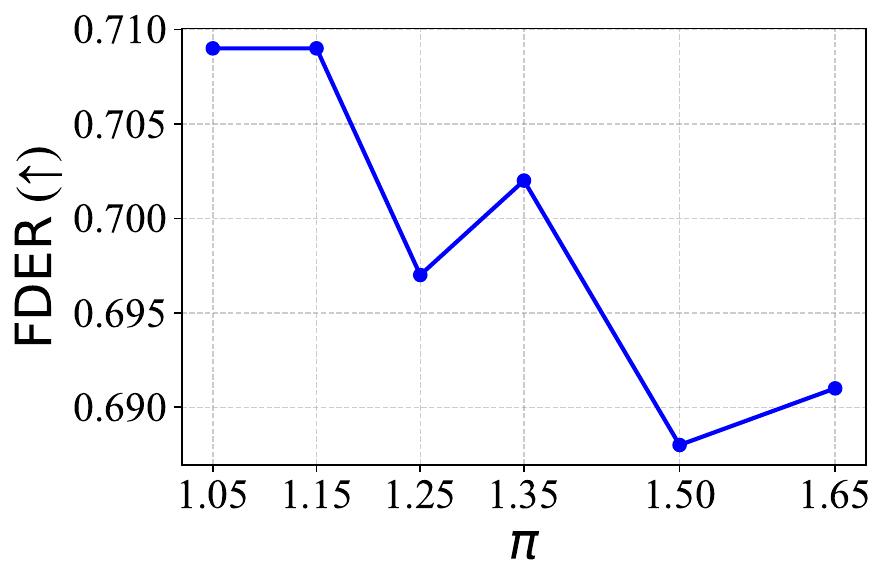}
    \end{subfigure}
    \vspace{-5pt}
    \caption{Defense performance of \methodname{} (\FDER{}) with  different scaling factor $\pi$ under \textbf{BackTime} attack on the \textbf{PEMS03} dataset with the FEDformer, SimpleTM, and TimesNet, respectively. }                
    \label{app-fig:different-pi-fder}
\end{figure}

We also vary $K$ and $\pi$ under the BackTime attack on the Weather dataset, as shown in \figureautorefname~\ref{app-fig:different-k-fder-weather}--\ref{app-fig:different-pi-fder-weather}, and under the Random attack on the PEMS03 dataset, as shown in \figureautorefname~\ref{app-fig:different-k-fder-random}--\ref{app-fig:different-pi-fder-random}. Consistent with the results on PEMS03 under BackTime in \sectionautorefname~\ref{sec:exp-main-result}, \methodname{} is relatively insensitive to the choice of $K$, and we recommend selecting $\pi \leq 1.5$. These results also suggest that the effects of $K$ and $\pi$ in \methodname{} are consistent across different datasets and attack scenarios.
\begin{figure}[!htbp]
    \centering
    \begin{subfigure}{0.24\linewidth}
        \includegraphics[width=\linewidth]{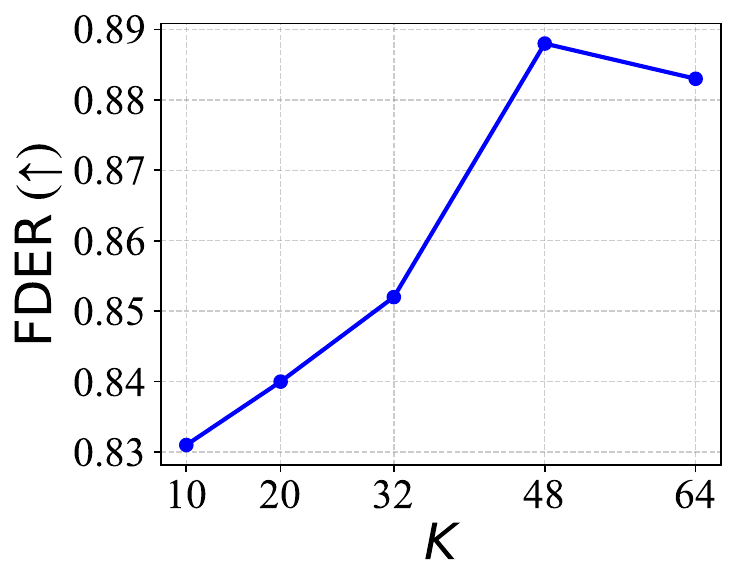}
    \end{subfigure}
    \hfill
    \begin{subfigure}{0.24\linewidth}
        \includegraphics[width=\linewidth]{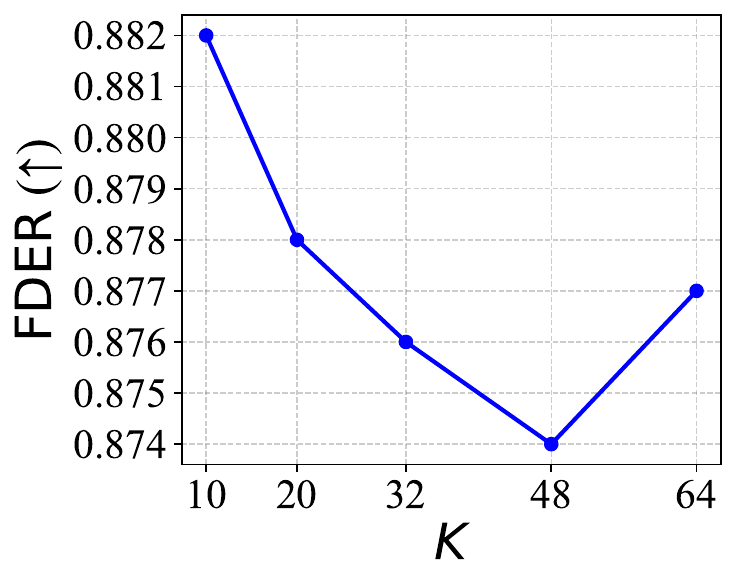}
    \end{subfigure}
    \hfill
    \begin{subfigure}{0.24\linewidth}
        \includegraphics[width=\linewidth]{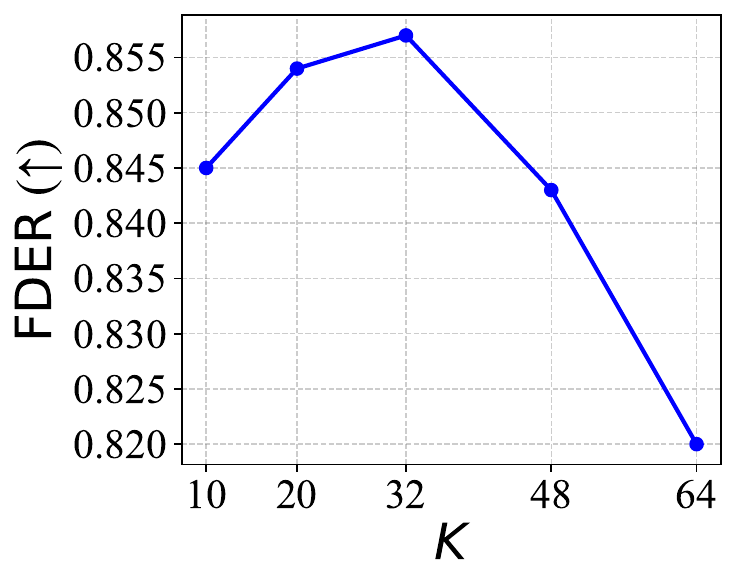}
    \end{subfigure}
    \hfill
    \begin{subfigure}{0.24\linewidth}
        \includegraphics[width=\linewidth]{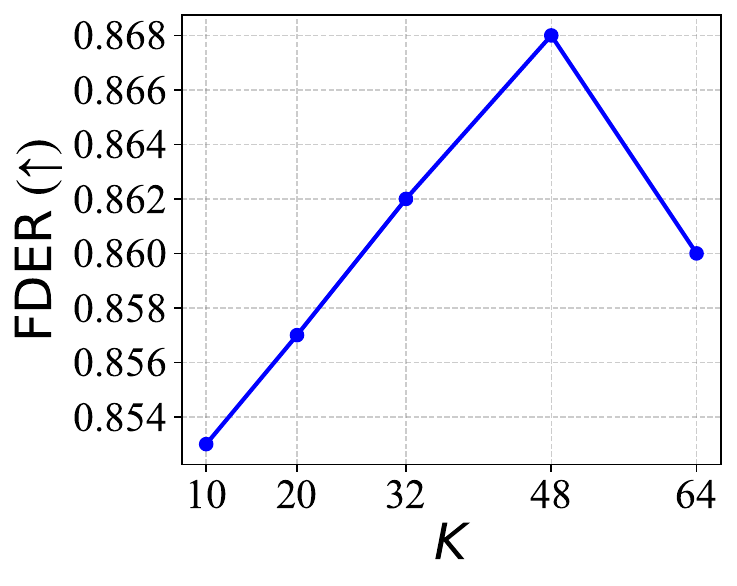}
    \end{subfigure}
    \vspace{-5pt}
    \caption{Defense performance of \methodname{} in terms of \FDER{} with different neighborhood sizes $K$ under the \textbf{BackTime} attack on the \textbf{Weather} dataset, reported for FEDformer, SimpleTM, TimesNet, and their average.}
    \label{app-fig:different-k-fder-weather}
    \vspace{-5pt}
\end{figure}

\begin{figure}[!htbp]
    \centering
    \begin{subfigure}{0.24\linewidth}
        \includegraphics[width=\linewidth]{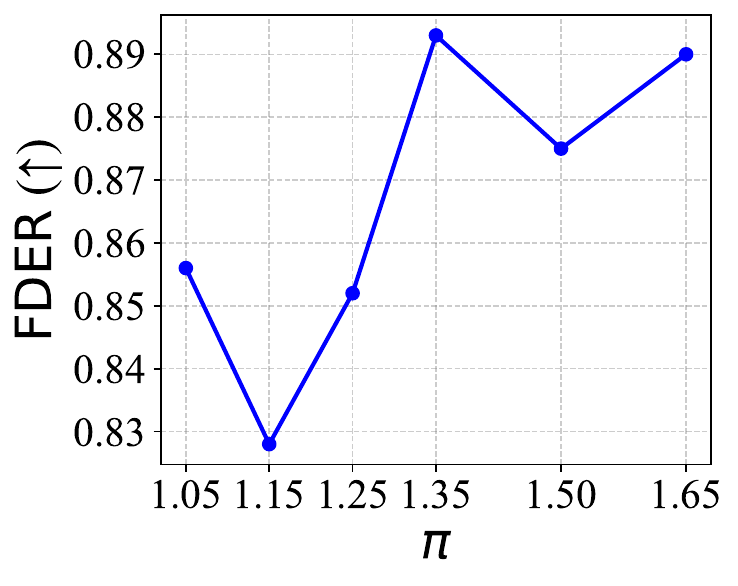}
    \end{subfigure}
    \hfill
    \begin{subfigure}{0.24\linewidth}
        \includegraphics[width=\linewidth]{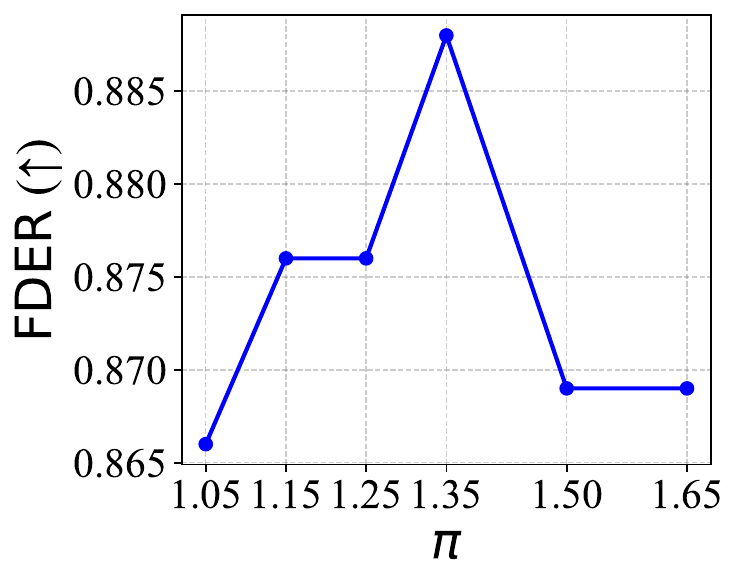}
    \end{subfigure}
    \hfill
    \begin{subfigure}{0.24\linewidth}
        \includegraphics[width=\linewidth]{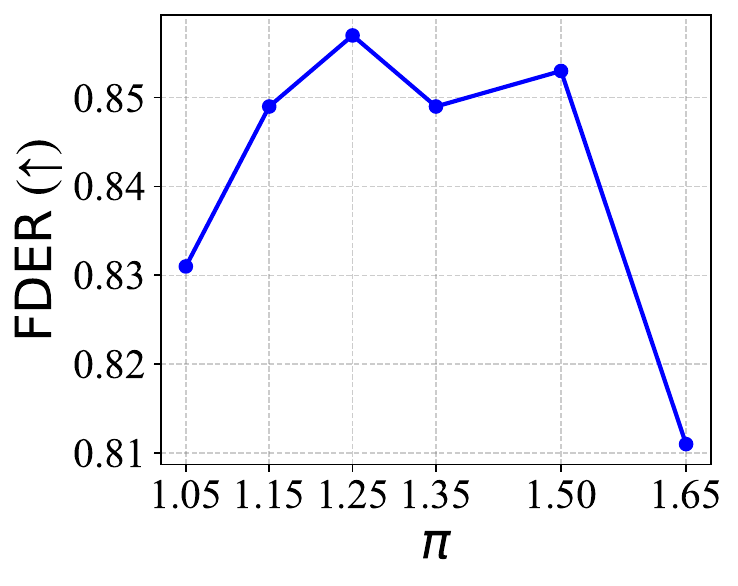}
    \end{subfigure}
    \hfill
    \begin{subfigure}{0.24\linewidth}
        \includegraphics[width=\linewidth]{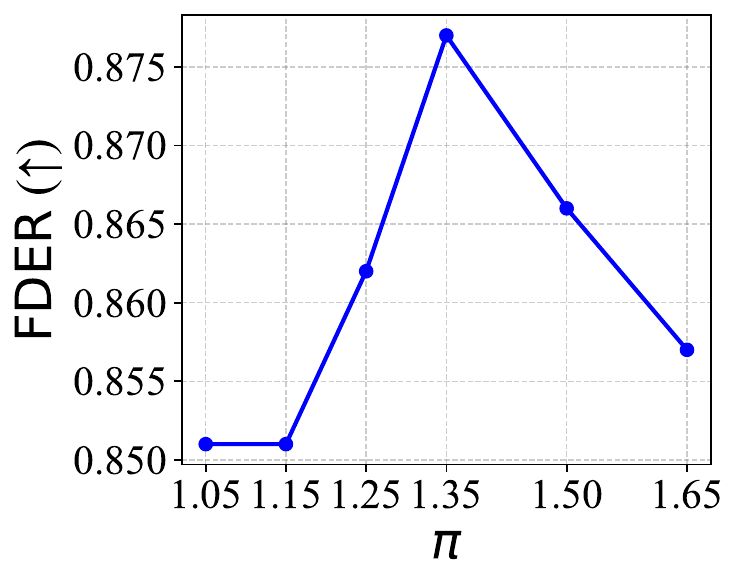}
    \end{subfigure}
    \vspace{-5pt}
    \caption{Defense performance of \methodname{} in terms of \FDER{} with different scaling factors $\pi$ under the \textbf{BackTime} attack on the \textbf{Weather} dataset, reported for FEDformer, SimpleTM, TimesNet, and their average.}   \label{app-fig:different-pi-fder-weather}
    \vspace{-5pt}
\end{figure}

\begin{figure}[!htbp]
    \centering
    \begin{subfigure}{0.24\linewidth}
        \includegraphics[width=\linewidth]{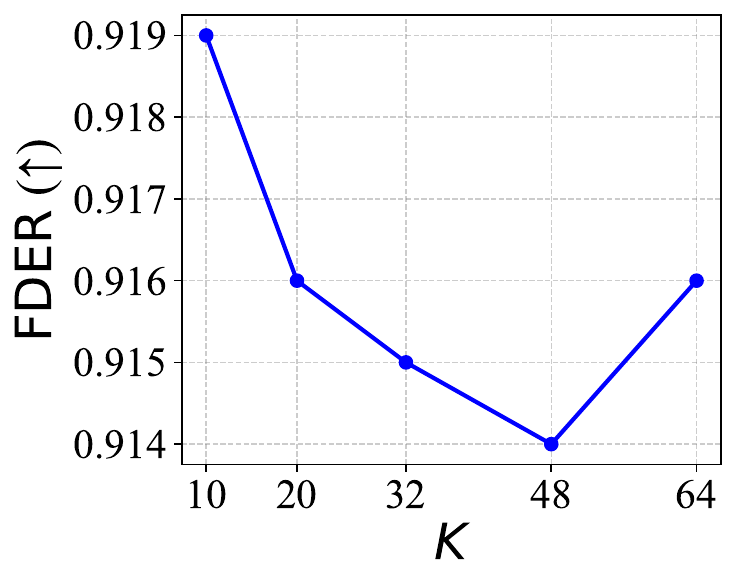}
    \end{subfigure}
    \hfill
    \begin{subfigure}{0.24\linewidth}
        \includegraphics[width=\linewidth]{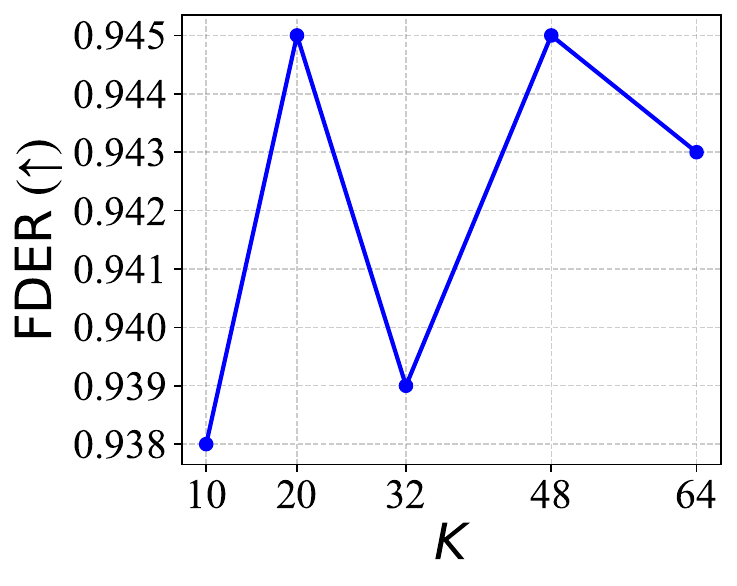}
    \end{subfigure}
    \hfill
    \begin{subfigure}{0.24\linewidth}
        \includegraphics[width=\linewidth]{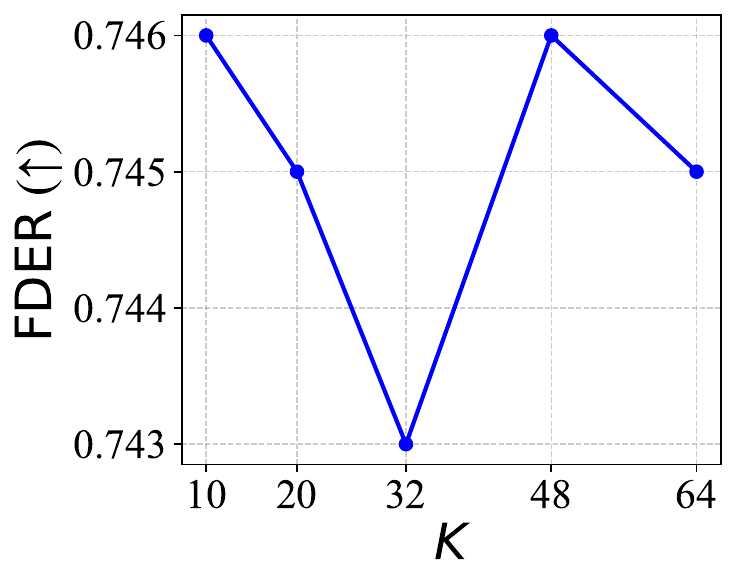}
    \end{subfigure}
    \hfill
    \begin{subfigure}{0.24\linewidth}
        \includegraphics[width=\linewidth]{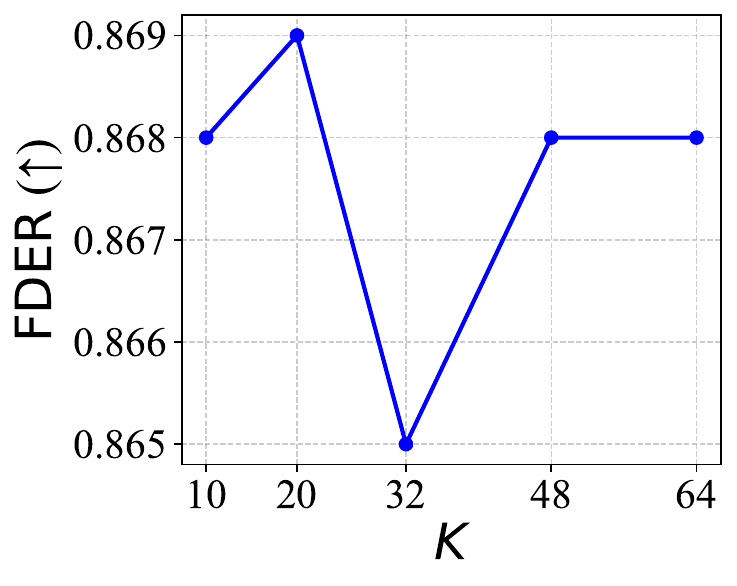}
    \end{subfigure}
    \vspace{-5pt}
    \caption{Defense performance of \methodname{} in terms of \FDER{} with different neighborhood sizes $K$ under the \textbf{Random} attack on the \textbf{PEMS03} dataset, reported for FEDformer, SimpleTM, TimesNet, and their average.}    \label{app-fig:different-k-fder-random}
    \vspace{-5pt}
\end{figure}

\begin{figure}[!htbp]
    \centering
    \begin{subfigure}{0.24\linewidth}
        \includegraphics[width=\linewidth]{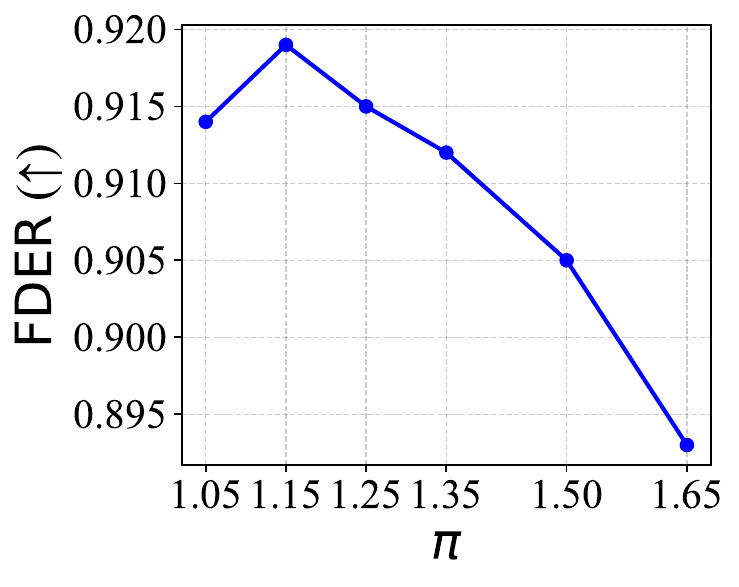}
    \end{subfigure}
    \hfill
    \begin{subfigure}{0.24\linewidth}
        \includegraphics[width=\linewidth]{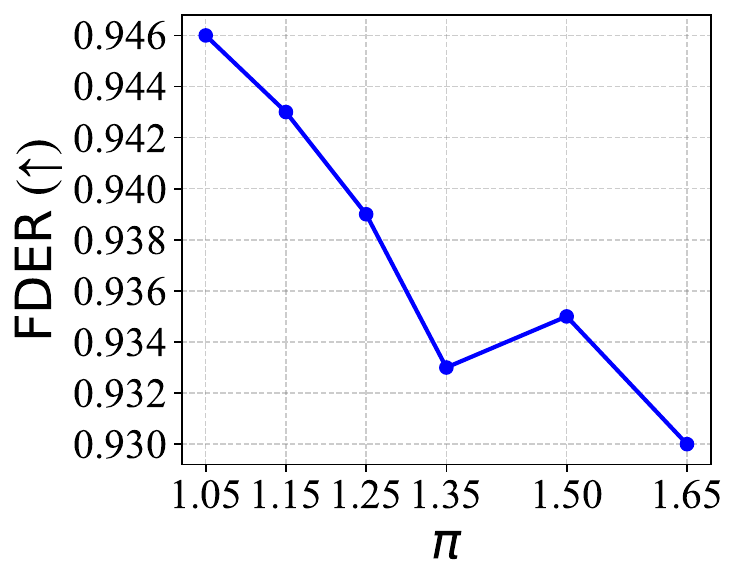}
    \end{subfigure}
    \hfill
    \begin{subfigure}{0.24\linewidth}
        \includegraphics[width=\linewidth]{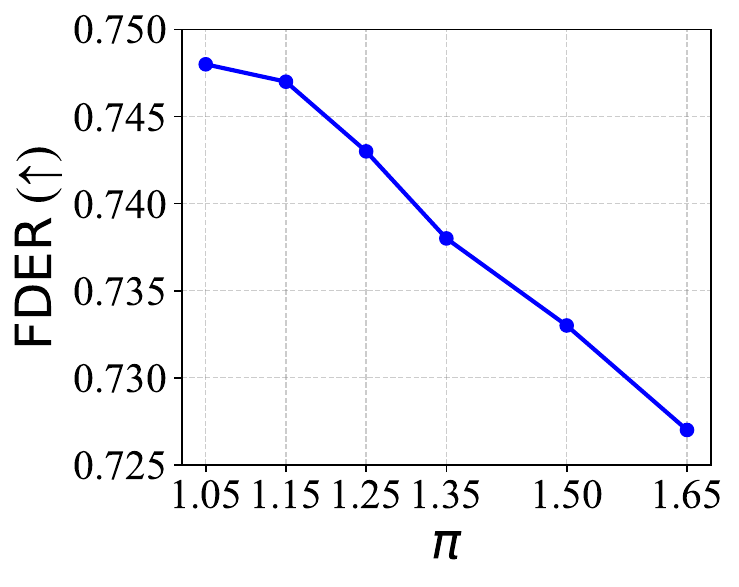}
    \end{subfigure}
    \hfill
    \begin{subfigure}{0.24\linewidth}
        \includegraphics[width=\linewidth]{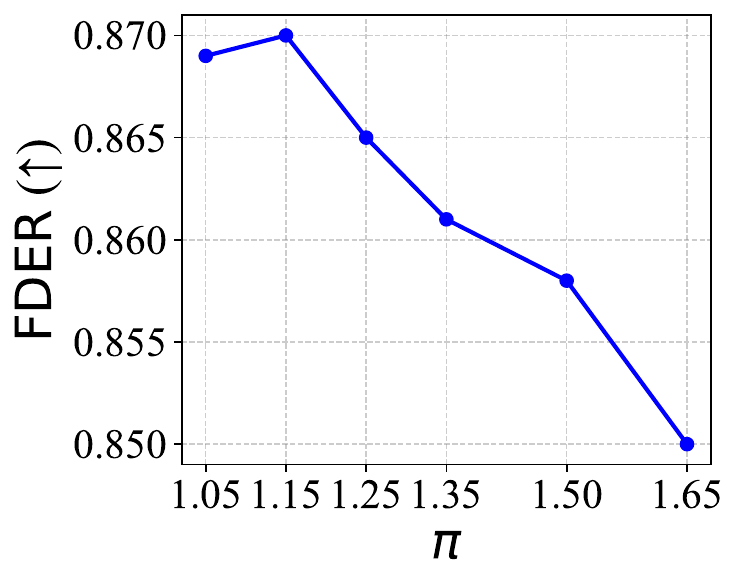}
    \end{subfigure}
    \vspace{-5pt}
    \caption{Defense performance of \methodname{} in terms of \FDER{} with different scaling factors $\pi$ under the \textbf{Random} attack on the \textbf{PEMS03} dataset, reported for FEDformer, SimpleTM, TimesNet, and their average.}    \label{app-fig:different-pi-fder-random}
    \vspace{-5pt}
\end{figure}

\textbf{Influence of $T_b$.} We further study the sensitivity to the number of backcaster training epochs $T_b$ used in the BLS module. Figure~\ref{app-fig:different-tb-fder} reports the defense performance of \methodname{} on PEMS03 under the BackTime attack for all three victim models while $T_b \in \{0, 5, 10, 20, 40, 60\}$. When $T_b < 10$, the backcaster is likely under-trained, making its loss signal less reliable for separating clean and poisoned samples, which results in suboptimal defense performance. To balance robustness and training cost, we set $T_b{=}10$ by default, which adds only roughly 10\% overhead to the training pipeline while reducing the risk of overfitting, where the backcaster may start fitting poisoned hard samples.

\begin{figure}[!htbp]
    \centering
    \begin{subfigure}{0.31\linewidth}
        \includegraphics[width=\linewidth]{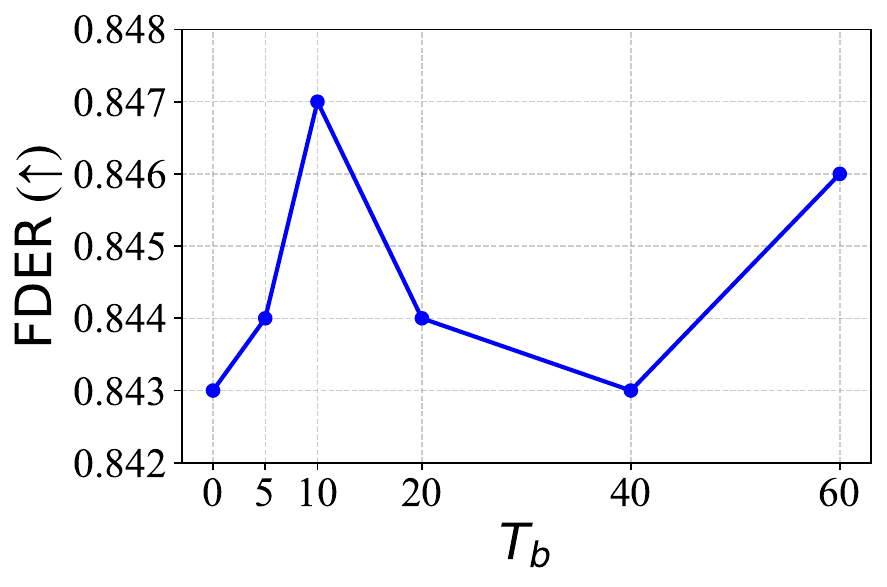}
    \end{subfigure}
    \hfill
    \begin{subfigure}{0.31\linewidth}
        \includegraphics[width=\linewidth]{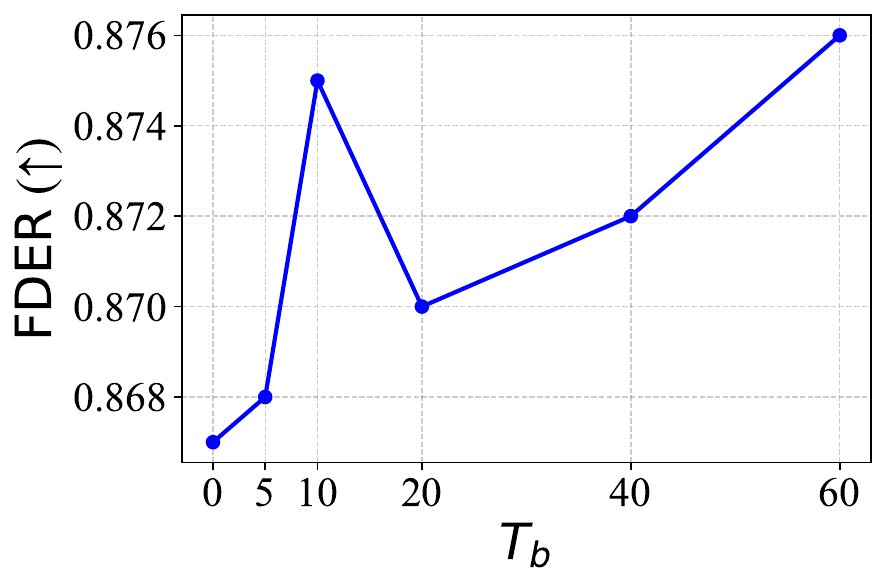}
    \end{subfigure}
    \hfill
    \begin{subfigure}{0.31\linewidth}
        \includegraphics[width=\linewidth]{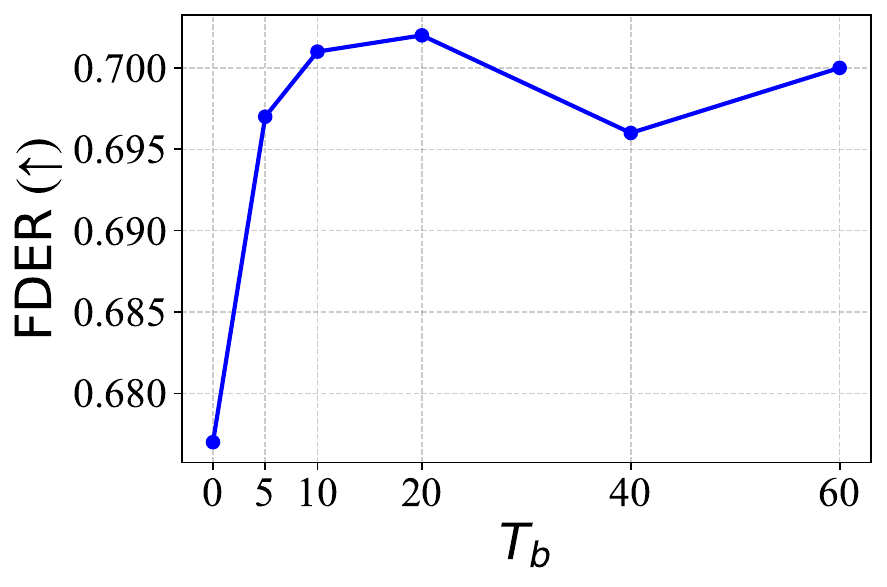}
    \end{subfigure}
    \vspace{-5pt}
    \caption{Defense performance of \methodname{} (\FDER{}) with varying backcaster $b_\phi$ training epoch $T_b$ under BackTime attack on the PEMS03 dataset with the FEDformer, SimpleTM, and TimesNet, respectively.}           
    \label{app-fig:different-tb-fder}
    \vspace{-12pt}
\end{figure}

\textbf{Influence of $T_1$ and $T_2$.} We study the sensitivity to the stage-wise training budgets $T_1$ (Stage I) and $T_2$ (Stage II) of \methodname{}, while fixing the total budget to $T_1 + T_2 = 100$. Figure~\ref{app-fig:different-t1-fder} reports the defense performance on PEMS03 under the BackTime attack for all three victim models, varying $T_1 \in \{0,5,10,20,30,40\}$ with the corresponding $T_2 \in \{100,95,90,80,70,60\}$. Overall, performance tends to degrade as $T_1$ increases, suggesting that overly long Stage I training may overfit the initial reliable pool and leave insufficient budget for incorporating newly admitted reliable samples in Stage II. Notably, for SimpleTM and TimesNet, a short Stage I training ($T_1  \leq 10$) improves subsequent progressive training, since Stage II relies on the current model's loss signal for selecting reliable samples. We set $T_1{=}10$ and $T_2{=}90$ by default, which is also consistent with common two-stage training schedules used in backdoor defenses~\cite{li2021anti, gao2023backdoor}.
\begin{figure}[!htbp]
    \centering
    \begin{subfigure}{0.31\linewidth}
        \includegraphics[width=\linewidth]{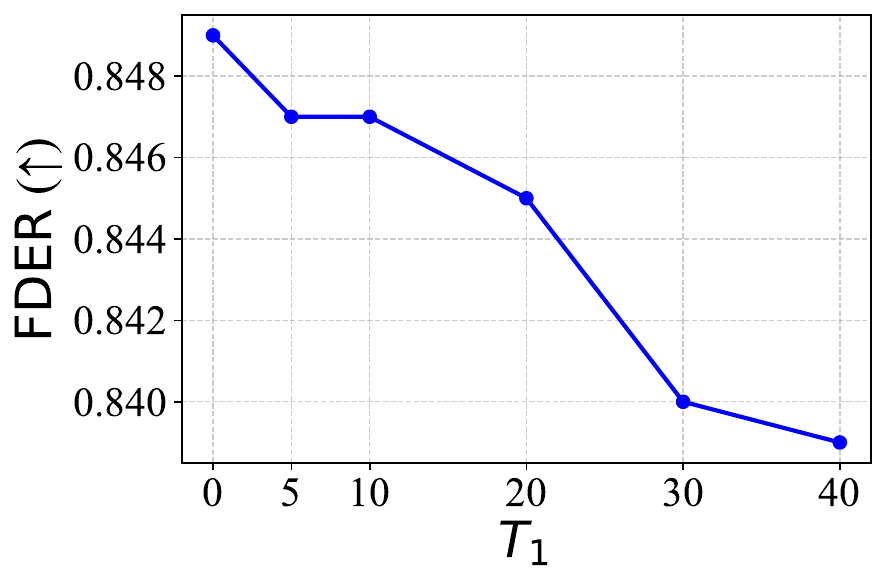}
    \end{subfigure}
    \hfill
    \begin{subfigure}{0.31\linewidth}
        \includegraphics[width=\linewidth]{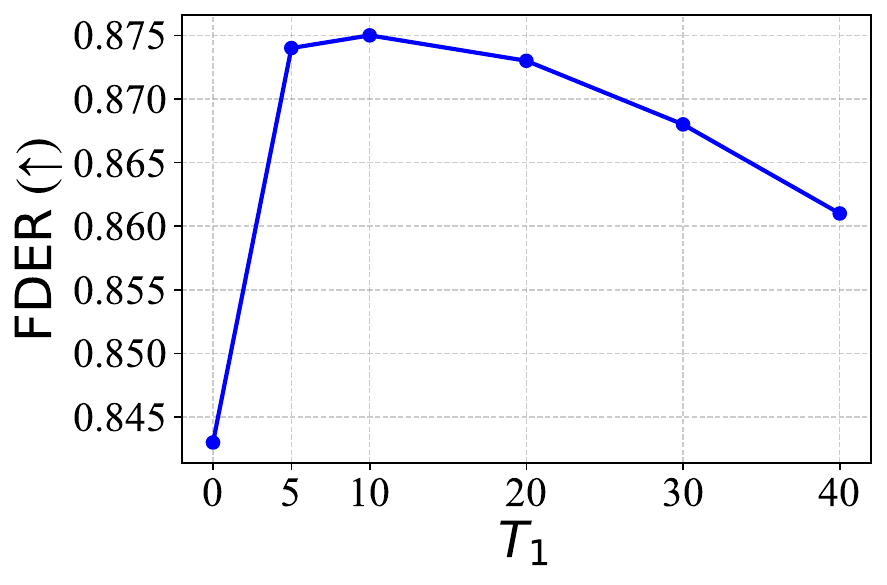}
    \end{subfigure}
    \hfill
    \begin{subfigure}{0.31\linewidth}
        \includegraphics[width=\linewidth]{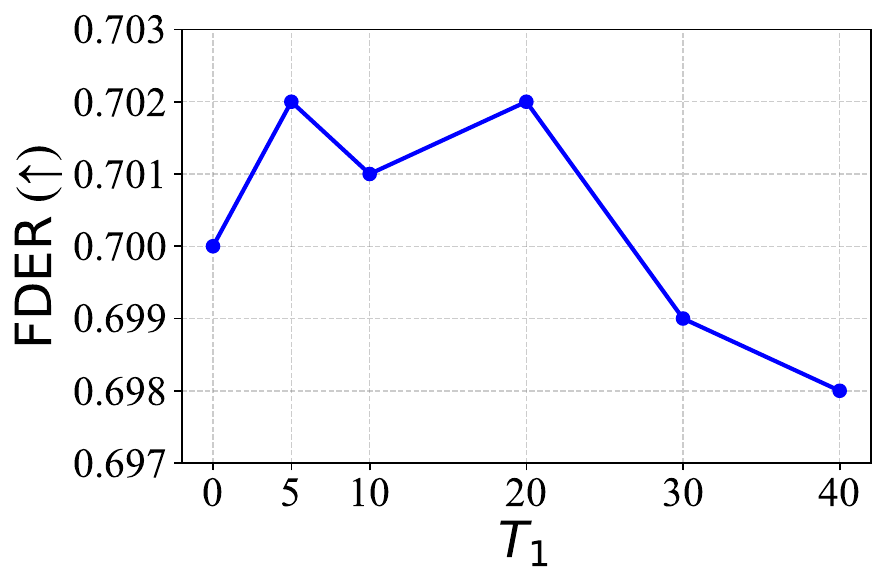}
    \end{subfigure}
    \vspace{-5pt}
    \caption{Defense performance of \methodname{} (\FDER{}) with varying training epoch $T_1$ under BackTime attack on the PEMS03 dataset with the FEDformer, SimpleTM, and TimesNet, respectively.}                
    \label{app-fig:different-t1-fder}
    \vspace{-5pt}
\end{figure}

\subsection{Detailed Efficiency Analysis}
\label{app:efficiency-analysis}

    As shown in \tableautorefname~\ref{tab:efficiency-analysis}, \methodname{} introduces additional training-time overhead, mainly from Stage I backcaster training and neighborhood-based filtering. However, this cost is limited and temporary: Stage I runs for only $T_b$ epochs, which we set to approximately $10\%$ of the total training budget by default. After Stage I, the backcaster is discarded, so its additional memory overhead does not persist into the main training stage. Importantly, this differs from the repeated inference-time overhead discussed in \tableautorefname~\ref{tab:inference-based-preliminary-table}, since \methodname{} introduces no additional latency during deployment.
        
    To reduce this cost, we implement neighborhood search in a precompute-and-reuse manner instead of recomputing all distances repeatedly. Specifically, we compute the channel-wise neighbor graph only once before the main Stage II training, cache the top-$K_\mathrm{max}$ neighbors for each sample, and reuse these cached neighborhoods throughout filtering and selection rather than recomputing kNN every epoch. By default, we set $K_\mathrm{max}=2K$; our preliminary experiments show that \methodname{} is insensitive to the choice of $K_\mathrm{max}$. This keeps the defense model-agnostic while avoiding repeated full-distance searches during training. As shown in \tableautorefname~\ref{app-tab:large-scale-model}, \methodname{} incurs training time comparable to PDB, the leading baseline, while achieving better defense performance on large-scale time-series foundation models based on LLaMA-7B~\cite{touvron2023llama}. Similar scalability trends are also observed on large-scale datasets, as shown in \tableautorefname~\ref{app-tab:gba-dataset}.

\subsection{Potential Adaptive Attacks}  \label{app:adaptive-attack}
    \begin{table}[htbp]
        \centering
        \caption{Defense performance of \methodname{} under BackTime and adaptive attacks on PEMS03 dataset, where FEDFormer, SimpleTM, and TimesNet are the victim models. Best results under adaptive attack are in \textbf{bold}.}
        \label{app-tab:adaptive-attack}
        \setlength{\tabcolsep}{2pt} 
        \renewcommand{\arraystretch}{1.2}
        \tiny
        \renewcommand{\aboverulesep}{0pt}
        \renewcommand{\belowrulesep}{0pt}
        \setlength\cellspacetoplimit{2pt}
        \setlength\cellspacebottomlimit{2pt}
        \resizebox{\linewidth}{!}{
        \begin{tabular}{clccccccccccccc}
        \toprule
        \multirow{2}{*}{\textbf{\begin{tabular}[c]{@{}c@{}}Attack\end{tabular}}}
        & \textbf{Model~→}
        & \multicolumn{3}{c}{\textbf{FEDformer}}
        & \multicolumn{3}{c}{\textbf{SimpleTM}}
        & \multicolumn{3}{c}{\textbf{TimesNet}} 
        & \multicolumn{3}{c}{\textbf{\textsc{Average}}} 
        \\
        \cmidrule(l{0pt}){2-2}
        \cmidrule(l){3-5}
        \cmidrule(l){6-8}
        \cmidrule(l){9-11}
        \cmidrule(l){12-14}
        & \textbf{Defense~↓}
        & \MAEC~↓ & \MAEP~↑ & \FDER~↑
        & \MAEC~↓ & \MAEP~↑ & \FDER~↑
        & \MAEC~↓ & \MAEP~↑ & \FDER~↑ 
        & \MAEC~↓ & \MAEP~↑ & \FDER~↑ 
        \\
        \midrule
        \multirow{2}{*}{BackTime} & No Defense 
            & 16.093 & 10.760 & --
            & 17.268 & 9.131 & --
            & 19.459 & 22.713 & -- 
            & 17.607 & 14.201 & -- \\
        &\textbf{\methodname}
            &16.840 & 41.232 & 0.847
            &17.243 & 36.626 & 0.875
            &20.061 & 40.052 & 0.701 
            &18.048 & 39.303 & 0.808\\
        \midrule
        \multirow{4}{*}{Adaptive} & No Defense 
            & 16.383 & 13.779 & --
            & 17.491 & 10.744 & --
            & 22.498 & 21.507 & -- 
            & 18.791 & 15.343 & -- \\
         &\textbf{\methodname}
            & 17.066 & \textbf{30.035} & \textbf{0.751}
            & 17.707 &\textbf{22.393} & \textbf{0.754}
            & \textbf{20.540} & \textbf{39.298} & \textbf{0.726} 
            & \textbf{18.438} & \textbf{30.575} & \textbf{0.744}\\
         &\textbf{\methodname{}} w/o NDF
            & \textbf{17.012} & 29.622 & 0.749
            & \textbf{17.287} & 22.053 & 0.756
            & 21.392 & 37.411 & 0.713 
            & 18.564 & 29.695 & 0.739\\
        &\textbf{\methodname{}} w/o DRLS
            & 19.035 & 14.558 & 0.457
            & 21.093 & 15.808 & 0.575
            & 22.460 & 26.712 & 0.597 
            & 20.863 & 19.026 & 0.543\\
        \bottomrule
        \end{tabular}
        }
    \vspace{-12pt}
    \end{table}    
    
    We consider a worst-case scenario in which the attacker deliberately adapts the attack strategy to circumvent our defense.

   \textbf{Design.} We construct an adaptive variant of BackTime~\cite{lin2024backtime} by augmenting the trigger-generator training objective. To challenge our unidirectional trigger-to-target assumption, we assume the attacker has access to a pre-trained backcaster $b_{\phi}$ trained on the clean dataset. The attacker then adds a reverse-consistency regularizer $L_{\text{uni}}$ that encourages the induced target pattern (specified by the attack target $\mathbf{P}$) to reconstruct the generated trigger $\mathbf{G}$, making the trigger and target more mutually predictive. In addition, to weaken our neighborhood similarity signal based on weighted Pearson correlation, the attacker introduces a similarity regularizer $L_{\text{sim}}$ that maintains a buffer of previously generated poisoned samples and penalizes the distance between the current poisoned sample and the buffer, thereby encouraging diversification. The resulting adaptive objective is:
    \begin{equation*}
    L_{\text{adap}} = L_{bd} + \lambda_1 L_{\text{uni}} + \lambda_2 L_{\text{sim}},
    \end{equation*}
    where $L_{bd}$ is the original BackTime backdoor objective and $\lambda_1,\lambda_2$ are attacker-controlled hyperparameters. In our implementation, we grid-search $\lambda_1 \in \{0.1, 0.5, 1\}$ and $\lambda_2 \in \{10, 100, 1000\}$.
    
    \textbf{Results.} \tableautorefname~\ref{app-tab:adaptive-attack} shows that, on PEMS03, this adaptive attack attains an average \MAEC{} of 18.791 and an average \MAEP{} of 15.343, which is slightly worse than the original BackTime attack (17.607 \MAEC{} and 14.201 \MAEP{}, averaged over models). This suggests that enforcing additional constraints, namely reverse consistency and similarity regularization, can hinder the attacker. This observation aligns with our analysis that effective TSF backdoors rely on generating highly similar trigger-induced patterns, as discussed in \theoremautorefname~\ref{thm:main-tsf-backdoor-success}.

    Meanwhile, \methodname{} remains effective under this adaptive threat, achieving 18.438 \MAEC{}, 30.575 \MAEP{}, and 0.744 \FDER{}, which remains within a strong defense-performance range. This trend is consistent across all forecasting models. We attribute this robustness primarily to the distance-aware criteria, which exploit the attacker’s structural need to produce highly correlated poisoned samples for successful backdoor activation.

\subsection{Neighborhood Distance Analysis} \label{app:neighborhood-analysis}

As \methodname{} relies on a hand-designed neighborhood metric, e.g., correlation-/distance-based $k$NN on normalized windows, input-space distances may degrade under strong distribution shifts or nonstationarity. To examine whether learned representations can provide a more robust neighborhood signal, we replace our Gaussian-weighted input-space distance with a TS2Vec embedding distance~\cite{yue2022ts2vec} on the Weather dataset under the BackTime attack. For implementation, we train TS2Vec on the full training set with hidden dimension 64, output dimension 128, and depth 6, and then use the resulting sample embeddings to compute neighborhood distances. Since this representation-learning step is computationally expensive, taking approximately 80,000 seconds, we evaluate this variant only on the moderate-scale Weather dataset.
\begin{table}[!htbp]
    \centering
        \vspace{-5pt}
        \caption{Defense performance of PDB, \methodname{}, and $\textsc{TimeGuard}_\mathrm{emb}$, a variant of \methodname{} that uses TS2Vec embeddings~\cite{yue2022ts2vec} as sample representations for neighborhood-distance computation, under the BackTime attack on the Weather dataset. FEDformer, SimpleTM, and TimesNet are used as victim models. Best results are shown in \textbf{bold}.}        
        \label{app-tab:t2v-embedding}
        \setlength{\tabcolsep}{2pt} 
        \renewcommand{\arraystretch}{1.2}
        \tiny
        \renewcommand{\aboverulesep}{0pt}
        \renewcommand{\belowrulesep}{0pt}
        \setlength\cellspacetoplimit{2pt}
        \setlength\cellspacebottomlimit{2pt}
        \resizebox{\linewidth}{!}{
        \begin{tabular}{lcccccccccccc}
            \toprule
            \textbf{Model~→}
            & \multicolumn{3}{c}{\textbf{FEDformer}}
            & \multicolumn{3}{c}{\textbf{SimpleTM}}
            & \multicolumn{3}{c}{\textbf{TimesNet}} 
            & \multicolumn{3}{c}{\textbf{\textsc{Average}}} 
            \\
            \cmidrule(l{0pt}){1-1}
            \cmidrule(l){2-4}
            \cmidrule(l){5-7}
            \cmidrule(l){8-10}
            \cmidrule(l){11-13}
            \textbf{Defense~↓}
            & \MAEC~↓ & \MAEP~↑ & \FDER~↑
            & \MAEC~↓ & \MAEP~↑ & \FDER~↑
            & \MAEC~↓ & \MAEP~↑ & \FDER~↑ 
            & \MAEC~↓ & \MAEP~↑ & \FDER~↑ 
            \\
            \midrule
            No Defense 
                & 9.609 & 8.020 & --
                & 7.752 & 15.301 & --
                & 14.943 & 24.417 & -- 
                & 10.768 & 15.913 & -- \\
            PDB~\cite{wei2024mitigating}
                & 10.254 & 35.951 & 0.857
                & 8.040 & 50.108 & 0.829
                & 16.903 & 83.259 & 0.795 
                & 11.732 & 56.439 & 0.827 \\
            \textbf{\methodname}
                & \textbf{10.089} & \textbf{43.244} & \textbf{0.883}
                & \textbf{7.934} & \textbf{69.357} & \textbf{0.878}
                & \textbf{14.125} & \textbf{87.000} & \textbf{0.860} 
                & \textbf{10.716} & \textbf{66.534} & \textbf{0.874} \\
            $\textbf{\textsc{TimeGuard}}_\mathrm{emb}$
                & 10.142 & 34.681 & 0.858
                & 8.256 & 64.011 & 0.850
                & 17.289 & 85.908 & 0.790 
                & 11.896 & 61.534 & 0.833 \\
            \bottomrule
        \end{tabular}
        }
        \vspace{-5pt}
    \end{table}  
    
\textbf{Results.} As shown in \tableautorefname~\ref{app-tab:t2v-embedding}, the embedding-based variant of \methodname{} still outperforms PDB across all models. However, it does not improve over the original input-space version. This suggests that generic learned embeddings such as TS2Vec do not automatically provide a stronger signal for TSF backdoor defense in our setting, consistent with previous observations on the limited effectiveness of generic data embeddings for time series forecasting~\cite{nematirad2025are}. We leave the design of more dedicated embeddings for TSF backdoor defense to future work.

\subsection{Clean Performance under No Attack} \label{app:no-attack}
     \begin{table}[!htbp]
        \centering
        \caption{Clean performance (\MAEC~↓) of in-training backdoor defenses under no attack scenario on PEMS03. Best results are in \textbf{bold}.}
        \label{app-tab:clean-no-attack}
        \setlength{\tabcolsep}{3.5pt}
        \renewcommand{\arraystretch}{1.25}
        \renewcommand{\aboverulesep}{0pt}
        \renewcommand{\belowrulesep}{0pt}
        \setlength\cellspacetoplimit{2pt}
        \setlength\cellspacebottomlimit{2pt}
        \begin{tabular}{lcccc}
        \toprule
        \textbf{Model →}   & \multirow{2}{*}{\textbf{SimpleTM}} & \multirow{2}{*}{\textbf{FEDformer}} & \multirow{2}{*}{\textbf{TimesNet}} & \multirow{2}{*}{\textbf{\textsc{Average}}} \\
        \textbf{Defense ↓} &&&& \\ \midrule
        Vanilla Training & 16.794 & 15.680 & 20.257 & 17.577 \\
        \midrule
        ABL       & \textbf{17.129} & 16.928 & 21.011 & 18.356 \\
        PDB       & 17.828 & 16.843 & 22.308 & 18.993 \\
        ESTI      & 17.396 & \textbf{15.915} & 20.119 & \textbf{17.810} \\
        \textbf{\methodname{}} & 17.197 & 16.695 & \textbf{19.804} & 17.899 \\
        \bottomrule
        \end{tabular}
        \vspace{-5pt}
    \end{table}

    In realistic scenarios, the defender may not know whether the training set has been poisoned. We therefore evaluate an extreme setting where no poisoning is present (\emph{No Attack}). Table~\ref{app-tab:clean-no-attack} reports the clean forecasting performance of four in-training defenses on the PEMS03 dataset under this setting. Overall, ESTI and \methodname{} preserve clean accuracy well, with at most a 3.5\% degradation across three  models compared to vanilla training, and they even improve performance for TimesNet in some cases; both outperform ABL and PDB in this no-attack regime. However, ESTI incurs substantially higher training cost and is more prone to failing under attacks in our TSF setting (\appendixautorefname~\ref{app:efficiency-analysis} and \sectionautorefname~\ref{sec:exp-main-result}). Taken together, these results suggest that \methodname{} offers a more practical trade-off when the poisoning status of the training data is unknown.

\subsection{Reliable Pool Dynamics Illustration} \label{app:pool-dynamics}

To illustrate how \methodname{} maintains a reliable pool during training, we plot (i) the number of poisoned samples admitted into the reliable pool at each epoch, together with the corresponding (ii) clean performance (\MAEC{}) and poisoned performance (\MAEP{}) of FEDformer. We compare \methodname{} against its loss-only variant (w/o NDF+DRLS) under the Random attack on three datasets, as shown in Figures~\ref{app-fig:clean-pool-illustration-pems03}–\ref{app-fig:clean-pool-illustration-ettm1} for PEMS03, Weather, and ETTm1, respectively. Overall, \methodname{} consistently admits fewer poisoned samples into the reliable pool than the loss-only variant, helping explain its strong robustness and competitive clean performance. These dynamics also highlight the importance of incorporating neighborhood-distance cues beyond loss-only criteria.

\begin{figure}[!htbp]
    \centering
    \begin{subfigure}{0.49\linewidth}
        \includegraphics[width=\linewidth]{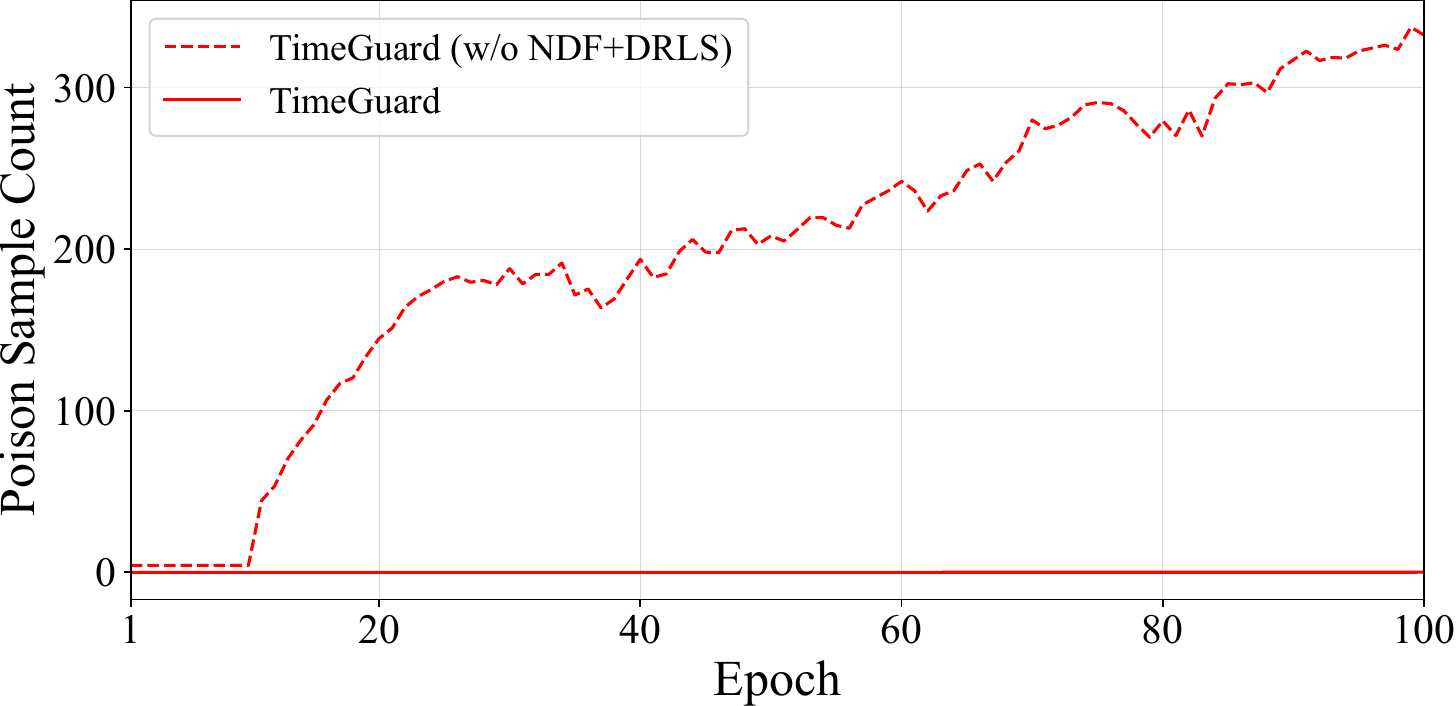}
        \caption{Number of poisoned samples in the reliable pool.}
    \end{subfigure}
    \begin{subfigure}{0.49\linewidth}
        \includegraphics[width=\linewidth]{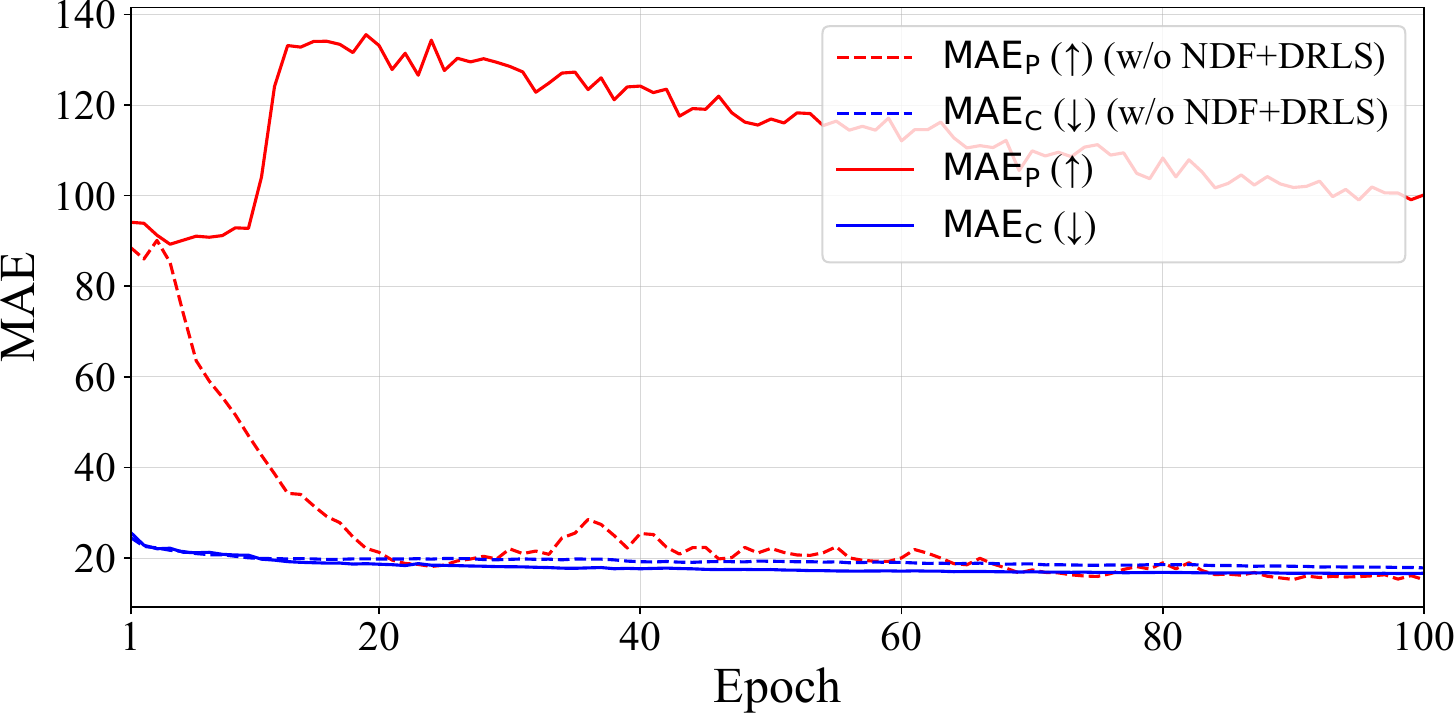}
        \caption{Clean and defense performance (\MAEC{} and \MAEP{}).}

    \end{subfigure}
    \caption{Dynamic illustration of \methodname{} at each training epoch under Random attack on PEMS03 dataset of FEDformer model.}                
    \label{app-fig:clean-pool-illustration-pems03}
    \vspace{-5pt}
\end{figure}

\begin{figure}[!htbp]
    \centering
    \begin{subfigure}{0.49\linewidth}
        \includegraphics[width=\linewidth]{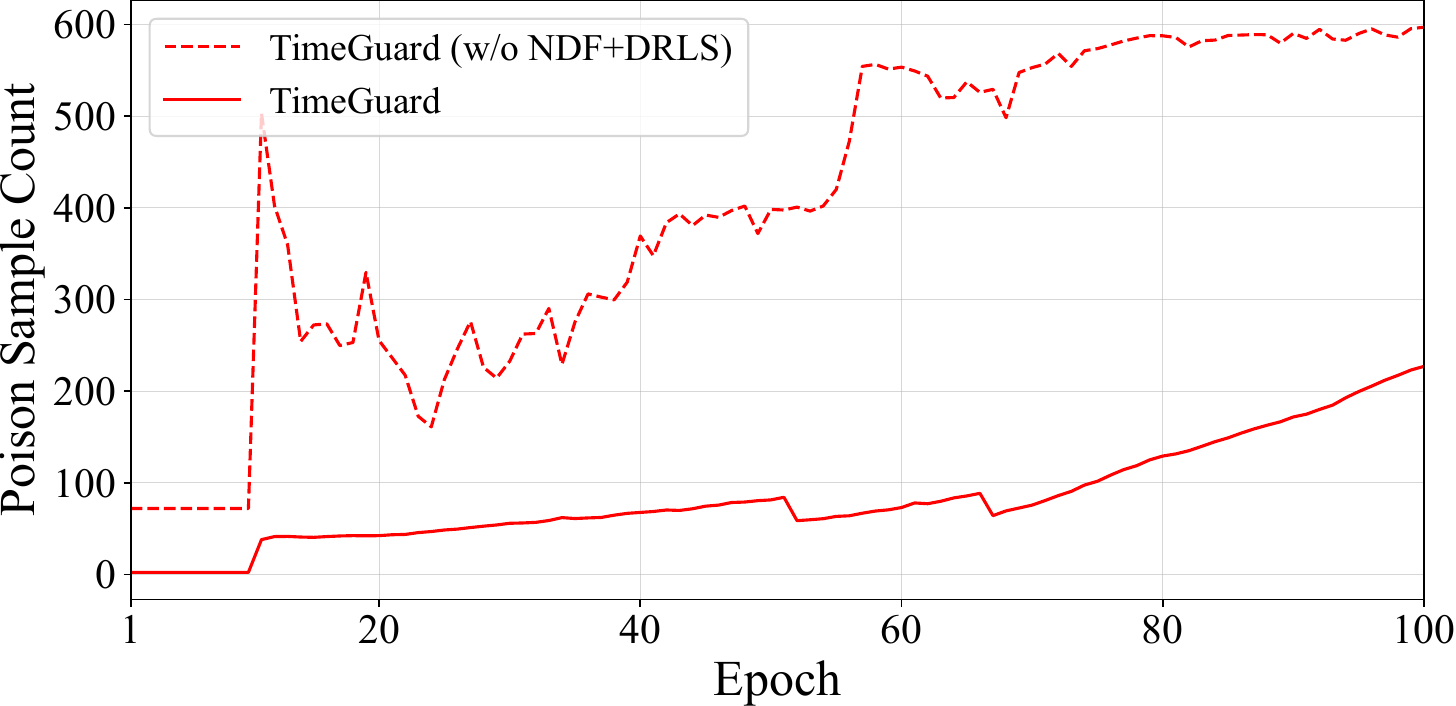}
        \caption{Number of poisoned samples in the reliable pool.}
    \end{subfigure}
    \begin{subfigure}{0.49\linewidth}
        \includegraphics[width=\linewidth]{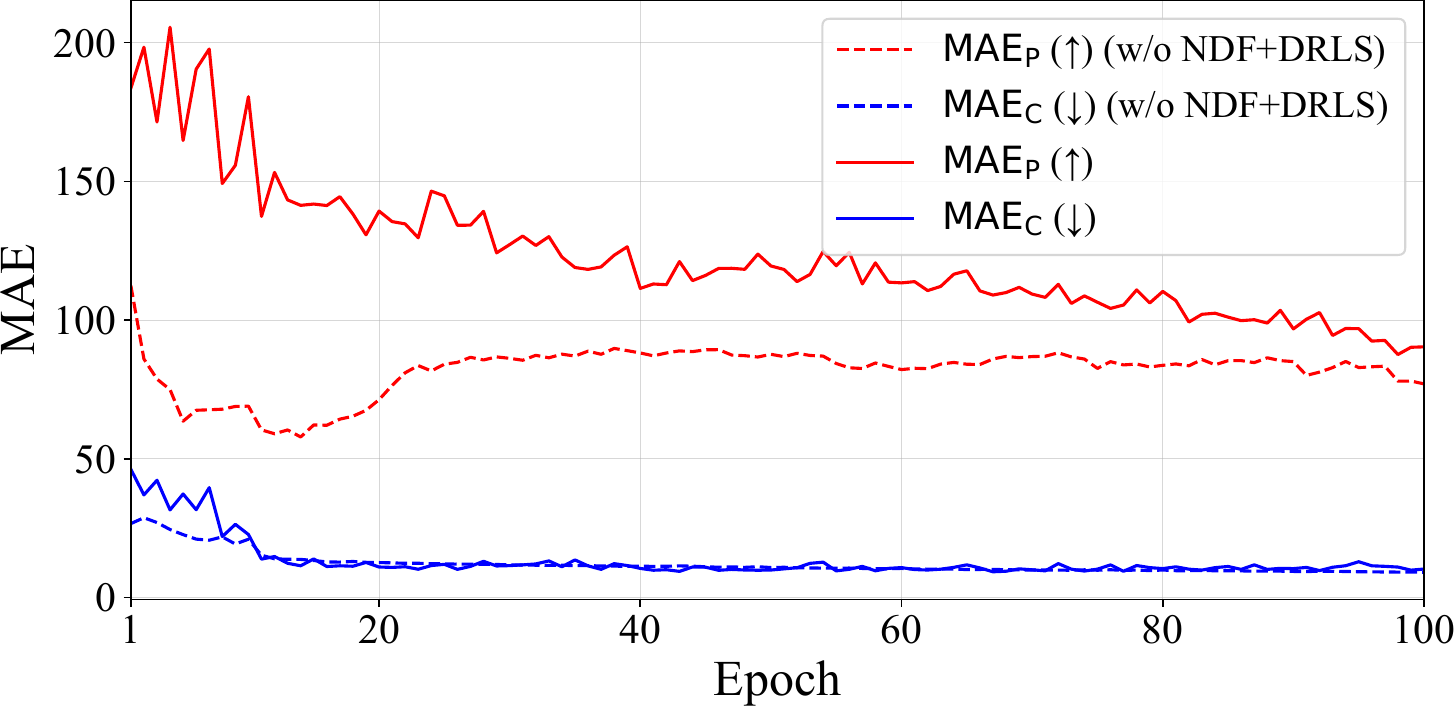}
        \caption{Clean and defense performance (\MAEC{} and \MAEP{}).}

    \end{subfigure}
    \caption{Dynamic illustration of \methodname{} at each training epoch under Random attack on Weather dataset of FEDformer model.}                
    \label{app-fig:clean-pool-illustration-weather}
\end{figure}

\begin{figure}[!htbp]
    \centering
    \begin{subfigure}{0.49\linewidth}
        \includegraphics[width=\linewidth]{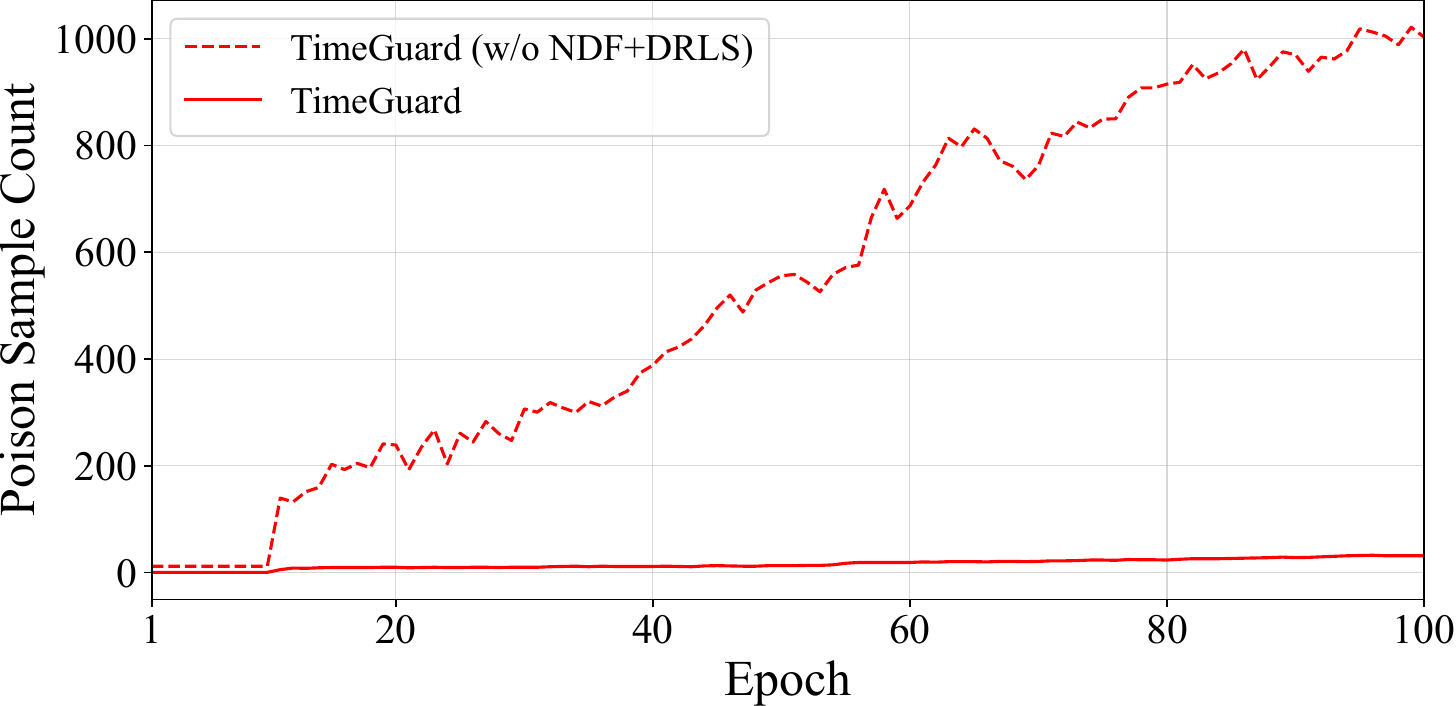}
        \caption{Number of poisoned samples in the reliable pool.}
    \end{subfigure}
    \begin{subfigure}{0.49\linewidth}
        \includegraphics[width=\linewidth]{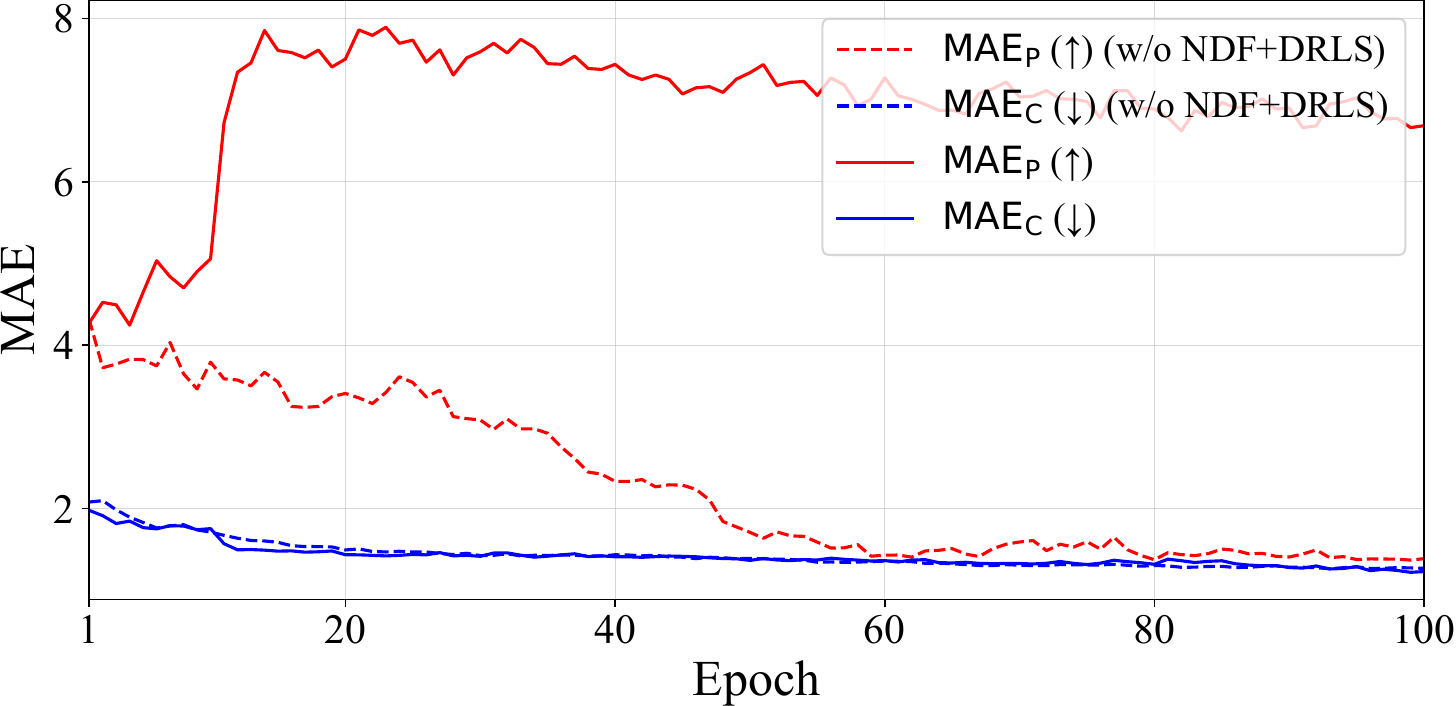}
        \caption{Clean and defense performance (\MAEC{} and \MAEP{}).}
    \end{subfigure}
    \caption{Dynamic illustration of \methodname{} at each training epoch under Random attack on ETTm1 dataset of FEDformer model.}                
    \label{app-fig:clean-pool-illustration-ettm1}
\end{figure}

\section{Showcases} \label{app:showcases}

To better visualize the effectiveness of \methodname{}, we provide an inference-time prediction example for the FEDformer model under the BackTime attack on PEMS03 in Figure~\ref{app-fig:clean-pool-illustration-pems03}, where the showcased triggers are sampled from different channels and different test samples. Overall, \methodname{} preserves accurate forecasts on clean channels while substantially mitigating trigger-induced manipulation on poisoned channels. Moreover, even when the input window is perturbed by the trigger, \methodname{} can partially recover the underlying future trend. We also observe that the generated triggers exhibit similar shapes despite BackTime using sample-dependent triggers, which further supports our analysis in \theoremautorefname~\ref{thm:main-tsf-backdoor-success}.
\begin{figure}[!htbp]
    \centering
    \begin{subfigure}{0.24\linewidth}
        \includegraphics[width=\linewidth]{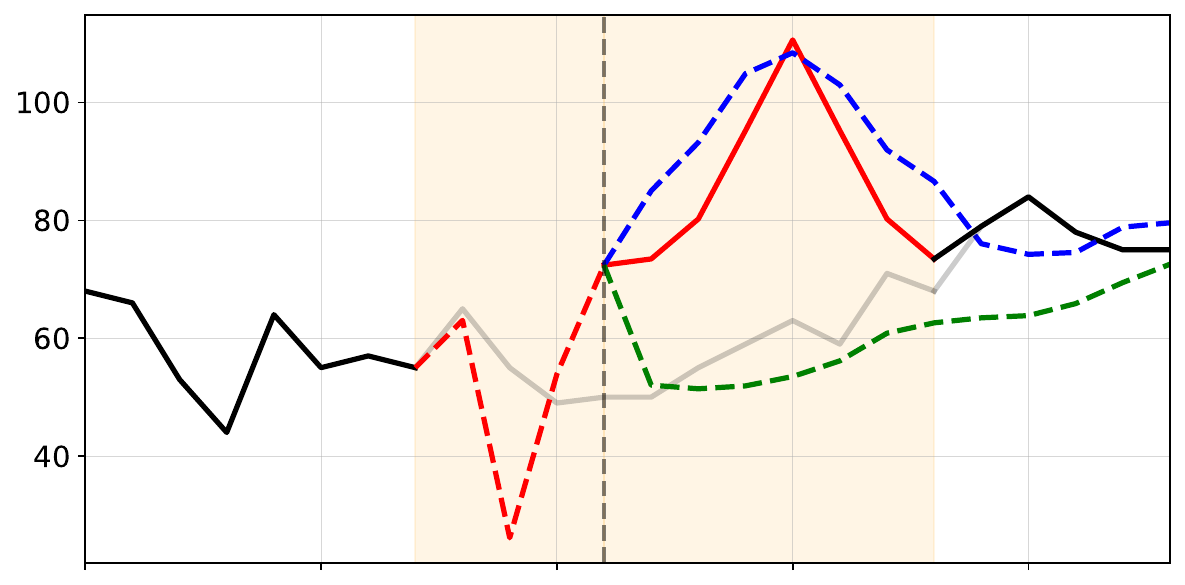}
    \end{subfigure}
    \hfill
    \begin{subfigure}{0.24\linewidth}
        \includegraphics[width=\linewidth]{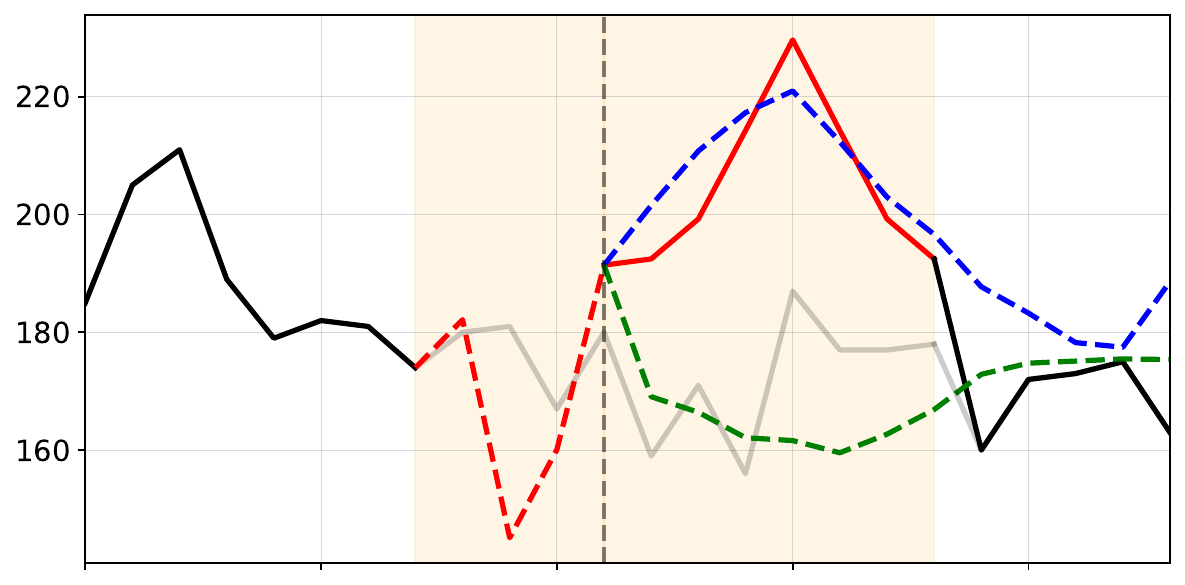}
    \end{subfigure}
    \hfill
     \begin{subfigure}{0.24\linewidth}
        \includegraphics[width=\linewidth]{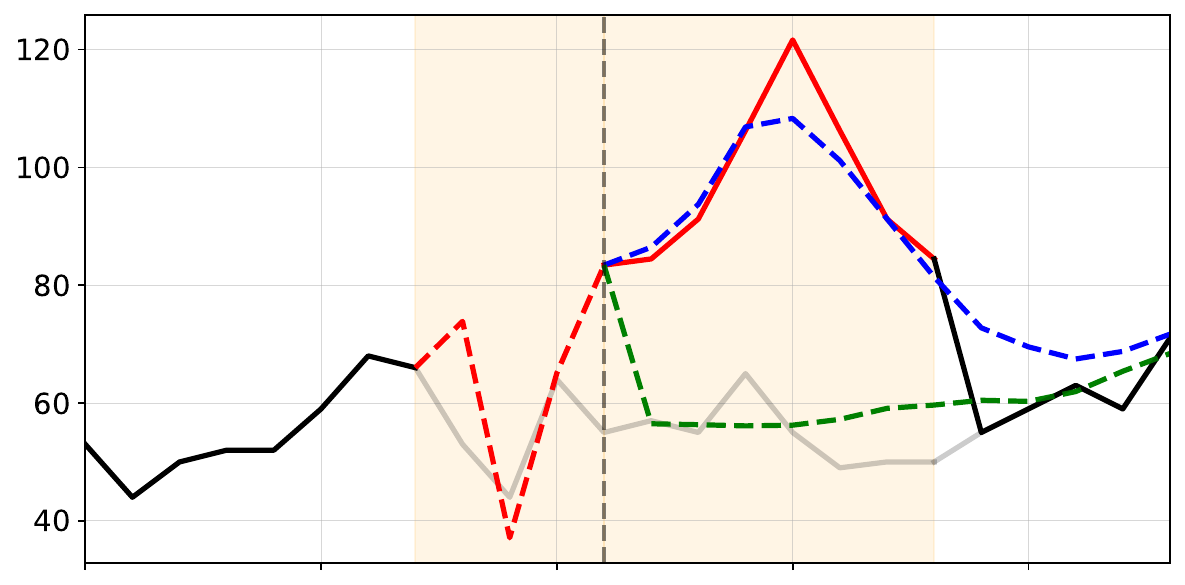}
    \end{subfigure}
    \hfill
     \begin{subfigure}{0.24\linewidth}
        \includegraphics[width=\linewidth]{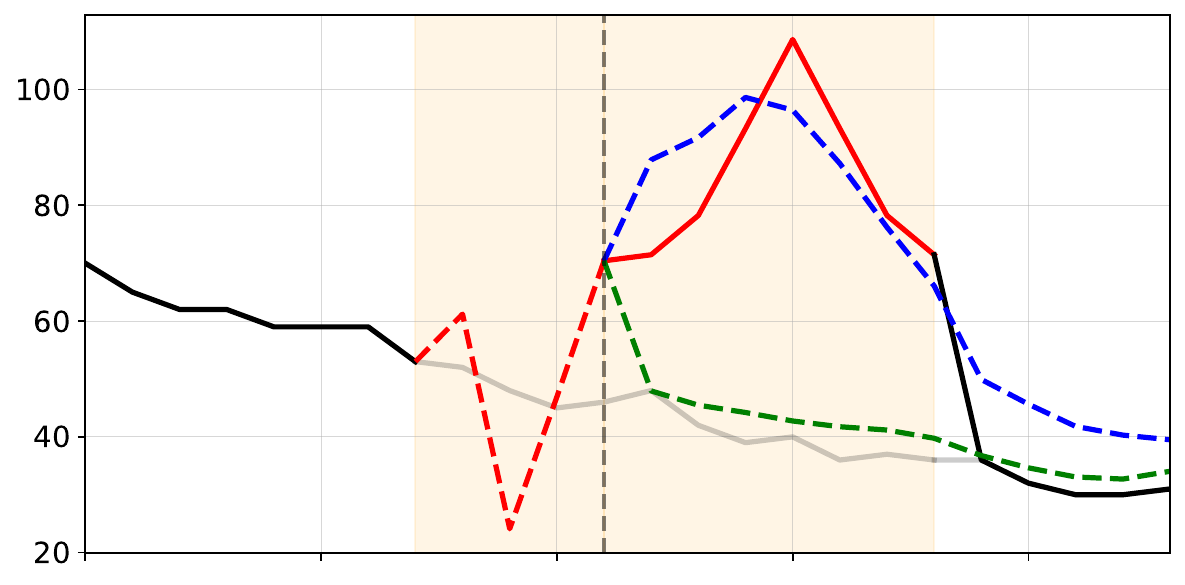}
    \end{subfigure}
    \vspace{0.6em} 
    \begin{subfigure}{0.24\linewidth}
        \includegraphics[width=\linewidth]{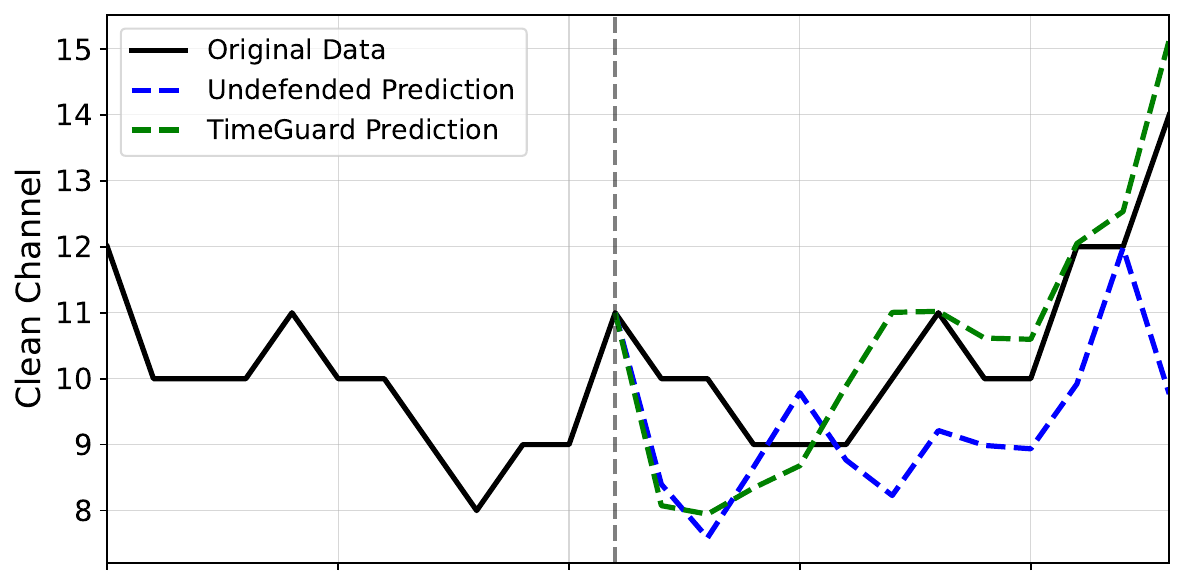}
    \end{subfigure}
    \hfill
    \begin{subfigure}{0.24\linewidth}
        \includegraphics[width=\linewidth]{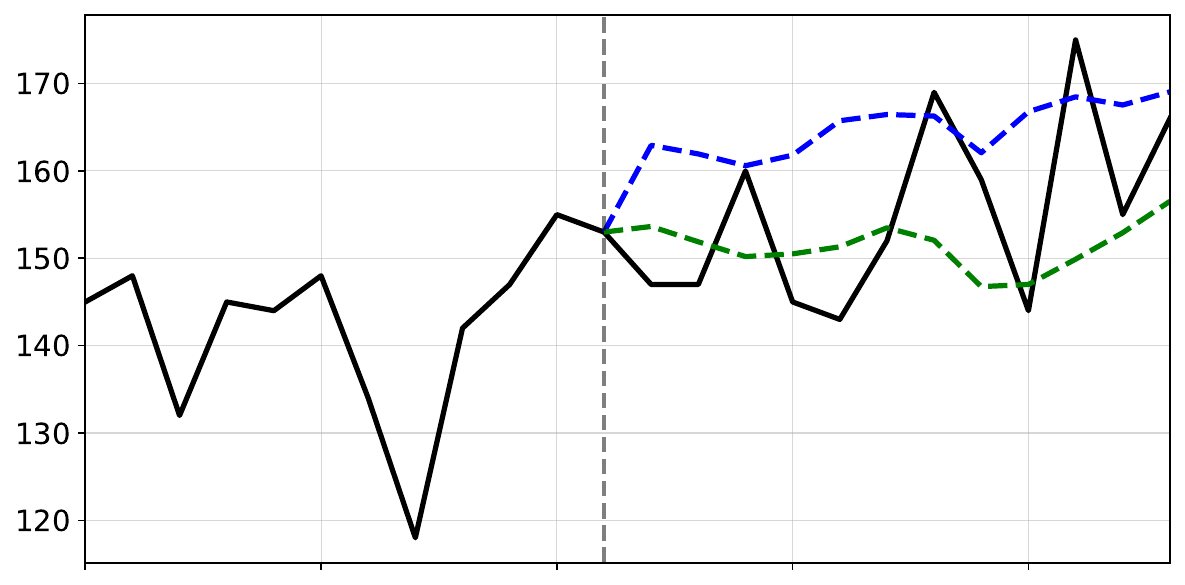}
    \end{subfigure}
    \hfill
     \begin{subfigure}{0.24\linewidth}
        \includegraphics[width=\linewidth]{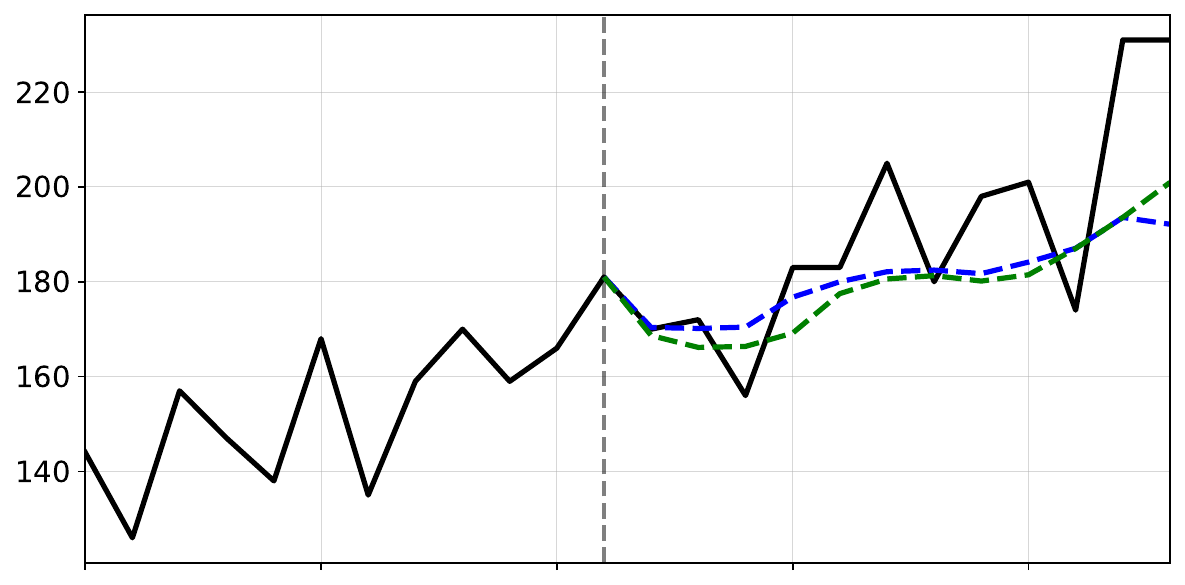}
    \end{subfigure}
    \hfill
     \begin{subfigure}{0.24\linewidth}
        \includegraphics[width=\linewidth]{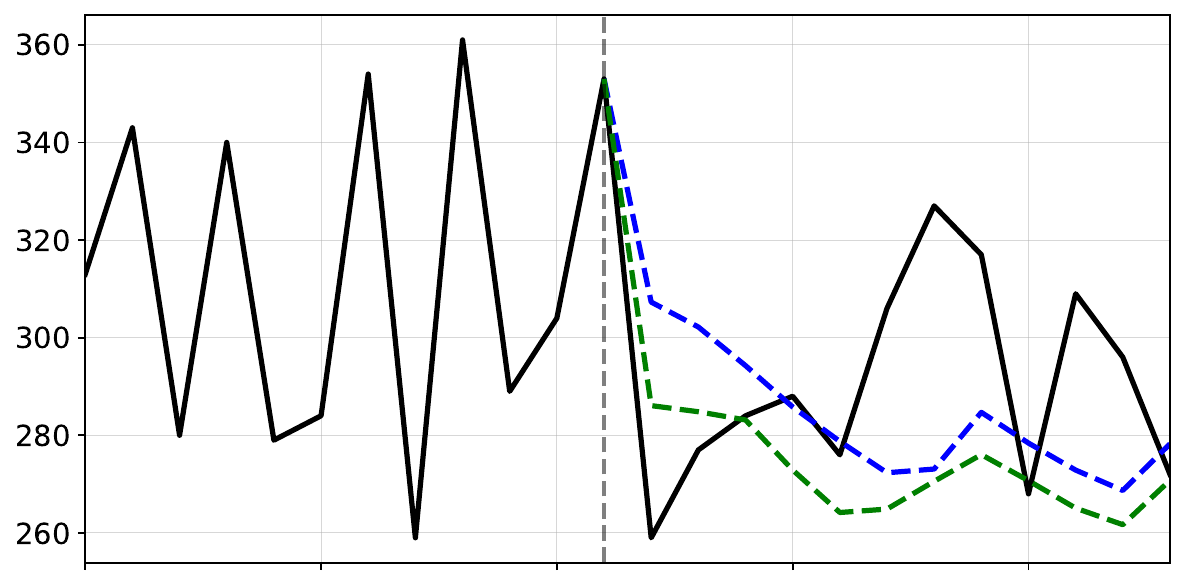}
    \end{subfigure}
    \vspace{0.6em}
     \begin{subfigure}{0.24\linewidth}
        \includegraphics[width=\linewidth]{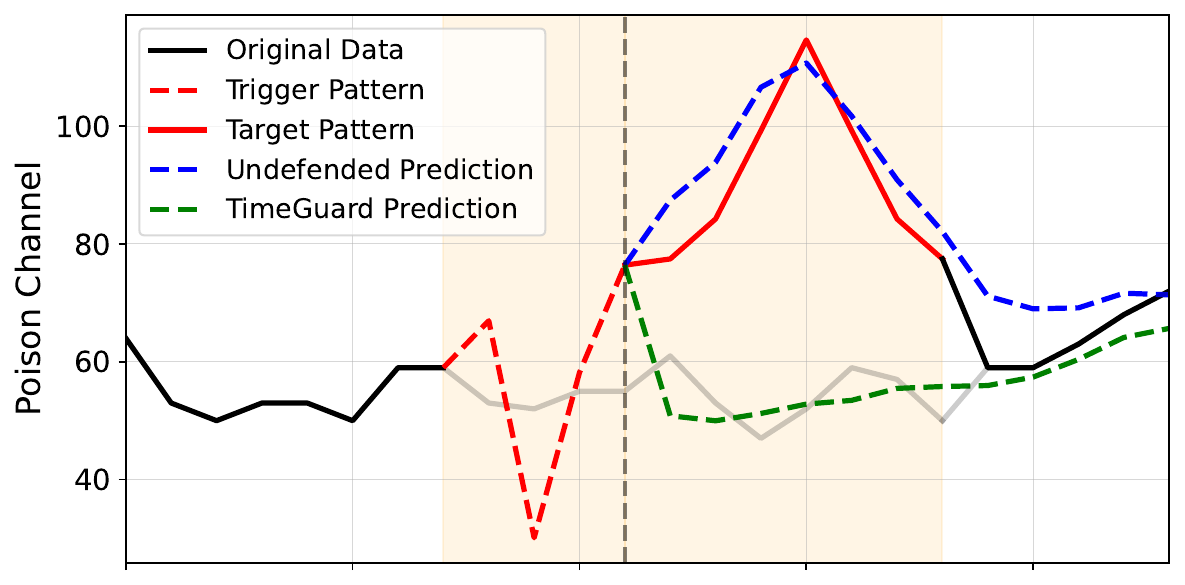}
    \end{subfigure}
    \hfill
    \begin{subfigure}{0.24\linewidth}
        \includegraphics[width=\linewidth]{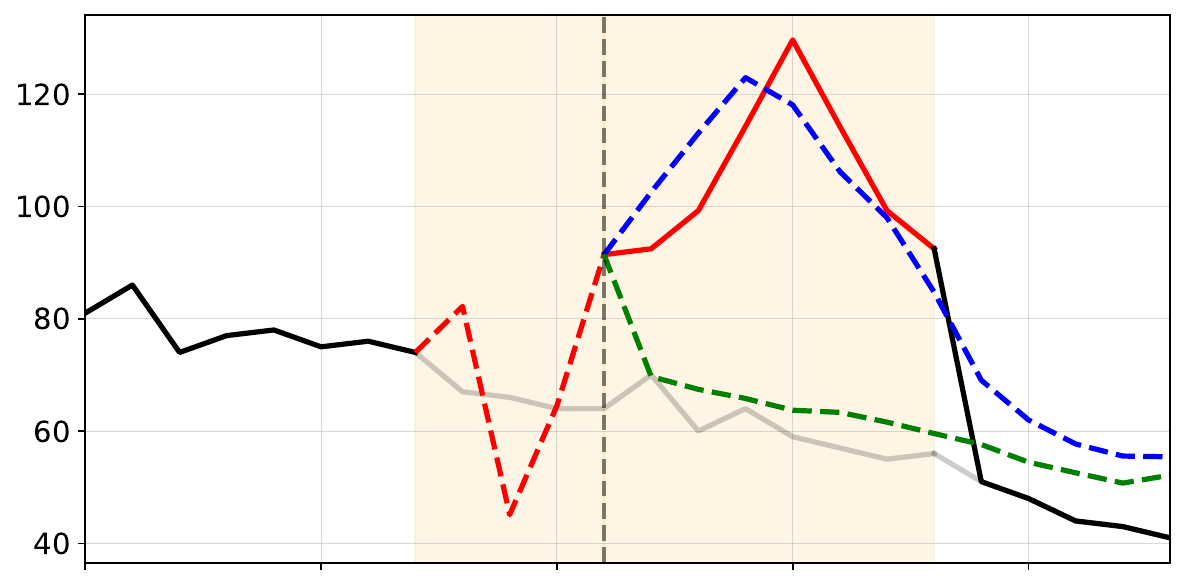}
    \end{subfigure}
    \hfill
     \begin{subfigure}{0.24\linewidth}
        \includegraphics[width=\linewidth]{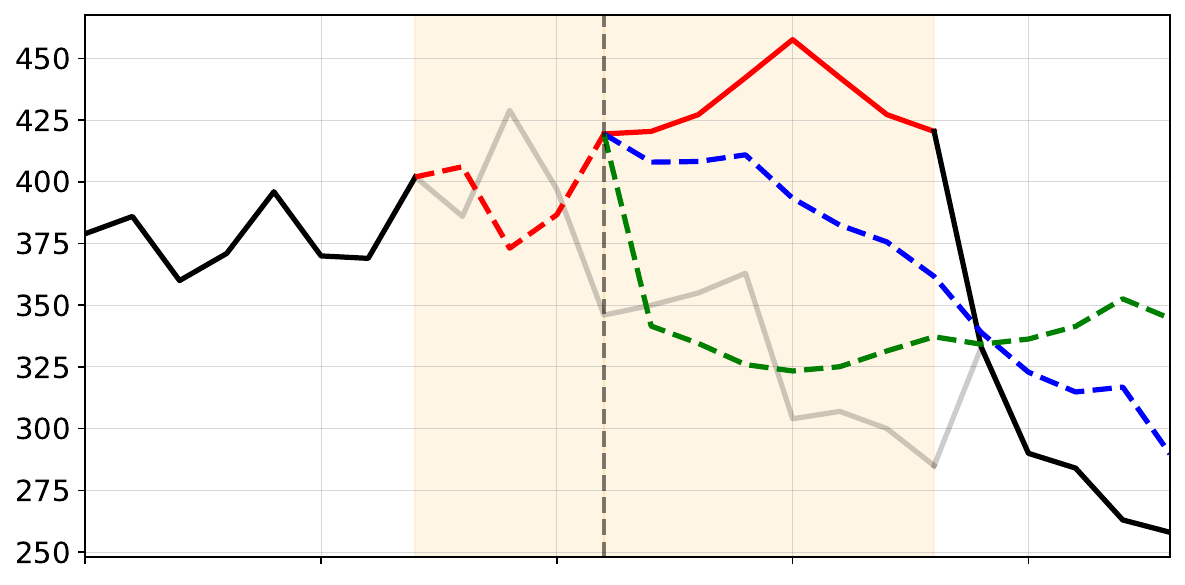}
    \end{subfigure}
    \hfill
     \begin{subfigure}{0.24\linewidth}
        \includegraphics[width=\linewidth]{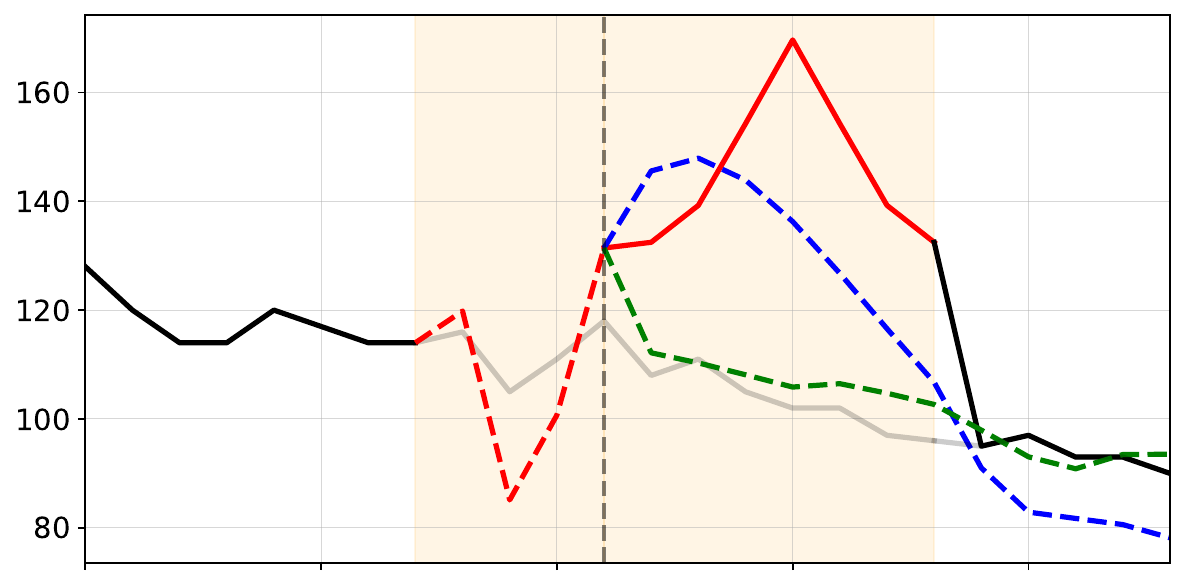}
    \end{subfigure}
    \vspace{0.6em} 
     \begin{subfigure}{0.24\linewidth}
        \includegraphics[width=\linewidth]{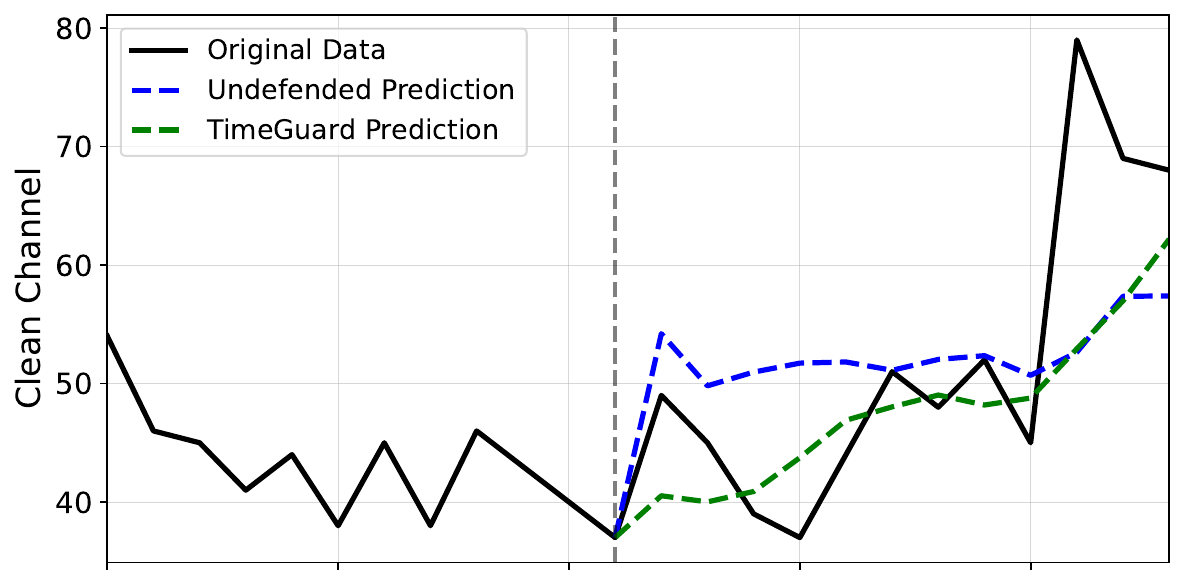}
    \end{subfigure}
    \hfill
    \begin{subfigure}{0.24\linewidth}
        \includegraphics[width=\linewidth]{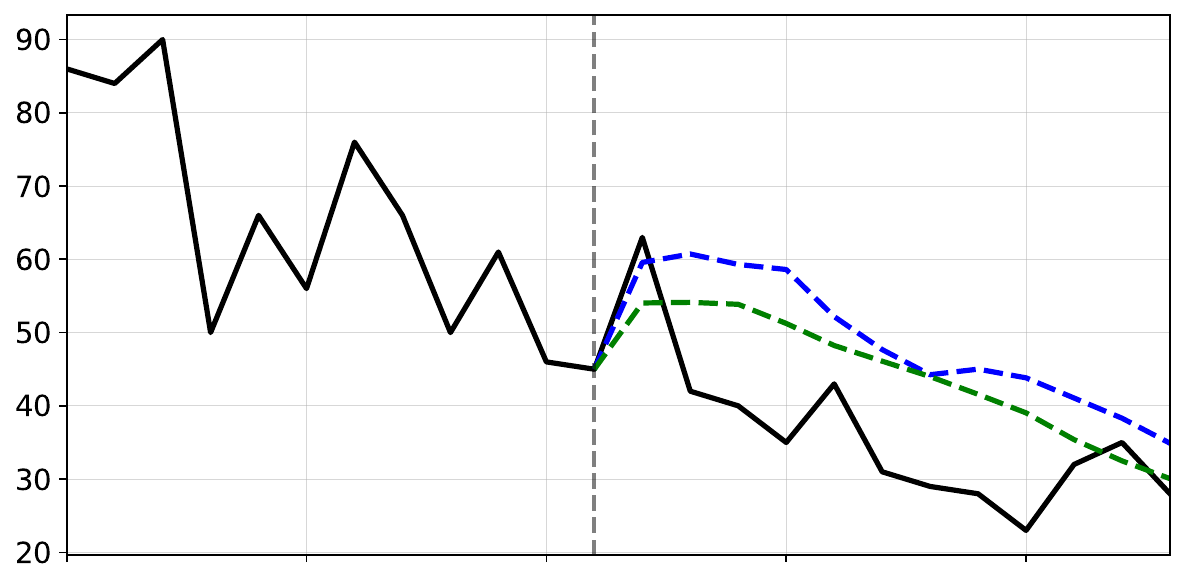}
    \end{subfigure}
    \hfill
     \begin{subfigure}{0.24\linewidth}
        \includegraphics[width=\linewidth]{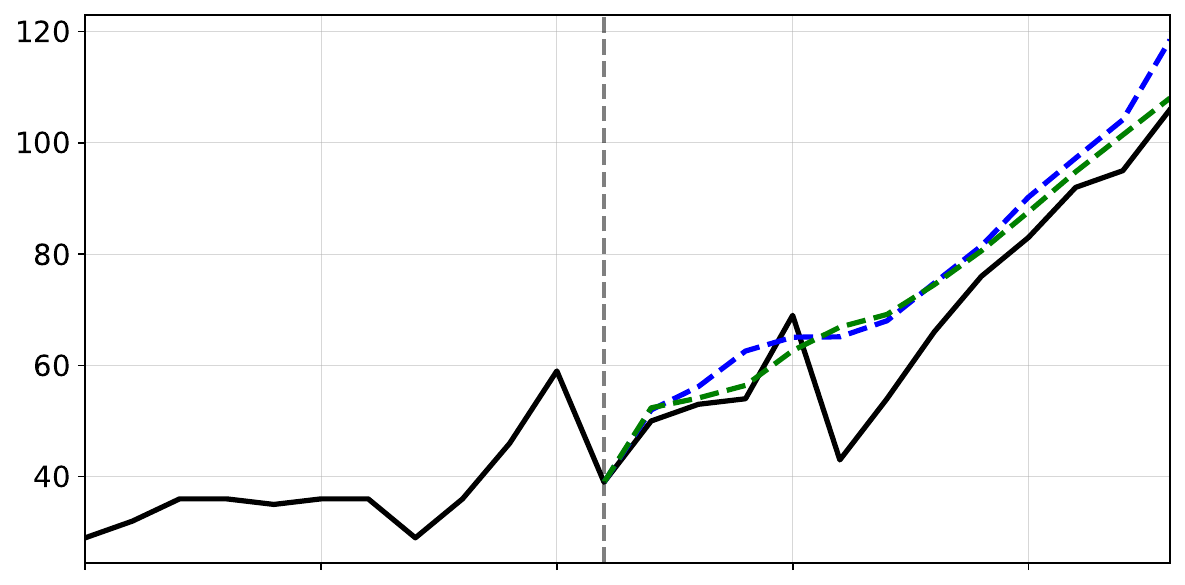}
    \end{subfigure}
    \hfill
     \begin{subfigure}{0.24\linewidth}
        \includegraphics[width=\linewidth]{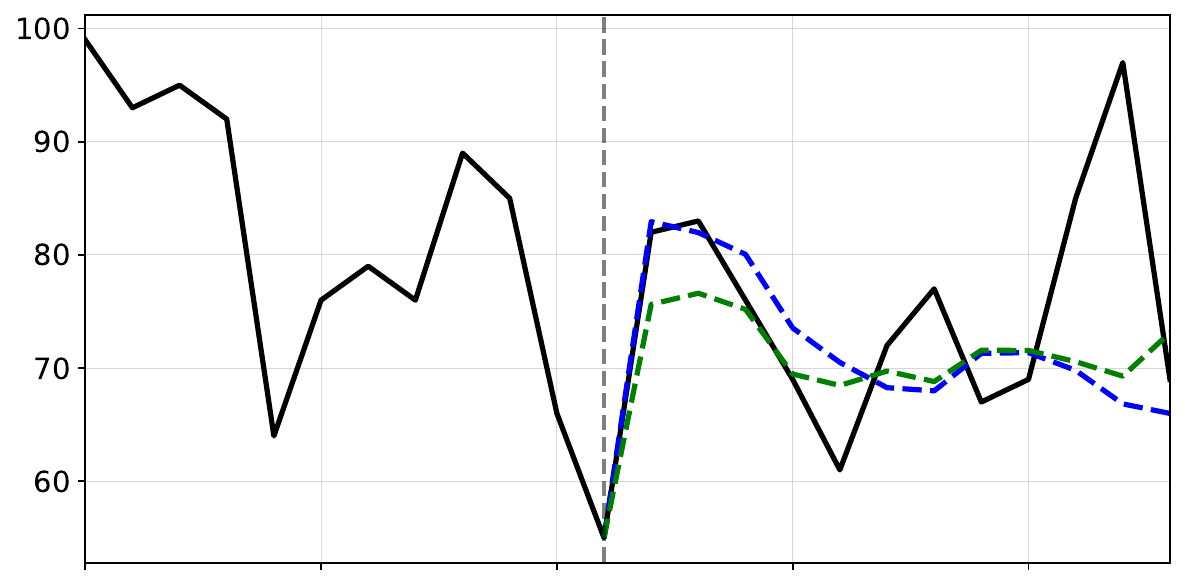}
    \end{subfigure}
    \caption{Inference-time prediction showcases of \methodname{} under the BackTime attack on PEMS03 of FEDformer model, visualized on alternating poisoned and clean channels. We display a randomly selected test sample with a randomly selected channel.}
    \label{app-fig:show-case}
\end{figure}

\section{Limitations and Future Work}
\label{app:limitation-future-work}

\noindent \textbf{Limitations.} First, \methodname{} relies on a hand-designed neighborhood metric, e.g., correlation-/distance-based $k$NN on normalized windows, to construct and refine the reliable pool. Such input-space distances may degrade under strong distribution shifts or nonstationarity. Our preliminary experiments suggest that this degradation is moderate rather than catastrophic; however, existing training-phase defenses also suffer under these challenging settings, as discussed in \appendixautorefname~\ref{app:defense-performance}. Second, TSF backdoor defenses in general face an inherent precision--recall trade-off when the trigger and induced target are not ``out-of-distribution'' relative to clean dynamics. If the trigger and target mimic prevalent motifs (e.g., a near-linear upward trend), poisoned and clean windows can be ambiguous in both learning-based and neighborhood structure: filtering/detecting may remove frequent clean patterns, while retaining them may preserve backdoor influence.

\noindent \textbf{Future work.}  A natural direction is to augment our input-space $k$NN with TSF-specific embedding spaces where neighborhoods better reflect forecasting semantics, e.g., via self-supervised or contrastive representations~\cite{zhang2024self, zheng2025st}. Although we conduct preliminary experiments using TS2Vec embeddings~\cite{yue2022ts2vec}, as discussed in \appendixautorefname~\ref{app:neighborhood-analysis}, the results remain unsatisfactory and do not outperform our original input-space implementation. Future work could explore representations specifically tailored to TSF backdoor defense. Another direction is to utilize (rather than discard) the unreliable pool with semi-supervised learning~\cite{cho2025comres}; however, current TSF semi-supervised methods are often architecture-dependent, motivating deeper study of architecture-agnostic formulations under backdoor settings. Finally, multivariate TSF offers opportunities to leverage cross-channel structure (e.g., dependency graphs or causal signals~\cite{qiu2025duet, han2025root}) to localize corrupted channels while improving clean and recovery forecasting performance.

More broadly, we hope this work encourage TSF-specific backdoor defense research and time series security in general, including standardized benchmarks, stronger adaptive attacks/defenses, and principled evaluation protocols.

\end{document}